\documentclass[12pt]{article}
\usepackage{amssymb}
\usepackage{amsfonts}
\usepackage{amsmath}
\usepackage{eurosym}
\usepackage{makeidx}
\usepackage{amssymb}
\usepackage{times}
\usepackage{anysize}
\usepackage{cite}
\usepackage[polutonikogreek,english]{babel}
\usepackage[utf8x]{inputenx}

\newcommand{\greek}[1]{{\selectlanguage{polutonikogreek}#1}}
\marginsize{2cm}{2cm}{1cm}{1cm}
\newtheorem{theorem}{Theorem}[section]

\newtheorem{condition}[theorem]{Condition}

\newtheorem{corollary}[theorem]{Corollary}

\newtheorem{definition}[theorem]{Definition}

\newtheorem{lemma}[theorem]{Lemma}
\newtheorem{notation}[theorem]{Notation}

\newtheorem{proposition}[theorem]{Proposition}
\newtheorem{remark}[theorem]{Remark}

\newenvironment{proof}[1][Proof]{\noindent\textbf{#1.} }{\ \rule{0.5em}{0.5em}}
\input{tcilatex}
\begin{document}

\title{Classical Dynamics from Self-Consistency Equations in Quantum
Mechanics -- Extended Version}
\author{J.-B. Bru \and W. de Siqueira Pedra}
\date{\today }
\maketitle

\begin{abstract}
\bigskip During the last three decades, P. B\'{o}na has developed a
non-linear generalization of quantum mechanics, which is based on symplectic
structures for normal states. One important application of such a
generalization is a general setting which is very convenient to study the
emergence of macroscopic classical dynamics from microscopic quantum
processes. We propose here a new mathematical approach to Bona's non-linear
quantum mechanics. It is based on $C_{0}$-semigroup theory and has a domain
of applicability which is much brother than Bona's original one. It
highlights the central role of self-consistency. This leads to a
mathematical framework in which the classical and quantum worlds are
naturally entangled. In this new mathematical approach, we build a Poisson
bracket for the polynomial functions on the hermitian weak$^{\ast }$
continuous functionals on any $C^{\ast }$-algebra. This is reminiscent of a
well-known construction for finite-dimensional Lie algebras. We then
restrict this Poisson bracket to states of this $C^{\ast }$-algebra, by
taking quotients with respect to Poisson ideals. This leads to densely
defined symmetric derivations on the commutative $C^{\ast }$-algebras of
real-valued functions on the set of states. Up to a closure, these are
proven to generate $C_{0}$-groups of contractions. As a matter of fact, in
general commutative $C^{\ast }$-algebras, even the closableness of unbounded
symmetric derivations is a non-trivial issue. Some new mathematical concepts
are introduced, which are possibly interesting by themselves: the convex weak%
$^{\ast }$ G\^{a}teaux derivative, state-dependent $C^{\ast }$-dynamical
systems and the weak$^{\ast }$-Hausdorff hypertopology, a new hypertopology
used to prove, among other things, that convex weak$^{\ast }$-compact sets
generically have weak$^{\ast }$-dense extreme boundary in infinite
dimension.\ Our recent results on macroscopic dynamical properties of
lattice-fermion and quantum-spin systems with long-range, or mean-field,
interactions corroborate the relevance of the general approach we present
here.\bigskip

 \bigskip

\end{abstract}

%TCIMACRO{\TeXButton{\tableofcontents }{\tableofcontents}}%
%BeginExpansion
\tableofcontents%
%EndExpansion

\section{Introduction}

An indian, son of a dangerous witch,\textit{... said to his wife:
\textquotedblleft It is my wish that you return with me to my mother's
lodge\ -- my home.\textquotedblright\ His wife, knowing well who he was and
who his mother was, readily consented to accompany him; by so doing she was
faithfully carrying out the policy which her blind brother had advised her
to pursue toward him. On their way homeward, while the husband was leading
the trail, they came to a point where the path divided into two divergent
ways which, however, after forming an oblong loop, reunited, forming once
more only a single path. Here the woman was surprised to see her husband's
body divide into two forms, one following the one path and the other the
other trail. She was indeed greatly puzzled by this phenomenon, for she was
at a loss to know which of the figures to follow as her husband.
Fortunately, she finally resolved to follow the one leading to the right.
After following this path for some distance, the wife saw that the two
trails reunited and also that the two figures of her husband coalesced into
one. It is said that this circumstance gave rise to the name of this strange
man, which was Degiyan\={e}'g\v{e}\~{n}`; that is to say, \textquotedblleft
They are two trails running parallel.\textquotedblright }\smallskip

\hfill Part of a Seneca\footnote{%
The Seneca was an important tribe of the Iroquois, the so-called Five
Nations of New York. There is still a Seneca nation nowadays in the United
States.} legend \cite{Seneca}\textit{\bigskip }

Recently, it was proven \cite{Ammari2018} that the Gross-Pitaevskii and
Hartree hierarchies, which are infinite systems of coupled PDEs
mathematically describing Bose gases with mean-field interactions, are
equivalent to Liouville's equations for functions on a suitable phase space.
This result is reminiscent of Hepp and Lieb's seminal paper \cite%
{Hepp-Lieb73} from year 1973, making explicit, for the first time, the
existence of Poisson brackets in some space of functions, related to the
classical effective dynamics for a permutation-invariant quantum-spin system
with mean-field interactions. This research line was further developed by
many other authors, at least until the nineties. For more details, see \cite[%
Section 1]{BruPedra-MFII}. We focus here on B\'{o}na's impressive series of
papers on the subject, starting in 1975 with \cite{Bona75}. In the middle of
the eighties, his works \cite{Bona83,Bona86} lead him to consider a
non-linear generalization of quantum mechanics. Based on his decisive
progresses \cite{Bona87,Bona88,Bona89,Bona90} on permutation-invariant
quantum-spin systems with mean-field interactions, B\'{o}na presents a
full-fledged abstract theory in 1991 \cite{Bona91}, which is improved later
in a mature textbook published in 2000 (and revised in 2012) \cite{Bono2000}%
. This theory does not seem to be incorporated by the physics and
mathematics communities, yet.

Following \cite[Section 1.1-a]{Bono2000}, B\'{o}na's original motivation was
to \textquotedblleft \textit{understand connections between quantum and
classical mechanics more satisfactorily than via the limit }$\hbar
\rightarrow 0$.\textquotedblright\ This last limit refers to the
semi-classical analysis, a well-developed research field in mathematics. In
physics, it refers to Weyl quantization or, more generally, the quantization
of classical systems with $\hbar $ as a deformation parameter. See, e.g., 
\cite[Chapter 13]{quantum theory}. This is the common understanding\footnote{%
At least in many textbooks on quantum mechanics. See for instance \cite[%
Section 12.4.2, end of the 4th paragraph of page 178]{quantum theory}.} of
the relation between quantum and classical mechanics, which is seen as a
limiting case of quantum mechanics, even if there exist physical features
(such as the spin of quantum particles) which do not have a clear classical
counterpart. Nonetheless, classical mechanics does not only appear in the
limit $\hbar \rightarrow 0$, as explained for instance in \cite%
{landsmann07,extra-ref0000}. B\'{o}na's major conceptual contribution is to
highlight the possible emergence of classical mechanics without the
disappearance of the quantum world, offering a general mathematical
framework which is appropriate to study macroscopic coherence in large
quantum systems.

Note that B\'{o}na's view point is different from recent approaches of
theoretical physics like \cite%
{extra-ref1988,extra-ref,extra-ref2,extra-ref3,extra-ref00,extra-ref1} (see
also references therein) which propose a general formalism to get a
consistent description of interactions between classical and quantum
systems, having in mind chemical reactions, decoherence or the quantum
measurement theory. In these approaches \cite%
{extra-ref1988,extra-ref0000,extra-ref,extra-ref2,extra-ref3,extra-ref00,extra-ref1}%
, neither B\'{o}na's papers nor Hepp and Lieb's results are mentioned, even
if theoretical physicists are of course aware of the emergence of classical
dynamics in presence of mean-field interactions. See, e.g., \cite%
{extra-ref0000} where the mean-field (classical) theory corresponds to the
leading term of a \textquotedblleft large $N$\textquotedblright\ expansion
while the quantum part of the theory (quantum fluctuations) is related to
the next-to-leading order term. The approaches \cite%
{extra-ref1988,extra-ref,extra-ref2,extra-ref3,extra-ref00,extra-ref1} (see
also references therein) refer to quantum-classical hybrid theories for
which the classical space exists by definition, in a \emph{ad hoc }way,
because of measuring instruments for instance. By contrast, the classical
dynamics in B\'{o}na's view point emerges as an \emph{intrinsic} property of
macroscopic quantum systems, like in \cite{extra-ref0bis}. This is also
similar to \cite{extra-ref0}, which is however a much more elementary example%
\footnote{%
It corresponds to a quantum systems with two species of particles in an
extreme mass ratio limit: the particles of one species become infinitely
more massive than the particles in the other one. In this limit, the species
of massive particles, like nuclei, becomes classical while the other one,
like electrons, stays quantum mechanical.} referring to the Ehrenfest
dynamics.

In the present paper we revisit B\'{o}na's conceptual lines, but propose a
new method to mathematically implement them, with a broader domain of
applicability than B\'{o}na's original version \cite{Bono2000} (see also 
\cite{landsmann07,Odzijewicz} and references therein). In contrast with all
previous approaches, including those of theoretical physics (see, e.g., \cite%
{extra-ref1988,extra-ref,extra-ref2,extra-ref3,extra-ref00,extra-ref1,extra-ref0000,extra-ref0bis}%
), in ours the classical and quantum worlds are \emph{entangled}, with \emph{%
backreaction}\footnote{%
We do not mean here the so-called \emph{quantum backreaction}, commonly used
in physics, which refers to the backreaction effect of quantum fluctuations
on the classical degrees of freedom. Note further that the phase spaces we
consider are, generally, much more complex than those related to the
position and momentum of simple classical particles.
\par
{}} (that is, feedbacks), as expected. Differently from B\'{o}na's technical
setting, ours has advantage of highlighting inherent \emph{self-consistency}
aspects, which are absolutely not exploited in \cite{Bono2000}, as well as
in quantum-classical hybrid theories of physics (e.g., \cite%
{extra-ref1988,extra-ref,extra-ref2,extra-ref3,extra-ref00,extra-ref1,extra-ref0,extra-ref0bis}%
).

The relevance of the abstract setting we propose here is corroborated by
recent results \cite{BruPedra-MFII,BruPedra-MFIII} on the \emph{macroscopic}
(i.e., infinite-volume) dynamical properties of lattice-fermion and
quantum-spin systems with long-range, or mean-field, interactions. Note that
a simple illustration of them is available in \cite%
{Bru-pedra-proceeding,Bru-pedra-MF-IV}. In fact, the outcomes of \cite%
{BruPedra-MFII,BruPedra-MFIII} refer to objects that are far more general
than quasi-free states or permutation invariant models and required the
development of an appropriate mathematical framework to accommodate the
macroscopic long-range dynamics, which turns out to be generally equivalent
to an intricate combination of classical and short-range quantum dynamics. 
\cite{BruPedra-MFII,BruPedra-MFIII} are therefore a strong motivation for a
change of perspective, which is thus presented in a \emph{systematic} way in
the present paper. Several key ingredients of \cite%
{BruPedra-MFII,BruPedra-MFIII} refer to abstract constructions discussed
this paper, like the Poisson structures elaborated here. In other words, 
\cite{BruPedra-MFII,BruPedra-MFIII} represent important applications, to the
quantum many-body problem, of the general setting presented here.

To set up our approach, we use the algebraic formalism for quantum and
classical mechanics \cite[Chapter 12]{quantum theory}. The most basic
element of our mathematical framework is a generic non-commutative unital $%
C^{\ast }$-algebra $\mathcal{X}$, which will be called here the
\textquotedblleft primordial\textquotedblright\ algebra. For instance, $%
\mathcal{X}$ is the so-called CAR $C^{\ast }$-algebra for fermion systems or
the spin $C^{\ast }$-algebras in the case of quantum spins. Then, the
classical objects associated with $\mathcal{X}$ are defined as follows:

\begin{itemize}
\item \emph{State and phase spaces} (Sections \ref{Phase Space}-\ref{generic
convex set}). The state space is the convex weak$^{\ast }$-compact set $E$
of \emph{all} states on $\mathcal{X}$. We define the phase space as being
the (weak$^{\ast }$) closure\footnote{%
More properly, the phase space should be taken as being the set $\mathcal{E}%
(E)$ of extreme states on $\mathcal{X}$ itself. Note, however, that what is
relevant in the algebraic approach is the algebra of continuous functions on
the given topological space, and not the space itself. The algebras of
continuous functions on the closure of $\mathcal{E}(E)$ is, of course, $\ast 
$-isomorphic to a $C^{\ast }$-subalgebra of continuous functions on $%
\mathcal{E}(E)$ and the closure of $\mathcal{E}(E)$ is taken to get a
compact phase space, only.}\ of the subset $\mathcal{E}(E)\subseteq E$ of
all extreme points of $E$. Interestingly, in the case that the $C^{\ast }$%
-algebra $\mathcal{X}$ is antiliminal and simple (e.g., the CAR algebra
associated with any separable infinite-dimensional one-particle Hilbert
space, the spin algebra of any infinite countable lattice, etc.), the phase
and state spaces coincide. More generally, by using a new (weak$^{\ast }$)
hypertopology, we show that this surprising property of the state space is 
\emph{not} accidental, but \emph{generic} in infinite-dimensional separable
Banach spaces. Note that our definitions of phase and state spaces differ
from B\'{o}na's ones: he does not really distinguish both spaces and
considers instead the set of all density matrices associated with a \emph{%
fixed} Hilbert space \cite[Section 2.1, see also 2.1-c]{Bono2000}. In
particular, B\'{o}na's definition of the phase/state space is
representation-dependent, in contrast with our approach. In fact, in \cite[%
Sections 2.1c, footnote]{Bono2000}, B\'{o}na proposes as a mathematically
and physically interesting problem to \textquotedblleft \textit{formulate
analogies of }[his]\textit{\ constructions on the space of all positive
normalized functionals on }$\mathcal{B}(\mathcal{H})$. \textit{This leads to
technical complications}.\textquotedblright\ In Section \ref{new section
jundiai copy(1)} we propose a solution to this problem for any $C^{\ast }$%
-algebra $\mathcal{X}$.

\item \emph{Classical algebra }(Section \ref{Classical algebra}). The
classical (i.e., commutative) unital $C^{\ast }$-algebra\footnote{%
Analogously to the above distinction between phase and state spaces, more
properly, the algebra related to the \textquotedblleft classical
world\textquotedblright\ should rather be the one of continuous functions on
the phase space, but we expose in Section \ref{Classical algebra copy(1)}
the conceptual limitations of the use of this algebra in quantum physics.
Moreover, in the case the $C^{\ast }$-algebra $\mathcal{X}$ is antiliminal
and simple, both classical algebras are $\ast $-isomorphic to each other. In
fact, the phase space turns out to be always conserved by the classical
flows (in the state space) and we show that the classical dynamics studied
in the present paper can always be pushed forward, by restriction of
functions, from $\mathfrak{C}$ to the algebra of weak$^{\ast }$ continuous
functions on the phase space.} in our approach is the algebra $\mathfrak{C}%
\doteq C\left( E;\mathbb{C}\right) $ of continuous and complex-valued
functions on the state space $E$.

\item \emph{Poisson structures }(Sections \ref{Convex Frechet Derivative}-%
\ref{Poisson Structure}). By generalizing the well-known construction of a
Poisson bracket for the polynomial functions on the dual space of finite
dimensional Lie algebras \cite[Section 7.1]{Poission}, we define a Poisson
bracket for the polynomial functions on the hermitian continuous functionals
(like the states) on any $C^{\ast }$-algebra $\mathcal{X}$. Then, the
Poisson bracket is localized on the state or phase space associated with
this algebra by taking quotients with respect to conveniently chosen Poisson
ideals. This leads to a Poisson bracket for polynomial functions of the
classical $C^{\ast }$-algebra $\mathfrak{C}$.
\end{itemize}

\noindent In our setting, we introduce state-dependent $C^{\ast }$-dynamical
systems associated with the primordial algebra $\mathcal{X}$, as follows:

\begin{itemize}
\item \emph{Secondary quantum algebra} (Section \ref{Quantum Algebras}).
Similar to quantum-classical hybrid theories of theoretical physics like in 
\cite{extra-ref1988,extra-ref,extra-ref2,extra-ref3,extra-ref00,extra-ref1}
we introduce an extended quantum algebra as being the\footnote{%
Because commutative $C^{\ast }$-algebras are nuclear, the norm making the
completion of the algebraic tensor product $\mathfrak{C}\otimes \mathcal{X}$
into a $C^{\ast }$-algebra is unique.} tensor product $\mathfrak{C}\otimes 
\mathcal{X}$ of the commutative $C^{\ast }$-algebra $\mathfrak{C}$ with the
primordial one $\mathcal{X}$. This tensor product is nothing else (up to
some $\ast $-isomorphism) than the unital $C^{\ast }$-algebra $\mathfrak{X}%
\doteq C(E;\mathcal{X})$, named here the \textit{secondary} algebra
associated with the \textit{primordial} one, $\mathcal{X}$. There are
natural inclusions $\mathcal{X}\subseteq \mathfrak{X}\ $and$\ \mathfrak{C}%
\subseteq \mathfrak{X}$ by identifying elements of $\mathcal{X}$ with
constant functions and elements of $\mathfrak{C}$ with functions whose
values are scalar multiples of the unit of the primordial algebra $\mathcal{X%
}$. Note that in B\'{o}na's approach, self-adjoint elements of $\mathfrak{X}$
refer to what he calls \textquotedblleft \textit{non-linear
observables\textquotedblright } \cite[Section 1.2.3]{Bono2000}.

\item \emph{State-dependent quantum dynamics} (Section \ref{Quantum Algebras}%
). As in $\mathcal{X}$, a (possibly non-autonomous) quantum dynamics on $%
\mathfrak{X}$ is, by definition, a strongly continuous two-parameter family $%
\mathfrak{T}\equiv (\mathfrak{T}_{t,s})_{s,t\in \mathbb{R}}$ of $\ast $%
-automorphisms of $\mathfrak{X}$\ satisfying the reverse cocycle property:%
\begin{equation*}
\forall s,r,t\in \mathbb{R}:\qquad \mathfrak{T}_{t,s}=\mathfrak{T}%
_{r,s}\circ \mathfrak{T}_{t,r}\ .
\end{equation*}%
If $\mathfrak{T}$ preserves the classical algebra $\mathfrak{C}\subseteq 
\mathfrak{X}$, then we name the pair $(\mathfrak{X},\mathfrak{T})$
state-dependent, or secondary, $C^{\ast }$-dynamical system associated with
the primordial algebra $\mathcal{X}$.
\end{itemize}

\noindent In this setting, the classical (i.e., commutative) and quantum
(i.e., non-commutative) objects are strongly related to each other as
follows:

\begin{itemize}
\item Any state-dependent $C^{\ast }$-dynamical system $(\mathfrak{X},%
\mathfrak{T})$ associated with $\mathcal{X}$, in the above sense, yields a
classical dynamics on $\mathfrak{C}$, as explained in Section \ref{Extended
Quantum Dynamics}. This classical dynamics then induces a \emph{Feller
evolution system} \cite{Feller}, which in turn implies the existence of
corresponding Markov transition kernels on $E$ (which can be canonically
identified with the Gelfand spectrum of the commutative unital $C^{\ast }$%
-algebra $\mathfrak{C}$). The full dynamics for (quantum) states on the
primordial algebra $\mathcal{X}$\ can then be recovered from the Markov
transition kernels. A Feller evolution with similar properties also exists
for the phase space (i.e., the closure of $\mathcal{E}(E)$).

\item More interestingly, we remark in Section \ref{Section QM} that any 
\textit{classical} differentiable Hamiltonian from $\mathfrak{C}$ is
associated with a \emph{state-dependent} quantum dynamics on the primordial $%
C^{\ast }$-algebra $\mathcal{X}$, in a natural way. This observation is then
used to derive, mathematically, in Sections \ref{Self-Consistency Equations}-%
\ref{Section CM}, classical dynamics associated with the Poisson structure
of the (polynomial subalgebra of the) classical algebra $\mathfrak{C}$.
These define again Feller evolution systems which turn out to be related to
a self-consistency problem (Theorem \ref{theorem sdfkjsdklfjsdklfj copy(3)}%
). By Lemma \ref{lemma supercigare bis copy(1)}, this yields, in turn,
state-dependent quantum dynamics on the secondary (quantum) $C^{\ast }$%
-algebra $\mathfrak{X}$ of continuous ($\mathcal{X}$-valued) functions on
states, associated with the primordial (quantum) algebra $\mathcal{X}$.
\end{itemize}

On the one hand, the classical world is embedded in the quantum world, as
mathematically expressed by the above defined inclusion $\mathfrak{C}%
\subseteq \mathfrak{X}$. On the other hand, our approach entangles the
quantum and classical worlds through self-consistency. As a result, \emph{%
non-autonomous} and \emph{non-linear} dynamics can emerge. Seeing both
entangled worlds, quantum and classical, as \textquotedblleft two sides of
the same coin\textquotedblright\ looks like an oxymoron, but there is no
contradiction there since everything refers to a \emph{primordial} quantum
world mathematically encoded in the structure of the non-commutative
(unital) $C^{\ast }$-algebra $\mathcal{X}$. In fact, the quantum algebra $%
\mathcal{X}$ is the \emph{arche}\textit{\footnote{%
Following Aristotle's use of the presocratic philosophical term
\textquotedblleft arche\textquotedblright\ (%
%TCIMACRO{\TeXButton{\greek}{\greek{>arq\`h}}}%
%BeginExpansion
\greek{>arq\`h}%
%EndExpansion
), here it means \textquotedblleft the element or principle of a thing
which, although undemonstrable and intangible in itself, provides the
conditions of the possibility of that thing\textquotedblright .\ See \cite[%
p. 143]{arche}.}} of the theory. For instance, the state space $E$ is the
imprint left by $\mathcal{X}$ in the classical world, whose observables are
the self-adjoint elements of the \textit{commutative} $C^{\ast }$-algebra $%
\mathfrak{C}\doteq C\left( E;\mathbb{C}\right) $, i.e., the continuous
complex-valued functions on $E$. If $\mathcal{X}$ were a commutative
algebra, note that the corresponding Poisson bracket and, hence, the
associated classical dynamics would be trivial.

Note that the abstract setting proposed in this paper is not really useful
to portray quantum dynamics of finite systems. In fact, in this case, the
time evolution is \emph{not} state-dependent. Nevertheless, as discussed
above, such a mathematical framework turn out to be natural for the study of
macroscopic dynamics of lattice-fermion or quantum-spin systems with
long-range, or mean-field, interactions, because, in this case, the
Heisenberg dynamics turns out to be effectively state dependent, in the
thermodynamic limit. See again \cite{BruPedra-MFII,BruPedra-MFIII}, which
uses self-consistency equations in a essential way, similar to Theorem \ref%
{theorem sdfkjsdklfjsdklfj copy(3)}. Moreover, since quantum many-body
systems in the continuum are also expected to have, in general, a \emph{%
state-dependent} Heisenberg dynamics in the thermodynamic limit (see, e.g., 
\cite[Section 6.3]{BratteliRobinson}), the approach presented here is very
likely relevant for future studies in this context. We thus consider
important to have a systematic approach that can be used beyond specific
applications, like \cite{BruPedra-MFII,BruPedra-MFIII}.

Our approach is not too far, in its spirit, to the one developed in \cite%
{Bono2000}, although it differs in its mathematical formulation. In
comparison with \cite{Bono2000}, our formulation is more general in the case
of an infinite-dimensional underlying $C^{\ast }$-algebra, which generally
has several inequivalent irreducible representations, as a consequence of
the Rosenberg theorem \cite{Ros}: Whereas \cite{Bono2000} has to use a
representation of the underlying $C^{\ast }$-algebra\ to be able to define
Poisson brackets in the associated classical algebra, we provide a
definition for such brackets with no reference to representations. This is
explained in more detail in Section \ref{new section jundiai}. Notice at
this point that, in condensed matter physics, the non-uniqueness of
irreducible representations is intimately related to the existence of
various inequivalent thermodynamically stable phases of the same material.

Last but not least, we observe that a large set of symmetric derivations can
be defined on all polynomial elements of $\mathfrak{C}$ by using the Poisson
bracket. See Section \ref{Bounded Derivations}. These (unbounded)
derivations are not a priori \emph{closed} operators, but this property is
necessary to generate (classical) dynamics, in its Hamiltonian formulation,
via \emph{strongly continuous semigroups}. In contrast with our approach, B%
\'{o}na avoids this problem by using Hamiltonian flows in symplectic leaves
of the corresponding Poisson manifold and by \textquotedblleft
gluing\textquotedblright\ together the flows within the leaves by showing
continuity properties \cite[Section 2.1-d]{Bono2000}.

The closabledness of a symmetric derivation is usually proven from its
dissipativity \cite[Definition 1.4.6, Proposition 1.4.7]{Bratelli-derivation}%
, which results from \cite[Theorem 1.4.9]{Bratelli-derivation} and the
assumption that the square root of each positive element of the domain of
the derivation also belongs to the same domain. We cannot expect this
property to be satisfied for symmetric derivations acting on a dense domain
of $\mathfrak{C}$. As a matter of fact, the closableness of unbounded
symmetric derivations in commutative $C^{\ast }$-algebras like $\mathfrak{C}$
is, in general, a non-trivial issue. This property might not even be true
since there exists norm-densely defined derivations of $C^{\ast }$-algebras
that are \emph{not} closable \cite{bra-robin-1975}. For instance, in \cite[%
p. 306]{BrattelliRobinsonI}, it is even claimed that \textquotedblleft
Herman has constructed an extension of the usual differentiation on $C(0,1)$
which is a non-closable derivation of $C(0,1)$.\textquotedblright\ A
complete classification of all closed symmetric derivations of functions on
a compact subset of a \emph{one-dimensional} space was obtained around 1990.
However, quoting \cite[Section 1.6.4, p. 27]{Bratelli-derivation},
\textquotedblleft for more than 2 dimensions only \textit{sporadic} results
in this direction are known.\textquotedblright\ See, e.g., \cite[Section
1.6.4]{Bratelli-derivation}, \cite{Tomiyama}, \cite{Kuroselastpaper,Kurose},
and later \cite[p. 306]{BrattelliRobinsonI}. Since then, no progress has
been made on this classification problem, at least to our knowledge.

In Section \ref{Closed derivation} (Theorem \ref{Closed Poissonian symmetric
derivations}), via the analysis of certain self-consistency problems
together with the one-parameter semigroups theory \cite{EngelNagel}, we
naturally obtain infinitely many closed symmetric derivations with dense
domain in $\mathfrak{C}$. As it turns out, this method is very natural and
efficient for the state space $E$, that is, a weak$^{\ast }$-compact convex
subset of the dual $\mathcal{X}^{\ast }$ of the unital (not necessarily
separable) $C^{\ast }$-algebra $\mathcal{X}$, which is, in general, infinite%
\emph{-}dimensional. In particular, $E$ is generally \emph{not} a subset of
a finite-dimensional space. This construction of closed derivations of a
commutative $C^{\ast }$-algebra via self-consistency problems is
non-conventional and may motivate further studies. For more information, see
Section \ref{Closed derivation}. \bigskip

\noindent \textbf{Main results and structure of the paper.} Recall that $E$
is the state space of a non-commutative unital $C^{\ast }$-algebra $\mathcal{%
X}$. Our main results are the following:

\begin{itemize}
\item The \emph{weak}$^{\ast }$-\emph{Hausdorff hypertopology} (Definitions %
\ref{hypertopology} and \ref{hypertopology0}) is a new notion, proposed here
in order to characterize generic convex weak$^{\ast }$-compact sets, by
extending \cite{Klee,FonfLindenstrauss} to weak$^{\ast }$ topological
structures. We show in particular that convex weak$^{\ast }$-compact sets of
the dual space $\mathcal{X}^{\ast }$ of a (real or complex) Banach space $%
\mathcal{X}$\ have generically weak$^{\ast }$-dense set of extreme points in
infinite dimension, in the sense of this new (hyper)topology. This refers to
Theorems \ref{theorem dense cool1} and \ref{theorem density2}. These results
has been extended in \cite{article-hypertopology} for the dual space $%
\mathcal{X}^{\ast }$, endowed with its weak$^{\ast }$-topology, of any
infinite-dimensional, separable topological vector space $\mathcal{X}$.

\item Corollary \ref{proposition sympatoch copy(1)} defines, in a natural
way, a Poisson bracket $\{\cdot ,\cdot \}$ on polynomial functions of $C(E;%
\mathbb{C})$, while Corollary \ref{proposition sympatoch copy(2)} shows that
the restriction of $\{\cdot ,\cdot \}$ to the phase space $\overline{%
\mathcal{E}(E)}$ also lead to a Poisson bracket on polynomial functions of $%
C(\overline{\mathcal{E}(E)};\mathbb{C})$. These Poisson brackets were
previously used, for instance, in \cite{BruPedra-MFII,BruPedra-MFIII}.

\item The \emph{convex weak}$^{\ast }$ \emph{Gateaux derivative} (Definition %
\ref{convex Frechet derivative}) is used to give an explicit expression for
the Poisson bracket for functions on the state space $E$. This refers to
Proposition \ref{Poisson algebra prop}, which is is an important result
because it allows us to perform more explicit computations, both in this
paper and in \cite{BruPedra-MFII,BruPedra-MFIII}.

\item Theorem \ref{theorem sdfkjsdklfjsdklfj copy(3)} shows the
well-posedness of self-consistency equations, allowing us to define, for an
appropriate continuous family $h\equiv (h(t))_{t\in \mathbb{R}}\subseteq
C^{1}\left( E;\mathbb{R}\right) $, a classical flow%
\begin{equation*}
\rho \mapsto \mathbf{\varpi }^{h}\left( s,t\right) \left( \rho \right)
,\qquad s,t\in \mathbb{R},
\end{equation*}%
in the state space $E$ and thus, a (generally non-autonomous classical)
dynamics $(V_{t,s}^{h})_{s,t\in \mathbb{R}}$ on $\mathfrak{C}\doteq C\left(
E;\mathbb{C}\right) $: 
\begin{equation*}
V_{t,s}^{h}\left( f\right) \doteq f\circ \mathbf{\varpi }^{h}\left(
s,t\right) \ ,\qquad f\in \mathfrak{C},\ s,t\in \mathbb{R}\ .
\end{equation*}%
Physically, the functions $h(t)$, $t\in \mathbb{R}$, are time-dependent
classical energies. In Corollary \ref{corollary conservation}, we show that
the classical flow conserves both the set $\mathcal{E}(E)$ of extreme states
and its weak$^{\ast }$ closure $\overline{\mathcal{E}(E)}$, which is the
phase space.

\item Proposition \ref{lemma poisson copy(1)} proves that $%
(V_{t,s}^{h})_{s,t\in \mathbb{R}}$ is a strongly continuous two-parameter
family of $\ast $-automorphisms of $\mathfrak{C}$ satisfying the reverse
cocycle property, i.e., the classical dynamics is a Feller evolution system 
\cite{Feller}.

\item Theorem \ref{Closed Poissonian symmetric derivations} shows that,
given an appropriate function $h\in C^{1}\left( E;\mathbb{R}\right) $, the
Poissonian symmetric derivations$\ $%
\begin{equation*}
f\mapsto \left\{ h,f\right\} \in \mathfrak{C}
\end{equation*}%
defined for any polynomial functions $f$ in $\mathfrak{C}$ is closable and
is directly related to the generator of the $C_{0}$-group $%
(V_{t,0}^{h})_{t\in \mathbb{R}}$ for the constant energy function $h$.

\item Theorem \ref{proposition dynamic classique II} shows the
non-autonomous evolution equations%
\begin{equation*}
\partial _{t}V_{t,s}^{h}\left( f\right) =V_{t,s}^{h}\left( \left\{ h\left(
t\right) ,f\right\} \right) \qquad \text{and}\qquad \partial
_{s}V_{t,s}^{h}\left( f\right) =-\left\{ h\left( s\right)
,V_{t,s}^{h}(f)\right\}
\end{equation*}%
for any appropriate $h\equiv (h(t))_{t\in \mathbb{R}}\subseteq C^{1}\left( E;%
\mathbb{R}\right) $, times $s,t\in \mathbb{R}$ and polynomial function $f$
of $\mathfrak{C}$. In the autonomous case, i.e., when $h\in C^{1}\left( E;%
\mathbb{R}\right) $, one gets Liouville's equation (Corollary \ref%
{proposition dynamic classique IIbis}), i.e., 
\begin{equation*}
\partial _{t}V_{t,0}^{h}\left( f\right) =V_{t,0}^{h}\left( \left\{
h,f\right\} \right) =\left\{ h,V_{t,0}^{h}(f)\right\} \ .
\end{equation*}

\item In Section \ref{Extended Quantum Dynamics}, we show how the above
classical dynamics defines a state-dependent quantum dynamics with \emph{%
fixed-point} algebra including\footnote{%
This is a very important property, excluding the definition given by
Equation (\ref{def fausse}).} the classical algebra $\mathfrak{C}$. This
lead us to define a \emph{state-dependent} $C^{\ast }$-dynamical system
(Definition \ref{Extended dynamical systems}). See Lemma \ref{lemma
supercigare bis copy(1)} and discussions afterwards. Such a quantum dynamics
is relevant in the study of macroscopic dynamics of lattice-fermion systems
or quantum-spin systems with long-range, or mean-field, interactions
performed in \cite{BruPedra-MFII,BruPedra-MFIII}.
\end{itemize}

\noindent The paper is organized as follows: We first introduce, in Section %
\ref{Phase Space copy(2)}, classical systems associated with arbitrary
unital $C^{\ast }$-algebras. The Poisson structures for these systems are
built in Section \ref{Poisson Structures in Quantum Mechanics}. Section \ref%
{Poisson Structures in Quantum Mechanics} also gathers all the necessary
definitions to describe, in Section \ref{Closed derivation}, classical
dynamics generated by a Poisson bracket, as is usual classical mechanics.
Section \ref{section extended C* dynamical systeme} then explains the final
setting of the theory. In Section \ref{Symmetry Group} we discuss the role
of symmetries as well as the notion of \textquotedblleft
reduction\textquotedblright\ of the classical dynamics. This is important in
applications to simplify the self-consistency equations. Section \ref%
{Hausdorff Hypertopology} gives all arguments to deduce Theorems \ref%
{theorem dense cool1}-\ref{theorem density2} by defining and studying the
weak$^{\ast }$-Hausdorff hypertopology. The proof of the most important
result, that is, Theorem \ref{theorem sdfkjsdklfjsdklfj copy(3)}, is
performed in Section \ref{Well-posedness sect copy(1)}, which also collects
additional results used in Section \ref{Section CM}. Finally, Section \ref%
{Liminal appendix} is an appendix on liminal, postliminal and antiliminal $%
C^{\ast }$-algebras. Though these are standard notions in $C^{\ast }$%
-algebra theory, they may not be known by non-experts, but have major
consequences on the structure of the set of states, which can be highly
non-trivial and are relevant in our discussions.

\begin{notation}
\label{remark constant}\mbox{
}\newline
\emph{(i)} A norm on a generic vector space $\mathcal{X}$ is denoted by $%
\Vert \cdot \Vert _{\mathcal{X}}$ and the identity map of $\mathcal{X}$ by $%
\mathbf{1}_{\mathcal{X}}$. The space of all bounded linear operators on $(%
\mathcal{X},\Vert \cdot \Vert _{\mathcal{X}}\mathcal{)}$ is denoted by $%
\mathcal{B}(\mathcal{X})$. The unit element of any algebra $\mathcal{X}$ is
denoted by $\mathfrak{1}$, provided it exists. The scalar product of any
Hilbert space $\mathcal{H}$ is denoted by $\langle \cdot ,\cdot \rangle _{%
\mathcal{H}}$. \newline
\emph{(ii)} For all topological spaces $\mathcal{X}$ and $\mathcal{Y}$, $%
C\left( \mathcal{X};\mathcal{Y}\right) $ denotes, as usual, the space of
continuous maps from $\mathcal{X}$ to $\mathcal{Y}$. If $\mathcal{X}$ is a
locally compact topological space and $\mathcal{Y}$ is a Banach space, then $%
C_{b}\left( \mathcal{X};\mathcal{Y}\right) $ denotes the Banach space of
bounded continuous maps from $\mathcal{X}$ to $\mathcal{Y}$ along with the
topology of uniform convergence. For any $p,n\in \mathbb{N}$, in the special
case $\mathcal{X}=\mathbb{R}^{n}$ and $\mathcal{Y}=\mathbb{R}$, $C_{b}^{p}(%
\mathbb{R}^{n};\mathbb{R})$ denotes the Banach space of bounded continuous,
real-valued, functions on $\mathbb{R}^{n}$ along with the topology of
uniform convergence for the functions and all its $m$-th derivatives, where $%
m\in \{1,\ldots ,p\}$. \newline
\emph{(iii)} We adopt the term \textquotedblleft
automorphism\textquotedblright\ in the sense of category theory and its
precise meaning thus depends on the structure of the corresponding domain:
An automorphism of a $\ast $-algebra is a bijective $\ast $-homomorphism
from this algebra to itself, whereas an automorphism of a topological space
is a self-homeomorphism, that is, a homeomorphism of the space to itself. In
fact, in the category of topological spaces the morphisms are precisely the
continuous maps and the morphisms of the category of $\ast $-algebras are
the $\ast $-homomorphisms. Recall that in category theory invertible
morphisms are called isomorphisms and isomorphisms whose domain and codomain
coincide are called automorphisms.\newline
\emph{(iv)} In the sequel, a primordial $C^{\ast }$-algebra $\mathcal{X}$ is
fixed and its state space is denoted by $E$. Then, various sets of functions
on $E$ are defined. The most important are $\mathfrak{C}\doteq C(E;\mathbb{C}%
)$, $\mathfrak{X}\doteq C(E;\mathcal{X})$ and $\mathfrak{Y}\equiv \mathfrak{Y%
}\left( \mathcal{Y}\right) \doteq C^{1}\left( E;\mathcal{Y}\right) $, $%
\mathcal{Y}$ being a Banach space, like $\mathbb{R}$ or $\mathbb{C}$. These
spaces appear many times and we always use a shorter notation than the usual
ones, like $C(E;\mathbb{C})$, the letter of the codomain within the Fraktur
alphabet. More generally, any capital letter in Fraktur alphabet always
refers to a space of functions on $E$. To denote real subspaces, we add the
superscript $\mathbb{R}$ like in $\mathfrak{C}^{\mathbb{R}}\doteq C(E;%
\mathbb{R})$ or in $\mathcal{X}^{\mathbb{R}}$, which is the real Banach
space of all self-adjoint elements of $\mathcal{X}$.
\end{notation}

\section{Classical View on Quantum Systems\label{Phase Space copy(2)}}

\subsection{State Space of $C^{\ast }$-Algebras\label{Phase Space}}

\noindent \textit{Perhaps the philosophically most relevant feature of
modern science is the emergence of abstract symbolic structures as the hard
core of objectivity behind -- as Eddington puts it -- the colorful tale of
the subjective storyteller mind.}\smallskip

\hfill Weyl, 1949 \cite[Appendix B, p. 237]{weyl}\bigskip

Fix once and for all a $C^{\ast }$-algebra 
\begin{equation*}
\mathcal{X}\equiv \left( \mathcal{X},+,\cdot _{{\mathbb{C}}},\times ,^{\ast
},\left\Vert \cdot \right\Vert _{\mathcal{X}}\right) \ ,
\end{equation*}%
that is, a (complex) Banach algebra endowed with an antilinear involution $%
A\mapsto A^{\ast }$ such that%
\begin{equation*}
(AB)^{\ast }=B^{\ast }A^{\ast }\text{\qquad and\qquad }\left\Vert A^{\ast
}A\right\Vert _{\mathcal{X}}=\left\Vert A\right\Vert _{\mathcal{X}}^{2},%
\text{\qquad }A,B\in \mathcal{X}\ .
\end{equation*}%
Here, $AB\equiv A\times B$. We always assume that $\mathcal{X}$ is unital,
i.e., the product of $\mathcal{X}$ has a unit $\mathfrak{1}\in \mathcal{X}$.
The (real) Banach subspace of all self-adjoint elements of $\mathcal{X}$ is
denoted by 
\begin{equation}
\mathcal{X}^{\mathbb{R}}\doteq \left\{ A\in \mathcal{X}:A=A^{\ast }\right\}
\equiv \left( \mathcal{X}^{\mathbb{R}},+,\cdot _{{\mathbb{R}}},\left\Vert
\cdot \right\Vert _{\mathcal{X}}\right) \ .  \label{self-adjoint elements}
\end{equation}%
The $C^{\ast }$-algebra $\mathcal{X}$ is named the \emph{primordial} $%
C^{\ast }$-algebra. Note that it is not necessarily separable.

By \cite[Theorem 3.10]{Rudin}, the dual space $\mathcal{X}^{\ast }$ of $%
\mathcal{X}$ endowed with its weak$^{\ast }$ topology (i.e., the $\sigma (%
\mathcal{X}^{\ast },\mathcal{X})$-topology of $\mathcal{X}^{\ast }$) is a
locally convex space (in the sense of \cite[Section 1.6]{Rudin}) whose dual
is $\mathcal{X}$. Recall that $\mathcal{X}^{\ast }$ is a Banach space when
it is endowed with the usual norm for linear functionals on a normed space.
A subset of $\mathcal{X}^{\ast }$ which is pivotal in the algebraic
formulation of quantum mechanics is the \emph{state space }of $\mathcal{X}$,
defined as follows: 

\begin{definition}[State space]
\label{state space}\mbox{ }\newline
Let $\mathcal{X}$ be a unital $C^{\ast }$-algebra. The state space is the
convex and weak$^{\ast }$-closed set 
\begin{equation*}
E\doteq \bigcap_{A\in \mathcal{X}}\left\{ \rho \in \mathcal{X}^{\ast }:\rho
\left( A^{\ast }A\right) \geq 0,\ \rho \left( \mathfrak{1}\right) =1\right\}
\end{equation*}%
of all positive normalized linear functionals $\rho \in \mathcal{X}^{\ast }$.
\end{definition}

Equivalently, $\rho \in \mathcal{X}^{\ast }$ is a state iff $\rho (\mathfrak{%
1})=1$ and $\Vert \rho \Vert _{\mathcal{X}^{\ast }}=1$. Note that any state
is hermitian: for all $\rho \in E$ and $A\in \mathcal{X}$, $\rho (A^{\ast })=%
\overline{\rho (A)}$. From the Banach-Alaoglu theorem \cite[Theorem 3.15]%
{Rudin}, $E$ is a weak$^{\ast }$-compact subset of the unit ball of $%
\mathcal{X}^{\ast }$. Therefore, the Krein-Milman theorem \cite[Theorem 3.23]%
{Rudin} tells us that $E$ is the weak$^{\ast }$ closure of the convex hull
of the (nonempty) set $\mathcal{E}(E)$ of its extreme points\footnote{%
I.e., the points which cannot be written as -- non--trivial -- convex
combinations of other elements of $E$.}: 
\begin{equation}
E=\overline{\mathrm{co}\mathcal{E}\left( E\right) }\ .
\label{closure of the convex hull}
\end{equation}%
The set $\mathcal{E}(E)$ is also called the extreme boundary of $E$. If $%
\mathcal{X}$ is separable then the weak$^{\ast }$ topology is metrizable on
any weak$^{\ast }$-compact subset of $\mathcal{X}^{\ast }$, by \cite[Theorem
3.16]{Rudin}. In particular, the state space $E$ of Definition \ref{state
space} is metrizable, in this case, and by the Choquet theorem \cite[p. 14]%
{Phe}, for any $\rho \in E$, there is a probability measure $\mu _{\rho }$
with support in $\mathcal{E}(E)$ such that, for any affine weak$^{\ast }$%
-continuous complex-valued function $g$ on $E$, 
\begin{equation}
g\left( \rho \right) =\int_{\mathcal{E}(E)}g\left( \hat{\rho}\right) \mathrm{%
d}\mu _{\rho }\left( \hat{\rho}\right) \ .  \label{affine decomposition}
\end{equation}%
The measure $\mu _{\rho }$ is unique for all $\rho \in E$, i.e., $E$ is a
Choquet simplex \cite[Theorem 4.1.15]{BrattelliRobinsonI}, iff the $C^{\ast
} $-algebra $\mathcal{X}$ is commutative, by \cite[Example 4.2.6]%
{BrattelliRobinsonI}.

If $E$ is not metrizable, meaning that $\mathcal{X}$ is not separable, note
that such a probability measure $\mu _{\rho }$ is only pseudo--supported by $%
\mathcal{E}(E)$, i.e., $\mu _{\rho }(\mathcal{B})=1$ for all Baire sets $%
\mathcal{B}\supseteq \mathcal{E}(E)$. This refers to the Choquet-Bishop-de
Leeuw theorem \cite[p. 17]{Phe}. Recall that the Baire sets are the elements
of the $\sigma $-algebra generated by the compact $G_{\delta }$ sets. If $%
\mathcal{E}(E)$ is a Baire set then $E$ must be metrizable \cite{MacGibbon}.
The weak$^{\ast }$ closure $\overline{\mathcal{E}(E)}$ may even not be a $%
G_{\delta }$ set, or more generally a Baire set, when $E$ is not metrizable.
In fact, in the non-metrizable case, $\mathcal{E}(E)$ can have very
surprising properties like being a \emph{zero-measure} Borel set for $\mu
_{\rho }$ (cf. \cite{Mokobodski}).

We use the state space $E$ in the next section to define a classical
algebra, the space $C(E;\mathbb{C})$ of complex-valued weak$^{\ast }$%
-continuous functions on $E$. Note that our (quantum) state space $E$ is
different from the one considered in \cite[Section 2.1, see also 2.1-c]%
{Bono2000}. In B\'{o}na's paper, the state space is defined to be the set of
density matrices associated with a fixed Hilbert space. In relation to our
approach, it corresponds to take, instead of \emph{all} states of $\mathcal{X%
}$, only those which are $\pi $-normal, for some \emph{fixed} represetation $%
\pi $\ of the $C^{\ast }$-algebra $\mathcal{X}$. Recall that the state $\rho
\in E$ is called \textquotedblleft $\pi $-normal\textquotedblright\ if the
state $\rho \circ \pi $ on $\pi (\mathcal{X})$ has a (unique) normal
extention to the von Neumann algebra $\pi (\mathcal{X})^{\prime \prime
}\supseteq \pi (\mathcal{X})$. By contrast, our definition of the (quantum)
state space is \emph{not} representation-dependent.

\subsection{Phase Space of $C^{\ast }$-Algebras\label{Phase Space copy(1)}}

Before the pioneer works of Jacobi and Boltzmann, then of Gibbs and Poincar%
\'{e}, the motion of a point-like particle was seen as a trajectory within
the three-dimensional space. However, in classical mechanics, fixing only
the position at a fixed time does not completely determine the trajectory,
which only becomes unique after fixing the momentum. This leads to the term 
\emph{phase}:\bigskip

\noindent \textit{If we regard a phase as represented by a point in space of 
}$2n$\textit{\ dimensions, the changes which take place in the course of
time in our ensemble of systems will be represented by a current in such
space.}\smallskip

\hfill Gibbs, 1902 \cite[p. 11, footnote]{Gibbs}\bigskip

\noindent This view point required the idea of \textquotedblleft high
dimensional\textquotedblright\ spaces, which widespread only in the first
decade of the 20th century. This space refers to the illustrious concept of 
\emph{phase space}, which seems to first appear in print in 1911 \cite%
{erenferst}.

The historical origins of the notion of phase space can be found in \cite%
{phase}, which makes explicit the \textquotedblleft \textit{tangle of
independent discovery and misattributions that persist today}%
\textquotedblright , even if this concept is seen as \textquotedblleft 
\textit{one of the most powerful inventions of modern science}%
\textquotedblright . For instance, the terminology of phase space is widely
used in classical mechanics, and also in \cite[Section 2.1]{Bono2000}, but
its use is regularly confusing in many textbooks, which often view the state
and phase spaces as the same thing.

The precise definition of phase space is an important, albeit non-trivial,
issue in the understanding of a physical system because it is usually
supposed to describe all its observable properties together with a
deterministic motion, once the initial coordinates of the system is fixed in
this phase space. In particular, it has to be sufficiently large to support
a deterministic, or causal, motion.

In classical physics, the phase space is a locally compact Hausdorff space%
\footnote{%
I.e., a topological space whose open sets separate points ($\rightarrow $%
Hausdorff) and whose points always have a compact neighborhood ($\rightarrow 
$locally compact).} $K$, like $\mathbb{R}^{6}$. In the algebraic formulation
of classical mechanics \cite[Chapter 3]{Landsman-livre}, one starts with a
commutative $C^{\ast }$-algebra. By the Gelfand theorem (see, for instance, 
\cite[Theorem 2.1.11A]{BrattelliRobinsonI} or \cite[Theorem 3.1]%
{Landsman-livre}), such an algebra is $\ast $-isomorphic to the algebra $%
C_{0}(K;\mathbb{C})$ of all continuous functions $f:K\rightarrow \mathbb{C}$
vanishing at infinity, where $K$ is a unique (up to a homeomorphism) locally
compact Hausdorff space. In this case, $K$ is, by definition, the phase
space of the physical system. The phase space $K$ is compact iff the
commutative $C^{\ast }$-algebra\ is unital.

For \emph{non-}commutative unital $C^{\ast }$-algebras, the definition of
the associated phase space is less straightforward. To motivate the
definition adopted here (Definition \ref{phase space}) for this space, we
exhibit the relation between the phase space $K$ and the state space $E$ of
Definition \ref{state space} for a commutative unital $C^{\ast }$-algebra
seen as an algebra\footnote{$C\left( K;\mathbb{C}\right) $ is separable iff $%
K$ is metrizable. See \cite[Problem (d) p. 245]{topology}.} 
\begin{equation*}
C\left( K;\mathbb{C}\right) \equiv \left( C\left( K;\mathbb{C}\right)
,+,\cdot _{{\mathbb{C}}},\times ,\overline{\left( \cdot \right) },\left\Vert
\cdot \right\Vert _{C(K;\mathbb{C})}\right)
\end{equation*}%
of continuous complex-valued functions on the compact Hausdorff space $K$.
Extreme points of $E$ are the so-called characters of this $C^{\ast }$%
-algebra: 
\begin{equation*}
\mathcal{E}\left( E\right) =\left\{ \mathfrak{c}\left( x\right) \in E:x\in
K\right\} \ ,
\end{equation*}%
where $\mathfrak{c}$ is the continuous and injective map from $K$ to $E$
defined by%
\begin{equation}
\left[ \mathfrak{c}\left( x\right) \right] \left( f\right) \doteq f\left(
x\right) \ ,\qquad f\in C\left( K;\mathbb{C}\right) ,\ x\in K\ .
\label{homeomorphisms}
\end{equation}%
Recall that the characters of a given $C^{\ast }$-algebra are, by
definition, the unital $\ast $-homomorphisms from this algebras to $\mathbb{C%
}$ (i.e., the multiplicative hermitian functionals on the algebra). See \cite%
[Proposition 2.3.27]{BrattelliRobinsonI}. In this special case, $\mathcal{E}%
(E)$ is weak$^{\ast }$-compact, like $K$, and the map $\mathfrak{c}$ is a
homeomorphism. In particular, the map $f\mapsto \hat{f}$ from $C\left( K;%
\mathbb{C}\right) $ to $C(\mathcal{E}(E);\mathbb{C})$ defined by%
\begin{equation}
\hat{f}\left( \mathfrak{c}\left( x\right) \right) =\left[ \mathfrak{c}\left(
x\right) \right] \left( f\right) \ ,\qquad f\in C\left( K;\mathbb{C}\right)
,\ x\in K\ ,  \label{f chapeau}
\end{equation}%
is a $\ast $-isomorphism of the commutative unital $C^{\ast }$-algebras $%
C\left( K;\mathbb{C}\right) $ and $C(\mathcal{E}(E);\mathbb{C})$. (See again 
\cite[Theorem 2.1.11A]{BrattelliRobinsonI} or \cite[Theorem 3.1]%
{Landsman-livre}.) Therefore, as is usual, the phase space of any
commutative unital $C^{\ast }$-algebra $\mathcal{X}$ can be identified with
the weak$^{\ast }$-compact set $\mathcal{E}(E)$ of extreme states of this
algebra. The set of all characters of the commutative $C^{\ast }$-algebra $%
\mathcal{X}$ is called its (Gelfand) \emph{spectrum} and its generalization
to arbitrary $C^{\ast }$-algebras is not straightforward: Remark, for
instance, that the algebra of $N\times N$ complex matrices, $N\geq 3$, has
no characters, in the above sense, at all, by the celebrated
Bell-Kochen-Specker theorem \cite[Theorem 6.5]{Landsman-livre}. The problem
of properly defining a notion of spectrum for a general $C^{\ast }$-algebra
is adressed, for instance, in \cite[Chapters 3 \& 4]{Dixmier} in the context
of decompositions of general representations of such an algebra in terms of
its \emph{irreducible} representations.

Now, with regard to the definition of the phase space as the set $\mathcal{E}%
(E)\neq E$ of extreme states, we want to emphasize that, for a \emph{%
non-commutative} unital $C^{\ast }$-algebra $\mathcal{X}$, this set does 
\emph{not} have to be weak$^{\ast }$-closed (in $E$), and so weak$^{\ast }$%
-compact. See, e.g., Lemma \ref{liminal copy(3)}. As explained above, a
classical physical system refers to the algebra of (complex-valued)
continuous functions decaying at infinity on a locally compact Hausdorff
space. Such an algebra is canonically $\ast $-isomorphic, via the
restriction of functions, to a $C^{\ast }$-algebra of functions defined on
any dense set of this Hausdorff space. Therefore, a natural definition of
the (classical) phase space associated with a general quantum system,
ensuring its compactness, is the weak$^{\ast }$ closure$\mathcal{\ }%
\overline{\mathcal{E}(E)}$, instead of the set $\mathcal{E}(E)$ of extreme
states itself:

\begin{definition}[Phase space]
\label{phase space}\mbox{ }\newline
Let $\mathcal{X}$ be a unital $C^{\ast }$-algebra. The associated phase
space is the weak$^{\ast }$ closure $\overline{\mathcal{E}(E)}$ of the
extreme boundary of the state space $E$ of Definition \ref{state space}.
\end{definition}

The phase space is, by definition, only a weak$^{\ast }$-closed \emph{subset}
of the state space. However, in mathematical physics, the unital $C^{\ast }$%
-algebra associated with an infinitely extended (quantum) system is usually
an approximately finite-dimensional (AF) $C^{\ast }$-algebra, i.e., it is
generated by an increasing family of \emph{finite-dimensional} $C^{\ast }$%
-subalgebras. They are all \emph{antiliminal} (Definition \ref{liminal
copy(2)}) and \emph{simple} (Definition \ref{simple def}). See Section \ref%
{Liminal appendix} for more details. In this case, by Lemma \ref{liminal
copy(3)}, $\mathcal{E}(E)$\ is weak$^{\ast }$-dense in $E$, i.e., 
\begin{equation}
E=\overline{\mathcal{E}(E)}\ .  \label{density extreme states}
\end{equation}%
In other words, in general, the phase space of Definition \ref{phase space}
is the same as the state space of Definition \ref{state space} for
infinitely extended quantum systems. The set $E$ of states has therefore a 
\emph{fairly complicated }geometrical structure. Compare, indeed, Equation (%
\ref{density extreme states}) with (\ref{closure of the convex hull}).
Provided the $C^{\ast }$-algebra\ $\mathcal{X}$ is separable, note that,
surprisingly, (\ref{closure of the convex hull}) and (\ref{density extreme
states}) do not prevent\ $E$ from having a unique center\footnote{%
I.e., a sort of maximally mixed point.} \cite{Lim}.

\subsection{Generic Weak$^{\ast }$-Compact Convex Sets in Infinite Dimension 
\label{generic convex set}}

\noindent \textit{Accidens vero est quod adest et abest praeter subiecti
corruptionem}.\textit{\footnote{%
Fr.: \textit{L'accident est ce qui arrive et s'en va sans provoquer la perte
du sujet.} See \cite[V. L'accident]{Isagoge}. It means that an accident is
what is present or absent in a subject without affecting its essence. This
comes from the \textit{Isagoge }(%
%TCIMACRO{\TeXButton{\greek}{\greek{EISAGWGH}}}%
%BeginExpansion
\greek{EISAGWGH}%
%EndExpansion
, originally in greek)\textit{\ }\cite{Isagoge} written in the IIIe century
by the Syrian Porphyry (of Tyr) as an introduction to \textit{Aristotle's
Categories}. The \textit{Isagoge} was a pivotal textbook in medieval
philosophy and more generaly on early logic during more than a millennium.
Its reception by medieval (scholastic) philosophers has, in particular,
initiated and fueled the celebrated \textit{problem of universals} \cite%
{querelle universaux}\ from the XIIe to the XIVe century.}}\smallskip

\hfill An accident in the Middle Ages\textit{\bigskip }

The existence of convex sets with dense extreme boundary is well-known in
infinite-dimensional vector spaces. For instance, the unit ball of any
infinite-dimensional Hilbert space has a dense extreme boundary in the weak
topology. In fact, a convex compact set with dense extreme boundary \emph{is
not an accident} in infinite-dimensional spaces, like Hilbert spaces or in
the dual space of an antiliminal unital $C^{\ast }$-algebras (cf. (\ref%
{density extreme states}) and Lemma \ref{liminal copy(3)}).

In 1959, Klee shows \cite{Klee} that, for convex norm-compact sets within a
Banach space, the property of having a dense set of extreme points is \emph{%
generic} in infinite dimension. More precisely, by \cite[Proposition 2.1,
Theorem 2.2]{Klee}, the set of all such convex compact subsets of an
infinite-dimensional separable\footnote{\cite[Proposition 2.1, Theorem 2.2]%
{Klee} seem to lead to the asserted property for all (possibly
non-separable) Banach spaces, as claimed in \cite%
{Klee,FonfLindenstrauss,infinite dim convexity}. However, \cite[Theorem 1.5]%
{Klee}, which assumes the separability of the Banach space, is clearly
invoked to prove the corresponding density stated in \cite[Theorem 2.2]{Klee}%
. We do not know how to remove the separability condition.} Banach space $%
\mathcal{Y}$ is generic\footnote{%
That is, the complement of a meagre set, i.e., a nowhere dense set.} in the
complete metric space of compact convex subsets of $\mathcal{Y}$, endowed
with the well-known Hausdorff metric topology \cite[Definition 3.2.1]{Beer}.
Klee's result is refined in 1998 by Fonf and Lindenstraus \cite[Section 4]%
{FonfLindenstrauss} for bounded norm-closed (but not necessarily
norm-compact) convex subsets of $\mathcal{Y}$ having so-called empty
quasi-interior (as a necessary condition). In this case, \cite[Theorem 4.3]%
{FonfLindenstrauss} shows that such sets can be approximated in the
Hausdorff metric topology by closed convex sets with a norm-dense set of
strongly exposed points\footnote{$x\in K$ is a strongly exposed point of a
convex set $K\subseteq \mathcal{Y}$ when there is $f\in \mathcal{Y}^{\ast }$
satisfying $f(x)=1$ and such that the diameter of $\{y\in K:f(y)\geq
1-\varepsilon \}$ tends to $0$ as $\varepsilon \rightarrow 0^{+}$.
(Strongly) exposed points are extreme elements of $K$.}. See, e.g., \cite[%
Section 7]{infinite dim convexity} for a recent review on this subject.

In this section we demonstrate the same genericity in the dual space $%
\mathcal{X}^{\ast }$ of an infinite-dimensional, separable unital $C^{\ast }$%
-algebra $\mathcal{X}$, endowed with its weak$^{\ast }$-topology. Of course,
if one uses the usual norm topology on $\mathcal{X}^{\ast }$ for continuous
linear functionals, then one can directly apply previous results \cite%
{Klee,FonfLindenstrauss} to the separable Banach space $\mathcal{X}^{\ast }$%
. This is not anymore possible if one considers the weak$^{\ast }$-topology.
In particular, \cite[Theorem 4.3]{FonfLindenstrauss} cannot be used because,
in general, weak$^{\ast }$-compact sets do not have an empty interior, in
the sense of the norm topology. However, generic properties of convex weak$%
^{\ast }$-compact sets, like the state space $E$ of Definition \ref{state
space}, are relevent in the present paper. We thus prove, in this situation,
results similar to \cite{Klee,FonfLindenstrauss} in order to better
understand the disconcerting structure of the state and phase spaces,
respectively $E$ and $\overline{\mathcal{E}(E)}$ defined above.

In order to talk about generic properties of convex weak$^{\ast }$-compact
sets, we first need to define an appropriate topological space of subsets of 
$\mathcal{X}^{\ast }$. It is naturally based on the set 
\begin{equation}
\mathbf{CK}\left( \mathcal{X}^{\ast }\right) \doteq \left\{ K\subseteq 
\mathcal{X}^{\ast }:K\neq \emptyset \text{ is convex and weak}^{\ast }\text{%
-compact}\right\} \ .  \label{ZDZD}
\end{equation}%
By Equation (\ref{hyperspace}) and Lemma \ref{dddddddddddddddddd}, note that 
\begin{equation}
\mathbf{CK}\left( \mathcal{X}^{\ast }\right) =\left\{ K\subseteq \mathcal{X}%
^{\ast }:K\neq \emptyset \text{ is convex, weak}^{\ast }\text{-closed and }%
\sup_{\sigma \in K}\left\Vert \sigma \right\Vert _{\mathcal{X}^{\ast
}}<\infty \right\} \ .  \label{sdfklsdjfklsdjf}
\end{equation}%
This is a set of weak$^{\ast }$-closed sets in a locally convex Hausdorff
space $\mathcal{X}^{\ast }$. See, e.g., \cite[Theorem 10.8]{BruPedra2}.

We now make $\mathbf{CK}(\mathcal{X}^{\ast })$ into a topological
(hyper)space by defining a hypertopology on it. Recall that topologies for
sets of closed subsets of topological spaces have been studied since the
beginning of the last century and when such topologies, restricted to
singletons, coincide with the original topology of the underlying space, we
talk about hypertopologies and hyperspaces of closed sets. There exist
several standard hypertopologies on the set of nonempty closed convex
subsets of a topological space like, for instance, the slice topology \cite[%
Section 2.4]{Beer}, the scalar and the linear topologies \cite[Section 4.3]%
{Beer}. Because of \cite[Theorem 2.4.5]{Beer}, note that the slice topology
is unappropriate here since it is not related to the weak$^{\ast }$-topology
of $\mathcal{X}^{\ast }$, but rather to its norm topology. In fact, we do
not use any of those standard hypertopologies, but another natural topology
on $\mathbf{CK}(\mathcal{X}^{\ast })$ given by a family of pseudometrics%
\footnote{%
Recall that a pseudometric $d$ satisfies all properties of a metric but the
identity of indiscernibles. In fact, $d(x,x)=0$ but possibly $d(x,y)=0$ for $%
x\neq y$.} inspired by the Hausdorff metric topology for closed subsets of $%
\mathbb{C}$:

\begin{definition}[Weak$^{\ast }$-Hausdorff hypertopology for convex sets]
\label{hypertopology}\mbox{ }\newline
The weak$^{\ast }$-Hausdorff hypertopology on $\mathbf{CK}(\mathcal{X}^{\ast
})$ is the topology induced by the family of Hausdorff pseudometrics $%
d_{H}^{(A)}$ defined, for all $A\in \mathcal{X}$, by%
\begin{equation}
d_{H}^{(A)}(K,\tilde{K})\doteq \max \left\{ \max_{\sigma \in K}\min_{\tilde{%
\sigma}\in \tilde{K}}\left\vert \left( \sigma -\tilde{\sigma}\right) \left(
A\right) \right\vert ,\max_{\tilde{\sigma}\in \tilde{K}}\min_{\sigma \in
K}\left\vert \left( \sigma -\tilde{\sigma}\right) \left( A\right)
\right\vert \right\} ,\qquad K,\tilde{K}\in \mathbf{CK}\left( \mathcal{X}%
^{\ast }\right) \ .  \label{pseudometrics}
\end{equation}
\end{definition}

\noindent Compare (\ref{pseudometrics}) with the definition of the Hausdorff
distance, given by (\ref{Hausdorf}). Definition \ref{hypertopology} is a
restriction of the weak$^{\ast }$-Hausdorff hypertopology of Definition \ref%
{hypertopology0}. In this topology, an arbitrary net $(K_{j})_{j\in J}$
converges to $K_{\infty }$ iff, for all $A\in \mathcal{X}$, 
\begin{equation}
\lim_{J}d_{H}^{(A)}(K_{j},K_{\infty })=0\ .  \label{induced topology}
\end{equation}%
This condition defines a unique topology in $\mathbf{CK}(\mathcal{X}^{\ast
}) $, by \cite[Chapter 2, Theorem 9]{topology}. In fact, because this
topology is generated by a family of pseudometrics, it is a uniform
topology, see, e.g., \cite[Chapter 6]{topology}.

It is completely obvious from the definition that any net $(\sigma
_{j})_{j\in J}$ in $\mathcal{X}^{\ast }$ converges to $\sigma \in \mathcal{X}%
^{\ast }$ in the weak$^{\ast }$ topology iff the net $(\{\sigma
_{j}\})_{j\in J}$ converges in $\mathbf{CK}(\mathcal{X}^{\ast })$ to $%
\{\sigma \}$ in the weak$^{\ast }$-Hausdorff (hyper)topology. In other
words, the embedding of $\mathcal{X}^{\ast }$ into $\mathbf{CK}(\mathcal{X}%
^{\ast })$ is a bicontinuous bijection on its image. This justifies the use
of the name weak$^{\ast }$-Hausdorff \emph{hyper}topology. We are not aware
whether this particular hypertopology has already been considered in the
past. We thus give in Section \ref{Hausdorff Hypertopology} its complete
study along with interesting connections to other fields of mathematics and
results that are more general than those stated in Section \ref{generic
convex set}.

Endowed with the weak$^{\ast }$-Hausdorff hypertopology, $\mathbf{CK}(%
\mathcal{X}^{\ast })$ is a \emph{Hausdorff} hyperspace. See Corollary \ref%
{convexity corrolary}. Observe also that the limit of weak$^{\ast }$%
-Hausdorff convergent nets within $\mathbf{CK}(\mathcal{X}^{\ast })$ is
directly related to lower and upper limits \`{a} la Painlev\'{e} \cite[\S\ 29%
]{topology-painleve}, as explained in Section \ref{Hyperconvergences}. See,
in particular, Equations (\ref{Li}) and (\ref{Ls}). When $\mathcal{X}$ is a
separable Banach space, Corollary \ref{Solution selfbaby copy(4)} tells us
that any weak$^{\ast }$-Hausdorff convergent net $(K_{j})_{j\in J}\subseteq 
\mathbf{CK}(\mathcal{X}^{\ast })$ converges to its Kuratowski-Painlev\'{e}
limit $K_{\infty }$, which is thus the set of all weak$^{\ast }$
accumulation points of nets $(\sigma _{j})_{j\in J}$ with $\sigma _{j}\in
K_{j}$.

Recall that, by the Krein-Milman theorem \cite[Theorem 3.23]{Rudin}, any
nonempty convex weak$^{\ast }$-compact set $K\in \mathbf{CK}(\mathcal{X}%
^{\ast })$ is the weak$^{\ast }$-closure of the convex hull of the
(nonempty) set $\mathcal{E}(K)$ of its extreme points: 
\begin{equation*}
K=\overline{\mathrm{co}\mathcal{E}\left( K\right) }\ .
\end{equation*}%
The property $K=\overline{\mathcal{E}\left( K\right) }$ (with respect to the
weak$^{\ast }$ topology) looks very peculiar. Nonetheless, as a matter of
fact, typical elements of $\mathbf{CK}(\mathcal{X}^{\ast })$ have this
property:

\begin{theorem}[Generic convex weak$^{\ast }$-compact sets]
\label{theorem dense cool1}\mbox{ }\newline
Let $\mathcal{X}$ be an infinite-dimensional separable Banach space. Then,
the set $\mathcal{D}$ of all nonempty convex weak$^{\ast }$-compact sets $K$
with a weak$^{\ast }$-dense set $\mathcal{E}(K)$ of extreme points is a weak$%
^{\ast }$-Hausdorff-dense $G_{\delta }$ subset of $\mathbf{CK}(\mathcal{X}%
^{\ast })$.
\end{theorem}

\begin{proof}
Combine Proposition \ref{Solution selfbaby copy(5)+00} with Theorem \ref%
{Solution selfbaby copy(5)+0}. Note that the proof of Theorem \ref{Solution
selfbaby copy(5)+0} is crafted by following original Poulsen's intuitive
construction \cite{Poulsen}, like in the proof of \cite[Theorem 4.3]%
{FonfLindenstrauss}. The Hahn-Banach separation theorem \cite[Theorem 3.4 (b)%
]{Rudin} plays a crucial role in this context.
\end{proof}

\noindent As a consequence, $\mathcal{D}$ is generic in the hyperspace $%
\mathbf{CK}(\mathcal{X}^{\ast })$, that is, the complement of a meagre set,
i.e., a nowhere dense set. In other words, $\mathcal{D}$ is of second
category in $\mathbf{CK}(\mathcal{X}^{\ast })$.

The weak$^{\ast }$-Hausdorff hypertopology on $\mathbf{CK}(\mathcal{X}^{\ast
})$ is finner than the scalar topology \cite[Section 4.3]{Beer} restricted
to weak$^{\ast }$-closed sets. The linear topology on the set of nonempty
closed convex subsets is the supremum of the scalar and Wijsman topologies.
Since the Wijsman topology \cite[Definition 2.1.1]{Beer} requires a metric
space, one has to use the norm on $\mathcal{X}^{\ast }$ and the linear
topology is not comparable with the weak$^{\ast }$-Hausdorff hypertopology.
If one uses the metric (\ref{metrics0}) generated the weak$^{\ast }$
topology on balls of $\mathcal{X}^{\ast }$ for a separable Banach space $%
\mathcal{X}$, then the Wijsman and linear topologies for norm-closed balls
of $\mathcal{X}^{\ast }$ are coarser than the weak$^{\ast }$-Hausdorff
hypertopology, by Theorem \ref{Solution selfbaby copy(4)+1}. As a matter of
fact, the Hausdorff metric topology is very fine, as compared to various
standard hypertopologies (apart from the Vietoris\footnote{%
Vietoris and Hausdorff metric topologies are not comparable.}
hypertopology). Consequently, the weak$^{\ast }$-Hausdorff hypertopology can
be seen as a very fine, weak$^{\ast }$-type, topology on $\mathbf{CK}(%
\mathcal{X}^{\ast })$. It shows that the density of the subset of all convex
weak$^{\ast }$ compact sets with weak$^{\ast }$-dense set of extreme points
stated in Theorem \ref{theorem dense cool1} is a very \emph{strong}
property. Moreover, the genericity of such sets even holds true \emph{inside}
the state space $E$ of any separable unital $C^{\ast }$-algebra:

\begin{theorem}[Generic weak$^{\ast }$-compact convex subset of the state
space]
\label{theorem density2}\mbox{ }\newline
Let $\mathcal{X}$ be a infinite-dimensional, separable and unital $C^{\ast }$%
-algebra and $E$ the state space (Definition \ref{state space}). Denote by $%
\mathbf{CK}(E)$ the set of all nonempty convex weak$^{\ast }$-compact
subsets of $E$ and by $\mathcal{D}(E)$ the set of all $K\in \mathbf{CK}(E)$
with a weak$^{\ast }$-dense set $\mathcal{E}(K)$ of extreme points. Then,
endowed with the weak$^{\ast }$-Hausdorff hypertopology, $\mathbf{CK}(E)$ is
a compact and completely metrizable hyperspace with $\mathcal{D}(E)$ being a
dense $G_{\delta }$ subset.
\end{theorem}

\begin{proof}
Since any state $\rho \in E$ has norm equal to $\Vert \rho \Vert _{\mathcal{X%
}^{\ast }}=1$, we deduce from Theorem \ref{Solution selfbaby copy(4)+1} that 
$\mathbf{CK}(E)$ belongs to the weak$^{\ast }$-Hausdorff-compact and
completely metrizable hyperspace $\mathbf{CK}_{1}(\mathcal{X}^{\ast })$,
defined by (\ref{ZD}). By Corollary \ref{Solution selfbaby copy(4)} and
because $E$ is a weak$^{\ast }$-closed set, $\mathbf{CK}(E)$ is weak$^{\ast
} $-Hausdorff-closed, and thus a compact and completely metrizable
hyperspace. It remains to prove that $\mathcal{D}(E)$ is a dense $G_{\delta
} $ subset of $\mathbf{CK}(E)$.

The fact that $\mathcal{D}(E)$ is a $G_{\delta }$ subset of $\mathbf{CK}(E)$
can directly be deduced from the proof of Proposition \ref{Solution selfbaby
copy(5)+00} by repacing $\mathcal{F}_{D,m}$ with 
\begin{equation*}
\mathcal{F}_{m}\left( E\right) \doteq \left\{ K\in \mathbf{CK}(E):\exists
\omega \in K,\ B\left( \omega ,1/m\right) \cap \mathcal{E}\left( K\right)
=\emptyset \right\} \subseteq \mathbf{CK}(E)\ .
\end{equation*}%
To prove the weak$^{\ast }$-Hausdorff-density of $\mathcal{D}(E)\subseteq 
\mathbf{CK}(E)$, it suffices to reproduce the proof of Theorem \ref{Solution
selfbaby copy(5)+0}, by adding one essential ingredient: the decomposition
of any continuous linear functional into non-negative components proven in 
\cite{decomposition} for real Banach spaces. By noting that (i) $\mathcal{X}%
^{\mathbb{R}}$ (\ref{self-adjoint elements}) is a real Banach space, (ii)
all states are hermitian functionals over $\mathcal{X}$, (iii) $(\mathcal{X}%
^{\mathbb{R}})^{\ast }$ is canonically identify with the real space of
hermitian elements of $\mathcal{X}^{\ast }$, and (iv) any $\sigma \in 
\mathcal{X}^{\ast }$ is decomposed as $\sigma =\mathrm{Re}\{\sigma \}+i%
\mathrm{Im}\{\sigma \}$ with $\mathrm{Re}\{\sigma \},\mathrm{Im}\{\sigma
\}\in (\mathcal{X}^{\mathbb{R}})^{\ast }$, we deduce from \cite%
{decomposition} that any $\sigma \in \mathcal{X}^{\ast }$ can be decomposed
as 
\begin{equation}
\sigma =c_{1}\rho _{1}-c_{2}\rho _{2}+i\left( c_{3}\rho _{3}-c_{4}\rho
_{4}\right) \ ,\qquad c_{1},c_{2},c_{3},c_{4}\in \mathbb{R}_{0}^{+},\ \rho
_{1},\rho _{2},\rho _{3},\rho _{4}\in E\ .  \label{decoposition}
\end{equation}%
At \emph{Step 1} of the proof of Theorem \ref{Solution selfbaby copy(5)+0},
because of (\ref{decoposition}), we observe that there is a\ non-zero
positive functional 
\begin{equation*}
\sigma _{1}\in (\mathcal{X}^{\ast }\backslash \mathrm{span}\{\omega
_{1},\ldots ,\omega _{n_{\varepsilon }}\})\ .
\end{equation*}%
So, we proceed by using $\sigma _{1}$ as a (non-zero) positive functional
with norm $\Vert \sigma _{1}\Vert _{\mathcal{X}^{\ast }}\leq 1$ and the
state 
\begin{equation*}
\omega _{n_{\varepsilon }+1}\doteq \left( 1-\lambda _{1}\Vert \sigma
_{1}\Vert _{\mathcal{X}^{\ast }}\right) \varpi _{1}+\lambda _{1}\sigma
_{1}\in E\ ,
\end{equation*}%
instead of (\ref{omega1}). One then iterates the arguments, as explained in
the proof of Theorem \ref{Solution selfbaby copy(5)+0}, using always a
(non-zero) positive functional $\sigma _{n}$ with norm $\Vert \sigma
_{n}\Vert _{\mathcal{X}^{\ast }}\leq 1$ and 
\begin{equation*}
\omega _{n_{\varepsilon }+n}\doteq \left( 1-\lambda _{n}\Vert \sigma
_{n}\Vert _{\mathcal{X}^{\ast }}\right) \varpi _{n}+\lambda _{n}\sigma
_{n}\in E\ ,
\end{equation*}%
instead of (\ref{definition omegan}), as already explained. In doing so, we
ensure that the convex weak$^{\ast }$-compact set $K_{\infty }$ of Equation (%
\ref{equaion}) belongs to $\mathcal{D}(E)\subseteq \mathbf{CK}(E)$.
\end{proof}

\noindent Note that Theorem \ref{theorem density2} does \emph{not} directly
follow from Theorem \ref{theorem dense cool1} because the complement of $%
\mathbf{CK}(E)$ is open and dense in $\mathbf{CK}(\mathcal{X}^{\ast })$.

Important examples of (antiliminal\emph{\ }and simple) $C^{\ast }$-algebras
with state space $E\in \mathcal{D}(E)\subseteq \mathcal{D}\subseteq \mathbf{%
CK}(\mathcal{X}^{\ast })$, i.e., satisfying (\ref{density extreme states}),
are the (even subalgebra of the) CAR $C^{\ast }$-algebras for
(non-relativistic) fermions on the lattice. Quantum-spin systems, i.e.,
infinite tensor products of copies of some elementary finite dimensional
matrix algebra, referring to a spin variable, are also important examples.\
They are, for instance, widely used in quantum information theory as well as
in condensed matter physics. In all these physical situations, the
corresponding (non-commutative) $C^{\ast }$-algebra $\mathcal{X}$ is
separable and $E$ is thus a metrizable weak$^{\ast }$-compact convex set. It
is \emph{not} a simplex \cite[Example 4.2.6]{BrattelliRobinsonI}, but 
\begin{equation}
E=\overline{\bigcup\limits_{n\in \mathbb{N}}\mathfrak{P}_{n}}
\label{property}
\end{equation}%
is the weak$^{\ast }$-closure of the union of a strictly increasing sequence 
$(\mathfrak{P}_{n})_{n\in \mathbb{N}}\subseteq \mathcal{D}(E)$ of Poulsen
simplices\footnote{%
It is the (unique up to a homeomorphism) metrizable simplex with dense
extreme boundary.} \cite{Poulsen}. Equation (\ref{property}) is a
consequence of well-known results (see, e.g., \cite{Israel,BruPedra2}) and
we give its complete proof in \cite{BruPedra-MFII}. In other words, by
Proposition \ref{Solution selfbaby copy(5)+1}, $E$ is the \emph{weak}$^{\ast
}$\emph{-Hausdorff limit} of the increasing sequence $(\mathfrak{P}%
_{n})_{n\in \mathbb{N}}$ within the set $\mathcal{D}(E)$ of all $K\in 
\mathbf{CK}(E)$ with weak$^{\ast }$-dense set of extreme points.

Note that the Poulsen simplex $\mathfrak{P}$ is \emph{not only} a metrizable
simplex with dense extreme boundary $\mathcal{E}(\mathfrak{P})$. It has also
the following remarkable properties:

\begin{itemize}
\item It is \emph{unique}, up to an affine homeomorphism. Indeed, any two
compact metrizable simplexes with dense extreme boundary are mapped into
each other by an affine homeomorphism, by \cite[Theorem 2.3]%
{Lindenstrauss-etal}.

\item It is \emph{universal} in the sense that every compact metrizable
simplex is affinely homeomorphic to a (closed) face\footnote{%
A face $F$ of a convex set $K$ is defined to be a subset of $K$ with the
property that, if $\rho =\lambda _{1}\rho _{1}+\cdots +\lambda _{n}\rho
_{n}\in F$ with $\rho _{1},\ldots ,\rho _{n}\in K$, $\lambda _{1},\ldots
,\lambda _{n}\in (0,1)$ and $\lambda _{1}+\cdots +\lambda _{n}=1$, then $%
\rho _{1},\ldots ,\rho _{n}\in F$.} of $\mathfrak{P}$, by \cite[Theorem 2.5]%
{Lindenstrauss-etal}. As a consequence, by \cite[Example 4.2.6]%
{BrattelliRobinsonI}, the state space of \emph{any} classical system with
separable phase space can be seen as a face of $\mathfrak{P}$. Moreover, by 
\cite{Haydon}, \emph{every }Polish space\footnote{%
I.e., a separable topological space that is homeomorphic to a complete
metric space.} is homeomorphic to the extreme boundary of a face of $%
\mathfrak{P}$.

\item It is \emph{homogeneous} in the sense that any two proper closed
isomorphic\footnote{%
I.e, there is an affine homeomorphism between both faces.} faces of $%
\mathfrak{P}$ are mapped into each other by an affine automorphism of $%
\mathfrak{P}$. See \cite[Theorem 2.3]{Lindenstrauss-etal}.
\end{itemize}

\noindent Together with Equation (\ref{property}) this demonstrates, for
infinite-dimensional quantum systems, the amazing structural richness of the
state space $E$, while making mathematically clear the possible
identification of the phase space $\overline{\mathcal{E}(E)}$ as the state
space $E$.

In fact, because of Theorems \ref{theorem dense cool1}-\ref{theorem density2}%
, if the \textquotedblleft primordial\textquotedblright\ (non-commutative)\
algebra $\mathcal{X}$ has \emph{infinite dimension}, then, as is done
without much attention in many textbooks, one should expect that the state
and phase spaces, as we define them in the present paper, are identical,
even if this feature has to be mathematically proven in each case (like for
antiliminal\emph{\ }and simple $\mathcal{X}$). For instance, if $\mathcal{X}$
is an infinite-dimensional, commutative and unital $C^{\ast }$-algebra, then
the state and phase spaces, respectively $E$ and $\overline{\mathcal{E}(E)}$%
, are cleary different from each other, even if $E$ can always be
approximated in the weak$^{\ast }$-Hausdorff hypertopology by a convex weak$%
^{\ast }$-compact set $K\subseteq E$ with weak$^{\ast }$-dense extreme
boundary, by Theorem \ref{theorem density2}.

\subsection{Classical $C^{\ast }$-Algebra of Continuous Functions on the
State Space\label{Classical algebra}}

The space $C(E;\mathbb{C})$ of complex-valued weak$^{\ast }$-continuous
functions on the state space $E$ of Definition \ref{state space}, endowed
with the point-wise operations and complex conjugation, is a unital \emph{%
commutative} $C^{\ast }$-algebra denoted by%
\begin{equation}
\mathfrak{C}\doteq \left( C\left( E;\mathbb{C}\right) ,+,\cdot _{{\mathbb{C}}%
},\times ,\overline{\left( \cdot \right) },\left\Vert \cdot \right\Vert _{%
\mathfrak{C}}\right) \ ,  \label{metaciagre set 2}
\end{equation}%
where%
\begin{equation}
\left\Vert f\right\Vert _{\mathfrak{C}}\doteq \max_{\rho \in E}\left\vert
f\left( \rho \right) \right\vert \ ,\qquad f\in \mathfrak{C}\ .
\label{metaciagre set 2bis}
\end{equation}%
The (real) Banach subspace of all real-valued functions from $\mathfrak{C}$
is denoted by $\mathfrak{C}^{\mathbb{R}}\varsubsetneq \mathfrak{C}$. If $%
\mathcal{X}$ is separable then $\mathfrak{C}$ is also separable, $E$ being
in this case metrizable. See, e.g., \cite[Problem (d) p. 245]{topology}.

Similar to the mapping defined by Equation (\ref{f chapeau}) for commutative 
$C^{\ast }$-algebras, elements of the unital $C^{\ast }$-algebra $\mathcal{X}
$ canonically define continuous affine functions $\hat{A}\in \mathfrak{C}$
by 
\begin{equation}
\hat{A}\left( \rho \right) \doteq \rho \left( A\right) \ ,\qquad \rho \in
E,\ A\in \mathcal{X}\ .  \label{fA}
\end{equation}%
This is the well-known \emph{Gelfand transform}. Note that $A\neq B$ yields $%
\hat{A}\neq \hat{B}$, as states separates elements of $\mathcal{X}$. Since $%
\mathcal{X}$ is a (unital) $C^{\ast }$-algebra, 
\begin{equation}
\left\Vert A\right\Vert _{\mathcal{X}}=\max_{\rho \in E}\left\vert \rho
\left( A\right) \right\vert \ ,\qquad A\in \mathcal{X}^{\mathbb{R}}\ ,
\label{norm properties}
\end{equation}%
and hence, the map $A\mapsto \hat{A}$ defines a linear isometry from the
Banach space $\mathcal{X}^{\mathbb{R}}$ of all self-adjoint elements (cf.
Equation (\ref{self-adjoint elements})) to the space $\mathfrak{C}^{\mathbb{R%
}}$ of all real-valued functions on $E$.

For any self-adjoint\footnote{%
This means that $A\in \mathcal{B}$ implies $A^{\ast }\in \mathcal{B}$.}
subspace $\mathcal{B}\subseteq \mathcal{X}$, we define the $\ast $%
-subalgebras 
\begin{equation}
\mathfrak{C}_{\mathcal{B}}\equiv \mathfrak{C}_{\mathcal{B}}\left( E\right)
\doteq \mathbb{C}[\{\hat{A}:A\in \mathcal{B}\}]\subseteq \mathfrak{C}\quad 
\text{and}\quad \mathfrak{C}_{\mathcal{B}}^{\mathbb{R}}\equiv \mathfrak{C}_{%
\mathcal{B}}^{\mathbb{R}}\left( E\right) \doteq \mathbb{R}[\{\hat{A}:A\in 
\mathcal{B\cap X}^{\mathbb{R}}\}]\subseteq \mathfrak{C}^{\mathbb{R}}\ ,
\label{def frac Cb}
\end{equation}%
where $\mathbb{K}[\mathcal{Y}]\subseteq \mathfrak{C}$ denotes the $\mathbb{K}
$-algebra generated by $\mathcal{Y}$, i.e., the subspace of polynomials in
the elements of $\mathcal{Y}$, with coefficients in the field $\mathbb{K}$ ($%
=\mathbb{R},\mathbb{C}$). The unit $\mathfrak{\hat{1}}\in \mathfrak{C}$,
being the constant map $\mathfrak{\hat{1}}(\rho )=1$ for $\rho \in E$ (cf.
Definition \ref{state space}), belongs, by definition, to $\mathfrak{C}_{%
\mathcal{B}}$ and $\mathfrak{C}_{\mathcal{B}}^{\mathbb{R}}\subseteq 
\mathfrak{C}_{\mathcal{B}}$. If $\mathcal{B}$ is dense in $\mathcal{X}$ then 
$\mathfrak{C}_{\mathcal{B}}$ separates states. Therefore, by the
Stone-Weierstrass theorem \cite[Chap. V, \S 8]{Conway}, for any dense
self-adjoint subset $\mathcal{B}\subseteq \mathcal{X}$, $\mathfrak{C}_{%
\mathcal{B}}$ is dense in $\mathfrak{C}$, i.e., $\mathfrak{C}=\overline{%
\mathfrak{C}_{\mathcal{B}}}$.

\subsection{Classical $C^{\ast }$-Algebra of Continuous Functions on the
Phase Space\label{Classical algebra copy(1)}}

If the weak$^{\ast }$-compact set $\overline{\mathcal{E}(E)}$ is supposed to
play the role of a phase space (cf. Definition \ref{phase space}), then a
classical dynamics should be defined on the space $C(\overline{\mathcal{E}(E)%
};\mathbb{C})$ of complex-valued weak$^{\ast }$-continuous functions on $%
\overline{\mathcal{E}(E)}$. Endowed with the usual point-wise operations and
complex conjugation, it is again a unital \emph{commutative} $C^{\ast }$%
-algebra. Of course, there is a natural $\ast $-homomorphism $\mathfrak{C}%
\rightarrow C(\overline{\mathcal{E}(E)};\mathbb{C})$, by restriction on $%
\overline{\mathcal{E}(E)}$ of functions from $\mathfrak{C}$. Recall that $C(%
\overline{\mathcal{E}(E)};\mathbb{C})$ is canonically $\ast $-isomorphic,
via the restriction on $\mathcal{E}(E)$ of functions, to a $C^{\ast }$%
-subalgebra of $C(\mathcal{E}(E);\mathbb{C})$. In Corollary \ref{corollary
conservation} and Equation (\ref{corollary conservation0}), we show that the
classical dynamics constructed in the present paper \emph{can be pushed
forward}, through the restriction map, from $\mathfrak{C}$ to either $C(%
\overline{\mathcal{E}(E)};\mathbb{C})$ or $C(\mathcal{E}(E);\mathbb{C})$.
The generator of the dynamics on $C(\overline{\mathcal{E}(E)};\mathbb{C})$
can be expressed on polymonials via the Poisson bracket of Corollary \ref%
{proposition sympatoch copy(2)}, by Proposition \ref{Poisson algebra prop}.

In standard classical mechanics, in the case of compact phase spaces, even
if the $C^{\ast }$-algebra $\mathfrak{C}$ is always well-defined, note that $%
\mathfrak{C}$ is usually never used, but rather $C(\overline{\mathcal{E}(E)};%
\mathbb{C})$, and a classical system is always supposed to be in some
extreme state. In fact, the same physical object cannot be at the same time
on two distinct points of the phase space, according to the spatio-temporal
identity of classical mechanics \cite{FK}. This refers to Leibniz's
Principle of Identity of Indiscernibles\footnote{%
Leibniz's Principle of Identity of Indiscernibles \cite[p. 1]{FK}:
\textquotedblleft \textit{Two objects which are indistinguishable, in the
sense of possessing all properties in common, cannot, in fact, be two
objects at all. In effect, the Principle provides a guarantee that
individual objects will always be distinguishable.}\textquotedblright}. This
is related to the fact that any extreme classical state is dispersion-free,
see \cite[Eq. (6.3), $V$ being the state]{Landsman-livre}. In the classical
situation, the space $\mathfrak{C}$ is therefore \emph{not} fundamental: In
this case, by the Riesz--Markov theorem, the state space is the same as the
set of probability measures on the phase space $\overline{\mathcal{E}(E)}$
and a mixed, or non-extreme, state $\rho \in E\backslash \mathcal{E}(E)$ of
a classical system is only used to reflect the lack of knowledge on the
physical object along with a probabilistic interpretation. Compare with (\ref%
{affine decomposition}).

For quantum systems, this property is not as evident as it is for classical
ones, as conceptually discussed for instance in \cite{FK}. The
spatio-temporal identity of classical mechanics is questionable in quantum
mechanics. This is correlated with the celebrated EPR paradox of Einstein,
Podolsky and Rosen. See also Einstein's conceptual opposition to quantum
mechanics: \bigskip

\noindent \textit{If one asks what, irrespective of quantum mechanics, is
characteristic of the world of ideas of physics, one is first of all struck
by the following: the concepts of physics relate to a real outside world...
it is further characteristic of these physical objects that they are thought
of as a range in a space-time continuum. An essential aspect of this
arrangement of things in physics is that they lay clamed, at a certain time,
to an existence independent of one another, provided these objects
\textquotedblleft are situated in different parts of space\textquotedblright
.}\smallskip

\hfill Einstein, 1948 \cite{einstein}\bigskip

\noindent The non-locality of quantum mechanics was in fact Einstein's main
criticism on this theory \cite{Bell}, more than its weakly deterministic
features.

The non-locality of quantum mechanics has been experimentally verified, for
instance via Bell's inequalities, and it is not the subject of the present
paper to discuss further related topics, like the existence of hidden
variables in quantum physics. The point in this brief discussion is that
there is no clear reason to restrict ourselves to the phase space $\overline{%
\mathcal{E}(E)}$ and not also consider the whole state space $E$, as, in
contrast to classical physics, extreme states are not anymore\
dispersion-free for quantum systems. See, e.g., \cite[Proposition 2.10]%
{Landsman-livre}; cf. also the Bell-Kochen-Specker theorem \cite[Theorem 6.5]%
{Landsman-livre}. As a matter of fact, important phenomena, like the
breakdown of the $U(1)$-gauge symmetry in the BCS theory of
superconductivity, are related with non-extreme states. See, as an example, 
\cite[Theorem 6.5]{BruPedra1}. What's more, the phase space and the state
space turn out to be \emph{identical} for important classes of (infinitely
extended) quantum systems in condensed matter physics, as already explained.
See Equation (\ref{density extreme states}).

\section{Poisson Structures in Quantum Mechanics\label{Poisson Structures in
Quantum Mechanics}}

If $\mathfrak{g}$ is a finite dimensional Lie algebra, there is a standard
contruction of a Poisson bracket for the polynomial functions on its dual
space $\mathfrak{g}^{\ast }$. See, for instance, \cite[Section 7.1]{Poission}%
. Observe that the (real) space $\mathcal{X}^{\mathbb{R}}$ of all
self-adjoint elements of an arbitrary $C^{\ast }$-algebra $\mathcal{X}$
forms a Lie algebra by endowing it with the Lie bracket $i[\cdot ,\cdot ]$,
i.e., the skew-symmetric biderivation on $\mathcal{X}^{\mathbb{R}}$ defined
by the commutator 
\begin{equation}
i\left[ A,B\right] \doteq i\left( AB-BA\right) \in \mathcal{X}^{\mathbb{R}}\
,\qquad A,B\in \mathcal{X}^{\mathbb{R}}\ .  \label{commutator}
\end{equation}%
One of the aims of our paper is to extend such a construction of a Poisson
bracket to polynomial functions on the dual space of $\mathcal{X}^{\mathbb{R}%
}$, which is possibly infinite-dimensional. Before doing that, we first
briefly present B\'{o}na's setting \cite[Sections 2.1b, 2.1c]{Bono2000},
which motivated the present work.

\subsection{B\'{o}na's Poisson Structures\label{new section jundiai}}

B\'{o}na \cite[Sections 2.1b, 2.1c]{Bono2000} proposes a Poisson structure
for polynomial functions on the \emph{predual} (instead of the dual) of a $%
C^{\ast }$-algebra. Recall that, if $\mathcal{X}$ is the $C^{\ast }$-algebra 
$\mathcal{B}(\mathcal{H})$\ of all bounded operators on a Hilbert space $%
\mathcal{H}$, then its predual $\mathcal{X}_{\ast }$\ can be identified with
the Banach space $\mathcal{L}^{1}(\mathcal{H})$ of trace-class operators on $%
\mathcal{H}$, with the (trace) norm 
\begin{equation*}
\left\Vert A\right\Vert _{1}\doteq \mathrm{Tr}_{\mathcal{H}}\sqrt{A^{\ast }A}%
\ ,\qquad A\in \mathcal{L}^{1}(\mathcal{H})\ .
\end{equation*}%
More precisely, for all $A\in \mathcal{B}(\mathcal{H})$ ($=\mathcal{X}$),
the linear map $\hat{A}$ defined by 
\begin{equation*}
\sigma \mapsto \mathrm{Tr}_{\mathcal{H}}(\sigma A)
\end{equation*}%
from $\mathcal{L}^{1}(\mathcal{H})$ to $\mathbb{C}$ is continuous and,
conversely, any linear continuous functional $\hat{A}:\mathcal{L}^{1}(%
\mathcal{H})\rightarrow \mathbb{C}$ is of this form for a unique $A\in 
\mathcal{B}(\mathcal{H})$. From this, one concludes that the dual of the
real Banach space $\mathcal{L}_{\mathbb{R}}^{1}(\mathcal{H})$ of
self-adjoint trace-class operators on $\mathcal{H}$ is the real Banach space 
$\mathcal{B}(\mathcal{H})^{\mathbb{R}}$ of self-adjoint bounded operators on
the Hilbert space $\mathcal{H}$. Thus, 
\begin{equation}
\mathcal{B}(\mathcal{H})^{\mathbb{R}}\equiv (\mathcal{L}_{\mathbb{R}}^{1}(%
\mathcal{H}))^{\ast }\subseteq C(\mathcal{L}_{\mathbb{R}}^{1}(\mathcal{H});%
\mathbb{R}).  \label{dual}
\end{equation}%
Let 
\begin{equation*}
\mathfrak{C}_{\mathcal{B}(\mathcal{H})^{\mathbb{R}}}^{\mathbb{R}}\doteq 
\mathbb{R}[\mathcal{B}(\mathcal{H})^{\mathbb{R}}]\subseteq C(\mathcal{L}_{%
\mathbb{R}}^{1}(\mathcal{H});\mathbb{R})
\end{equation*}%
be the subalgebra of polynomials in the elements of $\mathcal{B}(\mathcal{H}%
)^{\mathbb{R}}$ with real coefficients. The elements of this subalgebra are
called \textquotedblleft polynomial\textquotedblright\ functions on $%
\mathcal{L}_{\mathbb{R}}^{1}(\mathcal{H})$, the \emph{predual} of the Lie
algebra $(\mathcal{B}(\mathcal{H})^{\mathbb{R}},i[\cdot ,\cdot ])$. In \cite[%
Sections 2.1c]{Bono2000}, B\'{o}na proves the existence of a unique Poisson
bracket $\{\cdot ,\cdot \}$ on $\mathfrak{C}_{\mathcal{B}(\mathcal{H})^{%
\mathbb{R}}}^{\mathbb{R}}$, i.e., of a skew-symmetric biderivation
satisfying the Jacobi identity on polynomial functions, such that 
\begin{equation*}
\{\hat{A},\hat{B}\}\left( \sigma \right) =\mathrm{Tr}_{\mathcal{H}%
}(i[A,B]\sigma )=\widehat{i[A,B]}(\sigma ),\qquad A,B\in \mathcal{B}(%
\mathcal{H})^{\mathbb{R}},\ \sigma \in \mathcal{L}_{\mathbb{R}}^{1}(\mathcal{%
H})\ .
\end{equation*}%
It turns out that the Poisson manifold $(\mathcal{L}_{\mathbb{R}}^{1}(%
\mathcal{H}),\{\cdot ,\cdot \})$ has a non-trivial symplectic foliation: For
any $\sigma \in \mathcal{L}_{\mathbb{R}}^{1}(\mathcal{H})$, we define its
unitary orbit by%
\begin{equation}
\mathrm{O}(\sigma )\doteq \left\{ \mathrm{U}\sigma \mathrm{U}^{\ast }:%
\mathrm{U}\text{ a unitary operator on }\mathcal{H}\right\} \subseteq 
\mathcal{L}_{\mathbb{R}}^{1}(\mathcal{H})\ .  \label{follu0}
\end{equation}%
If $\sigma \in \mathcal{L}_{\mathbb{R}}^{1}(\mathcal{H})$ has
finite-dimensional range (i.e., $\dim \mathrm{ran}(\sigma )<\infty $), then $%
\mathrm{O}(\sigma )$ is a symplectic leaf of the Poisson manifold $(\mathcal{%
L}_{\mathbb{R}}^{1}(\mathcal{H}),\{\cdot ,\cdot \})$. In particular, the
restriction on such a leaf of the Poisson bracket of two functions $f,g$
only depends on the restriction of $f,g$ on the same leaf. Meanwhile, B\'{o}%
na observes in \cite[Lemma 2.1.7]{Bono2000} that the union%
\begin{equation*}
\bigcup \left\{ \mathrm{O}(\sigma ):\sigma \in \mathcal{L}_{\mathbb{R}}^{1}(%
\mathcal{H}),\ \sigma \geq 0,\ \mathrm{Tr}_{\mathcal{H}}(\sigma )=1,\ \dim 
\mathrm{ran}(\sigma )<\infty \right\}
\end{equation*}%
is dense in the set $\mathcal{S}_{\ast }$ of all normalized positive
elements (i.e., density matrices) of $\mathcal{L}_{\mathbb{R}}^{1}(\mathcal{H%
})$. Using this observation, B\'{o}na defines the Poisson bracket for
polynomial functions defined on $\mathcal{S}_{\ast }\subseteq \mathcal{L}_{%
\mathbb{R}}^{1}(\mathcal{H})$, but he proposes \cite[Sections 2.1c, footnote]%
{Bono2000} as a mathematically and physically interesting problem to
\textquotedblleft \textit{formulate analogies of }[his]\textit{\
constructions on the space of all positive normalized functionals on }$%
\mathcal{B}(\mathcal{H})$. \textit{This leads to technical complications}%
.\textquotedblright\ In Sections \ref{new section jundiai copy(1)} and \ref%
{new section jundiai copy(2)} we give such a construction for the dual space
of any $C^{\ast }$-algebra $\mathcal{X}$ (and not only for the special case $%
\mathcal{X=B}(\mathcal{H})$). Sections \ref{Convex Frechet Derivative}-\ref%
{Poisson Structure} contribute an alternative, more explicit, construction
of the same Poisson structure.

\begin{remark}
\mbox{
}\newline
The construction given in the recent paper \cite{Kryukov} for a Hamiltonian
flow associated with Schr\"{o}dinger's dynamics of one quantum particle
corresponds to B\'{o}na's symplectic leaf $\mathrm{O}(\sigma )$ of density
matrices $\sigma $ of dimension one, i.e., $\dim \mathrm{ran}(\sigma )=1$.
However, the author of \cite{Kryukov} does not seem to be aware of B\'{o}%
na's works.
\end{remark}

\subsection{Poisson Algebra of Polynomial Functions on the Continuous
Self-Adjoint Functionals on a $C^{\ast }$-Algebra\label{new section jundiai
copy(1)}}

Recall that $(\mathcal{X}^{\mathbb{R}},i[\cdot ,\cdot ])$ is a (possibly
infinite-dimensional) Lie algebra. See (\ref{commutator}). It is easy to
check that the continuous (real) linear functionals $\mathcal{X}^{\mathbb{R}%
}\rightarrow \mathbb{R}$ are in one-to-one correspondance to the hermitian
continuous (complex) linear functionals $\mathcal{X}\rightarrow \mathbb{C}$,
simply by restriction to $\mathcal{X}^{\mathbb{R}}\subseteq \mathcal{X}$.
Recall that a (complex) linear functional $\sigma :\mathcal{X}\rightarrow 
\mathbb{C}$ is, by definition, hermitian when 
\begin{equation*}
\sigma (A^{\ast })=\overline{\sigma (A)}\text{ },\text{\qquad }A\in \mathcal{%
X}\text{ }.
\end{equation*}%
We denote by $\mathcal{X}_{\mathbb{R}}^{\ast }$ the (real) space of all
hermitian elements of the (topological) dual space $\mathcal{X}^{\ast }$ and
use the identification 
\begin{equation*}
\mathcal{X}_{\mathbb{R}}^{\ast }\equiv (\mathcal{X}^{\mathbb{R}})^{\ast }\ ,
\end{equation*}%
as already done in the proof of Theorem \ref{theorem density2}. The space $%
\mathcal{X}_{\mathbb{R}}^{\ast }$ with $\mathcal{X=B}(\mathcal{H})$ plays in
our setting an analogous role as $\mathcal{L}_{\mathbb{R}}^{1}(\mathcal{H})$
in B\'{o}na's approach \cite[Sections 2.1b, 2.1c]{Bono2000}. See Section \ref%
{new section jundiai}.

Similar to (\ref{fA}), for any $A\in \mathcal{X}$, we define the weak$^{\ast
}$-continuous (complex) linear functional $\hat{A}:\mathcal{X}^{\ast
}\rightarrow \mathbb{C}$ by%
\begin{equation}
\hat{A}(\sigma )\doteq \sigma (A)\text{ },\text{\qquad }\sigma \in \mathcal{X%
}^{\ast }\ .  \label{sdfsdfkljsdlfkj}
\end{equation}%
(Note that we use the same notation as in (\ref{fA}), for the canonical
identification of $A\in \mathcal{X}$ with a linear functional on $\mathcal{X}%
^{\ast }$.) Any element of $\mathcal{X}^{\ast \ast }$ is of this form,
keeping in mind that the dual space $\mathcal{X}^{\ast }$ of $\mathcal{X}$
is here endowed with its weak$^{\ast }$ topology, see discussions before
Definition \ref{state space}. Note also that any weak$^{\ast }$-continuous
(real) linear functional on $\mathcal{X}_{\mathbb{R}}^{\ast }$ uniquely
extends to a weak$^{\ast }$-continuous (complex) linear hermitian functional
on $\mathcal{X}^{\ast }$. In this case, by hermiticity, the corresponding $%
A\in \mathcal{X}$ belongs to $\mathcal{X}^{\mathbb{R}}$. Conversely, any $%
A\in \mathcal{X}^{\mathbb{R}}$ defines a weak$^{\ast }$-continuous (real)
linear functional $\hat{A}:\mathcal{X}_{\mathbb{R}}^{\ast }\rightarrow 
\mathbb{C}$, by restriction of (\ref{sdfsdfkljsdlfkj}) to $\mathcal{X}_{%
\mathbb{R}}^{\ast }$. Therefore, we identify the real Banach space $\mathcal{%
X}^{\mathbb{R}}$ of self-adjoint elements of the $C^{\ast }$-algebra $%
\mathcal{X}$ with the space of all weak$^{\ast }$-continuous (real) linear
functionals $\mathcal{X}_{\mathbb{R}}^{\ast }\rightarrow \mathbb{R}$, i.e.,%
\begin{equation}
\mathcal{X}^{\mathbb{R}}\equiv (\mathcal{X}_{\mathbb{R}}^{\ast })^{\ast }\ .
\label{soigjosdfj}
\end{equation}%
In this view point, $\mathcal{X}^{\mathbb{R}}\subseteq C(\mathcal{X}_{%
\mathbb{R}}^{\ast };\mathbb{R})$. Let%
\begin{equation*}
\mathfrak{C}_{\mathcal{X}^{\mathbb{R}}}^{\mathbb{R}}\equiv \mathfrak{C}_{%
\mathcal{X}^{\mathbb{R}}}^{\mathbb{R}}\left( \mathcal{X}_{\mathbb{R}}^{\ast
}\right) \doteq \mathbb{R}[\mathcal{X}^{\mathbb{R}}]\subseteq C(\mathcal{X}_{%
\mathbb{R}}^{\ast };\mathbb{R})
\end{equation*}%
be the subalgebra of polynomials in the elements of $\mathcal{X}^{\mathbb{R}}
$, with real coefficients. (Compare with (\ref{def frac Cb}) for $\mathcal{B}%
=\mathcal{X}^{\mathbb{R}}$.) The elements of this subalgebra are again
called \textquotedblleft polynomial\textquotedblright\ functions on $%
\mathcal{X}_{\mathbb{R}}^{\ast }$, the \emph{dual} of the Lie algebra $(%
\mathcal{X}_{\mathbb{R}},i[\cdot ,\cdot ])$.

Note that such\ polynomials are Gateaux differentiable and, for any $f\in 
\mathfrak{C}_{\mathcal{X}^{\mathbb{R}}}^{\mathbb{R}}$ and any $\sigma \in 
\mathcal{X}_{\mathbb{R}}^{\ast }$, the Gateaux derivative $\mathrm{d}%
^{G}f\left( \sigma \right) $ is linear and weak$^{\ast }$ continuous, i.e., $%
\mathrm{d}^{G}f\left( \sigma \right) \in \mathcal{X}^{\mathbb{R}}$ (see (\ref%
{soigjosdfj})). In particular, for any $A\in \mathcal{X}$, by (\ref%
{sdfsdfkljsdlfkj}), 
\begin{equation}
\mathrm{d}^{G}\hat{A}\left( \sigma \right) =A\text{ },\text{\qquad }\sigma
\in \mathcal{X}_{\mathbb{R}}^{\ast }\ .  \label{equation super trivial}
\end{equation}%
Thus, we can define a skew-symmetric biderivation $\{\cdot ,\cdot \}_{0}$ on 
$\mathfrak{C}_{\mathcal{X}^{\mathbb{R}}}^{\mathbb{R}}$ as follows:

\begin{definition}[Poisson bracket]
\label{convex Frechet derivative copy(2)}\mbox{ }\newline
The skew-symmetric biderivation $\{\cdot ,\cdot \}_{0}$ on $\mathfrak{C}_{%
\mathcal{X}^{\mathbb{R}}}^{\mathbb{R}}$ is defined by%
\begin{equation*}
\left\{ f,g\right\} _{0}\left( \sigma \right) \doteq \sigma \left( i\left[ 
\mathrm{d}^{G}f\left( \sigma \right) ,\mathrm{d}^{G}g\left( \sigma \right) %
\right] \right) \ ,\qquad f,g\in \mathfrak{C}_{\mathcal{X}^{\mathbb{R}}}^{%
\mathbb{R}}\ .
\end{equation*}
\end{definition}

\noindent This skew-symmetric biderivation satisfies the Jacobi identity:

\begin{proposition}[Usual properties of Poisson brackets]
\label{lemma poisson copy(2)}\mbox{
}\newline
$\{\cdot ,\cdot \}_{0}$ is a Poisson bracket, i.e., it is a skew-symmetric
biderivation satisfying the Jacobi identity 
\begin{equation*}
\left\{ f,\left\{ g,h\right\} _{0}\right\} _{0}+\left\{ h,\left\{
f,g\right\} _{0}\right\} _{0}+\left\{ g,\left\{ h,f\right\} _{0}\right\}
_{0}=0\ ,\qquad f,g,h\in \mathfrak{C}_{\mathcal{X}^{\mathbb{R}}}^{\mathbb{R}%
}\doteq \mathbb{R}[\mathcal{X}^{\mathbb{R}}]\ .
\end{equation*}
\end{proposition}

\begin{proof}
$\{\cdot ,\cdot \}_{0}$ is clearly skew-symmetric, by (\ref{commutator}) and
Definition \ref{convex Frechet derivative copy(2)}. Note additionally that,
for any $f,g\in \mathfrak{C}_{\mathcal{X}^{\mathbb{R}}}^{\mathbb{R}}$,%
\begin{equation}
\mathrm{d}^{G}\left( f+g\right) =\mathrm{d}^{G}f+\mathrm{d}^{G}g\quad \text{%
and}\quad \mathrm{d}^{G}\left( fg\right) =f\mathrm{d}^{G}g+g\mathrm{d}^{G}f\
,  \label{Leipniz0}
\end{equation}%
where the products in the last equality are meant point-wise. As a
consequence, $\{\cdot ,\cdot \}_{0}$ is bilinear and satisfies Leibniz's
rule with respect to both arguments, by (\ref{commutator}). In other words, $%
\{\cdot ,\cdot \}_{0}$ is a skew-symmetric biderivation. Finally, by
bilinearity, it suffices to prove the Jacobi identity for $f,g,h$ being
monomials in the elements of $\mathcal{X}^{\mathbb{R}}$. If the\ sum of the
degree of the three monomials is $0$, $1$, or $2$, then the Jacobi identity
follows trivially. If the sum is exactly $3$ then the Jacobi identity
follows from the corresponding one for the commutators (\ref{commutator}).
(If one of the three monomials has zero degree then all terms in the Jacobi
identity trivially vanish.) If the sum is bigger than $3$ then at least one
of the monomial has degree bigger than $1$. Assume, without loss of
generality, that this monomial is $f$. Then $f=f_{1}f_{2}$ where\ the
monomials $f_{1}$ and $f_{2}$ have degree at least $1$, and, explicit
computations using Leibniz's rule and the skew-symmetry yield 
\begin{multline*}
\left\{ f,\left\{ g,h\right\} _{0}\right\} _{0}+\left\{ h,\left\{
f,g\right\} _{0}\right\} _{0}+\left\{ g,\left\{ h,f\right\} _{0}\right\}
_{0}=f_{1}\left( \left\{ f_{2},\left\{ g,h\right\} _{0}\right\} _{0}+\left\{
h,\left\{ f_{2},g\right\} _{0}\right\} _{0}+\left\{ g,\left\{
h,f_{2}\right\} _{0}\right\} _{0}\right) \\
+f_{2}\left( \left\{ f_{1},\left\{ g,h\right\} _{0}\right\} _{0}+\left\{
h,\left\{ f_{1},g\right\} _{0}\right\} _{0}+\left\{ g,\left\{
h,f_{1}\right\} _{0}\right\} _{0}\right) \ .
\end{multline*}%
Since $f_{1}$ and $f_{2}$ have in this case degree stricly smaller than the
degree of $f$, the Jacobi identity follows by induction.
\end{proof}

\begin{corollary}[Poisson algebra]
\mbox{ }\newline
The subspace $\mathfrak{C}_{\mathcal{X}^{\mathbb{R}}}^{\mathbb{R}}$ of
polynomials in the elements of $\mathcal{X}^{\mathbb{R}}\subseteq C(\mathcal{%
X}_{\mathbb{R}}^{\ast };\mathbb{R})$ with real coefficients, endowed with $%
\{\cdot ,\cdot \}_{0}$ and the pointwise multiplications of $\mathfrak{C}_{%
\mathcal{X}^{\mathbb{R}}}^{\mathbb{R}}$, is a Poisson algebra, in the sense
of \cite[Definition 1.1]{Poission}.
\end{corollary}

\subsection{Poisson Ideals Associated with State and Phase Spaces\label{new
section jundiai copy(2)}}

Let $F\subseteq \mathcal{X}_{\mathbb{R}}^{\ast }$ be any nonempty subset of $%
\mathcal{X}_{\mathbb{R}}^{\ast }$ and define the algebra 
\begin{equation}
\mathfrak{C}_{\mathcal{X}^{\mathbb{R}}}^{\mathbb{R}}\left( F\right) \doteq
\left\{ f|_{F}:f\in \mathfrak{C}_{\mathcal{X}^{\mathbb{R}}}^{\mathbb{R}%
}\left( \mathcal{X}_{\mathbb{R}}^{\ast }\right) \right\}  \label{affines0}
\end{equation}%
of polynomials on $F$. \ If the restriction to $F$ of the Poisson bracket $%
\{f,g\}_{0}$ of two polynomials $f,g\in \mathfrak{C}_{\mathcal{X}^{\mathbb{R}%
}}^{\mathbb{R}}\left( \mathcal{X}_{\mathbb{R}}^{\ast }\right) $ (Definition %
\ref{convex Frechet derivative copy(2)}) only depends on the corresponding
restrictions of $f,g$, then%
\begin{equation*}
\{f|_{F},g|_{F}\}\doteq \{f,g\}_{0}|_{F},\qquad f,g\in \mathfrak{C}_{%
\mathcal{X}^{\mathbb{R}}}^{\mathbb{R}}\left( \mathcal{X}_{\mathbb{R}}^{\ast
}\right) \ ,
\end{equation*}%
is a well-defined Poisson bracket on $\mathfrak{C}_{\mathcal{X}^{\mathbb{R}%
}}^{\mathbb{R}}\left( F\right) $. Equivalently, this means that the
subalgebra 
\begin{equation*}
\mathfrak{I}_{F}\doteq \left\{ f\in \mathfrak{C}_{\mathcal{X}^{\mathbb{R}}}^{%
\mathbb{R}}:f\left( F\right) =\left\{ 0\right\} \right\}
\end{equation*}%
of polynomials that vanish on $F\subseteq \mathcal{X}_{\mathbb{R}}^{\ast }$
is a Poisson ideal of the Poisson algebra $(\mathfrak{C}_{\mathcal{X}^{%
\mathbb{R}}}^{\mathbb{R}},\{\cdot ,\cdot \}_{0})$. Recall that a subalgebra $%
\mathfrak{I}$ of a Poisson algebra $(\mathcal{P},\{\cdot ,\cdot \})$ is
called a Poisson ideal whenever, for all $f\in \mathfrak{I}$ and $g\in 
\mathcal{P}$, $fg\in \mathfrak{I}$ and $\{f,g\}\in \mathfrak{I}$. See, e.g., 
\cite[Section 2.2.1]{Poission}. As a consequence of this fact, the Poisson
algebras 
\begin{equation*}
(\mathfrak{C}_{\mathcal{X}^{\mathbb{R}}}^{\mathbb{R}}\left( F\right)
,\{\cdot ,\cdot \})\qquad \text{and}\qquad (\mathfrak{C}_{\mathcal{X}^{%
\mathbb{R}}}^{\mathbb{R}}\left( \mathcal{X}_{\mathbb{R}}^{\ast }\right)
,\{\cdot ,\cdot \}_{0})/\mathfrak{I}_{F}
\end{equation*}%
are isomorphic. See \cite[Section 2.2.1]{Poission} for the definition of the
quotient of a Poisson algebra by one of its Poisson ideals. See also \cite[%
Proposition 2.8]{Poission}.

For any state $\rho \in E$, we apply these observations to the follium $%
E_{\rho }$ of states, defined by 
\begin{equation*}
E_{\rho }\doteq \left\{ \left\langle \varphi ,\pi _{\rho }\left( \cdot
\right) \varphi \right\rangle _{\mathcal{H}_{\rho }}:\varphi \in \mathcal{H}%
_{\rho },\ \left\Vert \varphi \right\Vert _{\mathcal{H}_{\rho }}=1\right\}
\subseteq E\subseteq \mathcal{X}_{\mathbb{R}}^{\ast }\ ,
\end{equation*}%
where the triplet $(\mathcal{H}_{\rho },\pi _{\rho },\Omega _{\rho })$ is
the GNS representation \cite[Section 2.3.3]{BrattelliRobinsonI} of $\rho $.

\begin{proposition}[Folia of states and Poisson ideals]
\label{proposition sympatoch}\mbox{ }\newline
For any $\rho \in E$ and any $f,g\in \mathfrak{C}_{\mathcal{X}^{\mathbb{R}%
}}^{\mathbb{R}}\left( \mathcal{X}_{\mathbb{R}}^{\ast }\right) $, the
restriction $\{f,g\}_{0}|_{E_{\rho }}$ only depends on the corresponding
restriction of $f,g$. In particular, 
\begin{equation*}
\{f|_{E_{\rho }},g|_{E_{\rho }}\}\doteq \{f,g\}_{0}|_{E_{\rho }},\qquad
f,g\in \mathfrak{C}_{\mathcal{X}^{\mathbb{R}}}^{\mathbb{R}}\left( \mathcal{X}%
_{\mathbb{R}}^{\ast }\right) \ ,
\end{equation*}%
is a well-defined Poisson bracket on $\mathfrak{C}_{\mathcal{X}^{\mathbb{R}%
}}^{\mathbb{R}}\left( E_{\rho }\right) $.
\end{proposition}

\begin{proof}
For any state $\rho \in E$ with GNS representation $(\mathcal{H}_{\rho },\pi
_{\rho },\Omega _{\rho })$, we define the unit sphere 
\begin{equation*}
\mathrm{S}_{\rho }\doteq \left\{ \varphi \in \mathcal{H}_{\rho }:\left\Vert
\varphi \right\Vert _{\mathcal{H}_{\rho }}=1\right\} \ .
\end{equation*}%
For any $f\in \mathfrak{C}_{\mathcal{X}^{\mathbb{R}}}^{\mathbb{R}}\left( 
\mathcal{X}_{\mathbb{R}}^{\ast }\right) $, we define the continuous function 
$f_{\rho }\in C(\mathrm{S}_{\rho };\mathbb{R})$ by%
\begin{equation*}
f_{\rho }\left( \varphi \right) \doteq f(\left\langle \varphi ,\pi _{\rho
}\left( \cdot \right) \varphi \right\rangle _{\mathcal{H}_{\rho }})\ ,\qquad
\varphi \in \mathrm{S}_{\rho }\ .
\end{equation*}%
Let%
\begin{equation*}
\mathfrak{C}^{(\rho )}\doteq \left\{ f_{\rho }:f\in \mathfrak{C}_{\mathcal{X}%
^{\mathbb{R}}}^{\mathbb{R}}\left( \mathcal{X}_{\mathbb{R}}^{\ast }\right)
\right\} \subseteq C(\mathrm{S}_{\rho };\mathbb{R})\ .
\end{equation*}%
Then, we prove the existence of a skew-symmetric biderivation $\{\cdot
,\cdot \}^{(\rho )}$ on $\mathfrak{C}^{(\rho )}$ satisfying 
\begin{equation}
\{\hat{A}_{\rho },\hat{B}_{\rho }\}^{(\rho )}=(\{\hat{A},\hat{B}%
\}_{0})_{\rho }\ ,\qquad A,B\in \mathcal{X}^{\mathbb{R}}\ .
\label{holds true0}
\end{equation}%
This last equality yields 
\begin{equation*}
\{f_{\rho },g_{\rho }\}^{(\rho )}=\left( \{f,g\}_{0}\right) _{\rho }\
,\qquad f,g\in \mathfrak{C}_{\mathcal{X}^{\mathbb{R}}}^{\mathbb{R}}\left( 
\mathcal{X}_{\mathbb{R}}^{\ast }\right) \ ,
\end{equation*}%
by linearity and Leibniz's rule. In particular, for any $\varphi \in \mathrm{%
S}_{\rho }$, 
\begin{equation*}
\{f,g\}_{0}(\left\langle \varphi ,\pi _{\rho }\left( \cdot \right) \varphi
\right\rangle _{\mathcal{H}_{\rho }})=\{f_{\rho },g_{\rho }\}^{(\rho )}\text{
}.
\end{equation*}%
As $f_{\rho },g_{\rho }$ only depend on the restrictions $f|_{E_{\rho }}$
and $g|_{E_{\rho }}$, respectively, the assertion follows.

Now, in order to prove the existence of a skew-symmetric biderivation $%
\{\cdot ,\cdot \}^{(\rho )}$ satisfying (\ref{holds true0}), let $\mathcal{L}%
_{\mathbb{R}}^{1}(\mathcal{H}_{\rho })$ be the real Banach space of all
self-adjoint trace-class operators on $\mathcal{H}_{\rho }$. For any $f\in 
\mathfrak{C}^{(\rho )}$ and $\varphi \in \mathrm{S}_{\rho }$, we denote by $%
\mathrm{d}_{\rho }^{G}f\left( \varphi \right) $ the Gateaux derivative at $%
A=0$ of the map 
\begin{equation*}
A\mapsto f\left( \mathrm{e}^{iA}\varphi \right)
\end{equation*}%
from $\mathcal{L}_{\mathbb{R}}^{1}(\mathcal{H}_{\rho })$ to $\mathbb{R}$.
For any $f\in \mathfrak{C}^{(\rho )}$, this Gateaux derivative is linear and
continuous, i.e., $\mathrm{d}_{\rho }^{G}f\left( \varphi \right) \in 
\mathcal{B}(\mathcal{H}_{\rho })^{\mathbb{R}}$. See, e.g., (\ref{dual}).
Therefore, we can define a skew-symmetric biderivation $\{\cdot ,\cdot
\}^{(\rho )}$ on $\mathfrak{C}_{\rho }$ by 
\begin{equation*}
\{f,g\}^{(\rho )}\left( \varphi \right) =\left\langle \varphi ,i\left[ 
\mathrm{d}_{\rho }^{G}f\left( \varphi \right) ,\mathrm{d}_{\rho }^{G}g\left(
\varphi \right) \right] \varphi \right\rangle _{\mathcal{H}_{\rho }}\
,\qquad f,g\in \mathfrak{C}^{(\rho )}\ .
\end{equation*}%
For any $A\in \mathcal{X}^{\mathbb{R}}$ and $\varphi \in \mathrm{S}_{\rho }$%
, observe that 
\begin{equation*}
\mathrm{d}_{\rho }^{G}\hat{A}_{\rho }\left( \varphi \right) \left( B\right)
=i\left\langle \varphi ,\left[ \pi _{\rho }\left( A\right) ,B\right] \varphi
\right\rangle _{\mathcal{H}_{\rho }}=i\mathrm{Tr}_{\mathcal{H}_{\rho }}(%
\left[ P_{\varphi },\pi _{\rho }\left( A\right) \right] B)\ ,
\end{equation*}%
where $P_{\varphi }$ is the orthogonal projection whose range is $\mathbb{C}%
\varphi $. In other words, 
\begin{equation*}
\mathrm{d}_{\rho }^{G}\hat{A}_{\rho }\left( \varphi \right) =i\left[
P_{\varphi },\pi _{\rho }\left( A\right) \right] \in \mathcal{B}(\mathcal{H}%
_{\rho })\ ,\qquad \varphi \in \mathrm{S}_{\rho },\ A\in \mathcal{X}^{%
\mathbb{R}}\ .
\end{equation*}%
Since $\pi _{\rho }:\mathcal{X\rightarrow B}(\mathcal{H}_{\rho })$ is a $%
\ast $-homomorphism, by Equation (\ref{equation super trivial}) and
Definition \ref{convex Frechet derivative copy(2)}, it follows that, for any 
$A,B\in \mathcal{X}^{\mathbb{R}}$,%
\begin{equation*}
\{\hat{A}_{\rho },\hat{B}_{\rho }\}^{(\rho )}\left( \varphi \right)
=i\left\langle \varphi ,\pi _{\rho }\left( \left[ A,B\right] \right) \varphi
\right\rangle _{\mathcal{H}_{\rho }}=(\{\hat{A},\hat{B}\}_{0})_{\rho }\left(
\varphi \right) \ ,\qquad \varphi \in \mathrm{S}_{\rho }\ ,
\end{equation*}%
i.e., Equation (\ref{holds true0}) holds true.
\end{proof}

The folia $E_{\rho }$, $\rho \in E$, play here an analogous role as the
symplectic leaves $\mathrm{O}(\sigma )$ (\ref{follu0}) of the Poisson
manifold $(\mathcal{L}_{\mathbb{R}}^{1}(\mathcal{H}),\{\cdot ,\cdot \})$ in B%
\'{o}na's approach \cite[Sections 2.1b, 2.1c]{Bono2000}. See Section \ref%
{new section jundiai}.

\begin{corollary}[State space and Poisson ideals]
\label{proposition sympatoch copy(1)}\mbox{ }\newline
For any $f,g\in \mathfrak{C}_{\mathcal{X}^{\mathbb{R}}}^{\mathbb{R}}\left( 
\mathcal{X}_{\mathbb{R}}^{\ast }\right) $, the restriction $\{f,g\}_{0}|_{E}$
only depends on the corresponding restriction of $f,g$. In particular, 
\begin{equation*}
\{f|_{E},g|_{E}\}\doteq \{f,g\}_{0}|_{E},\qquad f,g\in \mathfrak{C}_{%
\mathcal{X}^{\mathbb{R}}}^{\mathbb{R}}\left( \mathcal{X}_{\mathbb{R}}^{\ast
}\right) \ ,
\end{equation*}%
is a well-defined Poisson bracket on $\mathfrak{C}_{\mathcal{X}^{\mathbb{R}%
}}^{\mathbb{R}}\left( E\right) \subseteq \mathfrak{C}^{\mathbb{R}}\doteq C(E;%
\mathbb{R})$.
\end{corollary}

\begin{proof}
The assertion is a direct consequence of Proposition \ref{proposition
sympatoch} together with the obvious equality%
\begin{equation*}
E=\bigcup \left\{ E_{\rho }:\rho \in E\right\} \ .
\end{equation*}%
Equivalently, use that 
\begin{equation*}
\mathfrak{I}_{E}=\bigcap \left\{ \mathfrak{I}_{E_{\rho }}:\rho \in E\right\} 
\text{ },
\end{equation*}%
is a Poisson ideal of the Poisson algebra $(\mathfrak{C}_{\mathcal{X}^{%
\mathbb{R}}}^{\mathbb{R}},\{\cdot ,\cdot \}_{0})$.
\end{proof}

\noindent In Sections \ref{Convex Frechet Derivative}-\ref{Poisson Structure}%
, we also present an explicit construction of the Poisson bracket of
Corollary \ref{proposition sympatoch copy(1)}, because it is\ technically
more convenient for the subsequent sections.

Finally, recall that the phase space is the weak$^{\ast }$ closure $%
\overline{\mathcal{E}(E)}$ of the set $\mathcal{E}(E)$ of extreme points of
the state space $E$, see Definition \ref{phase space}. Similar to Corollary %
\ref{proposition sympatoch copy(1)}, we prove from Proposition \ref%
{proposition sympatoch} the existence of a Poisson bracket for polynomials
acting on the phase space:

\begin{corollary}[Phase space and Poisson ideals]
\label{proposition sympatoch copy(2)}\mbox{ }\newline
For any $f,g\in \mathfrak{C}_{\mathcal{X}^{\mathbb{R}}}^{\mathbb{R}}\left( 
\mathcal{X}_{\mathbb{R}}^{\ast }\right) $, the restriction $\{f,g\}_{0}|_{%
\overline{\mathcal{E}(E)}}$ only depends on the corresponding restriction of 
$f,g$. In particular, 
\begin{equation*}
\{f|_{\overline{\mathcal{E}(E)}},g|_{\overline{\mathcal{E}(E)}}\}\doteq
\{f,g\}_{0}|_{\overline{\mathcal{E}(E)}},\qquad f,g\in \mathfrak{C}_{%
\mathcal{X}^{\mathbb{R}}}^{\mathbb{R}}\left( \mathcal{X}_{\mathbb{R}}^{\ast
}\right) \ ,
\end{equation*}%
is a well-defined Poisson bracket on $\mathfrak{C}_{\mathcal{X}^{\mathbb{R}%
}}^{\mathbb{R}}(\overline{\mathcal{E}(E)})\subseteq C(\overline{\mathcal{E}%
(E)};\mathbb{R})$.
\end{corollary}

\begin{proof}
For any extreme (or pure) state $\rho \in \mathcal{E}(E)$, we infer from 
\cite[Proposition 2.2.4]{Landsmanlivre0} that the folium $E_{\rho }\subseteq 
\mathcal{E}(E)$ is a subset of extreme states and, hence, 
\begin{equation*}
\mathcal{E}(E)=\bigcup \left\{ E_{\rho }:\rho \in \mathcal{E}(E)\right\} \ .
\end{equation*}%
By Proposition \ref{proposition sympatoch} and continuity of polynomials, it
follows that 
\begin{equation*}
\mathfrak{I}_{\overline{\mathcal{E}(E)}}=\mathfrak{I}_{\mathcal{E}%
(E)}=\bigcap \left\{ \mathfrak{I}_{E_{\rho }}:\rho \in \mathcal{E}(E)\right\}
\end{equation*}%
is again a Poisson ideal of the Poisson algebra $(\mathfrak{C}_{\mathcal{X}^{%
\mathbb{R}}}^{\mathbb{R}},\{\cdot ,\cdot \}_{0})$.
\end{proof}

\subsection{Convex Weak$^{\ast }$ Gateaux Derivative\label{Convex Frechet
Derivative}}

In order to construct the Poisson bracket $\{\cdot ,\cdot \}$ of Corollary %
\ref{proposition sympatoch copy(1)} more explicitly, as well as to analyze
its properties as generator of (generally non-autonomous) classical
dynamics, we introduce the notion of \emph{convex} Gateaux derivative on the
space $C(E;\mathcal{Y})$ of weak$^{\ast }$-continuous functions on the
convex and weak$^{\ast }$-compact set $E$ of states with values in an
arbitrary Banach space 
\begin{equation*}
\mathcal{Y}\equiv \left( \mathcal{Y},+,\cdot _{{\mathbb{K}}},\left\Vert
\cdot \right\Vert _{\mathcal{Y}}\right) \ ,\qquad {\mathbb{K}}=\mathbb{R},%
\mathbb{C}\ .
\end{equation*}%
As far as only the construction of the Poisson bracket $\{\cdot ,\cdot \}$
of Corollary \ref{proposition sympatoch copy(1)} is concerned, the relevant
example is $\mathcal{Y}=\mathbb{R}={\mathbb{K}}$.

We first define the Banach space 
\begin{equation*}
\mathcal{A}\left( E;\mathcal{Y}\right) \doteq \left\{ f\in C\left( E;%
\mathcal{Y}\right) :\forall \lambda \in \left( 0,1\right) ,\ \rho ,\upsilon
\in E,\quad f\left( \left( 1-\lambda \right) \rho +\lambda \upsilon \right)
=\left( 1-\lambda \right) f\left( \rho \right) +\lambda f\left( \upsilon
\right) \right\}
\end{equation*}%
of all affine weak$^{\ast }$-continuous $\mathcal{Y}$-valued functions on $E$%
, endowed with the supremum norm%
\begin{equation}
\left\Vert f\right\Vert _{\mathcal{A}(E;\mathcal{Y})}\doteq \max_{\rho \in
E}\left\Vert f\left( \rho \right) \right\Vert _{\mathcal{Y}}\ ,\qquad f\in 
\mathcal{A}\left( E;\mathcal{Y}\right) \ .  \label{norm affine}
\end{equation}%
Again, the norm is \emph{not} used in the contruction of the Poisson bracket 
$\{\cdot ,\cdot \}$ of Corollary \ref{proposition sympatoch copy(1)}, but
only in Section \ref{Well-posedness sect copy(1)}.

The convex Gateaux derivative of a weak$^{\ast }$-continuous $\mathcal{Y}$%
-valued function on $E$ at a fixed state is an affine weak$^{\ast }$%
-continuous $\mathcal{Y}$-valued function on $E$ defined as follows:

\begin{definition}[Convex weak$^{\ast }$-continuous Gateaux derivative]
\label{convex Frechet derivative}\mbox{ }\newline
For any continuous function $f\in C(E;\mathcal{Y})$ and any state $\rho \in
E $, we say that $\mathrm{d}f\left( \rho \right) :E\rightarrow \mathcal{Y}$
is the (unique) convex weak$^{\ast }$-continuous Gateaux derivative of $f$
at $\rho \in E$ if $\mathrm{d}f\left( \rho \right) \in \mathcal{A}(E;%
\mathcal{Y}) $ and%
\begin{equation*}
\lim_{\lambda \rightarrow 0^{+}}\lambda ^{-1}\left( f\left( \left( 1-\lambda
\right) \rho +\lambda \upsilon \right) -f\left( \rho \right) \right) =\left[ 
\mathrm{d}f\left( \rho \right) \right] \left( \upsilon \right) \ ,\qquad
\rho ,\upsilon \in E\ .
\end{equation*}
\end{definition}

\noindent To our knowledge, the concept of convex weak$^{\ast }$-continuous
Gateaux derivative defined above is new.

A function $f\in C(E;\mathcal{Y})$ such that $\mathrm{d}f\left( \rho \right) 
$ exists for all $\rho \in E$ is called \emph{differentiable} and we use the
notation 
\begin{equation*}
\mathrm{d}f\equiv (\mathrm{d}f\left( \rho \right) )_{\rho \in
E}:E\rightarrow \mathcal{A}(E;\mathcal{Y})\ .
\end{equation*}%
Explicit examples of spaces of such differentiable functions are given, for
any $n\in \mathbb{N}$, by%
\begin{eqnarray}
\mathfrak{Y}_{n} &\equiv &\mathfrak{Y}\left( \mathcal{Y}\right) _{n}\doteq 
%TCIMACRO{\TeXButton{\Big\{}{\Big\{}}%
%BeginExpansion
\Big\{%
%EndExpansion
f\in C\left( E,\mathcal{Y}\right) :\exists \left\{ B_{j}\right\}
_{j=1}^{n}\subseteq \mathcal{X}^{\mathbb{R}},\ g\in C^{1}\left( \mathbb{R}%
^{n},\mathcal{Y}\right)  \notag \\
&&\left. \qquad \qquad \qquad \qquad \qquad \qquad \text{such that }f\left(
\rho \right) =g\left( \rho \left( B_{1}\right) ,\ldots ,\rho \left(
B_{n}\right) \right) \right. 
%TCIMACRO{\TeXButton{\Big\}}{\Big\}}}%
%BeginExpansion
\Big\}%
%EndExpansion
\ .  \label{Yn}
\end{eqnarray}%
Functions of this kind are said to be cylindrical. In fact, for any $n\in 
\mathbb{N}$ and $f\in \mathfrak{Y}_{n}$, 
\begin{equation}
\left[ \mathrm{d}f\left( \rho \right) \right] \left( \upsilon \right)
=\sum_{j=1}^{n}\left( \upsilon \left( B_{j}\right) -\rho \left( B_{j}\right)
\right) \partial _{x_{j}}g\left( \rho \left( B_{1}\right) ,\ldots ,\rho
\left( B_{n}\right) \right) ,\qquad \rho ,\upsilon \in E\ .  \label{df}
\end{equation}

We define the subspace of continuously differentiable $\mathcal{Y}$-valued
functions on the convex and weak$^{\ast }$-compact set $E$ by 
\begin{equation}
\mathfrak{Y}\equiv \mathfrak{Y}\left( \mathcal{Y}\right) \doteq C^{1}\left(
E;\mathcal{Y}\right) \doteq \left\{ f\in C\left( E;\mathcal{Y}\right) :%
\mathrm{d}f\in C\left( E;\mathcal{A}\left( E;\mathcal{Y}\right) \right)
\right\} \ .  \label{C1}
\end{equation}%
We endow this vector space with the norm%
\begin{equation}
\left\Vert f\right\Vert _{\mathfrak{Y}}\doteq \max_{\rho \in E}\left\Vert
f\left( \rho \right) \right\Vert _{\mathcal{Y}}+\max_{\rho \in E}\left\Vert 
\mathrm{d}f\left( \rho \right) \right\Vert _{\mathcal{A}(E;\mathcal{Y})}\
,\qquad f\in \mathfrak{Y}\ ,  \label{C1bis}
\end{equation}%
in order to obtain a Banach space, also denoted by $\mathfrak{Y}$. Note
again that we use \textquotedblleft $\max $\textquotedblright\ instead
\textquotedblleft $\sup $\textquotedblright\ in the definition of the norm,
because of the continuity of $f$ and $\mathrm{d}f$ together with the weak$%
^{\ast }$ compactness of $E$. Observe that%
\begin{equation*}
\partial _{\lambda }^{+}f_{\rho ,\upsilon }\left( \lambda \right) \doteq
\lim_{\varepsilon \rightarrow 0^{+}}\varepsilon ^{-1}\left( f_{\rho
,\upsilon }\left( \lambda +\varepsilon \right) -f_{\rho ,\upsilon }\left(
\lambda \right) \right) =\frac{1}{\left( 1-\lambda \right) }\left[ \mathrm{d}%
f\left( \left( 1-\lambda \right) \rho +\lambda \upsilon \right) \right]
\left( \upsilon \right)
\end{equation*}%
with $f_{\rho ,\upsilon }$ being the $\mathcal{Y}$-valued function on the
interval $[0,1)$ defined, at fixed states $\rho ,\upsilon \in E$, by%
\begin{equation*}
f_{\rho ,\upsilon }\left( \lambda \right) \doteq f\left( \left( 1-\lambda
\right) \rho +\lambda \upsilon \right) \ ,\qquad \lambda \in \lbrack 0,1)\ .
\end{equation*}%
Moreover, the operator $\partial _{\lambda }^{+}$ is closed on $C([a,b);%
\mathcal{Y})$ for any real parameters $a<b$, in the sense of the supremum
norm. Thus, by well-known properties of the uniform convergence of
continuous functions, the normed vector space $\mathfrak{Y}$ is complete.

Remark that the family $\{\mathfrak{Y}_{n}\}_{n\in \mathbb{N}}$ is
increasing with respect to inclusion and 
\begin{equation*}
\mathfrak{Y}_{\infty }\doteq \bigcup\limits_{n\in \mathbb{N}}\mathfrak{Y}%
_{n}\subseteq \mathfrak{Y}
\end{equation*}%
is the space of all cylindrical functions of $\mathfrak{Y}$. Additionally,
if $f\in \mathcal{A}(E;\mathcal{Y})$ then 
\begin{equation}
\mathrm{d}f\left( \rho \right) =f-f\left( \rho \right) \ ,\qquad \rho \in E\
,  \label{clear}
\end{equation}%
which means in particular that affine weak$^{\ast }$-continuous $\mathcal{Y}$%
-valued functions on $E$ are continuously differentiable, i.e., $\mathcal{A}%
(E;\mathcal{Y})\subseteq \mathfrak{Y}$.

\subsection{Explicit Construction of Poisson Brackets for Functions on the
State Space\label{Poisson Structure}}

We use the convex weak$^{\ast }$ Gateaux derivative in order to give an
explicit expression for the Poisson bracket $\{\cdot ,\cdot \}$ of Corollary %
\ref{proposition sympatoch copy(1)}. To this end, we only need the special
case $\mathcal{Y}=\mathbb{R}$ in Definition \ref{convex Frechet derivative}.
We also exploit the following result:

\begin{proposition}[Affine weak$^{\ast }$--continuous real-valued functions
over $E$]
\label{affine real valued functions}\mbox{
}\newline
For any unital $C^{\ast }$-algebra $\mathcal{X}$, $\mathcal{A}(E;\mathbb{R}%
)=\{\hat{A}:A\in \mathcal{X}^{\mathbb{R}}\}$, where $A\mapsto \hat{A}$ is
the linear isometry from $\mathcal{X}^{\mathbb{R}}$ to $\mathfrak{C}^{%
\mathbb{R}}$ defined by (\ref{fA}). In particular, by (\ref{affines0}), $%
\mathfrak{C}_{\mathcal{X}^{\mathbb{R}}}^{\mathbb{R}}\left( E\right) =\mathbb{%
R}[\mathcal{A}(E;\mathbb{R})]\subseteq \mathfrak{C}^{\mathbb{R}}\doteq C(E;%
\mathbb{R})$.
\end{proposition}

\begin{proof}
This statement is asserted without proof or references in \cite[p 339]%
{BrattelliRobinsonI}. A proof is only shortly sketched in \cite[p 161]%
{Takesaki-I} and we thus give it here for completeness and reader's
convenience. It is based on preliminary results of convex analysis together
with general properties of $C^{\ast }$-algebras: Clearly, $\{\hat{A}:A\in 
\mathcal{X}^{\mathbb{R}}\}\subseteq \mathcal{A}(E;\mathbb{R})$. Conversely,
fix $f\in \mathcal{A}(E;\mathbb{R})$. Since $E$ is a weak$^{\ast }$-compact
subset of $\mathcal{X}_{\mathbb{R}}^{\ast }$, we deduce from \cite[Corollary
6.3]{Takesaki-I} the existence of an increasing sequence $\{f_{n}\}_{n\in 
\mathbb{N}}$ of affine weak$^{\ast }$-continuous real-valued functions on $%
\mathcal{X}_{\mathbb{R}}^{\ast }$ that uniformly converges to $f$, as $%
n\rightarrow \infty $. Meanwhile, observe that any affine weak$^{\ast }$%
-continuous real-valued functions $g$ on $\mathcal{X}_{\mathbb{R}}^{\ast }$
is of the form 
\begin{equation*}
g\left( \sigma \right) =\sigma \left( A\right) +g\left( 0\right) \ ,\qquad
\sigma \in \mathcal{X}_{\mathbb{R}}^{\ast }\ ,
\end{equation*}%
for some self-adjoint element $A\in \mathcal{X}^{\mathbb{R}}$, because the
weak$^{\ast }$-continuous real-valued function $g-g(0)$ on $\mathcal{X}_{%
\mathbb{R}}^{\ast }$ is linear. We thus deduce the existence of a sequence $%
\{A_{n}\}_{n\in \mathbb{N}}\subseteq \mathcal{X}^{\mathbb{R}}$ such that 
\begin{equation*}
f_{n}\left( \sigma \right) =\sigma \left( A_{n}\right) +f_{n}\left( 0\right)
\ ,\qquad \sigma \in \mathcal{X}_{\mathbb{R}}^{\ast }\ .
\end{equation*}%
Since $\rho \left( \mathfrak{1}\right) =1$ for $\rho \in E$, by (\ref{norm
properties}), the uniform convergence of $\{f_{n}\}_{n\in \mathbb{N}}$ to $f$
on $E$ yields that $\{A_{n}+f_{n}\left( 0\right) \mathfrak{1}\}_{n\in 
\mathbb{N}}\subseteq \mathcal{X}^{\mathbb{R}}$ is a Cauchy sequence, which
thus converges to some $A\in \mathcal{X}^{\mathbb{R}}$, as $n\rightarrow
\infty $. It follows that $f=\hat{A}$.
\end{proof}

Recall that $A\neq B$ yields $\hat{A}\neq \hat{B}$ for any $A,B\in \mathcal{X%
}$. Therefore, by (\ref{norm properties}), (\ref{norm affine}), (\ref{C1})
and Proposition \ref{affine real valued functions}, for any continuously
differentiable real-valued function $f\in \mathfrak{Y}\left( \mathbb{R}%
\right) \varsubsetneq \mathfrak{C}$ there is a unique $\mathrm{D}f\in C(E;%
\mathcal{X}^{\mathbb{R}})$ such that 
\begin{equation}
\mathrm{d}f\left( \rho \right) =\widehat{\mathrm{D}f\left( \rho \right) }\
,\qquad \rho \in E\ .  \label{clear2}
\end{equation}%
For instance, one infers from (\ref{df}) that, for any $n\in \mathbb{N}$ and 
$f\in \mathfrak{Y}\left( \mathbb{R}\right) _{n}$, 
\begin{equation}
\mathrm{D}f\left( \rho \right) =\sum_{j=1}^{n}\left( A_{j}-\rho \left(
A_{j}\right) \mathfrak{1}\right) \partial _{x_{j}}g\left( \rho \left(
A_{1}\right) ,\ldots ,\rho \left( A_{n}\right) \right) ,\qquad \rho \in E\ .
\label{Ynbis}
\end{equation}%
By (\ref{norm properties}) and (\ref{norm affine}), note that 
\begin{equation}
\left\Vert \mathrm{D}f\left( \rho \right) \right\Vert _{\mathcal{X}%
}=\left\Vert \mathrm{d}f\left( \rho \right) \right\Vert _{\mathcal{A}(E;%
\mathbb{R})}\ ,\qquad \rho \in E\ .  \label{norm affine2}
\end{equation}%
Therefore, we can define a skew-symmetric biderivation on $\mathfrak{Y}(%
\mathbb{R})$ for continuously differentiable real-valued functions depending
on the state space:

\begin{definition}[Skew-symmetric biderivation on $\mathfrak{Y}(\mathbb{R})$]

\label{convex Frechet derivative copy(1)}\mbox{ }\newline
We define the map $\{\cdot ,\cdot \}:\mathfrak{Y}(\mathbb{R})\times 
\mathfrak{Y}(\mathbb{R})\rightarrow C(E;\mathbb{R})$ by%
\begin{equation*}
\left\{ f,g\right\} \left( \rho \right) \doteq \rho \left( i\left[ \mathrm{D}%
f\left( \rho \right) ,\mathrm{D}g\left( \rho \right) \right] \right) \
,\qquad f,g\in \mathfrak{Y}\left( \mathbb{R}\right) \ .
\end{equation*}
\end{definition}

\noindent This map $\{\cdot ,\cdot \}$ is clearly skew-symmetric, by (\ref%
{commutator}) and Definition \ref{convex Frechet derivative copy(1)}. This
skew-symmetric biderivation is precisely the one already constructed in
Corollary \ref{proposition sympatoch copy(1)} on polynomials:

\begin{proposition}[Poisson bracket]
\label{Poisson algebra prop}\mbox{ }\newline
Restricted to $\mathfrak{C}_{\mathcal{X}^{\mathbb{R}}}^{\mathbb{R}}\left(
E\right) $, the skew-symmetric biderivation of Definition \ref{convex
Frechet derivative copy(1)} coincides with the Poisson bracket defined by 
\begin{equation*}
\{f|_{E},g|_{E}\}\doteq \{f,g\}_{0}|_{E},\qquad f,g\in \mathfrak{C}_{%
\mathcal{X}^{\mathbb{R}}}^{\mathbb{R}}\left( \mathcal{X}_{\mathbb{R}}^{\ast
}\right) \ .
\end{equation*}%
See Corollary \ref{proposition sympatoch copy(1)}.
\end{proposition}

\begin{proof}
By Equation (\ref{Ynbis}),%
\begin{equation}
\mathrm{D}\hat{A}\left( \rho \right) =A-\rho \left( A\right) \mathfrak{1}\
,\qquad A\in \mathcal{X}^{\mathbb{R}}\ ,  \label{equation trivial}
\end{equation}%
and therefore, 
\begin{equation}
\{\hat{A},\hat{B}\}\left( \rho \right) =\rho \left( i\left[ A,B\right]
\right) \ ,\qquad \rho \in E,\ A,B\in \mathcal{X}^{\mathbb{R}}\ .
\label{poisson bracket}
\end{equation}%
Hence, by Definition \ref{convex Frechet derivative copy(2)} and Equation (%
\ref{equation super trivial}), 
\begin{equation}
\{\hat{A}|_{E},\hat{B}|_{E}\}=\{\hat{A},\hat{B}\}_{0}|_{E}\ ,\qquad A,B\in 
\mathcal{X}^{\mathbb{R}}\ .  \label{biderivation00}
\end{equation}%
Linearity and Leibniz's rule then lead to the assertion.
\end{proof}

The Poisson bracket can easily be extended to a complex Poisson bracket,
i.e., a Poisson bracket for complex-valued polynomials: Since the sum of two
affine functions stays affine, by Proposition \ref{affine real valued
functions}, observe that 
\begin{equation*}
\mathcal{A}(E;\mathbb{C})=\{\hat{A}:A\in \mathcal{X}\}\ .
\end{equation*}%
Moreover, by (\ref{norm properties}), (\ref{norm affine}) and (\ref{C1}),
for any continuously differentiable complex-valued function $f\in \mathfrak{Y%
}\left( \mathbb{C}\right) \varsubsetneq \mathfrak{C}$ there is a unique $%
\mathrm{D}f\in C(E;\mathcal{X})$ satisfying (\ref{clear2}). Then, the
Poisson bracket $\{\cdot ,\cdot \}$ of Definition \ref{convex Frechet
derivative copy(1)} can be extended to all $f,g\in \mathfrak{Y}\left( 
\mathbb{C}\right) $, as a skew-symmetric biderivation. In fact, since 
\begin{equation*}
\mathrm{D}\left( f+g\right) =\mathrm{D}f+\mathrm{D}g\quad \text{and}\quad 
\mathrm{D}\left( fg\right) =f\mathrm{D}g+g\mathrm{D}f\ ,
\end{equation*}%
this skew-symmetric biderivation satisfies 
\begin{equation}
\left\{ f,g\right\} =\left\{ \mathrm{Re}\left\{ f\right\} ,\mathrm{Re}%
\left\{ g\right\} \right\} -\left\{ \mathrm{Im}\left\{ f\right\} ,\mathrm{Im}%
\left\{ g\right\} \right\} +i\left( \left\{ \mathrm{Im}\left\{ f\right\} ,%
\mathrm{Re}\left\{ g\right\} \right\} +\left\{ \mathrm{Re}\left\{ f\right\} ,%
\mathrm{Im}\left\{ g\right\} \right\} \right)  \label{extension}
\end{equation}%
for all $f,g\in \mathfrak{Y}\left( \mathbb{C}\right) $. Note here that $%
\mathrm{Re}\left\{ f\right\} ,\mathrm{Im}\left\{ f\right\} \in \mathfrak{Y}%
\left( \mathbb{R}\right) $ for all $f\in \mathfrak{Y}\left( \mathbb{C}%
\right) $. Restricted to $\mathfrak{C}_{\mathcal{X}}\equiv \mathfrak{C}_{%
\mathcal{X}}\left( E\right) $, it is again a (complex) Poisson bracket,
since it satisfies the Jacobi identity, by Proposition \ref{Poisson algebra
prop} together with tedious computations.

\begin{remark}[Commutative case]
\mbox{ }\newline
If $\mathcal{X}$ is already a commutative unital $C^{\ast }$-algebra then
the Poisson bracket is of course trivial, being the zero biderivation, and
any classical dynamics generated by this Poisson bracket corresponds to the
identity map. This is reminiscent of the KMS dynamics, which becomes trivial
when the corresponding von Neumann algebra is commutative. (In this case,
the modular operator is the identity operator.)
\end{remark}

\subsection{Poissonian Symmetric Derivations\label{Bounded Derivations}}

A \emph{derivation} $\mathfrak{d}$ (on $\mathfrak{C}$) is a linear map from
a dense $\ast $-subalgebra $\mathrm{dom}(\mathfrak{d})$ (i.e., its domain)
of $\mathfrak{C}$ to\ the unital commutative $C^{\ast }$-algebra $\mathfrak{C%
}$ (\ref{metaciagre set 2}) of complex-valued weak$^{\ast }$-continuous
functions on $E$ such that 
\begin{equation}
\mathfrak{d}\left( fg\right) =\mathfrak{d}\left( f\right) g+f\mathfrak{d}%
\left( g\right) \ ,\qquad f,g\in \mathrm{dom}\left( \mathfrak{d}\right) \ .
\label{derivation0}
\end{equation}%
It is \emph{symmetric}, or a $\ast $-derivation, when%
\begin{equation}
\mathfrak{d}(\bar{f})=\overline{\mathfrak{d}\left( f\right) }\ ,\qquad f\in 
\mathrm{dom}\left( \mathfrak{d}\right) \ .  \label{derivation0bis}
\end{equation}%
For an exhaustive description of the theory of derivations, see \cite%
{BrattelliRobinsonI,Bratelli-derivation,Tomiyama} and references therein.

An important class of symmetric derivations can be defined by using the
Poisson bracket $\{\cdot ,\cdot \}$ of Definition \ref{convex Frechet
derivative copy(1)}:

\begin{definition}[Poissonian symmetric derivations]
\label{Poissonian Symmetric Derivations}\mbox{ }\newline
The Poissonian symmetric derivation associated with any continuously
differentiable real-valued function $h\in \mathfrak{Y}\left( \mathbb{R}%
\right) $ is the linear operator defined\ on its dense domain $\mathrm{dom}(%
\mathfrak{d}^{h})=\mathfrak{C}_{\mathcal{X}}\subseteq \mathfrak{C}$ by 
\begin{equation*}
\mathfrak{d}^{h}\left( f\right) \doteq \left\{ h,f\right\} \ ,\qquad f\in 
\mathfrak{C}_{\mathcal{X}}\ .
\end{equation*}
\end{definition}

\noindent Recall at this point that $\mathfrak{C}_{\mathcal{X}}\equiv 
\mathfrak{C}_{\mathcal{X}}(E)\subseteq \mathfrak{C}$ is the dense $\ast $%
-subalgebra of all polynomials in the elements of $\{\hat{A}:A\in \mathcal{X}%
\}$, with complex coefficients. See (\ref{def frac Cb}). Because of\
Definition \ref{convex Frechet derivative copy(1)} and Equations (\ref{C1})-(%
\ref{C1bis}), (\ref{clear2}), (\ref{norm affine2}) and (\ref{extension}), $%
\mathfrak{d}^{h}$ is a symmetric derivation satisfying 
\begin{equation*}
\left\Vert \mathfrak{d}^{h}\left( f\right) \right\Vert _{\mathfrak{C}}\leq
4\left\Vert h\right\Vert _{\mathfrak{Y}\left( \mathbb{C}\right) }\left\Vert
f\right\Vert _{\mathfrak{Y}\left( \mathbb{C}\right) }\ ,\qquad f\in 
\mathfrak{C}_{\mathcal{X}}\subseteq \mathfrak{Y}\left( \mathbb{C}\right) \ .
\end{equation*}%
In particular, $\mathfrak{d}^{h}$ could be extended as a bounded symmetric
derivation $\mathfrak{\tilde{d}}^{h}$ from $\mathfrak{Y}\left( \mathbb{C}%
\right) $ to $\mathfrak{C}$, i.e., as an element of $\mathcal{B}\left( 
\mathfrak{Y}\left( \mathbb{C}\right) ,\mathfrak{C}\right) $.

At first sight, the extension $\mathfrak{\tilde{d}}^{h}$ of $\mathfrak{d}%
^{h} $ to all continuously differentiable complex-valued functions of $%
\mathfrak{Y}\left( \mathbb{C}\right) $ seems to be natural, like for the
usual differentiation on functions of the compact set $[0,1]$. On second
thoughts, $\mathfrak{\tilde{d}}^{h}$ may \emph{not} be closable, even if
this property would be true for $\mathfrak{d}^{h}$. Of course, if $\mathfrak{%
\tilde{d}}^{h} $ is closable then $\mathfrak{d}^{h}$ is also closable.

For a large class of symmetric derivations, the closableness is proven from
dissipativity \cite[Definition 1.4.6, Proposition 1.4.7]{Bratelli-derivation}%
. This property is in turn deduced from a theorem proven by Kishimoto \cite[%
Theorem 1.4.9]{Bratelli-derivation}, which uses the assumption that the
square root of each positive element of the domain of the derivation also
belongs to the same domain. We cannot expect this last property to be
satisfied for symmetric derivations like $\mathfrak{d}^{h}$ or $\mathfrak{%
\tilde{d}}^{h}$.

The closableness of unbounded symmetric derivations of $C^{\ast }$-algebras
is, in general, a non-trivial issue, even in the commutative case like $%
\mathfrak{C}$. This property is not generally true: there exist norm-densely
defined derivations of $C^{\ast }$-algebras that are \emph{not} closable 
\cite{bra-robin-1975}. For instance, in \cite[p. 306]{BrattelliRobinsonI},
it is even claimed that \textquotedblleft \textit{Herman has constructed an
extension of the usual differentiation on }$C(0,1)$\textit{\ which is a
nonclosable derivation of }$C(0,1)$.\textquotedblright\ The general
characterization of closed symmetric derivations depends heavily on the
(Hausdorff) dimension of the locally compact set, here the weak$^{\ast }$%
-compact set $E$. Around 1990, a characterization of all closed symmetric
derivations were obtained by using spaces of functions acting on a compact
subset of a one-dimensional space. However, \textquotedblleft \textit{for
more than 2 dimensions only sporadic results are known}\textquotedblright ,
as quoted in \cite[Section 1.6.4, p. 27]{Bratelli-derivation}.\ See, e.g., 
\cite[Section 1.6.4]{Bratelli-derivation}, \cite{Tomiyama}, \cite%
{Kuroselastpaper,Kurose}, and later \cite[p. 306]{BrattelliRobinsonI}.

In our approach, the closableness of unbounded symmetric derivations like $%
\mathfrak{d}^{h}$ or $\mathfrak{\tilde{d}}^{h}$ is a necessary property to
make sense of a classical dynamics, in its Hamiltonian formulation, via $%
C_{0}$-groups. In Section \ref{Closed derivation}, we show that the
symmetric derivation $\mathfrak{d}^{h}$ is closable, at least for all
functions $h$ in a dense subset of $\mathfrak{C}$, including $\mathfrak{C}_{%
\mathcal{X}}$. This is performed via a self-consistency problem together
with the $C_{0}$-semigroup theory \cite{EngelNagel}. Our results are
non-trivial since $E$ is \emph{not} a subset of a finite-dimensional space
when $\mathcal{X}$ has infinite dimension. See proof of Theorem \ref{theorem
density2}.

\section{Hamiltonian Flows for States from Self-Consistent Quantum Dynamics 
\label{Closed derivation}}

Our approach to the construction of Hamiltonian flows and, in particular,
closed derivations of a commutative $C^{\ast }$-algebra via self-consistency
problems is \emph{non-conventional}. However, it shares some similarity with
the following simple example in the finite-dimensional case:\ Take $\mathcal{%
A}$ as being the commutative unital $C^{\ast }$-algebra of all continuous,
bounded and complex-valued functions on $\mathbb{R}^{2N}$, $N\in \mathbb{N}$%
, and fix a smooth and compactly supported function $h:\mathbb{R}%
^{2N}\rightarrow \mathbb{R}$. From the Picard-Lindel\"{o}f iteration
argument, the (Hamiltonian) vector field $J\nabla h$ (where $\nabla h$ is
the gradient of $h$ and $J$ is the $2N$-dimensional symplectic matrix)
generates a global smooth flow $\phi _{t}:\mathbb{R}^{2N}\rightarrow \mathbb{%
R}^{2N}$, $t\in \mathbb{R}$. Let the one-parameter group $\{V_{t}\}_{t\in 
\mathbb{R}}$ of $\ast $-automorphisms of $\mathcal{A}$\ be defined by%
\begin{equation*}
\lbrack V_{t}(f)](x)=f\circ \phi _{-t}(x)\text{ },\text{\qquad }x\in \mathbb{%
R}^{2N},\text{\ }t\in \mathbb{R}\text{\ }.
\end{equation*}%
Because of the compactness of the support of $h$, this one-parameter group
is strongly continuous and the corresponding generator is a closed
derivation in $\mathcal{A}$, denoted by $\delta _{h}$. Moreover, it is
straightforward to check that, in the dense set of smooth functions, this
derivation acts as $\delta _{h}=\{h,\cdot \}$, where $\{\cdot ,\cdot \}$ is
the canonical Poisson bracket%
\begin{equation*}
\{f,g\}(x)\doteq \sum_{k,l=1}^{2N}J_{ij}[\partial _{x_{i}}f(x)][\partial
_{x_{j}}g(x)]\text{ },\text{\qquad }x\in \mathbb{R}^{2N}\text{ },
\end{equation*}%
for smooth functions $f,g$ on $\mathbb{R}^{2N}$. The analogy of the results
presented in this section with this example is as follows: in our setting,
the space $E$ of all states on $\mathcal{X}$ replaces $\mathbb{R}^{2N}$ and
the analogue of the global Hamiltonian flow $\{\phi _{t}\}_{t\in \mathbb{R}}$
is a one-parameter family of weak$^{\ast }$ automorphisms of $E$ (or
self-homeomorphisms of $E$). Note that B\'{o}na uses such a construction
only on symplectic leaves of the corresponding Poisson manifold and
\textquotedblleft glues\textquotedblright\ them together in order to
construct the global flow \cite[Section 2.1-d]{Bono2000}. However, in strong
contrast to this simple example, in our case, it is not clear at all how to
construct the corresponding family of automorphisms from Hamiltonian vector
fields. Instead, we construct it as the solution to a \emph{self-consistency}
problem. Similar to the above example, the closed derivations we obtain for
the classical algebra $\mathfrak{C}$ are closed extensions of densely
defined derivations of the form $f\mapsto \{h,f\}$, $f,h\in \mathfrak{C}_{%
\mathcal{X}^{\mathbb{R}}}$, where $\mathfrak{C}_{\mathcal{X}^{\mathbb{R}%
}}\subseteq \mathfrak{C}$ is the dense subalgebra of polynomials in the
elements of $\mathcal{X}^{\mathbb{R}}$, defined by (\ref{def frac Cb}) for $%
\mathcal{B=X}^{\mathbb{R}}$.

All this construction is performed in this section, supplemented with
technical assertions proven in\ Section \ref{Well-posedness sect copy(1)}.
We start by somehow tedious, albeit necessary, definitions and notation in
Sections \ref{Preliminary Definitions}-\ref{Self-Consistency Equations}, the
self-consistency equations being asserted in Theorem \ref{theorem
sdfkjsdklfjsdklfj copy(3)} and exploited afterwards.

\subsection{Preliminary Definitions\label{Preliminary Definitions}}

Let $C_{b}\left( \mathbb{R};\mathfrak{Y}\left( \mathbb{R}\right) \right) $
be the Banach space of bounded continuous maps from $\mathbb{R}$ to $%
\mathfrak{Y}\left( \mathbb{R}\right) $ with the norm 
\begin{equation}
\left\Vert h\right\Vert _{C_{b}\left( \mathbb{R};\mathfrak{Y}\left( \mathbb{R%
}\right) \right) }\doteq \sup_{t\in \mathbb{R}}\left\Vert h\left( t\right)
\right\Vert _{\mathfrak{Y}\left( \mathbb{R}\right) }\ ,\qquad h\in
C_{b}\left( \mathbb{R};\mathfrak{Y}\left( \mathbb{R}\right) \right) \ .
\label{norm C}
\end{equation}%
We identify $\mathfrak{Y}\left( \mathbb{R}\right) $ with the subalgebra of
constant functions of $C_{b}\left( \mathbb{R};\mathfrak{Y}\left( \mathbb{R}%
\right) \right) $, i.e., 
\begin{equation}
\mathfrak{Y}\left( \mathbb{R}\right) \subseteq C_{b}\left( \mathbb{R};%
\mathfrak{Y}\left( \mathbb{R}\right) \right) \ .  \label{identifyidentify}
\end{equation}

Let $C\left( E;E\right) $ be the set of weak$^{\ast }$-continuous functions
from the state space $E$ to itself endowed with the topology of uniform
convergence. In other words, a net $(f_{j})_{j\in J}\subseteq C\left(
E;E\right) $ converges to $f\in C\left( E;E\right) $ whenever 
\begin{equation}
\lim_{j\in J}\max_{\rho \in E}\left\vert f_{j}(\rho )(A)-f(\rho
)(A)\right\vert =0\ ,\qquad \text{for all }A\in \mathcal{X}\ .
\label{uniform convergence weak*}
\end{equation}

We denote by $\mathrm{Aut}\left( E\right) \varsubsetneq C\left( E;E\right) $
the subspace of all automorphisms of $E$, i.e., elements of $C\left(
E;E\right) $ with weak$^{\ast }$-continuous inverse. Equivalently, $\mathrm{%
Aut}\left( E\right) $ is the set of all bijective maps in $C\left(
E;E\right) $, because $E$ is a compact Hausdorff space. Recall that, here,
the concept of\ an automorphism depends on the structure of the
corresponding domain: elements of $\mathrm{Aut}\left( E\right) $ are
self-homeomorphisms while a automorphism of a $C^{\ast }$-algebra is a $\ast 
$-automorphism of this $C^{\ast }$-algebra.

Any continuous function $h\in C_{b}\left( \mathbb{R};\mathfrak{Y}\left( 
\mathbb{R}\right) \right) $ defines a non-autonomous, state-dependent, \emph{%
quantum} dynamics on the $C^{\ast }$-algebra $\mathcal{X}$ via the family $\{%
\mathrm{D}h(t)\}_{t\in \mathbb{R}}\subseteq C(E;\mathcal{X}^{\mathbb{R}})$,
satisfying (\ref{clear2}) for each $t\in \mathbb{R}$. This quantum dynamics
can in turn be used to define a (classical) dynamics on the commutative $%
C^{\ast }$-algebra $\mathfrak{C}$ of all continous complex-valued functions
on $E$. This latter dynamics turns out to be the flow generated, as is usual
in classical mechanics, by the Poisson bracket $\{h(t),\cdot \}$ of
Definition \ref{convex Frechet derivative copy(1)} (see also Corollary \ref%
{proposition sympatoch copy(1)} and Proposition \ref{Poisson algebra prop}).
We start with the state-dependent quantum dynamics on the primordial $%
C^{\ast }$-algebra $\mathcal{X}$, in the next subsection.

\subsection{Dynamics on the Primordial $C^{\ast }$-Algebra\label{Section QM}}

Fix $h\in C_{b}\left( \mathbb{R};\mathfrak{Y}\left( \mathbb{R}\right)
\right) $, which plays the role of a time-dependent family of classical
Hamiltonians. Then, for each state $\rho \in E$ and time $t\in {\mathbb{R}}$%
, we define the symmetric bounded derivation $X_{t}^{\rho }\in \mathcal{B}(%
\mathcal{X})$ by%
\begin{equation}
X_{t}^{\rho }\left( A\right) \doteq i\left[ \mathrm{D}h\left( t;\rho \right)
,A\right] \doteq i\left( \mathrm{D}h\left( t;\rho \right) A-A\mathrm{D}%
h\left( t;\rho \right) \right) \ ,\qquad A\in \mathcal{X}\ ,
\label{flow baby0}
\end{equation}%
where $[\cdot ,\cdot ]$ is the usual commutator defined by (\ref{commutator}%
) and 
\begin{equation*}
\mathrm{D}h\left( t;\rho \right) \doteq \left[ \mathrm{D}h\left( t\right) %
\right] \left( \rho \right) \in \mathcal{X}^{\mathbb{R}}\ ,\qquad \rho \in
E,\ t\in \mathbb{R}\ .
\end{equation*}%
By Equations (\ref{C1bis}) and (\ref{norm affine2}), note that%
\begin{equation}
\sup_{\rho \in E}\Vert X_{t}^{\rho }\Vert _{\mathcal{B}(\mathcal{X})}\leq
2\left\Vert h\right\Vert _{C_{b}\left( \mathbb{R};\mathfrak{Y}\left( \mathbb{%
R}\right) \right) }  \label{bounded baby}
\end{equation}%
and, for any state-valued continuous function $\xi \in C\left( \mathbb{R}%
;E\right) $ and times $s,t\in \mathbb{R}$,%
\begin{equation*}
\Vert X_{t}^{\xi \left( t\right) }-X_{s}^{\xi \left( s\right) }\Vert _{%
\mathcal{B}(\mathcal{X})}\leq 2\left\Vert h\left( t\right) -h\left( s\right)
\right\Vert _{\mathfrak{Y}\left( \mathbb{R}\right) }+2\left\Vert \mathrm{D}%
h\left( s;\xi \left( t\right) \right) -\mathrm{D}h\left( s;\xi \left(
s\right) \right) \right\Vert _{\mathcal{X}}\ ,
\end{equation*}%
from (\ref{norm properties}) and (\ref{clear2}).

Since $\mathrm{D}f\in C(E;\mathcal{X}^{\mathbb{R}})$ when $f\in \mathfrak{Y}%
\left( \mathbb{R}\right) $, for any function $\xi \in C\left( \mathbb{R}%
;E\right) $, $(X_{t}^{\xi \left( t\right) })_{t\in \mathbb{R}}$ is a
norm-continuous family of bounded operators. Therefore, for any continuous
functions $h\in C_{b}\left( \mathbb{R};\mathfrak{Y}\left( \mathbb{R}\right)
\right) $ and $\xi \in C\left( \mathbb{R};E\right) $, a norm-continuous
two-para%
%TCIMACRO{\TeXButton{\-}{\-}}%
%BeginExpansion
\-%
%EndExpansion
meter family $(T_{t,s}^{\xi })_{s,t\in \mathbb{R}}$ of $\ast $-automorphisms
of $\mathcal{X}$ is uniquely defined in $\mathcal{B}(\mathcal{X})$ by the
non-auto%
%TCIMACRO{\TeXButton{\-}{\-}}%
%BeginExpansion
\-%
%EndExpansion
nomous evolution equation%
\begin{equation}
\forall s,t\in {\mathbb{R}}:\qquad \partial _{t}T_{t,s}^{\xi }=T_{t,s}^{\xi
}\circ X_{t}^{\xi \left( t\right) }\ ,\qquad T_{s,s}^{\xi }=\mathbf{1}_{%
\mathcal{X}}\ ,  \label{flow baby1}
\end{equation}%
or, equivalently, by 
\begin{equation}
\forall s,t\in {\mathbb{R}}:\qquad \partial _{s}T_{t,s}^{\xi }=-X_{s}^{\xi
\left( s\right) }\circ T_{t,s}^{\xi }\ ,\qquad T_{t,t}^{\xi }=\mathbf{1}_{%
\mathcal{X}}\ .  \label{flow baby1bis}
\end{equation}%
Note that $(T_{t,s}^{\xi })_{s,t\in \mathbb{R}}$ clearly satisfies the
(reverse) cocycle property%
\begin{equation}
\forall s,r,t\in \mathbb{R}:\qquad T_{t,s}^{\xi }=T_{r,s}^{\xi }\circ
T_{t,r}^{\xi }\ .  \label{babyreverse}
\end{equation}%
The existence and uniqueness of a solution to these evolution equations
follow from the usual theory of non-auto%
%TCIMACRO{\TeXButton{\-}{\-}}%
%BeginExpansion
\-%
%EndExpansion
nomous evolution equations for bounded norm-continuous generators, see,
e.g., \cite{Pazy}. In this case, it is explicitly given by Dyson series. The
fact that it defines a family of $\ast $-automorphisms of $\mathcal{X}$
results from the identity 
\begin{equation*}
\partial _{t}\left\{ T_{t,s}^{\xi }T_{s,t}^{\xi }\right\} =0\ ,\qquad s,t\in 
{\mathbb{R}}\ ,
\end{equation*}%
and the fact that the corresponding generators are symmetric derivations.

\subsection{Self-Consistency Equations\label{Self-Consistency Equations}}

Let $C\left( \mathbb{R};C\left( E;E\right) \right) $ be the set of
continuous functions from $\mathbb{R}$ to $C\left( E;E\right) $. Any $%
\mathbf{\xi }\in C\left( \mathbb{R};C\left( E;E\right) \right) $ defines a
function $\mathbf{\xi }\left( \cdot ;\rho \right) \in C\left( \mathbb{R}%
;E\right) $ by%
\begin{equation}
\mathbf{\xi }\left( t;\rho \right) \doteq \left[ \mathbf{\xi }\left(
t\right) \right] \left( \rho \right) \ ,\qquad \rho \in E,\ t\in \mathbb{R}\
.  \label{notation}
\end{equation}%
Then, for any continuous functions $h\in C_{b}\left( \mathbb{R};\mathfrak{Y}%
\left( \mathbb{R}\right) \right) $, $\mathbf{\xi }\in C\left( \mathbb{R}%
;C\left( E;E\right) \right) $ and state $\rho \in E$, the norm-continuous
two-para%
%TCIMACRO{\TeXButton{\-}{\-}}%
%BeginExpansion
\-%
%EndExpansion
meter family $(T_{t,s}^{\mathbf{\xi }\left( \cdot ;\rho \right) })_{s,t\in 
\mathbb{R}}$ of $\ast $-automorphisms of $\mathcal{X}$ defined above
(Section \ref{Section QM}) is used to define a family $(\phi _{t,s}^{(h,%
\mathbf{\xi })})_{s,t\in \mathbb{R}}$ of maps from the state space $E$ to
itself, as follows: 
\begin{equation}
\phi _{t,s}^{(h,\mathbf{\xi })}\left( \rho \right) \doteq \rho \circ
T_{t,s}^{\mathbf{\xi }\left( \cdot ;\rho \right) }\ ,\qquad \rho \in E,\
s,t\in \mathbb{R}\ .  \label{phits}
\end{equation}%
By the reverse cocycle property (\ref{babyreverse}) for $(T_{t,s}^{\xi
})_{s,t\in \mathbb{R}}$, $(\phi _{t,s}^{(h,\mathbf{\xi })})_{s,t\in \mathbb{R%
}}$ has the (non-reverse) cocycle property, i.e.,%
\begin{equation}
\phi _{t,s}^{(h,\mathbf{\xi })}=\phi _{t,r}^{(h,\mathbf{\xi })}\circ \phi
_{r,s}^{(h,\mathbf{\xi })}\ ,\qquad s,t,r\in \mathbb{R}\ .
\label{babyreverse2}
\end{equation}%
By Lemma \ref{Solution selfbaby copy(1)}, 
\begin{equation*}
(\phi _{t,s}^{(h,\mathbf{\xi })}\left( \rho \right) )_{s,t\in \mathbb{R}}\in
C\left( \mathbb{R}^{2};E\right) \ ,\qquad \rho \in E\ .
\end{equation*}%
As a consequence, the family $(\phi _{t,s}^{(h,\mathbf{\xi })})_{s,t\in 
\mathbb{R}}$ is a continuous flow on the state space $E$. Since $\{\mathrm{D}%
h(t)\}_{t\in \mathbb{R}}\subseteq C(E;\mathcal{X}^{\mathbb{R}})$, by Lemma %
\ref{Solution selfbaby copy(1)} and Lebesgue's dominated convergence
theorem, note additionally that $\{\phi _{t,s}^{(h,\mathbf{\xi })}\}_{s,t\in 
\mathbb{R}}$ is a family of automorphisms (self-homeomorphisms) of $E$, i.e. 
\begin{equation*}
\{\phi _{t,s}^{(h,\mathbf{\xi })}\}_{s,t\in \mathbb{R}}\subseteq \mathrm{Aut}%
\left( E\right) .
\end{equation*}

To understand the relevance of this flow with respect to classical dynamics,
it is enlightening to consider the autonomous case for which $h$ is the
constant function $\hat{H}$ for some $H\in \mathcal{X}^{\mathbb{R}}$. See (%
\ref{fA}) for the definition of the function $\hat{H}$, the Gelfand
transform of $H$. In this case, choose a state $\rho \in E$ and observe from
(\ref{flow baby0}), (\ref{flow baby1}) and (\ref{phits}), together with
Definition \ref{convex Frechet derivative copy(1)} and Equation (\ref%
{equation trivial}), that%
\begin{equation*}
\partial _{t}\hat{A}_{t,s}=\{h,\hat{A}_{t,s}\}\qquad \text{with\quad }\hat{A}%
_{t,s}\doteq \hat{A}\circ \phi _{t,s}^{(\hat{H})}\in \mathfrak{C}
\end{equation*}%
for any $A\in \mathcal{X}$ and $s,t\in \mathbb{R}$, noting that the flow $%
\phi _{t,s}^{(\hat{H})}\equiv \phi _{t,s}^{(h,\mathbf{\xi })}$, $s,t\in 
\mathbb{R}$, does not depend on $\mathbf{\xi }\in C\left( \mathbb{R};C\left(
E;E\right) \right) $. Since $\phi _{t,s}^{(\hat{H})}=\phi _{t-s,0}^{(\hat{H}%
)}$ for any $s,t\in \mathbb{R}$, the flow defined by $(\phi _{t,s}^{(\hat{H}%
)})_{s,t\in \mathbb{R}}$ is associated with an autonomous classical
dynamics, in the usual sense, on elementary elements $\{\hat{A}:A\in 
\mathcal{X}\}$.

In the general case of (non-autonomous) classical dynamics generated by
time-dependent Poissonian symmetric derivations of the form $\{h(t),\cdot \}$%
, $t\in \mathbb{R}$, a convenient (and non-trivial) choice of the function $%
\mathbf{\xi }$ in Equation (\ref{phits}) has to be made. We determine it via
a \emph{self-consistency equation}. This is our first main result:

\begin{theorem}[Self-consistency equations]
\label{theorem sdfkjsdklfjsdklfj copy(3)}\mbox{ }\newline
\emph{(a)} Let $\mathcal{X}$ be a unital $C^{\ast }$-algebra and $\mathfrak{B%
}$ a finite-dimensional real subspace of $\mathcal{X}^{\mathbb{R}}$%
.\smallskip \newline
\emph{(b)} Take $h\in C_{b}\left( \mathbb{R};\mathfrak{Y}\left( \mathbb{R}%
\right) \right) $ such that, for some constant $D_{0}\in \mathbb{R}^{+}$,%
\begin{equation*}
\sup_{t\in {\mathbb{R}}}\left\Vert \mathrm{D}h(t;\rho )-\mathrm{D}h(t;\tilde{%
\rho})\right\Vert _{\mathcal{X}}\leq D_{0}\sup_{B\in \mathfrak{B},\left\Vert
B\right\Vert =1}\left\vert \left( \rho -\tilde{\rho}\right) \left( B\right)
\right\vert \ ,\qquad \rho ,\tilde{\rho}\in E\ .
\end{equation*}%
Under Conditions (a)-(b), there is a unique function $\mathbf{\varpi }%
^{h}\in C\left( \mathbb{R}^{2};\mathrm{Aut}\left( E\right) \right) $ such
that 
\begin{equation}
\mathbf{\varpi }^{h}\left( s,t\right) =\phi _{t,s}^{(h,\mathbf{\varpi }%
^{h}\left( \alpha ,\cdot \right) )}|_{\alpha =s}\ ,\qquad s,t\in {\mathbb{R}}%
\ ,  \label{self-consitency}
\end{equation}%
where we recall that $\mathrm{Aut}\left( E\right) \varsubsetneq C\left(
E;E\right) $ is the subspace of all automorphisms\ (or self-homeomorphisms)
of $E$.
\end{theorem}

\begin{proof}
The theorem is a consequence of Lemmata \ref{Solution selfbaby} and \ref%
{lemma wellbabybaby}.
\end{proof}

\begin{remark}
\mbox{ }\newline
\emph{(i)} Stronger results than Theorem \ref{theorem sdfkjsdklfjsdklfj
copy(3)} are proven in Section \ref{Well-posedness sect copy(1)}. See, in
particular, Lemma \ref{Differentiability2baby}.\newline
\emph{(ii)} If $\mathcal{X}$ is separable, recall that the state space $E$
of Definition \ref{state space} is metrizable, which is a very useful
property. In Theorem \ref{theorem sdfkjsdklfjsdklfj copy(3)}, however, the
separability of $\mathcal{X}$ is not necessary at the cost of taking a
finite dimensional space $\mathfrak{B}$ in Condition (b).
\end{remark}

Condition (b) of Theorem \ref{theorem sdfkjsdklfjsdklfj copy(3)} is, for
instance, satisfied for any cylindrical function $h$ within the set 
\begin{gather}
\mathfrak{Z}\doteq 
%TCIMACRO{\TeXButton{\Big\{}{\Big\{}}%
%BeginExpansion
\Big\{%
%EndExpansion
\left( f\left( t\right) \right) _{t\in \mathbb{R}}\in \mathfrak{C}:f\left(
t;\rho \right) =g\left( t;\rho \left( B_{1}\right) ,\ldots ,\rho \left(
B_{n}\right) \right) \text{ for }t\in \mathbb{R}\text{ and }\rho \in E 
\notag \\
\qquad \qquad \qquad \qquad \qquad \qquad \text{with }n\in \mathbb{N}\text{, 
}\left\{ B_{j}\right\} _{j=1}^{n}\subseteq \mathcal{X}^{\mathbb{R}}\text{
and }g\in C_{b}\left( \mathbb{R};C_{b}^{3}\left( \mathbb{R}^{n},\mathbb{R}%
\right) \right) 
%TCIMACRO{\TeXButton{\Big\}}{\Big\}}}%
%BeginExpansion
\Big\}%
%EndExpansion
\ .  \label{Zn}
\end{gather}%
By (\ref{Yn}), note that, for any $h\in \mathfrak{Z}$, there is $n\in 
\mathbb{N}$ such that $h(t)\in \mathfrak{Y}_{n}$ for all $t\in \mathbb{R}$.
See also (\ref{Ynbis}). Observe that $\mathfrak{Z}\varsubsetneq \mathfrak{C}$
is a dense subset since $\mathfrak{C}_{\mathcal{X}^{\mathbb{R}}}^{\mathbb{R}%
}\subseteq \mathfrak{Z}$. In (\ref{Zn}) we are quite generous by assuming
that the function $g\left( t\right) $ belongs to $C_{b}^{3}(\mathbb{R}^{n},%
\mathbb{R})$ for some $n\in \mathbb{N}$, but even $C_{b}^{2}(\mathbb{R}^{n},%
\mathbb{R})$ would be sufficient to get Condition (b). We assume more
regularity for $g(t)$, $t\in \mathbb{R}$, to be able to prove Theorem \ref%
{proposition dynamic classique II}. Here, $C_{b}^{p}(\mathbb{R}^{n};\mathbb{R%
})$, $p\in \mathbb{N}$, denotes the Banach space of bounded real-valued $%
C^{p}$-functions on $\mathbb{R}^{n}$, whose norm is the $C^{p}$-norm, i.e.,
the sum of the supremum norm of all derivatives of order from $0$ to $p$.

As explained in Section \ref{Classical algebra copy(1)}, for quantum
systems, we shall not restrict our study to the phase space $\overline{%
\mathcal{E}(E)}$ of Definition \ref{phase space}, but we generally consider
the whole state space $E$ of Definition \ref{state space}. We show next that
both the set $\mathcal{E}(E)$ of extreme points and its weak$^{\ast }$
closure $\overline{\mathcal{E}(E)}$ are conserved by the flow of Theorem \ref%
{theorem sdfkjsdklfjsdklfj copy(3)}, which is defined on the whole state
space $E$:

\begin{corollary}[Conservation of the phase space]
\label{corollary conservation}\mbox{ }\newline
Under Conditions (a)-(b) of Theorem \ref{theorem sdfkjsdklfjsdklfj copy(3)},
for any $s,t\in \mathbb{R}$,%
\begin{equation*}
\mathbf{\varpi }^{h}\left( s,t\right) \left( \mathcal{E}(E)\right) \subseteq 
\mathcal{E}(E)\qquad \text{and}\qquad \mathbf{\varpi }^{h}\left( s,t\right) (%
\overline{\mathcal{E}(E)})\subseteq \overline{\mathcal{E}(E)}\ .
\end{equation*}
\end{corollary}

\begin{proof}
The proof is done by contradiction: Assume Conditions (a)-(b) of Theorem \ref%
{theorem sdfkjsdklfjsdklfj copy(3)}. Take $\rho \in \mathcal{E}(E)$ and
assume the existence of $s,t\in \mathbb{R}$, $\lambda \in (0,1)$ and two
distinct $\rho _{1},\rho _{2}\in E$ such that 
\begin{equation*}
\mathbf{\varpi }^{h}\left( s,t\right) \left( \rho \right) =\phi _{t,s}^{(h,%
\mathbf{\varpi }^{h}\left( s,\cdot \right) )}\left( \rho \right) =\left(
1-\lambda \right) \rho _{1}+\lambda \rho _{2}\ .
\end{equation*}%
See Theorem \ref{theorem sdfkjsdklfjsdklfj copy(3)}. By (\ref{babyreverse})
and (\ref{phits}), it follows that 
\begin{equation*}
\rho =\left( 1-\lambda \right) \rho _{1}\circ T_{s,t}^{\mathbf{\varpi }%
^{h}\left( s,\cdot \right) \left( \rho \right) }+\lambda \rho _{2}\circ
T_{s,t}^{\mathbf{\varpi }^{h}\left( s,\cdot \right) \left( \rho \right) }\ .
\end{equation*}%
This is not possible whenever $\rho \in \mathcal{E}(E)$ because 
\begin{equation*}
\rho _{1}\circ T_{s,t}^{\mathbf{\varpi }^{h}\left( s,\cdot \right) \left(
\rho \right) }\qquad \text{and}\qquad \rho _{2}\circ T_{s,t}^{\mathbf{\varpi 
}^{h}\left( s,\cdot \right) \left( \rho \right) }
\end{equation*}%
are two distinct states. This proves that the image of an extreme state by $%
\mathbf{\varpi }^{h}\left( s,t\right) $ is always an extreme state. $\mathbf{%
\varpi }^{h}\left( s,t\right) \in \mathrm{Aut}\left( E\right) $ and thus
preserves the phase space $\overline{\mathcal{E}(E)}$.
\end{proof}

\subsection{Classical Dynamics as Feller Evolution\label{Section CM}}

The continuous family $\mathbf{\varpi }^{h}$ of Theorem \ref{theorem
sdfkjsdklfjsdklfj copy(3)} yields a family $(V_{t,s}^{h})_{s,t\in \mathbb{R}%
} $ of $\ast $-automorphisms of $\mathfrak{C}\doteq C(E;\mathbb{C})$ defined
by 
\begin{equation}
V_{t,s}^{h}\left( f\right) \doteq f\circ \mathbf{\varpi }^{h}\left(
s,t\right) \ ,\qquad f\in \mathfrak{C},\ s,t\in \mathbb{R}\ .
\label{classical evolution familybaby}
\end{equation}%
By Corollary \ref{corollary conservation}, such a map can also be defined in
the same way on $C(\overline{\mathcal{E}(E)};\mathbb{C})$ or $C(\mathcal{E}%
(E);\mathbb{C})$, where we recall that $\overline{\mathcal{E}(E)}$ is the
phase space of Definition \ref{phase space}. In any case, it is a strongly
continuous two-parameter family defining a classical dynamics:

\begin{proposition}[Classical dynamics as Feller evolution system]
\label{lemma poisson copy(1)}\mbox{
}\newline
Under Conditions (a)-(b) of Theorem \ref{theorem sdfkjsdklfjsdklfj copy(3)}, 
$(V_{t,s}^{h})_{s,t\in \mathbb{R}}$ is a strongly continuous two-parameter
family of $\ast $-automorphisms of $\mathfrak{C}$ satisfying the reverse
cocycle property:%
\begin{equation}
\forall s,r,t\in \mathbb{R}:\qquad V_{t,s}^{h}=V_{r,s}^{h}\circ V_{t,r}^{h}\
.  \label{reverse}
\end{equation}%
If, additionally, $h\in \mathfrak{Y}\left( \mathbb{R}\right) $ (cf. (\ref%
{identifyidentify})), then $V_{t,s}^{h}=V_{t-s,0}^{h}$ for any $s,t\in 
\mathbb{R}$\ and $(V_{t,0}^{h})_{t\in \mathbb{R}}$ is a $C_{0}$-group of $%
\ast $-automorphisms of $\mathfrak{C}$.
\end{proposition}

\begin{proof}
The strong continuity of this family with respect to $s,t\in \mathbb{R}$ is
a consequence of $\mathbf{\varpi }^{h}\in C\left( \mathbb{R}^{2};\mathrm{Aut}%
\left( E\right) \right) $ and the fact that any continuous family of
continuous functions on compacta is uniformly continuous. Recall that the
topology of $\mathrm{Aut}\left( E\right) $ is the topology of uniform
convergence of weak$^{\ast }$-continuous functions from $E$ to itself. (To
prove continuity in such a strong sense, one could also use $V_{t,s}^{h}\in 
\mathcal{B}\left( \mathfrak{C}\right) $ and the density of $\mathfrak{C}_{%
\mathcal{X}}$ in $\mathfrak{C}$.) Equation (\ref{reverse}) follows from
Corollary \ref{Corollary bije+cocylbaby}. Finally, if $h\in \mathfrak{Y}%
\left( \mathbb{R}\right) $, while assuming Conditions (a)-(b) of Theorem \ref%
{theorem sdfkjsdklfjsdklfj copy(3)}, then the family $(T_{t,s}^{\xi
})_{s,t\in \mathbb{R}}$ defined by (\ref{flow baby1})-(\ref{flow baby1bis})
for any $\xi \in C\left( \mathbb{R};E\right) $ satisfies $T_{t,s}^{\xi
}=T_{t-s,0}^{\xi \left( \cdot +s\right) }$ for any $s,t\in \mathbb{R}$,
where $\xi \left( \cdot +s\right) \in C\left( \mathbb{R};E\right) $ is the
function $\xi $ translated by the real number $s$. As a consequence, at any
fixed $s\in \mathbb{R}$ and $\rho \in E$, the function $\xi \in C\left( 
\mathbb{R};E\right) $ defined by 
\begin{equation*}
\xi _{s}\left( t\right) \doteq \mathbf{\varpi }^{h}\left( 0,t-s;\rho \right)
\ ,\qquad t\in \mathbb{R}\ ,
\end{equation*}%
is a solution to Equation (\ref{self consitence equation1baby}). By Lemma %
\ref{Solution selfbaby}, it follows that 
\begin{equation*}
\mathbf{\varpi }^{h}\left( 0,t-s\right) =\mathbf{\varpi }^{h}\left(
s,t\right) \ ,\qquad s,t\in \mathbb{R}\ ,
\end{equation*}%
i.e., $V_{t,s}^{h}=V_{t-s,0}^{h}$ for any $s,t\in \mathbb{R}$. By using (\ref%
{reverse}) at $r=t-\alpha +s$ for any $\alpha \in \mathbb{R}$, one verifies
that the one-parameter family $(V_{t,0}^{h})_{t\in \mathbb{R}}$ satisfies
the group property.
\end{proof}

Under Conditions (a)-(b) of Theorem \ref{theorem sdfkjsdklfjsdklfj copy(3)}, 
$(V_{t,s}^{h})_{s,t\in \mathbb{R}}$ restricted on $\mathfrak{C}^{\mathbb{R}}$
is automatically a \emph{Feller evolution system} in the following sense:

\begin{itemize}
\item As a $\ast $-automorphism of a $C^{\ast }$-algebra, $V_{t,s}^{h}$ is
positivity preserving and $\Vert V_{t,s}^{h}\Vert _{\mathcal{B}\left( 
\mathfrak{C}^{\mathbb{R}}\right) }=1$;

\item $(V_{t,s}^{h})_{s,t\in \mathbb{R}}$ is a strongly continuous
two-parameter family satisfying (\ref{reverse}).
\end{itemize}

\noindent Therefore, the classical dynamics defined on the real space $%
\mathfrak{C}^{\mathbb{R}}$ from $(V_{t,s}^{h})_{s,t\in \mathbb{R}}$ can be
associated in this case with Feller processes\footnote{%
The positivity and norm-preserving property are reminiscent of Markov
semigroups.} in probability theory: By the Riesz-Markov representation
theorem and the monotone convergence theorem, there is a unique
two-parameter group $(p_{t,s}^{h})_{s,t\in \mathbb{R}}$\ of Markov
transition kernels $p_{t,s}^{h}(\cdot ,\cdot )$ on $E$ such that%
\begin{equation*}
V_{t,s}^{h}f\left( \rho \right) =\int_{E}f\left( \hat{\rho}\right)
p_{t,s}^{h}(\rho ,\mathrm{d}\hat{\rho})\ ,\qquad f\in \mathfrak{C}^{\mathbb{R%
}}\ .
\end{equation*}%
The right hand side of the above identity makes sense for bounded measurable
functions from $E$ to $\mathbb{R}$. In fact, one can naturally extend $%
(V_{t,s}^{h})_{s,t\in \mathbb{R}}$ to this more general class of functions
on $E$. See (\ref{classical evolution familybaby}).

Note that the notion of Feller evolution system, which is an extension of
Feller semigroups to non-auto%
%TCIMACRO{\TeXButton{\-}{\-}}%
%BeginExpansion
\-%
%EndExpansion
nomous two-parameter families, has been (probably) introduced (only) in 2014 
\cite{Feller}. In contrast with \cite{Feller}, here the usual cocycle
property is replaced by the reverse one and $C_{\infty }(\mathbb{R}^{d})$ by 
$\mathfrak{C}^{\mathbb{R}}$, similar to \cite[Section 8.1.15]{Feller-2} or 
\cite[Definition 1.6]{Feller-3}, because we do not have any differentiable
structure on $E$. In fact, the term \textquotedblleft Feller
semigroup\textquotedblright\ can have different definitions\footnote{%
Feller semigroups have usually the same properties, but they can be defined
on different classes of spaces in the litterature.} in the literature. See,
e.g., \cite[Section 8.1.15]{Feller-2} and \cite[Section 1.1]{Feller-3}.

For any constant function $h\in \mathfrak{Y}\left( \mathbb{R}\right) $
satisfying Conditions (a)-(b) of Theorem \ref{theorem sdfkjsdklfjsdklfj
copy(3)}, $(V_{t,0}^{h})_{t\in \mathbb{R}}$ is therefore a $C_{0}$-group of $%
\ast $-automorphisms of $\mathfrak{C}$ and we denote by $\daleth ^{h}$ its
(well-defined) generator. By \cite[Chap. II, Sect. 3.11]{EngelNagel}, it is
a closed (linear) operator densely defined in $\mathfrak{C}$. Since $%
V_{t,0}^{h}$, $t\in \mathbb{R}$, are $\ast $-automorphisms, we infer from
the Nelson theorem \cite[Theorem 1.5.4]{Bratelli-derivation}, or the
Lumer-Phillips theorem \cite[Theorem 3.1.16]{BrattelliRobinsonI}, that $\pm
\daleth ^{h}$ are dissipative operators, i.e., $\daleth ^{h}$ is
conservative. The $\ast $-morphism property of $V_{t,0}^{h}$, $t\in \mathbb{R%
}$, is reflected by the fact that $\daleth ^{h}$ has to be a symmetric
derivation of $\mathfrak{C}$. This closed derivation is directly related
with a Poissonian symmetric derivation:{}

\begin{theorem}[Generators as Poissonian symmetric derivations]
\label{Closed Poissonian symmetric derivations}\mbox{
}\newline
Assume Conditions (a)-(b) of Theorem \ref{theorem sdfkjsdklfjsdklfj copy(3)}%
. \newline
\emph{(i)} The Poissonian symmetric derivation $\mathfrak{d}^{h}$ of
Definition \ref{Poissonian Symmetric Derivations} is closable. Its closure $%
\mathfrak{\bar{d}}^{h}$ is conservative and equals the generator $\daleth
^{h}\supseteq \mathfrak{\bar{d}}^{h}$ on its domain.\newline
\emph{(ii)} If $\beth \supseteq \mathfrak{\bar{d}}^{h}$ is a conservative
closed operator generating a $C_{0}$-group, then $\beth =\daleth ^{h}$. 
\newline
\emph{(iii)} If $h\in \mathfrak{C}_{\mathcal{X}^{\mathbb{R}}}$ then $%
\mathfrak{\bar{d}}^{h}=\daleth ^{h}$ is the generator of the $C_{0}$-group $%
(V_{t,0}^{h})_{t\in \mathbb{R}}$.
\end{theorem}

\begin{proof}
Fix all assumptions of the theorem. Note first that one can compute $\daleth
^{h}$ for any (elementary) functions of $\{\hat{A}:A\in \mathcal{X}\}$, see (%
\ref{fA}). In the light of the self-consistency equation given by Theorem %
\ref{theorem sdfkjsdklfjsdklfj copy(3)}, which is combined with (\ref%
{notation})-(\ref{phits}) and (\ref{classical evolution familybaby}), note
that, for any $\rho \in E$, $s,t\in {\mathbb{R}}$ and $A\in \mathcal{X}$, 
\begin{equation*}
V_{t,s}^{h}(\hat{A})\left( \rho \right) =\rho \circ T_{t,s}^{\mathbf{\varpi }%
^{h}\left( s,\cdot ;\rho \right) }\left( A\right) \ ,
\end{equation*}%
which, by (\ref{flow baby1}), in turn leads to the equality%
\begin{equation}
\partial _{t}V_{t,s}^{h}(\hat{A})\left( \rho \right) =\mathbf{\varpi }%
^{h}\left( s,t;\rho \right) \circ X_{t}^{\mathbf{\varpi }^{h}\left( s,t;\rho
\right) }\left( A\right) \ .  \label{equality trival}
\end{equation}%
Using Definitions \ref{convex Frechet derivative copy(1)}, \ref{Poissonian
Symmetric Derivations}, Equations (\ref{extension}), (\ref{flow baby0}) and (%
\ref{classical evolution familybaby}) as well as the fact that $%
(V_{t,0}^{h})_{t\in \mathbb{R}}$ is generated by $\daleth ^{h}$, we deduce
from the last equality that%
\begin{equation*}
\daleth ^{h}(\hat{A})=\mathfrak{d}^{h}(\hat{A})\ ,\qquad A\in \mathcal{X}\ .
\end{equation*}%
Since both $\daleth ^{h}$ and $\mathfrak{d}^{h}$ are symmetric derivations,
it follows that 
\begin{equation}
\daleth ^{h}|_{\mathfrak{C}_{\mathcal{X}}}=\mathfrak{d}^{h}\ .
\label{restriction}
\end{equation}%
The operator $\mathfrak{d}^{h}$ is therefore (norm-) closable: For any
sequence $(f_{n})_{n\in \mathbb{N}}\subseteq \mathrm{dom}(\mathfrak{d}^{h})=%
\mathfrak{C}_{\mathcal{X}}$ converging to $0$, if $(\mathfrak{d}%
^{h}(f_{n}))_{n\in \mathbb{N}}$ is a Cauchy sequence then it converges to $0$%
, by (\ref{restriction}) and the closedness of $\daleth ^{h}$, as a
generator of a $C_{0}$-group. Since $\daleth ^{h}$ is conservative, we also
infer from (\ref{restriction}) that both the operator $\mathfrak{d}^{h}$ and
its closure of $\mathfrak{d}^{h}$ are conservative. (See, e.g., \cite[%
Proposition 3.1.15]{BrattelliRobinsonI}.) The generator $\daleth ^{h}$ is a
closed, not necessarily minimal, extension of $\mathfrak{d}^{h}$. This
concludes the proof of Assertion (i). The second one (ii) thus follows from 
\cite[Proposition 3.1.15]{BrattelliRobinsonI}.

To prove Assertion (iii) we use (ii) and the Nelson theorem \cite[Theorem
1.5.4]{Bratelli-derivation}: Pick $h_{1},h_{2}\in \mathfrak{C}_{\mathcal{X}}$%
. Assume without loss of generality that $h_{1},h_{2}$ are both not constant
functions. Then, for any $\ell \in \{1,2\}$ there are $n_{\ell }\in \mathbb{N%
}$, $\left\{ B_{\ell ,j}\right\} _{j=1}^{n_{\ell }}\subseteq \mathcal{X}^{%
\mathbb{R}}$, and $g_{\ell }:\mathbb{R}^{n_{\ell }}\rightarrow \mathbb{R}$
being a polynomial of degree $m_{\ell }\in \mathbb{N}$ such that%
\begin{equation*}
h_{\ell }\left( \rho \right) =g_{\ell }\left( \rho \left( B_{\ell ,1}\right)
,\ldots ,\rho \left( B_{\ell ,n_{\ell }}\right) \right) \ ,\qquad \rho \in
E\ .
\end{equation*}%
Then, from Equation (\ref{Ynbis}) and Definition \ref{convex Frechet
derivative copy(1)}, note that $\mathfrak{d}^{h_{1}}\left( h_{2}\right) \in 
\mathfrak{C}_{\mathcal{X}}$ with 
\begin{align}
\mathfrak{d}^{h_{1}}\left( h_{2}\right) \left( \rho \right) &
=\sum_{j_{1}=1}^{n_{1}}\sum_{j_{2}=1}^{n_{2}}\rho \left( i\left[
B_{1,j_{1}},B_{2,j_{2}}\right] \right) \partial _{x_{j_{1}}}g_{1}\left( \rho
\left( B_{1,1}\right) ,\ldots ,\rho \left( B_{1,n_{1}}\right) \right)
\label{sdfsdfsdf} \\
& \qquad \qquad \qquad \qquad \qquad \qquad \times \partial
_{x_{j_{2}}}g_{2}\left( \rho \left( B_{2,1}\right) ,\ldots ,\rho \left(
B_{2,n_{2}}\right) \right)  \notag
\end{align}%
for any $\rho \in E$. Note that, for any $k\in \mathbb{N}$, 
\begin{equation}
n_{1}^{k}\prod\limits_{j=0}^{k-1}\left( j\left( n_{1}+1\right) +n_{2}\right)
\leq n_{1}^{k}\left( k\left( n_{1}+1\right) +n_{2}\right) ^{k}\leq k!\exp
\left( n_{1}\left( k\left( n_{1}+1\right) +n_{2}\right) \right) \ ,
\label{sdfsdfsdfssdfsdfsdf}
\end{equation}%
because $x^{n}\leq n!\mathrm{e}^{x}$ for all $x\geq 0$ and $n\in \mathbb{N}$%
. Thus, using (\ref{sdfsdfsdf})-(\ref{sdfsdfsdfssdfsdfsdf}) together with
Equations (\ref{metaciagre set 2bis}), (\ref{norm properties}) and
straightforward estimates, one gets that%
\begin{equation*}
\left\Vert (\mathfrak{d}^{h_{1}})^{k}(h_{2})\right\Vert _{\mathfrak{C}}\leq
k!2^{k}\left( 1+D_{0}\right) ^{k}\left( 1+D_{1}\right) ^{k}\left(
1+D_{2}\right) \exp \left( n_{1}\left( k\left( n_{1}+1\right) +n_{2}\right)
\right) \ ,\qquad k\in \mathbb{N}\ .
\end{equation*}%
where 
\begin{equation*}
D_{0}\doteq \max_{\ell \in \{1,2\}}\max_{j\in \left\{ 1,\ldots ,n_{\ell
}\right\} }\left\Vert B_{\ell ,j}\right\Vert _{\mathcal{X}}\ ,\quad D_{\ell
}\doteq \max_{\underline{n}\in \mathbb{N}_{0}^{m_{\ell }}}\left\{ \max_{\rho
\in E}\left\vert \partial ^{\underline{n}}g_{\ell }\left( \rho \left(
B_{\ell ,1}\right) ,\ldots ,\rho \left( B_{\ell ,n_{\ell }}\right) \right)
\right\vert \right\}
\end{equation*}%
for $\ell \in \{1,2\}$. It follows that 
\begin{equation*}
\sum\limits_{k\in \mathbb{N}}\frac{t^{k}}{k!}\left\Vert (\mathfrak{d}%
^{h_{1}})^{k}(h_{2})\right\Vert _{\mathfrak{C}}<\infty
\end{equation*}%
for some positive time $t$ satisfying 
\begin{equation*}
0\leq t<\frac{\mathrm{e}^{-n_{1}\left( n_{1}+1\right) }}{2\left(
1+D_{0}\right) \left( 1+D_{1}\right) \left( 1+D_{2}\right) }\ .
\end{equation*}%
Therefore, by density of $\mathfrak{C}_{\mathcal{X}}$ in $\mathfrak{C}$, the
conservative, densely defined, closed operator $\mathfrak{\bar{d}}^{h_{1}}$
has a dense set of analytic elements. By the Nelson theorem \cite[Theorem
1.5.4]{Bratelli-derivation}, $\mathfrak{\bar{d}}^{h_{1}}$ is a conservative
closed operator generating a $C_{0}$-group of $\ast $-automorphisms of $%
\mathfrak{C}$, whence Assertion (iii), following (ii).
\end{proof}

Note that Equation (\ref{equality trival}) holds true for \emph{any} $h\in
C_{b}\left( \mathbb{R};\mathfrak{Y}\left( \mathbb{R}\right) \right) $
satisfying Conditions (a)-(b) of Theorem \ref{theorem sdfkjsdklfjsdklfj
copy(3)}.\ It follows that, for any $s,t\in \mathbb{R}$ and polynomial
function $f\in \mathfrak{C}_{\mathcal{X}}$, 
\begin{equation}
\partial _{t}V_{t,s}^{h}\left( f\right) =V_{t,s}^{h}\left( \left\{ h\left(
t\right) ,f\right\} \right) \ .  \label{ewr}
\end{equation}%
A similar expression for $\partial _{s}V_{t,s}^{h}$ like 
\begin{equation}
\partial _{s}V_{t,s}^{h}\left( f\right) =-\left\{ h\left( s\right)
,V_{t,s}^{h}(f)\right\}  \label{ewrbis}
\end{equation}%
is less obvious. First, we do not know, a priori, if\ $V_{t,s}^{h}$ maps
elements from $\mathfrak{C}_{\mathcal{X}}$ to continuously differentiable
complex-valued functions on $E$, i.e., if $V_{t,s}^{h}\left( \mathfrak{C}_{%
\mathcal{X}}\right) \subseteq \mathfrak{Y}\left( \mathbb{C}\right) $.
Secondly, even if $V_{t,s}^{h}\left( \mathfrak{C}_{\mathcal{X}}\right)
\subseteq \mathfrak{Y}\left( \mathbb{R}\right) $, one still has to prove
that Equation (\ref{ewrbis}) holds true. This is done in the next theorem:\ 

\begin{theorem}[Non-autonomous classical dynamics]
\label{proposition dynamic classique II}\mbox{
}\newline
Take $h\in \mathfrak{Z}$. Then, for any $s,t\in \mathbb{R}$ and $f\in 
\mathfrak{C}_{\mathcal{X}}$, (\ref{ewr})-(\ref{ewrbis}) hold true. See (\ref%
{Zn}) for the definition of $\mathfrak{Z}$.
\end{theorem}

\begin{proof}
Note that any function $h\in \mathfrak{Z}$ satisfies Conditions (a)-(b) of
Theorem \ref{theorem sdfkjsdklfjsdklfj copy(3)}. Equation (\ref{ewr}) is
already discussed before the theorem: it results from (\ref{equality trival}%
) for $h\in C_{b}\left( \mathbb{R};\mathfrak{Y}\left( \mathbb{R}\right)
\right) $ and properties of derivatives and symmetric derivations (linearity
and Leibniz's rule, see, e.g., (\ref{derivation0})). To prove (\ref{ewrbis}%
), it suffices to invoke Corollary \ref{Corollary bije+cocylbaby copy(1)},
which says that%
\begin{equation*}
\partial _{s}V_{t,s}^{h}(\hat{A})=-\{h\left( s\right) ,V_{t,s}^{h}(\hat{A})\}
\end{equation*}%
for any $s,t\in \mathbb{R}$ and $A\in \mathcal{X}$. Since $%
(V_{t,s}^{h})_{s,t\in \mathbb{R}}$ is a family of $\ast $-automorphisms of $%
\mathfrak{C}$, by using the (bi)linearity and Leibniz's rule satisfied by
the derivatives and the bracket $\{\cdot ,\cdot \}$, we deduce (\ref{ewrbis}%
) for all polynomial functions $f\in \mathfrak{C}_{\mathcal{X}}$.
\end{proof}

This theorem applied to the autonomous situation leads to the dynamical
equation of classical mechanics (see, e.g., \cite[Proposition 10.2.3]%
{classical-dynamics}), i.e., (autonomous) \emph{Liouville's equation}, which
reads in our case as follows:

\begin{corollary}[Autonomous Liouville's equation]
\label{proposition dynamic classique IIbis}\mbox{
}\newline
Take $h\in \mathfrak{Z}$ constant in time. Then, for any $t\in \mathbb{R}$
and $f\in \mathfrak{C}_{\mathcal{X}}$,%
\begin{equation}
\partial _{t}V_{t,0}^{h}\left( f\right) =V_{t,0}^{h}\circ \daleth ^{h}\left(
f\right) =V_{t,0}^{h}\left( \left\{ h,f\right\} \right) =\left\{
h,V_{t,0}^{h}(f)\right\} =\daleth ^{h}\circ V_{t,0}^{h}\left( f\right) \ .
\label{classical dynamicsbis}
\end{equation}
\end{corollary}

\begin{proof}
Combine Theorem \ref{proposition dynamic classique II} with Theorem \ref%
{Closed Poissonian symmetric derivations}.
\end{proof}

In the non-autonomous case, $(V_{t,s}^{h})_{s,t\in \mathbb{R}}$ is a
strongly continuous two-parameter family of $\ast $-automorphisms of $%
\mathfrak{C}$ solving the non-auto%
%TCIMACRO{\TeXButton{\-}{\-}}%
%BeginExpansion
\-%
%EndExpansion
nomous evolution equations 
\begin{equation*}
\forall s,t\in {\mathbb{R}}:\qquad \partial _{t}V_{t,s}^{h}=V_{t,s}^{h}\circ
\daleth ^{h\left( t\right) }\ ,\qquad V_{s,s}^{h}=\mathbf{1}_{\mathfrak{C}}\
,
\end{equation*}%
on $\mathfrak{C}_{\mathcal{X}}$, as explained before Theorem \ref%
{proposition dynamic classique II}. See also Theorem \ref{Closed Poissonian
symmetric derivations}. Theorem \ref{proposition dynamic classique II}
suggests that, under stronger conditions, $(V_{t,s}^{h})_{s,t\in \mathbb{R}}$
is the solution to the non-auto%
%TCIMACRO{\TeXButton{\-}{\-}}%
%BeginExpansion
\-%
%EndExpansion
nomous evolution equations%
\begin{equation}
\forall s,t\in {\mathbb{R}}:\qquad \partial _{s}V_{t,s}^{h}=-\daleth
^{h\left( s\right) }\circ V_{t,s}^{h}\ ,\qquad V_{s,s}^{h}=\mathbf{1}_{%
\mathfrak{C}}\ ,  \label{(*)}
\end{equation}%
on some dense subspace. To prove this, one could look for assumptions on $h$
such that the family $(\daleth ^{h\left( t\right) })_{t\in \mathbb{R}}$ of
closed dissipative operators satisfies sufficient conditions to generate an
evolution family solving (\ref{(*)}), as explained in \cite%
{Katobis,Caps,Schnaubelt1,Pazy,Bru-Bach}. Then, $(V_{t,s}^{h})_{s,t\in 
\mathbb{R}}$ would be \emph{the} solution to the non-auto%
%TCIMACRO{\TeXButton{\-}{\-}}%
%BeginExpansion
\-%
%EndExpansion
nomous evolution equation (\ref{(*)}). This looks doable, but at the cost of
many technical arguments. We thus refrain from doing such a study in this
paper.

\section{State-Dependent $C^{\ast }$-Dynamical Systems\label{section
extended C* dynamical systeme}}

\subsection{Quantum $C^{\ast }$-Algebras of Continuous Functions on State
Space\label{Quantum Algebras}}

The space $C(E;\mathcal{X})$ of $\mathcal{X}$-valued weak$^{\ast }$%
-continuous functions on the weak$^{\ast }$-compact space $E$ is a unital $%
C^{\ast }$-algebra with respect to the point-wise operations, denoted by 
\begin{equation}
\mathfrak{X}\doteq \left( C(E;\mathcal{X}),+,\cdot _{{\mathbb{C}}},\times
,^{\ast },\left\Vert \cdot \right\Vert _{\mathfrak{X}}\right)
\label{metaciagre set}
\end{equation}%
where%
\begin{equation}
\left\Vert f\right\Vert _{\mathfrak{X}}\doteq \max_{\rho \in E}\left\Vert
f\left( \rho \right) \right\Vert _{\mathcal{X}}\ ,\qquad f\in \mathfrak{X}\ .
\label{metaciagre set bis}
\end{equation}%
Clearly, $\mathfrak{X}$ is commutative iff $\mathcal{X}$ is commutative. The
(real) Banach subspace of all $\mathcal{X}^{\mathbb{R}}$-valued functions
from $\mathfrak{X}$ is denoted by $\mathfrak{X}^{\mathbb{R}}\varsubsetneq 
\mathfrak{X}$. $\mathfrak{X}$ is separable whenever $\mathcal{X}$ is
separable, $E$ being in this case metrizable.

We identify the primordial $C^{\ast }$-algebra $\mathcal{X}$, on which the
quantum dynamics is usually defined, with the subalgebra of constant
functions of $\mathfrak{X}$. Meanwhile, the classical dynamics appears in
the space $\mathfrak{C}\doteq C(E;\mathbb{C})$ of complex-valued weak$^{\ast
}$-continuous functions on $E$. See (\ref{metaciagre set 2})-(\ref%
{metaciagre set 2bis}). This unital \emph{commutative} $C^{\ast }$-algebra
is thus identified with the subalgebra of functions of $\mathfrak{X}$ whose
values are multiples of the unit $\mathfrak{1}\in \mathcal{X}$. Compare (\ref%
{metaciagre set})-(\ref{metaciagre set bis}) with (\ref{metaciagre set 2})-(%
\ref{metaciagre set 2bis}). Hence, we have the inclusions%
\begin{equation}
\mathcal{X}\subseteq \mathfrak{X}\qquad \text{and}\qquad \mathfrak{%
C\subseteq X}\ .  \label{embeding}
\end{equation}%
Both classical and quantum dynamics can then be extended to $\mathfrak{X}$.
This is explained in the next subsection.

\subsection{State-Dependent Quantum Dynamics\label{Extended Quantum Dynamics}%
}

Since $\mathfrak{C\subseteq X}$, there is a natural extension to $\mathfrak{X%
}$ of the classical dynamics on $\mathfrak{C}$: The continuous family $%
\mathbf{\varpi }^{h}$ of Theorem \ref{theorem sdfkjsdklfjsdklfj copy(3)}
yields a family $(\mathfrak{V}_{t,s}^{h})_{s,t\in \mathbb{R}}$ of $\ast $%
-automorphisms of $\mathfrak{X}$ defined by%
\begin{equation}
\mathfrak{V}_{t,s}^{h}\left( f\right) \doteq f\circ \mathbf{\varpi }%
^{h}\left( s,t\right) \ ,\qquad f\in \mathfrak{X},\ s,t\in \mathbb{R}\ .
\label{def fausse}
\end{equation}%
In particular, by (\ref{classical evolution familybaby}), $\mathfrak{V}%
_{t,s}^{h}|_{\mathfrak{C}}=V_{t,s}^{h}$\ for any $s,t\in \mathbb{R}$.
However, it is \emph{not} what we have in mind here: Emphasizing rather the
inclusion $\mathcal{X}\subseteq \mathfrak{X}$, the classical algebra $%
\mathfrak{C}$ will become a subalgebra of the \emph{fixed-point} algebra of
the state-dependent dynamics we define below on $\mathfrak{X}$.

In Section \ref{Section QM}, we explain how a fixed function $h\in
C_{b}\left( \mathbb{R};\mathfrak{Y}\left( \mathbb{R}\right) \right) $ is
used to define (possibly non-autonomous) quantum dynamics $(T_{t,s}^{\xi
})_{s,t\in \mathbb{R}}$ on the primordial $C^{\ast }$-algebra $\mathcal{X}$,
for any $\xi \in C\left( \mathbb{R};E\right) $. This primal dynamics induces
classical dynamics on the (classical) $C^{\ast }$-algebra $\mathfrak{C}%
\doteq C(E;\mathbb{C})$ of continuous functions on states, as discussed in
Sections \ref{Self-Consistency Equations}-\ref{Section CM}. By Theorem \ref%
{theorem sdfkjsdklfjsdklfj copy(3)}, it yields, in turn, a \emph{%
state-dependent} quantum dynamics, referring in this case to a
norm-continuous family 
\begin{equation*}
(\mathrm{T}_{t,s}^{\rho })_{\left( \rho ,s,t\right) \in E\times \mathbb{R}%
^{2}}=(T_{t,s}^{\mathbf{\varpi }^{h}\left( s_{0},\cdot ;\rho \right)
})_{\left( \rho ,s,t\right) \in E\times \mathbb{R}^{2}}
\end{equation*}%
of $\ast $-automorphisms of $\mathcal{X}$ for some \emph{fixed} $s_{0}\in 
\mathbb{R}$. This leads to a (state-dependent) dynamics on the (secondary) $%
C^{\ast }$-algebra $\mathfrak{X}$ of continuous functions on states.

As a matter of fact, any strongly continuous family $(\mathrm{T}^{\rho
})_{\rho \in E}$ of linear contractions from $\mathcal{X}$ to itself can be
viewed as a linear contraction $\mathfrak{T}$ from $\mathfrak{X}$ to itself
defined by 
\begin{equation}
\left[ \mathfrak{T}\left( f\right) \right] \left( \rho \right) \doteq 
\mathrm{T}^{\rho }\left( f\left( \rho \right) \right) \ ,\qquad \rho \in E,\
f\in \mathfrak{X}\ .  \label{cigare0}
\end{equation}%
Such contractions have the following properties:

\begin{lemma}[State-dependent quantum dynamics]
\label{lemma supercigare}\mbox{ }\newline
Let $\mathcal{X}$ be a unital $C^{\ast }$-algebra. For any $s,t\in \mathbb{R}%
^{2}$, let $(\mathrm{T}_{t,s}^{\rho })_{\rho \in E}$ be any strongly
continuous family of linear contractions from $\mathcal{X}$ to itself, and $%
\mathfrak{T}_{t,s}$ be defined by (\ref{cigare0}) with $\mathrm{T}^{\rho }=%
\mathrm{T}_{t,s}^{\rho }$. \newline
\emph{(i)} If $\mathrm{T}_{t,s}^{\rho }$ is a $\ast $-automorphism of $%
\mathcal{X}$ at $s,t\in \mathbb{R}$ for any $\rho \in E$, then $\mathfrak{T}%
_{t,s}$ is a $\ast $-automorphism of $\mathfrak{X}$ and the classical
subalgebra $\mathfrak{C}\subseteq \mathfrak{X}$ is contained in the
fixed-point algebra of $\mathfrak{T}_{t,s}$, i.e., 
\begin{equation*}
\mathfrak{T}_{t,s}(f)=f\ ,\qquad f\in \mathfrak{C}\ .
\end{equation*}%
\emph{(ii)} If $(\mathrm{T}_{t,s}^{\rho })_{s,t\in \mathbb{R}}$ satisfies a
reverse cocycle property for any $\rho \in E$, i.e., 
\begin{equation}
\mathrm{T}_{t,s}^{\rho }=\mathrm{T}_{r,s}^{\rho }\circ \mathrm{T}%
_{t,r}^{\rho }\ ,\qquad \rho \in E,\ s,t,r\in \mathbb{R}\ ,
\label{reverse property}
\end{equation}%
then $(\mathfrak{T}_{t,s})_{s,t\in \mathbb{R}}$ has also this property.%
\newline
\emph{(iii)} If $(\mathrm{T}_{t,s}^{\rho })_{\left( \rho ,s,t\right) \in
E\times \mathbb{R}^{2}}$ is a strongly continuous family of contractions
then so do $(\mathfrak{T}_{t,s})_{s,t\in \mathbb{R}}$.
\end{lemma}

\begin{proof}
Assertion (i)-(ii) directly follows from (\ref{cigare0}) and it remains to
prove (iii). By contradiction, suppose that the family is not strongly
continuous. Then, there is $f\in \mathfrak{X}$, times $s,t\in \mathbb{R}$,
two zero nets $(\eta _{j})_{j\in J},(\varkappa _{j})_{j\in J}\subseteq 
\mathbb{R}$, a net $(\rho _{j})_{j\in J}\subseteq E$ of states and a
positive constant $D>0$ such that%
\begin{equation*}
\inf_{j\in J}\left\Vert \mathrm{T}_{t+\eta _{j},s+\varkappa _{j}}^{\rho
_{j}}\left( f\left( \rho _{j}\right) \right) -\mathrm{T}_{t,s}^{\rho
_{j}}\left( f\left( \rho _{j}\right) \right) \right\Vert _{\mathcal{X}}\geq
D>0\ .
\end{equation*}%
By weak$^{\ast }$ compactness of $E$, we can assume without loss of
generality that $(\rho _{j})_{j\in J}$ converges to some $\rho \in E$.
Because $(\mathrm{T}_{t,s}^{\rho })_{\left( \rho ,s,t\right) \in E\times 
\mathbb{R}^{2}}$ is a family of contractions, the above bound yields 
\begin{equation*}
\liminf_{j\in J}\left\Vert \mathrm{T}_{t+\eta _{j},s+\varkappa _{j}}^{\rho
_{j}}\left( f\left( \rho \right) \right) -\mathrm{T}_{t,s}^{\rho _{j}}\left(
f\left( \rho \right) \right) \right\Vert _{\mathcal{X}}\geq D>0\ ,
\end{equation*}%
which contradicts the strong continuity of this family.
\end{proof}

If $(\mathrm{T}_{t,s}^{\rho })_{\left( \rho ,s,t\right) \in E\times \mathbb{R%
}^{2}}$ is a family of $\ast $-automorphisms of $\mathcal{X}$ then $(%
\mathfrak{T}_{t,s})_{s,t\in \mathbb{R}}$ is a family of $\ast $%
-automorphisms of $\mathfrak{X}$ and, by Lemma \ref{lemma supercigare} (i),
the classical subalgebra $\mathfrak{C}\subseteq \mathfrak{X}$ is contained
in the fixed-point algebra of the full quantum dynamics $(\mathfrak{T}%
_{t,s})_{s,t\in \mathbb{R}}$, i.e., for all $f\in \mathfrak{C}$ and all $%
s,t\in \mathbb{R}$, $\mathfrak{T}_{t,s}(f)=f$. Any family $(\mathfrak{T}%
_{t,s})_{s,t\in \mathbb{R}}$ of $\ast $-automorphisms of $\mathfrak{X}$
preserving each element of $\mathfrak{C}$ is of this form, at least when $%
\mathcal{X}$ is separable:

\begin{lemma}[State-dependent quantum dynamics and fixed-point algebra]
\label{lemma supercigare bis}\mbox{ }\newline
Let $\mathcal{X}$ be a separable, unital $C^{\ast }$-algebra. The classical
subalgebra $\mathfrak{C}\subseteq \mathfrak{X}$ is contained in the
fixed-point algebra of a strongly continuous two-parameter family $(%
\mathfrak{T}_{t,s})_{s,t\in \mathbb{R}}$ of $\ast $-automorphisms of $%
\mathfrak{X}$ iff there is a strongly continuous family $(\mathrm{T}%
_{t,s}^{\rho })_{\left( \rho ,s,t\right) \in E\times \mathbb{R}^{2}}$ of $%
\ast $-automorphisms of $\mathcal{X}$ satisfying (\ref{cigare0}).
\end{lemma}

\begin{proof}
In order to obtain the equivalence stated in the lemma, it only remains to
prove that any strongly continuous family $(\mathfrak{T}_{t,s})_{s,t\in 
\mathbb{R}}$ of $\ast $-automorphisms of $\mathfrak{X}$ whose fixed-point
algebra contains $\mathfrak{C}$ comes from the strongly continuous family $(%
\mathrm{T}_{t,s}^{\rho })_{\left( \rho ,s,t\right) \in E\times \mathbb{R}%
^{2}}$ of $\ast $-homomorphisms defined by 
\begin{equation}
\mathrm{T}_{t,s}^{\rho }\left( A\right) \doteq \left[ \mathfrak{T}%
_{t,s}\left( A\right) \right] \left( \rho \right) \ ,\qquad \rho \in E,\
A\in \mathcal{X}\subseteq \mathfrak{X},\ s,t\in \mathbb{R}\ .
\label{dfsdklfjsdklfjdslkjfsldf}
\end{equation}%
To this end, recall that, if $\mathcal{X}$ is separable then $E$ is
metrizable. So, take a distance $d(\cdot ,\cdot )$ generating the weak$%
^{\ast }$ topology on $E$. For any $\rho \in E$ define the sequence $%
\{g_{n}\}_{n\in \mathbb{N}}\subseteq \mathfrak{C}$ of continuous functions
by 
\begin{equation*}
g_{n}\left( \tilde{\rho}\right) =\frac{1}{1+nd\left( \tilde{\rho},\rho
\right) }\ ,\qquad \tilde{\rho}\in E,\ n\in \mathbb{N}\ .
\end{equation*}%
Since, by assumption, $\mathfrak{T}_{t,s}$ is a $\ast $-automorphism of $%
\mathfrak{X}$ satisfying $\mathfrak{T}_{t,s}(g_{n})=g_{n}$ for $s,t\in 
\mathbb{R}$ and $n\in \mathbb{N}$, we note that, for every fixed $\rho \in E$%
, $s,t\in \mathbb{R}$, $n\in \mathbb{N}$ and all functions $f\in \mathfrak{X}
$,%
\begin{equation*}
\left[ \mathfrak{T}_{t,s}\left( f\right) \right] \left( \rho \right) =\left[ 
\mathfrak{T}_{t,s}\left( fg_{n}-f\left( \rho \right) g_{n}\right) \right]
\left( \rho \right) +\mathrm{T}_{t,s}^{\rho }\left( f\left( \rho \right)
\right) \ .
\end{equation*}%
Because $\mathfrak{T}_{t,s}$ is a contraction (for it is a $\ast $%
-automorphism), by continuity of $f\in \mathfrak{X}$, it follows that 
\begin{equation*}
\lim_{n\rightarrow \infty }\left\Vert \mathfrak{T}_{t,s}\left(
fg_{n}-f\left( \rho \right) g_{n}\right) \right\Vert _{\mathfrak{X}%
}=\lim_{n\rightarrow \infty }\left\Vert fg_{n}-f\left( \rho \right)
g_{n}\right\Vert _{\mathfrak{X}}=0\ ,
\end{equation*}%
and hence,%
\begin{equation*}
\left[ \mathfrak{T}_{t,s}\left( f\right) \right] \left( \rho \right) =%
\mathrm{T}_{t,s}^{\rho }\left( f\left( \rho \right) \right) \ ,\qquad \rho
\in E,\ f\in \mathfrak{X},\ s,t\in \mathbb{R}\ .
\end{equation*}%
From the last equality we also conclude that $\mathrm{T}_{t,s}^{\rho }$ is a 
$\ast $-automorphism of $\mathcal{X}$ for all $\left( \rho ,s,t\right) \in
E\times \mathbb{R}^{2}$.
\end{proof}

The above situation motivates the following notion of \emph{state-dependent }%
$C^{\ast }$-dynamical system:

\begin{definition}[State-dependent $C^{\ast }$-dynamical systems]
\label{Extended dynamical systems}\mbox{ }\newline
If $\mathfrak{T}\equiv (\mathfrak{T}_{t,s})_{s,t\in \mathbb{R}}$ is a
strongly continuous two-parameter family of $\ast $-automorphisms of $%
\mathfrak{X}$ preserving each element of $\mathfrak{C}\subseteq \mathfrak{X}$
and satisfying the reverse cocycle property 
\begin{equation*}
\mathfrak{T}_{t,s}=\mathfrak{T}_{r,s}\circ \mathfrak{T}_{t,r}\ ,\qquad
s,t,r\in \mathbb{R}\ ,
\end{equation*}%
then we name the pair $(\mathfrak{X},\mathfrak{T})$ \textquotedblleft
state-dependent $C^{\ast }$-dynamical system\textquotedblright .
\end{definition}

\noindent An example of such a $C^{\ast }$-dynamical system is given from
Theorem \ref{theorem sdfkjsdklfjsdklfj copy(3)} via the family $\mathfrak{T}%
^{h,s_{0}}\equiv (\mathfrak{T}_{t,s}^{h,s_{0}})_{s,t\in \mathbb{R}}$ of $%
\ast $-automorphisms of $\mathfrak{X}$ defined by%
\begin{equation*}
\left[ \mathfrak{T}_{t,s}^{h,s_{0}}\left( f\right) \right] \left( \rho
\right) \doteq T_{t,s}^{\mathbf{\varpi }^{h}\left( s_{0},\cdot ;\rho \right)
}\left( f\left( \rho \right) \right) \ ,\qquad \rho \in E,\ f\in \mathfrak{X}%
,\ s,t\in \mathbb{R}\ ,
\end{equation*}%
for any fixed $s_{0}\in \mathbb{R}$ and every (time-depending classical
Hamiltonian) $h\in C_{b}\left( \mathbb{R};\mathfrak{Y}\left( \mathbb{R}%
\right) \right) $ satisfying all assumptions of Theorem \ref{theorem
sdfkjsdklfjsdklfj copy(3)}. This is a state-dependent $C^{\ast }$-dynamical
system:

\begin{lemma}[From self-consistency equations to state-dependent quantum
dynamics]
\label{lemma supercigare bis copy(1)}\mbox{ }\newline
Assume Conditions (a)-(b) of Theorem \ref{theorem sdfkjsdklfjsdklfj copy(3)}%
. Then, for any $s_{0}\in \mathbb{R}$, $(\mathfrak{X},\mathfrak{T}%
^{h,s_{0}}) $ is a state-dependent $C^{\ast }$-dynamical system.
\end{lemma}

\begin{proof}
Fix all parameters of the lemma. By Lemma \ref{lemma supercigare} (i), $%
\mathfrak{T}_{t,s}^{h,s_{0}}$ is $\ast $-automorphism of $\mathfrak{X}$ and
the classical subalgebra $\mathfrak{C}\subseteq \mathfrak{X}$ is contained
in the fixed-point algebra of $\mathfrak{T}_{t,s}^{h,s_{0}}$ for any $s,t\in 
\mathbb{R}$. From Lemma \ref{lemma supercigare} (ii), $\mathfrak{T}%
^{h,s_{0}} $ clearly satisfies the reverse cocycle property. Moreover, by
Lemma \ref{Solution selfbaby copy(1)} and Theorem \ref{theorem
sdfkjsdklfjsdklfj copy(3)}, we can infer from Lemma \ref{lemma supercigare}
(iii) that $\mathfrak{T}^{h,s_{0}}$ is strongly continuous.
\end{proof}

Exactly like the classical dynamics defined in Section \ref{Section CM},
state-dependent $C^{\ast }$-dynamical systems $(\mathfrak{X},\mathfrak{T})$
induce Feller dynamics within the (classical) commutative $C^{\ast }$%
-algebra $\mathfrak{C}$:

\begin{itemize}
\item Recall that $\mathrm{Aut}\left( E\right) $ is the space of all
automorphisms (or self-homeomorphisms) of the state space $E$, endowed with
the topology of uniform convergence of weak$^{\ast }$-continuous functions.

\item From the family $(\mathfrak{T}_{t,s})_{s,t\in \mathbb{R}}$, we define
a continuous family $(\phi _{t,s})_{s,t\in \mathbb{R}}\subseteq \mathrm{Aut}%
\left( E\right) $ by 
\begin{equation}
\phi _{t,s}\left( \rho \right) \doteq \rho \circ \mathrm{T}_{t,s}^{\rho }\
,\qquad \rho \in E,\ s,t\in \mathbb{R}\ ,  \label{phitsbis}
\end{equation}%
where $(\mathrm{T}_{t,s}^{\rho })_{\left( \rho ,s,t\right) \in E\times 
\mathbb{R}^{2}}$ is a strongly continuous family of $\ast $-automorphisms of 
$\mathcal{X}$ satisfying (\ref{cigare0}). See Lemma \ref{lemma supercigare
bis}. Compare with Equation (\ref{phits}). Similar to Corollary \ref%
{corollary conservation},%
\begin{equation}
\phi _{t,s}\left( \mathcal{E}(E)\right) \subseteq \mathcal{E}(E)\qquad \text{%
and}\qquad \phi _{t,s}(\overline{\mathcal{E}(E)})\subseteq \overline{%
\mathcal{E}(E)}\ .  \label{corollary conservation0}
\end{equation}

\item This family in turn yields a strongly continuous two-parameter family $%
(V_{t,s})_{s,t\in \mathbb{R}}$ of $\ast $-auto%
%TCIMACRO{\TeXButton{\-}{\-}}%
%BeginExpansion
\-%
%EndExpansion
morphisms of $\mathfrak{C}$ defined by 
\begin{equation}
V_{t,s}f\doteq f\circ \phi _{t,s}\ ,\qquad f\in \mathfrak{C},\ s,t\in 
\mathbb{R}\ .  \label{shorodinger dynamicsbis}
\end{equation}%
Compare with Equation (\ref{classical evolution familybaby}). Moreover, by (%
\ref{corollary conservation0}), this map can also be defined in the same way
on $C(\overline{\mathcal{E}(E)};\mathbb{C})$, where we recall that $%
\overline{\mathcal{E}(E)}$ is the phase space of Definition \ref{phase space}%
.

\item If (\ref{reverse property}) holds true, then $(V_{t,s})_{s,t\in 
\mathbb{R}}$ satisfies a reverse cocycle property, i.e., for $s,t,r\in 
\mathbb{R}$, $V_{t,s}=V_{t,r}\circ V_{r,s}$. This classical dynamics is a 
\emph{Feller evolution system}, as defined in Section \ref{Section CM}.
Compare with Proposition \ref{lemma poisson copy(1)}.

\item If, for any $\rho \in E$, the strongly continuous family $(\mathrm{T}%
_{t,s}^{\rho })_{s,t\in \mathbb{R}}$ of $\ast $-automorphisms defined by (%
\ref{cigare0}) satisfies in $\mathcal{B}(\mathcal{X})$ some non-auto%
%TCIMACRO{\TeXButton{\-}{\-}}%
%BeginExpansion
\-%
%EndExpansion
nomous evolution equation, then the family $(V_{t,s})_{s,t\in \mathbb{R}}$
would also satisfy some non-auto%
%TCIMACRO{\TeXButton{\-}{\-}}%
%BeginExpansion
\-%
%EndExpansion
nomous evolution equation, as discussed at the end of Section \ref{Section
CM}.
\end{itemize}

\subsection{State-Dependent Symmetries and Classical Dynamics\label{Symmetry
Group}}

Fix a state-dependent $C^{\ast }$-dynamical system $(\mathfrak{X},\mathfrak{T%
})$. See Definition \ref{Extended dynamical systems}. A \emph{%
state-dependent symmetry} of $(\mathfrak{X},\mathfrak{T})$ is defined as
follows:

\begin{definition}[State-dependent symmetry]
\label{symmetry 1}\mbox{ }\newline
A state-dependent symmetry $\mathfrak{G}$ of $(\mathfrak{X},\mathfrak{T})$
is a $\ast $-automorphism of $\mathfrak{X}$ satisfying 
\begin{equation*}
\mathfrak{G}\circ \mathfrak{T}_{t,s}=\mathfrak{T}_{t,s}\circ \mathfrak{G}\
,\qquad s,t\in \mathbb{R}\ ,
\end{equation*}%
and with fixed-point algebra containing $\mathfrak{C}\subseteq \mathfrak{X}$.
\end{definition}

If $\mathfrak{G}$ is a state-dependent symmetry of $(\mathfrak{X},\mathfrak{T%
})$, then, similar to Lemma \ref{lemma supercigare bis}, the equalities 
\begin{equation*}
G^{\rho }\left( A\right) \doteq \left[ \mathfrak{G}\left( A\right) \right]
\left( \rho \right) \ ,\qquad \rho \in E,\ A\in \mathcal{X}\subseteq 
\mathfrak{X}\ ,
\end{equation*}%
define a strongly continuous family $(G^{\rho })_{\rho \in E}$ of $\ast $%
-automorphisms of $\mathcal{X}$. In this case, we define the weak$^{\ast }$%
-compact space%
\begin{equation}
E_{\mathfrak{G}}\doteq \left\{ \rho \in E:\rho \circ G^{\rho }=\rho \right\}
\label{G-invariance states}
\end{equation}%
of $\mathfrak{G}$-invariant states. By Equation (\ref{phitsbis}) and
Definition \ref{symmetry 1}, together with (\ref{cigare0}) for $\mathfrak{T}=%
\mathfrak{T}_{t,s}$ and $\mathrm{T}^{\rho }=\mathrm{T}_{t,s}^{\rho }$, it
follows that%
\begin{equation}
\phi _{t,s}\left( E_{\mathfrak{G}}\right) \subseteq E_{\mathfrak{G}}\qquad 
\text{and}\qquad \phi _{t,s}\left( E\backslash E_{\mathfrak{G}}\right)
\subseteq E\backslash E_{\mathfrak{G}}\ ,\qquad s,t\in \mathbb{R}\ .
\label{symmetry group eq1}
\end{equation}%
In particular, by Equation (\ref{shorodinger dynamicsbis}), for any function 
$f\in \mathfrak{C}$ and times $s,t\in \mathbb{R}$, 
\begin{equation}
V_{t,s}\left( f|_{E_{\mathfrak{G}}}\right) \doteq \left( V_{t,s}f\right)
|_{E_{\mathfrak{G}}}\qquad \text{and}\qquad V_{t,s}\left( f|_{E\backslash E_{%
\mathfrak{G}}}\right) \doteq \left( V_{t,s}f\right) |_{E\backslash E_{%
\mathfrak{G}}}  \label{symmetry group eq2}
\end{equation}%
well define two two-parameter families of $\ast $-automorphisms respectively
acting on%
\begin{equation}
\left\{ f|_{E_{\mathfrak{G}}}:f\in \mathfrak{C}\right\} \qquad \text{and}%
\qquad \left\{ f|_{E\backslash E_{\mathfrak{G}}}:f\in \mathfrak{C}\right\} \
.  \label{symmetry group eq2bis}
\end{equation}

More generally, we can consider a faithful group homomorphism $g\mapsto 
\mathfrak{G}_{g}$ from a group $\mathrm{G}$ to the group of $\ast $%
-automorphisms of $\mathfrak{X}$. Then, a \emph{state-dependent symmetry
group} is defined as follows:

\begin{definition}[State-dependent symmetry group]
\label{Extended dynamical systems copy(1)}\mbox{ }\newline
A state-dependent symmetry group of $(\mathfrak{X},\mathfrak{T})$ is a group 
$(\mathfrak{G}_{g})_{g\in \mathrm{G}}$ of state-dependent symmetries of $(%
\mathfrak{X},\mathfrak{T})$.
\end{definition}

If $(\mathfrak{G}_{g})_{g\in \mathrm{G}}$ is a state-dependent symmetry
group of $(\mathfrak{X},\mathfrak{T})$, then, defining the weak$^{\ast }$%
-compact space 
\begin{equation}
E_{\mathrm{G}}\doteq \left\{ \rho \in E:\rho \circ G_{g}^{\rho }=\rho \quad 
\text{for}\ \text{all}\ g\in \mathrm{G}\right\}  \label{G-invariance states2}
\end{equation}%
of $\mathrm{G}$-invariant states, we observe that 
\begin{equation}
\phi _{t,s}\left( E_{\mathrm{G}}\right) \subseteq E_{\mathrm{G}}\ ,\qquad
\phi _{t,s}\left( E\backslash E_{\mathrm{G}}\right) \subseteq E\backslash E_{%
\mathrm{G}}\ ,\qquad s,t\in \mathbb{R}\ ,  \label{G-invariance states3}
\end{equation}%
(cf. (\ref{symmetry group eq1})) and, exactly like in Equations (\ref%
{symmetry group eq2})-(\ref{symmetry group eq2bis}), we infer from (\ref%
{shorodinger dynamicsbis}) the existence of two-parameter families of $\ast $%
-automorphisms respectively defined on%
\begin{equation*}
\left\{ f|_{E_{\mathrm{G}}}:f\in \mathfrak{C}\right\} \qquad \text{and}%
\qquad \left\{ f|_{E\backslash E_{\mathrm{G}}}:f\in \mathfrak{C}\right\} \ .
\end{equation*}

\subsection{Reduction of Classical Dynamics via Invariant Subspaces\label%
{Invariant subspaces}}

Any family $\mathcal{B}\subseteq \mathcal{X}$ defines an equivalence
relation 
\begin{equation*}
\heartsuit _{\mathcal{B}}\doteq \left\{ \left( \rho _{1},\rho _{2}\right)
\in E^{2}:\rho _{1}\left( A\right) =\rho _{2}\left( A\right) \text{ for all }%
A\in \mathcal{B}\right\}
\end{equation*}%
on the set $E$ of states. We say that the subset $E_{\mathcal{B}}\subseteq E$
represents $E$ with respect to $\mathcal{B}$ whenever, for all $\rho _{1}\in
E$, there is $\rho _{2}\in E_{\mathcal{B}}$ such that $\left( \rho _{1},\rho
_{2}\right) \in \heartsuit _{\mathcal{B}}$. In particular, one can identify
continuous functions $f\in \overline{\mathfrak{C}_{\mathcal{B}}}$ with their
restrictions to $E_{\mathcal{B}}$.

Fix now a state-dependent $C^{\ast }$-dynamical system $(\mathfrak{X},%
\mathfrak{T})$. See Definition \ref{Extended dynamical systems}. For any
self-adjoint subspace $\mathcal{B}\subseteq \mathcal{X}$, consider the
following conditions:

\begin{condition}[Reduction of dynamics]
\label{condition Lie copy(1)}\mbox{ }\newline
\emph{(i)} $(\mathcal{B}\cap \mathcal{X}^{\mathbb{R}},i[\cdot ,\cdot ])$ is
a real Lie algebra, $[\cdot ,\cdot ]$ being the usual commutator in $%
\mathcal{X}$. \newline
\emph{(ii)} $E_{\mathcal{B}}$ is a weak$^{\ast }$-compact space representing 
$E$ with respect to $\mathcal{B}$. \newline
\emph{(iii)} $\mathrm{T}_{t,s}^{\rho }\left( \mathcal{B}\right) \subseteq 
\mathcal{B}$ for all $\rho \in E$, $s,t\in \mathbb{R}$, with $\mathrm{T}%
_{t,s}^{\rho }$ defined by (\ref{cigare0}).\newline
\emph{(iv)} $\phi _{t,s}\left( E_{\mathcal{B}}\right) \subseteq E_{\mathcal{B%
}}$ for all $s,t\in \mathbb{R}$, with $(\phi _{t,s})_{s,t\in \mathbb{R}}$
defined by (\ref{phitsbis}).
\end{condition}

\noindent By (\ref{shorodinger dynamicsbis}), this condition yields that the
polynomial algebra $\mathfrak{C}_{\mathcal{B}}$ (\ref{def frac Cb}) is
preserved by the family $(V_{t,s})_{s,t\in \mathbb{R}}$, i.e., 
\begin{equation*}
V_{t,s}\left( \mathfrak{C}_{\mathcal{B}}\right) \subseteq \mathfrak{C}_{%
\mathcal{B}}\ ,\qquad s,t\in \mathbb{R}\ .
\end{equation*}%
In this case, the state space of the classical dynamics coming from $(%
\mathfrak{T}_{t,s})_{s,t\in \mathbb{R}}$ can be restricted to the weak$%
^{\ast }$-compact subset $E_{\mathcal{B}}\subseteq E$ with the corresponding
Poisson algebra for observables being the subalgebra $\mathfrak{C}_{\mathcal{%
B}}\subseteq C\left( E_{\mathcal{B}};\mathbb{C}\right) $.

\subsection{Other Constructions Involving Algebras of $C^{\ast }$-Valued
Functions\label{Quantum AlgebrasQuantum Algebras}}

We are not aware whether the $C^{\ast }$-algebra $\mathfrak{X}$ of $\mathcal{%
X}$-valued continuous functions on states has previously been systematically
studied. However, other constructions of $C^{\ast }$-algebras of $\mathcal{X}
$-valued, continuous or measurable, functions are well-known in the
literature. For instance, in \cite[Definition 1]{marseille-Kastler} $C^{\ast
}$-algebras of $\mathcal{X}$-valued measurable functions on a locally
compact group $\{\mathfrak{G}_{g}\}_{g\in \mathrm{G}}$ of $\ast $%
-automorphisms of $\mathcal{X}$ are introduced. This kind of construction
goes under the name \textquotedblleft covariance algebras\textquotedblright
. In contrast, note that the state space $E$ has no natural group structure.
Moreover, the product of covariances algebras are convolutions and not
point-wise products as in $\mathfrak{X}$.

Covariance algebras are reminiscent of crossed products of $C^{\ast }$%
-algebras by groups acting on these algebras. Such products are relatively
standard in the theory of operator algebras. For instance, they are
fundamental in Haagerup's approach to noncommutative $L_{p}$-spaces \cite%
{Haagerup}.

\section{The Weak$^{\ast }$-Hausdorff Hypertopology\label{Hausdorff
Hypertopology}}

\noindent \textit{Deep in the human unconscious is a pervasive need for a
logical universe that makes sense. But the real universe is always one step
beyond logic.}\smallskip

\hfill \textquotedblleft The Sayings of Muad'Dib\textquotedblright\ by the
Princess Irulan\textit{\footnote{\textit{Dune} by F. Herbert (1965).}}%
\bigskip

The aim of this section is to provide all arguments to deduce Theorems \ref%
{theorem dense cool1}-\ref{theorem density2}. We adopt a broad perspective
on the weak$^{\ast }$-Hausdorff hypertopology because it does not seem to
have been considered in the past. This leads, hopefully, to a good
understanding of this hypertopology along with interesting connections to
other fields of mathematics and more general results than those stated in
Section \ref{generic convex set}. This broader perspective also highlights
the role played by the convexity of weak$^{\ast }$-compact subsets in our
arguments.

Recall that, when the restriction to singletons of a topology for sets of
closed subsets of topological spaces coincide with the original topology of
the underlying space, we talk about hypertopologies and hyperspaces of
closed sets.

\subsection{Immeasurable Hyperspaces}

In all Section \ref{Hausdorff Hypertopology}, $\mathcal{X}$ is not
necessarily a $C^{\ast }$-algebra, but only a (real or complex) Banach
space. Unless it is explicitly mentioned, for convenience, we always
consider the complex case, as in all other sections. We study subsets of its
dual $\mathcal{X}^{\ast }$, which, endowed with the weak$^{\ast }$-topology,
is a locally convex Hausdorff space. See, e.g., \cite[Theorem 10.8]%
{BruPedra2}. As is usual in the theory of hyperspaces \cite{Beer}, we start
with the set%
\begin{equation*}
\mathbf{F}\left( \mathcal{X}^{\ast }\right) \doteq \left\{ F\subseteq 
\mathcal{X}^{\ast }:F\neq \emptyset \text{ is weak}^{\ast }\text{-closed}%
\right\}
\end{equation*}%
of all nonempty weak$^{\ast }$-closed subsets of $\mathcal{X}^{\ast }$. It
is endowed below with some hypertopology.

Recall that there are various standard hypertopologies on general sets of
nonempty closed subsets of a metric space $(\mathcal{Y},d)$: the Fell,
Vietoris, Wijsman, proximal or locally finite hypertopologies, to name a few
well-known examples. See, e.g., \cite{Beer}. The most well-studied and
well-known hypertopology, named the Hausdorff metric topology \cite[%
Definition 3.2.1]{Beer}, comes from the Hausdorff distance between two sets $%
F_{1},F_{2}$, associated with the metric $d$ on $\mathcal{Y}$: 
\begin{equation}
d_{H}\left( F_{1},F_{2}\right) \doteq \max \left\{ \sup_{x_{1}\in
F_{1}}\inf_{x_{2}\in F_{2}}d\left( x_{1},x_{2}\right) ,\sup_{x_{2}\in
F_{2}}\inf_{x_{1}\in F_{1}}d\left( x_{1},x_{2}\right) \right\} \in \mathbb{R}%
_{0}^{+}\cup \left\{ \infty \right\} \ .  \label{Hausdorf}
\end{equation}%
In this case, the corresponding hyperspace of nonempty closed subsets of $%
\mathcal{Y}$ is complete iff the metric space $(\mathcal{Y},d)$ is complete.
See, e.g., \cite[Theorem 3.2.4]{Beer}. The Hausdorff metric topology is the
hypertopology used in \cite{Klee,FonfLindenstrauss}, the metric $d$ being
the one associated with the norm of a separable Banach space $\mathcal{Y}$,
in order to prove the density of the set of convex compact subsets of $%
\mathcal{Y}$ with dense extreme boundary.

None of these well-known hypertopologies is used here for $\mathbf{F}(%
\mathcal{X}^{\ast })$. Instead, we use a weak$^{\ast }$ version of the
Hausdorff metric topology. This corresponds to the weak$^{\ast }$-Hausdorff
hypertopology of Definition \ref{hypertopology}, which is naturally extended
to all weak$^{\ast }$-closed sets of $\mathbf{F}(\mathcal{X}^{\ast })$:

\begin{definition}[Weak$^{\ast }$-Hausdorff hypertopology]
\label{hypertopology0}\mbox{ }\newline
The weak$^{\ast }$-Hausdorff hypertopology on $\mathbf{F}(\mathcal{X}^{\ast
})$ is the topology induced (see (\ref{induced topology})) by the family of
Hausdorff pseudometrics $d_{H}^{(A)}$ defined, for all $A\in \mathcal{X}$, by%
\begin{equation}
d_{H}^{(A)}(F,\tilde{F})\doteq \max \left\{ \sup_{\sigma \in F}\inf_{\tilde{%
\sigma}\in \tilde{F}}\left\vert \left( \sigma -\tilde{\sigma}\right) \left(
A\right) \right\vert ,\sup_{\tilde{\sigma}\in \tilde{F}}\inf_{\sigma \in
F}\left\vert \left( \sigma -\tilde{\sigma}\right) \left( A\right)
\right\vert \right\} \in \mathbb{R}_{0}^{+}\cup \left\{ \infty \right\}
,\quad F,\tilde{F}\in \mathbf{F}\left( \mathcal{X}^{\ast }\right) \ .
\label{def}
\end{equation}
\end{definition}

\noindent To our knowledge, this hypertopology has not been considered so
far and we thus give here a detailed study of its main properties. Recall
that it is an \emph{hyper}topology because any net $(\sigma _{j})_{j\in J}$
in $\mathcal{X}^{\ast }$ converges to $\sigma \in \mathcal{X}^{\ast }$ in
the weak$^{\ast }$ topology iff the net $(\{\sigma _{j}\})_{j\in J}$
converges in $\mathbf{F}(\mathcal{X}^{\ast })$ to $\{\sigma \}$ in the weak$%
^{\ast }$-Hausdorff (hyper)topology.

Observe that (\ref{def}) is always finite on the subspace%
\begin{equation}
\mathbf{K}\left( \mathcal{X}^{\ast }\right) \doteq \left\{ K\in \mathbf{F}%
\left( \mathcal{X}^{\ast }\right) :\sup_{\sigma \in K}\left\Vert \sigma
\right\Vert _{\mathcal{X}^{\ast }}<\infty \right\} \subseteq \mathbf{F}%
\left( \mathcal{X}^{\ast }\right)  \label{hyperspace}
\end{equation}%
of all nonempty weak$^{\ast }$-closed norm-bounded subsets of the dual space 
$\mathcal{X}^{\ast }$. Its complement, i.e., the set of all nonempty weak$%
^{\ast }$-closed norm-unbounded subsets of $\mathcal{X}^{\ast }$, is denoted
by 
\begin{equation}
\mathbf{K}^{c}\left( \mathcal{X}^{\ast }\right) \doteq \mathbf{F}\left( 
\mathcal{X}^{\ast }\right) \backslash \mathbf{K}\left( \mathcal{X}^{\ast
}\right) \ .  \label{set of unbounded}
\end{equation}%
Both sets are weak$^{\ast }$-Hausdorff closed since the weak$^{\ast }$%
-Hausdorff hypertopology \emph{immeasurably} separates norm-unbounded sets
from norm-bounded ones:

\begin{lemma}[Immeasurable separation of norm-unbounded sets from
norm-bounded ones]
\label{dddddddddddddddddd copy(1)}\mbox{ }\newline
Let $\mathcal{X}$ be a Banach space. For any norm-unbounded weak$^{\ast }$%
-closed set $F\in \mathbf{K}^{c}(\mathcal{X}^{\ast })$, there is $A\in 
\mathcal{X}$ such that 
\begin{equation}
d_{H}^{(A)}(F,K)=\infty \ ,\qquad K\in \mathbf{K}(\mathcal{X}^{\ast })\ .
\label{separation}
\end{equation}%
Additionally, the union of any weak$^{\ast }$-Hausdorff convergent net $%
(K_{j})_{j\in J}\subseteq \mathbf{K}(\mathcal{X}^{\ast })$ is norm-bounded.
\end{lemma}

\begin{proof}
Take any norm-unbounded $F\in \mathbf{K}^{c}(\mathcal{X}^{\ast })$. Then,
there is a net $(\sigma _{j})_{j\in J}\subseteq F$ such that 
\begin{equation*}
\lim_{J}\left\Vert \sigma _{j}\right\Vert _{\mathcal{X}^{\ast }}=\infty \ .
\end{equation*}%
By the uniform boundedness principle (see, e.g., \cite[Theorems 2.4 and 2.5]%
{Rudin}), there is $A\in \mathcal{X}$ such that 
\begin{equation}
\lim_{J}\left\vert \sigma _{j}\left( A\right) \right\vert =\infty \ .
\label{infinite}
\end{equation}%
Now, pick any $K\in \mathbf{K}(\mathcal{X}^{\ast })$. Then, by\ Definition %
\ref{hypertopology0} and the triangle inequality, for any $j\in J$,%
\begin{equation*}
d_{H}^{(A)}\left( F,K\right) \geq \inf_{\tilde{\sigma}\in K}\left\vert
\left( \sigma _{j}-\tilde{\sigma}\right) \left( A\right) \right\vert \geq
\left\vert \sigma _{j}\left( A\right) \right\vert -\sup_{\tilde{\sigma}\in
K}\left\vert \tilde{\sigma}\left( A\right) \right\vert \ .
\end{equation*}%
Since $K$ is, by definition, norm-bounded, by (\ref{infinite}), the limit
over $j$ of the last inequality obviously yields (\ref{separation}).

Finally, any weak$^{\ast }$-Hausdorff convergent net $(K_{j})_{j\in
J}\subseteq \mathbf{K}(\mathcal{X}^{\ast })$ has to converge in $\mathbf{K}(%
\mathcal{X}^{\ast })$, by the first part of the lemma. Therefore, using an
argument by contradiction and the uniform boundedness principle (see, e.g., 
\cite[Theorems 2.4 and 2.5]{Rudin}) as above, one also checks that the union
of any net $(K_{j})_{j\in J}\subseteq \mathbf{K}(\mathcal{X}^{\ast })$ that
weak$^{\ast }$-Hausdorff converges must be norm-bounded.
\end{proof}

Because of Lemma \ref{dddddddddddddddddd copy(1)}, we say that the
(nonempty) subhyperspaces $\mathbf{K}(\mathcal{X}^{\ast })$ and $\mathbf{K}%
^{c}(\mathcal{X}^{\ast })$ are weak$^{\ast }$-Hausdorff-\emph{immeasurable}
with respect to each other.

\begin{corollary}[Weak$^{\ast }$-Hausdorff-clopen subhyperspaces]
\label{convexity corrolary copy(1)}\mbox{ }\newline
Let $\mathcal{X}$ be a Banach space. Then, $\mathbf{K}(\mathcal{X}^{\ast })$
is a weak$^{\ast }$-Hausdorff-closed subset of $\mathbf{F}(\mathcal{X}^{\ast
})$.
\end{corollary}

\begin{proof}
The assertion is a consequence of Lemma \ref{dddddddddddddddddd copy(1)}.
Note that a subset of a topological space is closed iff it contains the set
of its accumulation points, by \cite[Chapter 1, Theorem 5]{topology}. The
accumulation points of a set are precisely the limits of nets whose elements
are in this set, by \cite[Chapter 2, Theorem 2]{topology}.
\end{proof}

Note that $\mathbf{K}(\mathcal{X}^{\ast })$ is also a connected hyperspace:

\begin{lemma}[$\mathbf{K}(\mathcal{X}^{\ast })$ as connected subhyperspace]
\label{connected sub-hyperspace}\mbox{ }\newline
Let $\mathcal{X}$ be a Banach space. Then, $\mathbf{K}(\mathcal{X}^{\ast })$
is convex and path-connected.
\end{lemma}

\begin{proof}
Take any $K_{0},K_{1}\in \mathbf{K}(\mathcal{X}^{\ast })$. Define the map $f$
from $[0,1]$ to $\mathbf{K}(\mathcal{X}^{\ast })$ by%
\begin{equation*}
f\left( \lambda \right) \doteq \left\{ \left( 1-\lambda \right) \sigma
_{0}+\lambda \sigma _{1}:\sigma _{0}\in K_{0},\ \sigma _{1}\in K_{1}\right\}
\ ,\qquad \lambda \in \left[ 0,1\right] \ .
\end{equation*}%
(This already demonstrates that $\mathbf{K}(\mathcal{X}^{\ast })$ is
convex.) By Definition \ref{hypertopology0}, for any $\lambda _{1},\lambda
_{2}\in \lbrack 0,1]$,%
\begin{equation*}
d_{H}^{(A)}\left( f\left( \lambda _{1}\right) ,f\left( \lambda _{2}\right)
\right) \leq \left\vert \lambda _{2}-\lambda _{1}\right\vert \max_{\sigma
\in (K_{0}-K_{1})}\left\vert \sigma \left( A\right) \right\vert \ ,\qquad
A\in \mathcal{X}\ .
\end{equation*}%
So, the map $f$ is a continuous function from $[0,1]$ to $\mathbf{K}(%
\mathcal{X}^{\ast })$ with $f\left( 0\right) =K_{0}$ and $f\left( 1\right)
=K_{1}$. Therefore, $\mathbf{K}(\mathcal{X}^{\ast })$ is path-connected. The
image under a continuous map of a connected set is connected and, by \cite[%
Chapter 1, Theorem 21]{topology}, $\mathbf{K}(\mathcal{X}^{\ast })$, being
path-connected, is connected.
\end{proof}

Note that one can prove that $\mathbf{K}(\mathcal{X}^{\ast })$ is even a
connected component\footnote{%
That is, a maximal connected subset.} of $\mathbf{F}(\mathcal{X}^{\ast })$.
There are possibly many disconnected components, or even non-trivial weak$%
^{\ast }$-Hausdorff-clopen subsets of $\mathbf{F}(\mathcal{X}^{\ast })$,
associated with different directions (characterized by some $A\in \mathcal{X}
$) where the weak$^{\ast }$-closed sets $F\in \mathbf{K}^{c}(\mathcal{X}%
^{\ast })$ are unbounded. This would lead to a whole collection of weak$%
^{\ast }$-Hausdorff-clopen sets, which could be used to form a Boolean
algebra whose lattice operations are given by the union and intersection, as
is usual in mathematical logics\footnote{%
See Stone's representation theorem for Boolean algebras.}. This is far from
the scope of the article and we thus refrain from doing such a study here.

Meanwhile, note that the weak$^{\ast }$-Hausdorff-closed set $\mathbf{K}(%
\mathcal{X}^{\ast })$ of all nonempty weak$^{\ast }$-closed norm-bounded
subsets of $\mathcal{X}^{\ast }$ is nothing else than the set of all
nonempty weak$^{\ast }$-compact subsets:

\begin{lemma}[Weak$^{\ast }$-compactness vs. norm-boundedness]
\label{dddddddddddddddddd}\mbox{ }\newline
Let $\mathcal{X}$ be a Banach space. Then, 
\begin{equation*}
\mathbf{K}\left( \mathcal{X}^{\ast }\right) =\left\{ K\subseteq \mathcal{X}%
^{\ast }:K\neq \emptyset \text{ is weak}^{\ast }\text{-compact}\right\} \ .
\end{equation*}
\end{lemma}

\begin{proof}
The proof of the norm-boundedness of a weak$^{\ast }$-compact set is
completely standard (see, e.g., \cite[Proposition 1.2.9]{Beer}) and is given
here only for completeness: Take any weak$^{\ast }$-compact set $K\subseteq 
\mathcal{X}^{\ast }$\ and use, for any $A\in \mathcal{X}$, the weak$^{\ast }$%
-continuity of the map $\hat{A}:\sigma \mapsto \sigma (A)$ from $\mathcal{X}%
^{\ast }$ to $\mathbb{C}$ (cf. (\ref{fA}) and (\ref{sdfsdfkljsdlfkj})) to
show that $\sigma (K)$ is a bounded set, by weak$^{\ast }$ compactness of $K$%
. Then, the norm-boundedness of any weak$^{\ast }$-compact set is a
consequence of the uniform boundedness principle, see, e.g., \cite[Theorems
2.4 and 2.5]{Rudin}. Since $\mathcal{X}^{\ast }$ is a Hausdorff space (see,
e.g., \cite[Theorem 10.8]{BruPedra2}), by \cite[Chapter 5, Theorem 7]%
{topology}, it follows that weak$^{\ast }$-compact set are weak$^{\ast }$%
-closed and norm-bounded subsets of $\mathcal{X}^{\ast }$. On the other
hand, by the Banach-Alaoglu theorem \cite[Theorem 3.15]{Rudin}, weak$^{\ast
} $-closed and norm-bounded subsets of $\mathcal{X}^{\ast }$ are also weak$%
^{\ast }$-compact and the assertion follows.
\end{proof}

By Lemma \ref{dddddddddddddddddd}, for any $K,\tilde{K}\in \mathbf{K}(%
\mathcal{X}^{\ast })$, the suprema and infima in (\ref{def}) become
respectively maxima and minima. In this case, Definition \ref{hypertopology0}
is the same as Definition \ref{hypertopology}, extended to all weak$^{\ast }$%
-compact sets. Of course, by Lemma \ref{dddddddddddddddddd}, $\mathbf{K}(%
\mathcal{X}^{\ast })$ includes the hyperspace 
\begin{equation}
\mathbf{CK}\left( \mathcal{X}^{\ast }\right) \doteq \left\{ K\subseteq 
\mathcal{X}^{\ast }:K\neq \emptyset \text{ is convex and weak}^{\ast }\text{%
-compact}\right\} \subseteq \mathbf{K}\left( \mathcal{X}^{\ast }\right)
\subseteq \mathbf{F}\left( \mathcal{X}^{\ast }\right)  \label{ZDZDZDZD}
\end{equation}%
of all nonempty convex weak$^{\ast }$-compact subsets of $\mathcal{X}^{\ast
} $, already defined by (\ref{ZDZD}) and used in\ Section \ref{generic
convex set}.

\subsection{Hausdorff Property and Closure Operator}

One fundamental property one shall ask regarding the hyperspace $\mathbf{F}(%
\mathcal{X}^{\ast })$ (or $\mathbf{K}(\mathcal{X}^{\ast })$) is whether it
is a Hausdorff space, with respect to the weak$^{\ast }$-Hausdorff
hypertopology, or not. The answer is \emph{negative} for real Banach spaces
of dimension greater than 1, as demonstrated in the next lemma:

\begin{lemma}[Non-weak$^{\ast }$-Hausdorff-separable points]
\label{convexity lemma copy(2)}\mbox{ }\newline
Let $\mathcal{X}$ be a real Banach space. Take any set $K\in \mathbf{CK}(%
\mathcal{X}^{\ast })$ with weak$^{\ast }$-path-connected weak$^{\ast }$%
-closed set $\mathcal{E}(K)\subseteq K$ of extreme points\footnote{%
Cf. the Krein-Milman theorem \cite[Theorem 3.23]{Rudin}.}. Then, $\mathcal{E}%
(K)\in \mathbf{K}(\mathcal{X}^{\ast })$ and $d_{H}^{(A)}(K,\mathcal{E}(K))=0$
for any $A\in \mathcal{X}$.
\end{lemma}

\begin{proof}
Let $\mathcal{X}$ be a real Banach space. Recall that any $A\in \mathcal{X}$
defines a weak$^{\ast }$-continuous linear functional $\hat{A}:\mathcal{X}%
^{\ast }\rightarrow \mathbb{R}$ by%
\begin{equation*}
\hat{A}(\sigma )\doteq \sigma (A)\text{ },\text{\qquad }\sigma \in \mathcal{X%
}^{\ast }\ .
\end{equation*}%
See (\ref{sdfsdfkljsdlfkj}). Observe next that 
\begin{equation}
d_{H}^{(A)}(K,\mathcal{E}(K))=\max \left\{ \max_{x_{1}\in \hat{A}\left(
K\right) }\min_{x_{2}\in \hat{A}\left( \mathcal{E}(K)\right) }\left\vert
x_{1}-x_{2}\right\vert ,\max_{x_{2}\in \hat{A}\left( \mathcal{E}(K)\right)
}\min_{x_{1}\in \hat{A}\left( K\right) }\left\vert x_{1}-x_{2}\right\vert
\right\} \ .  \label{rewritte0}
\end{equation}%
The right hand side is nothing else than the Hausdorff distance (\ref%
{Hausdorf}) between the sets $\hat{A}(K)$ and $\hat{A}(\mathcal{E}(K))$,
where the metric used in $\mathcal{Y}=\mathbb{R}$ is the absolute-value
distance. Now, clearly, 
\begin{equation}
\hat{A}\left( \mathcal{E}\left( K\right) \right) \subseteq \hat{A}\left(
K\right) \subseteq \left[ \min \hat{A}\left( K\right) ,\max \hat{A}\left(
K\right) \right] \ .  \label{rewritte}
\end{equation}%
By the Bauer maximum principle \cite[Lemma 10.31]{BruPedra2} together with
the affinity and weak$^{\ast }$-continuity of $\hat{A}$, 
\begin{equation*}
\min \hat{A}\left( K\right) =\min \hat{A}\left( \mathcal{E}\left( K\right)
\right) \qquad \text{and}\qquad \max \hat{A}\left( K\right) =\max \hat{A}%
\left( \mathcal{E}\left( K\right) \right) \ .
\end{equation*}%
In particular, we can rewrite (\ref{rewritte}) as 
\begin{equation}
\hat{A}\left( \mathcal{E}(K)\right) \subseteq \hat{A}\left( K\right)
\subseteq \left[ \min \hat{A}\left( \mathcal{E}\left( K\right) \right) ,\max 
\hat{A}\left( \mathcal{E}\left( K\right) \right) \right] \ .
\label{rewritte2}
\end{equation}%
Since $\mathcal{E}(K)$ is, by assumption, path-connected in the weak$^{\ast
} $ topology, there is a weak$^{\ast }$-continuous path $\gamma
:[0,1]\rightarrow \mathcal{E}(K)$ from a minimizer to a maximizer of $\hat{A}
$ in $\mathcal{E}(K)$. By weak$^{\ast }$-continuity of $\hat{A}$, it follows
that 
\begin{equation*}
\left[ \min \hat{A}\left( \mathcal{E}(K)\right) ,\max \hat{A}\left( \mathcal{%
E}(K)\right) \right] =\hat{A}\circ \gamma \left( \left[ 0,1\right] \right)
\subseteq \hat{A}\left( \mathcal{E}(K)\right)
\end{equation*}%
and we infer from (\ref{rewritte2}) that 
\begin{equation*}
\hat{A}\left( \mathcal{E}(K)\right) =\hat{A}\left( K\right) =\left[ \min 
\hat{A}\left( K\right) ,\max \hat{A}\left( K\right) \right] =\left[ \min 
\hat{A}\left( \mathcal{E}(K)\right) ,\max \hat{A}\left( \mathcal{E}%
(K)\right) \right] \ .
\end{equation*}%
Together with (\ref{rewritte0}), this last equality obviously leads to the
assertion. Note that $\mathcal{E}(K)\in \mathbf{K}(\mathcal{X}^{\ast })$
since it is, by assumption, a weak$^{\ast }$-closed subset of the weak$%
^{\ast }$-compact set $K$ (Lemma \ref{dddddddddddddddddd}).
\end{proof}

\begin{corollary}[Non-Hausdorff hyperspaces]
\label{Non-Hausdorff hyperspaces}\mbox{ }\newline
Let $\mathcal{X}$ be a real Banach space of dimension greater than 1. Then, $%
\mathbf{F}(\mathcal{X}^{\ast })$ and $\mathbf{K}(\mathcal{X}^{\ast })$ are
non-Hausdorff topological spaces.
\end{corollary}

\begin{proof}
This corollary is a direct consequence of Lemma \ref{convexity lemma copy(2)}
by observing that the dual of a real Banach space of dimension greater than
1 contains a two-dimensional closed disc.
\end{proof}

As a consequence, the Hausdorff property of the hyperspace $\mathbf{F}(%
\mathcal{X}^{\ast })$ does not hold true, in general. A restriction to the
sub-hyperspace $\mathbf{K}(\mathcal{X}^{\ast })$ is also not sufficient to
get the separation property. This fact, described in Lemma \ref{convexity
lemma copy(2)}, also appears for other well-established hypertopologies,
which cannot distinguish a set from its closed convex hull. The so-called
scalar topology for closed sets is a good example of this phenomenon, as
explained in \cite[Section 4.3]{Beer}. Actually, similar to the scalar
topology, $\mathbf{CK}(\mathcal{X}^{\ast })$ is a Hausdorff hyperspace. To
get a better intuition of this fact, the following proposition is
instructive:

\begin{proposition}[Separation of the weak$^{\ast }$-closed convex hull]
\label{convexity lemma copy(1)}\mbox{ }\newline
Let $\mathcal{X}$ be a Banach space and $K_{1},K_{2}\in \mathbf{K}(\mathcal{X%
}^{\ast })$. If $d_{H}^{(A)}(K_{1},K_{2})=0$ for all $A\in \mathcal{X}$,
then $\overline{\mathrm{co}K_{1}}=\overline{\mathrm{co}K_{2}}$, where $%
\overline{\mathrm{co}F}$ denotes the weak$^{\ast }$-closure of the convex
hull of any set $F\in \mathbf{F}(\mathcal{X}^{\ast })$.
\end{proposition}

\begin{proof}
Pick any weak$^{\ast }$-compact sets $K_{1},K_{2}$ satisfying $%
d_{H}^{(A)}(K_{1},K_{2})=0$ for all $A\in \mathcal{X}$. Let $\sigma _{1}\in
K_{1}$. By\ Definition \ref{hypertopology0}, it follows that%
\begin{equation}
\min_{\sigma _{2}\in K_{2}}\left\vert \left( \sigma _{1}-\sigma _{2}\right)
\left( A\right) \right\vert =0\ ,\qquad A\in \mathcal{X}\ .
\label{contradition00}
\end{equation}%
The dual space $\mathcal{X}^{\ast }$ of the Banach space $\mathcal{X}$ is a
locally convex (Hausdorff) space in the weak$^{\ast }$ topology and its dual
is $\mathcal{X}$. Note also that $\overline{\mathrm{co}K_{2}}$ is convex and
weak$^{\ast }$-compact, because it is a norm-bounded weak$^{\ast }$-closed
subset of $\mathcal{X}^{\ast }$, see Lemma \ref{dddddddddddddddddd}. Since $%
\{\sigma _{1}\}$ is a convex weak$^{\ast }$-closed set, if $\sigma
_{1}\notin \overline{\mathrm{co}K_{2}}$ then we infer from the Hahn-Banach
separation theorem \cite[Theorem 3.4 (b)]{Rudin} the existence of $A_{0}\in 
\mathcal{X}$ and $x_{1},x_{2}\in \mathbb{R}$ such that 
\begin{equation}
\max_{\sigma _{2}\in \overline{\mathrm{co}K_{2}}}\mathrm{Re}\left\{ \sigma
_{2}\left( A_{0}\right) \right\} <x_{1}<x_{2}<\mathrm{Re}\left\{ \sigma
_{1}\left( A_{0}\right) \right\} \ ,  \label{hahn banach}
\end{equation}%
which contradicts (\ref{contradition00}) for $A=A_{0}$. As a consequence, $%
\sigma _{1}\in \overline{\mathrm{co}K_{2}}$ and hence, $K_{1}\subseteq 
\overline{\mathrm{co}K_{2}}$. This in turn yields $\overline{\mathrm{co}K_{1}%
}\subseteq \overline{\mathrm{co}K_{2}}$. By switching the role of the weak$%
^{\ast }$-compact sets, we thus deduce the assertion.
\end{proof}

Proposition \ref{convexity lemma copy(1)} motivates the introduction of the 
\emph{weak}$^{\ast }$\emph{-closed convex hull operator}:

\begin{definition}[The weak$^{\ast }$-closed convex hull operator]
\label{convex hull operator}\mbox{ }\newline
The weak$^{\ast }$-closed convex hull operator is the map $\overline{\mathrm{%
co}}$ from $\mathbf{F}(\mathcal{X}^{\ast })$ to itself defined by%
\begin{equation*}
\overline{\mathrm{co}}\left( F\right) \doteq \overline{\mathrm{co}F}\
,\qquad F\in \mathbf{F}(\mathcal{X}^{\ast })\ ,
\end{equation*}%
where $\overline{\mathrm{co}F}$ denotes the weak$^{\ast }$-closure of the
convex hull of $F$ or, equivalently, the intersection of all weak$^{\ast }$%
-closed convex sets containing $F$.
\end{definition}

\noindent It is a \emph{closure\ }(or hull)\emph{\ }operator \cite[%
Definition 5.1]{Universal Algebra} since it satisfies the following
properties:

\begin{itemize}
\item For any $F\in \mathbf{F}(\mathcal{X}^{\ast })$, $F\subseteq \overline{%
\mathrm{co}}(F)$ (extensive);

\item For any $F\in \mathbf{F}(\mathcal{X}^{\ast })$, $\overline{\mathrm{co}}%
\left( \overline{\mathrm{co}}(F)\right) =\overline{\mathrm{co}}(F)$
(idempotent);

\item For any $F_{1},F_{2}\in \mathbf{F}(\mathcal{X}^{\ast })$ such that $%
F_{1}\subseteq F_{2}$, $\overline{\mathrm{co}}\left( F_{1}\right) \subseteq 
\overline{\mathrm{co}}(F_{2})$ (isotone).
\end{itemize}

\noindent Such a closure operator has probably been used in the past. It is
a composition of (i) an \emph{algebraic} (or finitary) closure operator \cite%
[Definition 5.4]{Universal Algebra} defined by $F\mapsto \mathrm{co}F$ with
(ii) a \emph{topological} (or Kuratowski) closure operator \cite[Chapter 1,
p.43]{topology} defined by $F\mapsto \overline{F}$ on $\mathbf{F}(\mathcal{X}%
^{\ast })$.

As is usual, weak$^{\ast }$-closed subsets $F\in \mathbf{F}(\mathcal{X}%
^{\ast })$ satisfying $F=\overline{\mathrm{co}}(F)$ are by definition $%
\overline{\mathrm{co}}$\emph{-closed }sets. In the light of Proposition \ref%
{convexity lemma copy(1)}, it is natural to propose the set $\overline{%
\mathrm{co}}\left( \mathbf{K}(\mathcal{X}^{\ast })\right) $ as the Hausdorff
hyperspace to consider. This set is nothing else than the set of all
nonempty convex weak$^{\ast }$-compact sets defined by (\ref{ZDZD}) or (\ref%
{ZDZDZDZD}):%
\begin{equation}
\overline{\mathrm{co}}\left( \mathbf{K}\left( \mathcal{X}^{\ast }\right)
\right) =\mathbf{CK}\left( \mathcal{X}^{\ast }\right) \ .  \label{surjective}
\end{equation}%
We thus deduce the following assertion:

\begin{corollary}[$\mathbf{CK}(\mathcal{X}^{\ast })$ as an Hausdorff
hyperspace]
\label{convexity corrolary}\mbox{ }\newline
Let $\mathcal{X}$ be a Banach space. Then, $\mathbf{CK}(\mathcal{X}^{\ast })$
is a Hausdorff hyperspace.
\end{corollary}

\begin{proof}
This is a direct consequence of Proposition \ref{convexity lemma copy(1)}.
\end{proof}

Note additionally that the restriction of the weak$^{\ast }$-closed convex
hull operator to $\mathbf{K}(\mathcal{X}^{\ast })$ is a weak$^{\ast }$%
-Hausdorff continuous map from the hyperspace $\mathbf{K}(\mathcal{X}^{\ast
})$ to $\mathbf{CK}(\mathcal{X}^{\ast })$:

\begin{proposition}[Weak$^{\ast }$-Hausdorff continuity of the weak$^{\ast }$%
-closed convex hull operator]
\label{convexity lemma}\mbox{ }\newline
Let $\mathcal{X}$ be a Banach space. Then, $\overline{\mathrm{co}}$
preserves the set (\ref{set of unbounded}) of all nonempty weak$^{\ast }$%
-closed norm-unbounded subsets of $\mathcal{X}^{\ast }$, i.e., 
\begin{equation}
\overline{\mathrm{co}}\left( \mathbf{K}^{c}\left( \mathcal{X}^{\ast }\right)
\right) \subseteq \mathbf{K}^{c}\left( \mathcal{X}^{\ast }\right) \doteq 
\mathbf{F}\left( \mathcal{X}^{\ast }\right) \backslash \mathbf{K}\left( 
\mathcal{X}^{\ast }\right) \ ,  \label{trivial equality}
\end{equation}%
and $\overline{\mathrm{co}}$ restricted to $\mathbf{K}(\mathcal{X}^{\ast })$
is a weak$^{\ast }$-Hausdorff continuous map onto $\mathbf{CK}(\mathcal{X}%
^{\ast })$.
\end{proposition}

\begin{proof}
Let $\mathcal{X}$ be a Banach space. Equation (\ref{trivial equality}) and
surjectivity of $\overline{\mathrm{co}}$ seen as a map from $\mathbf{K}(%
\mathcal{X}^{\ast })$ to $\mathbf{CK}(\mathcal{X}^{\ast })$ are both
obvious, by Definition \ref{convex hull operator} and (\ref{surjective}).
Now, take any weak$^{\ast }$-Hausdorff convergent net $(K_{j})_{j\in
J}\subseteq \mathbf{K}(\mathcal{X}^{\ast })$ with limit $K_{\infty }\in 
\mathbf{K}(\mathcal{X}^{\ast })$. Note that 
\begin{equation}
\max_{\sigma \in \overline{\mathrm{co}}\left( K_{\infty }\right) }\min_{%
\tilde{\sigma}\in \overline{\mathrm{co}}\left( K_{j}\right) }\left\vert
\left( \sigma -\tilde{\sigma}\right) \left( A\right) \right\vert
=\sup_{\sigma \in \mathrm{co}K_{\infty }}\min_{\tilde{\sigma}\in \overline{%
\mathrm{co}}\left( K_{j}\right) }\left\vert \left( \sigma -\tilde{\sigma}%
\right) \left( A\right) \right\vert \ ,\qquad A\in \mathcal{X}\ ,  \label{dd}
\end{equation}%
because, for any $A\in \mathcal{X}$, $j\in J$, $\sigma _{1},\sigma _{2}\in 
\overline{\mathrm{co}}\left( K_{\infty }\right) $ and $\tilde{\sigma}\in 
\overline{\mathrm{co}}\left( K_{j}\right) $,%
\begin{equation*}
\left\vert \left\vert \left( \sigma _{1}-\tilde{\sigma}\right) \left(
A\right) \right\vert -\left\vert \left( \sigma _{2}-\tilde{\sigma}\right)
\left( A\right) \right\vert \right\vert \leq \left\vert \left( \sigma
_{1}-\sigma _{2}\right) \left( A\right) \right\vert \ ,
\end{equation*}%
which yields%
\begin{equation*}
\left\vert \min_{\tilde{\sigma}\in K_{j}}\left\vert \left( \sigma _{1}-%
\tilde{\sigma}\right) \left( A\right) \right\vert -\min_{\tilde{\sigma}\in
K_{j}}\left\vert \left( \sigma _{2}-\tilde{\sigma}\right) \left( A\right)
\right\vert \right\vert \leq \left\vert \left( \sigma _{1}-\sigma
_{2}\right) \left( A\right) \right\vert
\end{equation*}%
for any $A\in \mathcal{X}$, $j\in J$ and $\sigma _{1},\sigma _{2}\in 
\overline{\mathrm{co}}\left( K_{\infty }\right) $. Fix $n\in \mathbb{N}$, $%
\sigma _{1},\ldots ,\sigma _{n}\in K_{\infty }$ and parameters $\lambda
_{1},\ldots ,\lambda _{n}\in \left[ 0,1\right] $ such that 
\begin{equation*}
\sum\limits_{k=1}^{n}\lambda _{k}=1\ .
\end{equation*}%
For any $A\in \mathcal{X}$ and $k\in \{1,\ldots ,n\}$, we define $\tilde{%
\sigma}_{k,j}\in K_{j}$ such that%
\begin{equation*}
\min_{\tilde{\sigma}\in K_{j}}\left\vert \left( \sigma _{k}-\tilde{\sigma}%
\right) \left( A\right) \right\vert =\left\vert \left( \sigma _{k}-\tilde{%
\sigma}_{k,j}\right) \left( A\right) \right\vert \ .
\end{equation*}%
Then, for all $j\in J$,%
\begin{equation*}
\min_{\tilde{\sigma}\in \overline{\mathrm{co}}\left( K_{j}\right)
}\left\vert \left( \sum\limits_{k=1}^{n}\lambda _{k}\sigma _{k}-\tilde{\sigma%
}\right) \left( A\right) \right\vert \leq \sum\limits_{k=1}^{n}\lambda
_{k}\left\vert \left( \sigma _{k}-\tilde{\sigma}_{k,j}\right) \left(
A\right) \right\vert \leq \max_{\sigma \in K_{\infty }}\min_{\tilde{\sigma}%
\in K_{j}}\left\vert \left( \sigma -\tilde{\sigma}\right) \left( A\right)
\right\vert \ .
\end{equation*}%
By using (\ref{dd}), we then deduce that, for all $j\in J$, 
\begin{equation}
\max_{\sigma \in \overline{\mathrm{co}}\left( K_{\infty }\right) }\min_{%
\tilde{\sigma}\in \overline{\mathrm{co}}\left( K_{j}\right) }\left\vert
\left( \sigma -\tilde{\sigma}\right) \left( A\right) \right\vert \leq
\max_{\sigma \in K_{\infty }}\min_{\tilde{\sigma}\in K_{j}}\left\vert \left(
\sigma -\tilde{\sigma}\right) \left( A\right) \right\vert \ ,\qquad A\in 
\mathcal{X}\ .  \label{convex1}
\end{equation}%
By switching the role of $K_{\infty }$ and $K_{j}$ for every $j\in J$, we
also arrive at the inequality 
\begin{equation}
\max_{\tilde{\sigma}\in \overline{\mathrm{co}}\left( K_{j}\right)
}\min_{\sigma \in \overline{\mathrm{co}}\left( K_{\infty }\right)
}\left\vert \left( \sigma -\tilde{\sigma}\right) \left( A\right) \right\vert
\leq \max_{\tilde{\sigma}\in K_{j}}\min_{\sigma \in K_{\infty }}\left\vert
\left( \sigma -\tilde{\sigma}\right) \left( A\right) \right\vert \ ,\qquad
A\in \mathcal{X}\ .  \label{convex2}
\end{equation}%
Since $(K_{j})_{j\in J}$ converges in the weak$^{\ast }$-Hausdorff
hypertopology to $K_{\infty }$, Inequalities (\ref{convex1})-(\ref{convex2})
combined with Definition \ref{hypertopology0} yield the weak$^{\ast }$%
-Hausdorff convergence of $(\overline{\mathrm{co}}\left( K_{j}\right)
)_{j\in J}$ to $\overline{\mathrm{co}}\left( K_{\infty }\right) $. By \cite[%
Chapter 3, Theorem 1]{topology}, $\overline{\mathrm{co}}$ restricted to $%
\mathbf{K}(\mathcal{X}^{\ast })$ is a weak$^{\ast }$-Hausdorff continuous
map onto $\mathbf{CK}(\mathcal{X}^{\ast })$.
\end{proof}

\begin{corollary}[$\mathbf{CK}(\mathcal{X}^{\ast })$ as a connected, weak$%
^{\ast }$-Hausdorff-closed set]
\label{convexity corrolary copy(2)}\mbox{ }\newline
Let $\mathcal{X}$ be a Banach space. Then, $\mathbf{CK}(\mathcal{X}^{\ast })$
is a convex, path-connected, weak$^{\ast }$-Hausdorff-closed subset of $%
\mathbf{K}(\mathcal{X}^{\ast })$.
\end{corollary}

\begin{proof}
By Corollary \ref{convexity corrolary}, $\mathbf{CK}(\mathcal{X}^{\ast })$
endowed with the weak$^{\ast }$-Hausdorff hypertopology is a Hausdorff
space. Hence, by \cite[Chapter 2, Theorem 3]{topology}, each convergent net
in this space converges in the weak$^{\ast }$-Hausdorff hypertopology to at
most one point, which, by Proposition \ref{convexity lemma}, must be a
convex weak$^{\ast }$-compact set. Additionally, by Lemma \ref{connected
sub-hyperspace}, Proposition \ref{convexity lemma} and the fact that the
image under a continuous map of a path-connected space is path-connected, $%
\mathbf{CK}(\mathcal{X}^{\ast })$ is also path-connected. Convexity of $%
\mathbf{CK}(\mathcal{X}^{\ast })$ is also obvious.
\end{proof}

As is usual, the weak$^{\ast }$-closed convex hull operator $\overline{%
\mathrm{co}}$ yields a notion of compactness, defined as follows: A set $%
K\in \mathbf{F}(\mathcal{X}^{\ast })$ is $\overline{\mathrm{co}}$\emph{%
-compact} iff it is $\overline{\mathrm{co}}$-closed and each family of $%
\overline{\mathrm{co}}$-closed subsets of $K$ which has the finite
intersection property has a non-empty intersection. Compare this definition
with \cite[Chapter 5, Theorem 1]{topology}. The set $\mathbf{CK}(\mathcal{X}%
^{\ast })$ of all nonempty convex weak$^{\ast }$-compact sets defined by (%
\ref{ZDZD}) or (\ref{ZDZDZDZD}) is precisely the set of $\overline{\mathrm{co%
}}$-compact sets:

\begin{proposition}[$\mathbf{CK}(\mathcal{X}^{\ast })$ as the space of $%
\overline{\mathrm{co}}$-compact sets]
\label{convexity lemma copy(3)}\mbox{ }\newline
Let $\mathcal{X}$ be a Banach space. Then, 
\begin{equation*}
\mathbf{CK}\left( \mathcal{X}^{\ast }\right) =\left\{ K\in \mathbf{F}\left( 
\mathcal{X}^{\ast }\right) :K\text{ is }\overline{\mathrm{co}}\text{-compact}%
\right\} \ .
\end{equation*}
\end{proposition}

\begin{proof}
By \cite[Chapter 5, Theorem 1]{topology}, we clearly have 
\begin{equation*}
\mathbf{CK}\left( \mathcal{X}^{\ast }\right) \subseteq \left\{ K\in \mathbf{F%
}\left( \mathcal{X}^{\ast }\right) :K\text{ is }\overline{\mathrm{co}}\text{%
-compact}\right\} \ .
\end{equation*}%
Conversely, take any $\overline{\mathrm{co}}$-compact element $K\in \mathbf{F%
}(\mathcal{X}^{\ast })$. If $K$ is not norm-bounded, then we deduce from the
uniform boundedness principle \cite[Theorems 2.4 and 2.5]{Rudin} the
existence of $A\in \mathcal{X}$ such that $\hat{A}(K)\subseteq \mathbb{C}$
is not bounded, where we recall that $\hat{A}:\mathcal{X}^{\ast }\rightarrow 
\mathbb{C}$ is the weak$^{\ast }$-continuous (complex) linear functional
defined by (\ref{sdfsdfkljsdlfkj}). Without loss of generality, assume that $%
\mathrm{Re}\{\hat{A}(K)\}$ is not bounded from above. Define for every $n\in 
\mathbb{N}$ the set 
\begin{equation*}
K_{n}\doteq \left\{ \sigma \in K:\mathrm{Re}\{\hat{A}(\sigma )\}\geq
n\right\} \ .
\end{equation*}%
Clearly, by convexity of $K$, $K_{n}$ is a convex weak$^{\ast }$-closed
subset of $K$ and the family $(K_{n})_{n\in \mathbb{N}}$ has the finite
intersection property, but, by construction, 
\begin{equation*}
\bigcap_{n\in \mathbb{N}}K_{n}=\emptyset \ .
\end{equation*}%
(The intersection of preimages is the preimage of the intersection.) This
contradicts the fact that $K$ is $\overline{\mathrm{co}}$-compact.
Therefore, $K$ is norm-bounded and, being $\overline{\mathrm{co}}$-compact,
it is also weak$^{\ast }$-closed and convex. Consequently, $K\in \mathbf{CK}%
\left( \mathcal{X}^{\ast }\right) $ (see, e.g., Equation \ref%
{sdfklsdjfklsdjf}).
\end{proof}

The last proposition goes beyond the specific topic of the present article,
and the proof of the weak$^{\ast }$-Hausdorff density of convex weak$^{\ast
} $-compact sets with dense extreme boundary. This is however discussed here
because, like (\ref{surjective}), it is an elegant abstract characterization
of $\mathbf{CK}(\mathcal{X}^{\ast })$, only given in terms of a closure
operator, namely the weak$^{\ast }$-closed convex hull operator. It
demonstrates connections with other mathematical fields, in particular with
mathematical logics where fascinating applications of closure operators have
been developed, already by Tarski himself during the 1930's.

\subsection{Weak$^{\ast }$-Hausdorff Hyperconvergence\label%
{Hyperconvergences}}

In this subsection, we study weak$^{\ast }$-Hausdorff convergent nets. Even
if only the hyperspace $\mathbf{CK}(\mathcal{X}^{\ast })$ of all nonempty
convex weak$^{\ast }$-compact sets is Hausdorff (Corollary \ref{convexity
corrolary}), we study the convergence within the hyperspace $\mathbf{K}(%
\mathcal{X}^{\ast })$ of all nonempty, weak$^{\ast }$-closed and
norm-bounded subsets of $\mathcal{X}^{\ast }$. Recall that $\mathbf{K}(%
\mathcal{X}^{\ast })$ is a (path-) connected weak$^{\ast }$-Hausdorff-closed
subset of $\mathbf{F}(\mathcal{X}^{\ast })$, by Corollary \ref{convexity
corrolary copy(1)}\ and Lemma \ref{connected sub-hyperspace}.

It is instructive to relate weak$^{\ast }$-Hausdorff limits of nets to lower
and upper limits of sets \`{a} la Painlev\'{e} \cite[\S\ 29]%
{topology-painleve}: The \emph{lower limit} of any net $(K_{j})_{j\in J}$ of
subsets of $\mathcal{X}^{\ast }$ is defined by 
\begin{equation}
\mathrm{Li}\left( K_{j}\right) _{j\in J}\doteq \left\{ \sigma \in \mathcal{X}%
^{\ast }:\sigma \text{ is a weak}^{\ast }\text{ limit of a net }(\sigma
_{j})_{j\in J}\text{ with }\sigma _{j}\in K_{j}\text{ for all }j\in
J\right\} \ ,  \label{Li}
\end{equation}%
while its \emph{upper limit} equals%
\begin{equation}
\mathrm{Ls}\left( K_{j}\right) _{j\in J}\doteq \left\{ \sigma \in \mathcal{X}%
^{\ast }:\sigma \text{ is a weak}^{\ast }\text{ accumulation point of }%
(\sigma _{j})_{j\in J}\text{ with }\sigma _{j}\in K_{j}\text{ for all }j\in
J\right\} \ .  \label{Ls}
\end{equation}%
Clearly, $\mathrm{Li}(K_{j})_{j\in J}\subseteq \mathrm{Ls}(K_{j})_{j\in J}$.
If $\mathrm{Li}(K_{j})_{j\in J}=\mathrm{Ls}(K_{j})_{j\in J}$ then $\left(
K_{j}\right) _{j\in J}$ is said to be convergent to this set. See \cite[\S\ %
29, I, III, VI]{topology-painleve}, which however defines $\mathrm{Li}$ and $%
\mathrm{Ls}$ within metric spaces. This refers in the literature to the 
\emph{Kuratowski} or \emph{Kuratowski-Painlev\'{e}}\footnote{%
The idea of upper and lower limits is due to Painlev\'{e}, as acknowledged
by Kuratowski himself in \cite[\S\ 29, Footnote 1, p. 335]{topology-painleve}%
. We thus use the name Kuratowski-Painlev\'{e} convergence.} convergence,
see e.g. \cite[Appendix B]{Lucchetti} and \cite[Section 5.2]{Beer}. By \cite[%
Theorem 1.22]{Rudin}, if $\mathcal{X}$ is an infinite-dimensional space,
then its dual $\mathcal{X}^{\ast }$, endowed with the weak$^{\ast }$ or norm
topology, is not locally compact. In this case, the Kuratowski-Painlev\'{e}
convergence is not topological \cite[Theorem B.3.2]{Lucchetti}. See also 
\cite[Chapter 5]{Beer}, in particular \cite[Theorem 5.2.6 and following
discussions]{Beer} which relates the Kuratowski-Painlev\'{e} convergence to
the so-called \emph{Fell }topology.

We start by proving the weak$^{\ast }$-Hausdorff convergence of
monotonically increasing nets which are bounded from above within $\mathbf{K}%
(\mathcal{X}^{\ast })$:

\begin{proposition}[Weak$^{\ast }$-Hausdorff hyperconvergence of increasing
nets]
\label{Solution selfbaby copy(5)+1}\mbox{ }\newline
Let $\mathcal{X}$ be a Banach space. Any increasing net $(K_{j})_{j\in
J}\subseteq \mathbf{K}(\mathcal{X}^{\ast })$ such that 
\begin{equation}
K\doteq \overline{\bigcup\limits_{j\in J}K_{j}}\in \mathbf{K}\left( \mathcal{%
X}^{\ast }\right) \varsubsetneq \mathbf{F}\left( \mathcal{X}^{\ast }\right)
\label{sdsdsdsd}
\end{equation}%
(with respect to the weak$^{\ast }$ closure) converges in the weak$^{\ast }$%
-Hausdorff hypertopology to the Kuratowski-Painlev\'{e} limit 
\begin{equation*}
K=\mathrm{Li}\left( K_{j}\right) _{j\in J}=\mathrm{Ls}\left( K_{j}\right)
_{j\in J}\ .
\end{equation*}
\end{proposition}

\begin{proof}
Let $(K_{j})_{j\in J}\subseteq \mathbf{K}(\mathcal{X}^{\ast })$ be any
increasing net, i.e., $K_{j_{1}}\subseteq K_{j_{2}}$ whenever $j_{1}\prec
j_{2}$, satisfying (\ref{sdsdsdsd}). Because $K\in \mathbf{K}(\mathcal{X}%
^{\ast })$, it is norm-bounded. By the convergence of increasing bounded
nets of real numbers, it follows that, for any $A\in \mathcal{X}$, 
\begin{equation*}
\lim_{J}\max_{\tilde{\sigma}\in K_{j}}\min_{\sigma \in K}\left\vert \left( 
\tilde{\sigma}-\sigma \right) \left( A\right) \right\vert =\sup_{j\in
J}\max_{\tilde{\sigma}\in K_{j}}\min_{\sigma \in K}\left\vert \left( \tilde{%
\sigma}-\sigma \right) \left( A\right) \right\vert \leq \max_{\tilde{\sigma}%
\in K}\min_{\sigma \in K}\left\vert \left( \tilde{\sigma}-\sigma \right)
\left( A\right) \right\vert =0\ .
\end{equation*}%
Therefore, by Definition \ref{hypertopology0}, if 
\begin{equation}
\limsup_{J}\max_{\sigma \in K}\min_{\tilde{\sigma}\in K_{j}}\left\vert
\left( \tilde{\sigma}-\sigma \right) \left( A\right) \right\vert =0\ ,\qquad
A\in \mathcal{X}\ ,  \label{limit prove}
\end{equation}%
then the increasing net $(K_{j})_{j\in J}$ converges in $\mathbf{K}(\mathcal{%
X}^{\ast })$ to $K$, which clearly equals the Kuratowski-Painlev\'{e} limit
of the net. To prove (\ref{limit prove}), assume by contradiction the
existence of $\varepsilon \in \mathbb{R}^{+}$ such that%
\begin{equation}
\limsup_{J}\max_{\sigma \in K}\min_{\tilde{\sigma}\in K_{j}}\left\vert
\left( \tilde{\sigma}-\sigma \right) \left( A\right) \right\vert \geq
\varepsilon \in \mathbb{R}^{+}  \label{contradiction}
\end{equation}%
for some fixed $A\in \mathcal{X}$. For any $j\in J$, take $\sigma _{j}\in K$
such that%
\begin{equation}
\max_{\sigma \in K}\min_{\tilde{\sigma}\in K_{j}}\left\vert \left( \tilde{%
\sigma}-\sigma \right) \left( A\right) \right\vert =\min_{\tilde{\sigma}\in
K_{j}}\left\vert \left( \tilde{\sigma}-\sigma _{j}\right) \left( A\right)
\right\vert \ .  \label{to prove2}
\end{equation}%
By weak$^{\ast }$-compactness of $K$ (Lemma \ref{dddddddddddddddddd}), there
is a subnet $(\sigma _{j_{l}})_{l\in L}$ converging in the weak$^{\ast }$
topology to $\sigma _{\infty }\in K$. Via Equation (\ref{to prove2}) and the
triangle inequality, we then get that, for any $l\in L$, 
\begin{equation*}
\max_{\sigma \in K}\min_{\tilde{\sigma}\in K_{j_{l}}}\left\vert \left( 
\tilde{\sigma}-\sigma \right) \left( A\right) \right\vert \leq \left\vert
\left( \sigma _{j_{l}}-\sigma _{\infty }\right) \left( A\right) \right\vert
+\min_{\tilde{\sigma}\in K_{j_{l}}}\left\vert \left( \tilde{\sigma}-\sigma
_{\infty }\right) \left( A\right) \right\vert \ .
\end{equation*}%
By (\ref{sdsdsdsd}) and the fact that $(K_{j})_{j\in J}\subseteq \mathbf{K}(%
\mathcal{X}^{\ast })$ is an increasing net, it follows that%
\begin{equation}
\lim_{L}\max_{\sigma \in K}\min_{\tilde{\sigma}\in K_{j_{l}}}\left\vert
\left( \tilde{\sigma}-\sigma \right) \left( A\right) \right\vert =0\ .
\label{contradiction0}
\end{equation}%
By the convergence of decreasing bounded nets of real numbers, note that 
\begin{equation*}
\limsup_{J}\max_{\sigma \in K}\min_{\tilde{\sigma}\in K_{j}}\left\vert
\left( \tilde{\sigma}-\sigma \right) \left( A\right) \right\vert
=\liminf_{J}\max_{\sigma \in K}\min_{\tilde{\sigma}\in K_{j}}\left\vert
\left( \tilde{\sigma}-\sigma \right) \left( A\right) \right\vert
\end{equation*}%
and hence, (\ref{contradiction0}) contradicts (\ref{contradiction}). As a
consequence, Equation (\ref{limit prove}) holds true.
\end{proof}

Nonmonotone weak$^{\ast }$-Hausdorff convergent nets in $\mathbf{K}(\mathcal{%
X}^{\ast })$ are not trivial to study, in general. In the next proposition
we give preliminary, but completely general, results on limits of convergent
nets.

\begin{proposition}[Weak$^{\ast }$-Hausdorff hypertopology vs. upper and
lower limits]
\label{Solution selfbaby copy(5)+000}\mbox{ }\newline
Let $\mathcal{X}$ be a Banach space and $K_{\infty }\in \mathbf{K}(\mathcal{X%
}^{\ast })$ any weak$^{\ast }$-Hausdorff limit of a convergent net $%
(K_{j})_{j\in J}\subseteq \mathbf{K}(\mathcal{X}^{\ast })$. Then, 
\begin{equation*}
\mathrm{Li}\left( K_{j}\right) _{j\in J}\subseteq \overline{\mathrm{co}}%
\left( K_{\infty }\right) \qquad \text{and}\qquad K_{\infty }\subseteq 
\overline{\mathrm{co}}\left( \mathrm{Ls}(K_{j})_{j\in J}\right) \ ,
\end{equation*}%
where we recall that $\overline{\mathrm{co}}$ is the weak$^{\ast }$-closed
convex hull operator (Definition \ref{convex hull operator}).
\end{proposition}

\begin{proof}
Let $\mathcal{X}$ be a Banach space and $(K_{j})_{j\in J}\subseteq \mathbf{K}%
(\mathcal{X}^{\ast })$ any net converging to $K_{\infty }$. Assume without
loss of generality that $\mathrm{Li}\left( K_{j}\right) _{j\in J}$ is
nonempty. Let $\sigma _{\infty }\in \mathrm{Li}\left( K_{j}\right) _{j\in J}$%
, which is, by definition, the weak$^{\ast }$ limit of a net $(\sigma
_{j})_{j\in J}$ with $\sigma _{j}\in K_{j}$ for all $j\in J$. Then, for any $%
A\in \mathcal{X}$, 
\begin{equation*}
\min_{\sigma \in K_{\infty }}\left\vert \left( \sigma -\sigma _{\infty
}\right) \left( A\right) \right\vert \leq \left\vert \left( \sigma
_{j}-\sigma _{\infty }\right) \left( A\right) \right\vert +\min_{\sigma \in
K_{\infty }}\left\{ \left\vert \left( \sigma -\sigma _{j}\right) \left(
A\right) \right\vert \right\} \ .
\end{equation*}%
Taking this last inequality in the limit with respect to $J$ and using
Definition \ref{hypertopology0}, we deduce that 
\begin{equation}
\min_{\sigma \in K_{\infty }}\left\vert \left( \sigma -\sigma _{\infty
}\right) \left( A\right) \right\vert =0\ ,\qquad A\in \mathcal{X}\ .
\label{contradition0}
\end{equation}%
If $\sigma _{\infty }\notin \overline{\mathrm{co}}\left( K_{\infty }\right) $
then, as it is done to prove (\ref{hahn banach}), we infer from the
Hahn-Banach separation theorem \cite[Theorem 3.4 (b)]{Rudin} the existence
of $A_{0}\in \mathcal{X}$ and $x_{1},x_{2}\in \mathbb{R}$ such that 
\begin{equation*}
\max_{\sigma \in \overline{\mathrm{co}}\left( K_{\infty }\right) }\mathrm{Re}%
\left\{ \sigma \left( A_{0}\right) \right\} <x_{1}<x_{2}<\mathrm{Re}\left\{
\sigma _{\infty }\left( A_{0}\right) \right\} \ ,
\end{equation*}%
which contradicts (\ref{contradition0}) for $A=A_{0}$. As a consequence, $%
\sigma _{\infty }\in \overline{\mathrm{co}}\left( K_{\infty }\right) $ and,
hence, $\mathrm{Li}\left( K_{j}\right) _{j\in J}\subseteq \overline{\mathrm{%
co}}\left( K_{\infty }\right) $.

Conversely, let $\sigma _{\infty }\in K_{\infty }$. Since $K_{\infty }$ is
by definition the limit of $(K_{j})_{j\in J}$ (see Definition \ref%
{hypertopology0}), we deduce that 
\begin{equation*}
\lim_{J}\min_{\sigma \in K_{j}}\left\vert \left( \sigma -\sigma _{\infty
}\right) \left( A\right) \right\vert =0\ ,\qquad A\in \mathcal{X}\ .
\end{equation*}%
By combining this equality with Lemma \ref{dddddddddddddddddd copy(1)} and
the Banach-Alaoglu theorem \cite[Theorem 3.15]{Rudin}, for any $A\in 
\mathcal{X}$, there is $\sigma _{A}\in \mathrm{Ls}\left( K_{j}\right) _{j\in
J}$ such that 
\begin{equation*}
\sigma _{A}\left( A\right) =\sigma _{\infty }\left( A\right) \ .
\end{equation*}%
Consequently, one infers from the Hahn-Banach separation theorem \cite[%
Theorem 3.4 (b)]{Rudin} that $\sigma _{\infty }$ belongs to the weak$^{\ast
} $-closed convex hull of the upper limit $\mathrm{Ls}\left( K_{j}\right)
_{j\in J}$.
\end{proof}

Applied to nonempty convex weak$^{\ast }$-compact subsets of the dual space $%
\mathcal{X}^{\ast }$, Proposition \ref{Solution selfbaby copy(5)+000} reads
as follows:

\begin{corollary}[Weak$^{\ast }$-Hausdorff hypertopology and convexity vs.
upper and lower limits]
\label{Solution selfbaby copy(2)}\mbox{ }\newline
Let $\mathcal{X}$ be a Banach space and $K_{\infty }\in \mathbf{CK}(\mathcal{%
X}^{\ast })$ any weak$^{\ast }$-Hausdorff limit of a convergent net $%
(K_{j})_{j\in J}\subseteq \mathbf{CK}(\mathcal{X}^{\ast })$. Then, 
\begin{equation*}
\overline{\mathrm{Li}(K_{j})_{j\in J}}=\overline{\mathrm{co}}\left( \mathrm{%
Li}(K_{j})_{j\in J}\right) \subseteq K_{\infty }\subseteq \overline{\mathrm{%
co}}\left( \mathrm{Ls}(K_{j})_{j\in J}\right) \ .
\end{equation*}
\end{corollary}

\begin{proof}
The assertion is an obvious application of Proposition \ref{Solution
selfbaby copy(5)+000} to the subset $\mathbf{CK}(\mathcal{X}^{\ast
})\subseteq \mathbf{K}(\mathcal{X}^{\ast })$ together with the idempotency
of the weak$^{\ast }$-closed convex hull operator $\overline{\mathrm{co}}$.
Note that $\mathrm{Li}(K_{j})_{j\in J}$\ is a convex set.
\end{proof}

\subsection{Metrizable Hyperspaces}

We are interested in investigating \emph{metrizable} sub-hyperspaces of $%
\mathbf{F}(\mathcal{X}^{\ast })$. Metrizable topological spaces are
Hausdorff, so, in the light of Corollaries \ref{Non-Hausdorff hyperspaces}
and \ref{convexity corrolary}, we restrict our analysis on the Hausdorff
hyperspace $\mathbf{CK}(\mathcal{X}^{\ast })$ of all nonempty convex weak$%
^{\ast }$-compact subsets of $\mathcal{X}^{\ast }$, already defined by
Equation (\ref{ZDZD}) or (\ref{ZDZDZDZD}).

For a separable Banach space $\mathcal{X}$, we show how the well-known
metrizability of the weak$^{\ast }$ topology on balls of $\mathcal{X}^{\ast
} $ leads to the metrizability of the weak$^{\ast }$-Hausdorff hypertopology
on uniformly norm-bounded subsets of $\mathbf{CK}(\mathcal{X}^{\ast })$: Let 
\begin{equation}
\mathbf{CK}_{D}\left( \mathcal{X}^{\ast }\right) \doteq \left\{ K\in \mathbf{%
CK}\left( \mathcal{X}^{\ast }\right) :K\subseteq \mathbf{B}\left( 0,D\right)
\right\}  \label{ZD}
\end{equation}%
where 
\begin{equation}
\mathbf{B}\left( 0,D\right) \doteq \left\{ \sigma \in \mathcal{X}^{\ast
}:\left\Vert \sigma \right\Vert _{\mathcal{X}^{\ast }}\leq D\right\}
\subseteq \mathcal{X}^{\ast }  \label{norm ball}
\end{equation}%
is the norm-closed ball of radius $D\in \mathbb{R}^{+}$ in $\mathcal{X}%
^{\ast }$. If $\mathcal{X}$ is separable then the weak$^{\ast }$ topology is
metrizable on any ball $\mathbf{B}(0,D)$, $D\in \mathbb{R}^{+}$, by the
Banach-Alaoglu theorem \cite[Theorem 3.15]{Rudin} and \cite[Theorem 3.16]%
{Rudin}. Take any countable dense set $(A_{n})_{n\in \mathbb{N}}$ of the
unit ball of $\mathcal{X}$ and define the metric%
\begin{equation}
d\left( \sigma _{1},\sigma _{2}\right) \doteq \sum\limits_{n\in \mathbb{N}%
}2^{-n}\left\vert \left( \sigma _{1}-\sigma _{2}\right) \left( A_{n}\right)
\right\vert \ ,\qquad \sigma _{1},\sigma _{2}\in \mathcal{X}^{\ast }\ .
\label{metrics0}
\end{equation}%
This metric is well-defined and induces the weak$^{\ast }$ topology on $%
\mathbf{B}(0,D)$. Denote by $d_{H}$ the Hausdorff distance between two
elements $K_{1},K_{2}\in \mathbf{CK}_{D}(\mathcal{X}^{\ast })$, associated
with the metric $d$, as defined by (\ref{Hausdorf}), that is\footnote{%
Minima in (\ref{metric1}) directly come from the compactness of sets and the
continuity of $d$. The following maxima in (\ref{metric1}) result from the
compactness of sets and the fact that the minimum over a continuous map
defines an upper semicontinuous function.}, 
\begin{equation}
d_{H}\left( K_{1},K_{2}\right) \doteq \max \left\{ \max_{\sigma _{1}\in
K_{1}}\min_{\sigma _{2}\in K_{2}}d\left( \sigma _{1},\sigma _{2}\right)
,\max_{\sigma _{2}\in K_{2}}\min_{\sigma _{1}\in K_{1}}d\left( \sigma
_{1},\sigma _{2}\right) \right\} \ .  \label{metric1}
\end{equation}%
This Hausdorff distance induces the weak$^{\ast }$-Hausdorff hypertopology
on $\mathbf{CK}_{D}(\mathcal{X}^{\ast })$:

\begin{theorem}[Complete metrizability of the weak$^{\ast }$-Hausdorff
hypertopology]
\label{Solution selfbaby copy(4)+1}\mbox{ }\newline
Let $\mathcal{X}$ be a separable Banach space and $D\in \mathbb{R}^{+}$. The
family 
\begin{equation*}
\left\{ \left\{ K_{2}\in \mathbf{CK}_{D}\left( \mathcal{X}^{\ast }\right)
:d_{H}\left( K_{1},K_{2}\right) <r\right\} :r\in \mathbb{R}^{\mathbb{+}},\
K_{1}\in \mathbf{CK}_{D}\left( \mathcal{X}^{\ast }\right) \right\}
\end{equation*}%
is a basis of the weak$^{\ast }$-Hausdorff hypertopology of $\mathbf{CK}_{D}(%
\mathcal{X}^{\ast })$. Additionally, $\mathbf{CK}_{D}(\mathcal{X}^{\ast })$
is weak$^{\ast }$-Hausdorff-compact and completely metrizable.
\end{theorem}

\begin{proof}
Recall that a topology is finer than a second one iff any convergent net of
the first topology converges also in the second topology to the same limit.
See, e.g., \cite[Chapter 2, Theorems 4, 9]{topology}. We first show that the
topology induced by the Hausdorff metric $d_{H}$ is finer than the weak$%
^{\ast }$-Hausdorff hypertopology of $\mathbf{CK}_{D}(\mathcal{X}^{\ast })$
at fixed radius $D\in \mathbb{R}^{+}$: Take any net $(K_{j})_{j\in J}$
converging in $\mathbf{CK}_{D}(\mathcal{X}^{\ast })$ to $K$ in the topology
induced by the Hausdorff metric (\ref{metric1}). Let $A\in \mathcal{X}$ and
assume without loss of generality that $\left\Vert A\right\Vert _{\mathcal{X}%
}\leq 1$. By density of $(A_{n})_{n\in \mathbb{N}}$ in the unit ball of $%
\mathcal{X}$, for any $\varepsilon \in \mathbb{R}^{\mathbb{+}}$, there is $%
n\in \mathbb{N}$ such that, for all $j\in J$, 
\begin{equation*}
d_{H}^{(A)}(K,K_{j})\leq \varepsilon +d_{H}^{(A_{n})}(K,K_{j})\leq
\varepsilon +2^{n}d_{H}(K,K_{j})\ .
\end{equation*}%
Thus, the net $(K_{j})_{j\in J}$ converges to $K$ also in the weak$^{\ast }$%
-Hausdorff hypertopology.

Endowed with the Hausdorff metric topology, the space of closed subsets of a
compact metric space is compact, by \cite[Theorem 3.2.4]{Beer}. In
particular, by weak$^{\ast }$ compactness of norm-closed balls, $\mathbf{CK}%
_{D}(\mathcal{X}^{\ast })$ endowed with the Hausdorff metric $d_{H}$ is a
compact hyperspace. By Corollary \ref{convexity corrolary copy(2)}, $\mathbf{%
CK}_{D}(\mathcal{X}^{\ast })$ is closed with respect to the weak$^{\ast }$%
-Hausdorff hypertopology, and thus closed with respect to the topology
induced by $d_{H}$, because this topology is coarser than the weak$^{\ast }$%
-Hausdorff hypertopology, as proven above. Hence, $\mathbf{CK}_{D}(\mathcal{X%
}^{\ast })$ is also compact with respect to the topology induced by $d_{H}$.
Since the weak$^{\ast }$-Hausdorff hypertopology is a Hausdorff topology
(Corollary \ref{convexity corrolary}), as it is well-known \cite[Section 3.8
(a)]{Rudin}, both topologies must coincide: Take any subset $\mathcal{K}%
\subseteq \mathbf{CK}_{D}(\mathcal{X}^{\ast })$ which is closed with respect
to the Hausdorff metric $d_{H}$. By compactness of $(\mathbf{CK}_{D}(%
\mathcal{X}^{\ast }),d_{H})$, $\mathcal{K}$ is compact with respect to the
Hausdorff metric $d_{H}$ (see, e.g., \cite[Chapter 5, p. 140]{topology})
and, hence, also with respect to the weak$^{\ast }$-Hausdorff hypertopology.
Because any compact set in a Hausdorff space is closed \cite[Chapter 5,
Theorem 7]{topology}, by Corollary \ref{convexity corrolary}, $\mathcal{K}$
is closed with respect to the weak$^{\ast }$-Hausdorff hypertopology.
\end{proof}

\noindent Note that Theorem \ref{Solution selfbaby copy(4)+1} is similar to
the assertion \cite[End of p. 91]{Beer}. It leads to a strong improvement of
Proposition \ref{Solution selfbaby copy(5)+1} and Corollary \ref{Solution
selfbaby copy(2)}:

\begin{corollary}[Weak$^{\ast }$-Hausdorff hypertopology and
Kuratowski-Painlev\'{e} convergence]
\label{Solution selfbaby copy(4)}\mbox{ }\newline
Let $\mathcal{X}$ be a separable Banach space. Then any weak$^{\ast }$%
-Hausdorff convergent net $(K_{j})_{j\in J}\subseteq \mathbf{CK}(\mathcal{X}%
^{\ast })$ converges to the Kuratowski-Painlev\'{e} limit%
\begin{equation*}
K_{\infty }=\mathrm{Li}(K_{j})_{j\in J}=\mathrm{Ls}(K_{j})_{j\in J}\in 
\mathbf{CK}\left( \mathcal{X}^{\ast }\right) \ .
\end{equation*}
\end{corollary}

\begin{proof}
Recall that $\mathbf{CK}(\mathcal{X}^{\ast })\subseteq \mathbf{K}(\mathcal{X}%
^{\ast })$, see (\ref{ZDZDZDZD}). By Lemma \ref{dddddddddddddddddd copy(1)},
the union of any weak$^{\ast }$-Hausdorff convergent net in $\mathbf{CK}(%
\mathcal{X}^{\ast })$ is norm-bounded and, as a consequence, we can
restrict, without loss of generality, the study of weak$^{\ast }$-Hausdorff
hyperconvergent nets to the sub-hyperspace $\mathbf{CK}_{D}(\mathcal{X}%
^{\ast })$ for some $D\in \mathbb{R}^{+}$. By Theorem \ref{Solution selfbaby
copy(4)+1} the weak$^{\ast }$-Hausdorff hypertopology is induced by the
Hausdorff distance $d_{H}$ defined by (\ref{metric1}). The assertion thus
follows from \cite[\S\ 29, Section IX, Theorem 2]{topology-painleve}.
\end{proof}

\subsection{Generic Hypersets in Infinite Dimensions}

By Corollary \ref{convexity corrolary}, recall that $\mathbf{CK}(\mathcal{X}%
^{\ast })$ is a weak$^{\ast }$-Hausdorff-closed subset of $\mathbf{K}\left( 
\mathcal{X}^{\ast }\right) $. Let 
\begin{equation}
\mathcal{D}\doteq \left\{ K\in \mathbf{CK}\left( \mathcal{X}^{\ast }\right)
:K=\overline{\mathcal{E}\left( K\right) }\right\} \subseteq \mathbf{CK}%
\left( \mathcal{X}^{\ast }\right)  \label{Zbis}
\end{equation}%
be the subset of all $K\in \mathbf{CK}(\mathcal{X}^{\ast })$ with weak$%
^{\ast }$-dense set $\mathcal{E}\left( K\right) $ of extreme points (cf. the
Krein-Milman theorem \cite[Theorem 3.23]{Rudin}).

Recall that the so-called \emph{exposed} points are particular examples of
extreme ones: a point $\sigma _{0}\in K$ in a convex subset $K\subseteq 
\mathcal{X}^{\ast }$ is \emph{exposed} if there is $A\in \mathcal{X}$ such
that the real part of the weak$^{\ast }$-continuous functional $\hat{A}%
:\sigma \mapsto \sigma (A)$ from $\mathcal{X}^{\ast }$ to $\mathbb{C}$ (cf. (%
\ref{sdfsdfkljsdlfkj})) takes its \emph{unique} maximum on $K$ at $\sigma
_{0}\in K$. Considering exposed points instead of general extreme points is
technically convenient because of the weak$^{\ast }$-density of the set of
exposed points in the set of extreme points \cite[Theorem 6.2]%
{Phelps-Asplund} is an important ingredient to show that $\mathcal{D}$ is a $%
G_{\delta }$ subset of $\mathbf{CK}(\mathcal{X}^{\ast })$:

\begin{proposition}[$\mathcal{D}$ as a $G_{\protect\delta }$ set]
\label{Solution selfbaby copy(5)+00}\mbox{ }\newline
Let $\mathcal{X}$ be a separable Banach space. Then $\mathcal{D}$ is a $%
G_{\delta }$ subset of $\mathbf{CK}(\mathcal{X}^{\ast })$.
\end{proposition}

\begin{proof}
Let $\mathcal{X}$ be a separable Banach space. For any $D\in \mathbb{R}^{+}$%
, we can use the metric $d$ defined by (\ref{metrics0}) and generating the
weak$^{\ast }$ topology on the norm-closed ball $\mathbf{B}(0,D)$ of radius $%
D$, defined by (\ref{norm ball}). For any $D\in \mathbb{R}^{+}$, we denote
by 
\begin{equation}
B\left( \omega ,r\right) \doteq \left\{ \sigma \in \mathbf{B}\left(
0,D\right) :d\left( \omega ,\sigma \right) <r\right\}  \label{ball weak}
\end{equation}%
the weak$^{\ast }$-open ball of radius $r\in \mathbb{R}^{+}$ centered at $%
\omega \in \mathbf{B}(0,D)$. Then, for any $D\in \mathbb{R}^{+}$ and $m\in 
\mathbb{N}$, let $\mathcal{F}_{D,m}$ be the set of all nonempty convex weak$%
^{\ast }$-compact subsets $K\subseteq \mathbf{B}(0,D)$ such that $B\left(
\omega ,1/m\right) \cap \mathcal{E}(K)=\emptyset $ for some $\omega \in K$,
i.e., 
\begin{equation}
\mathcal{F}_{D,m}\doteq \left\{ K\in \mathbf{CK}_{D}\left( \mathcal{X}^{\ast
}\right) :\exists \omega \in K,\ B\left( \omega ,1/m\right) \cap \mathcal{E}%
\left( K\right) =\emptyset \right\} \subseteq \mathbf{CK}_{D}\left( \mathcal{%
X}^{\ast }\right) \ .  \label{Fm}
\end{equation}%
Recall again that $\mathcal{E}(K)$ is the nonempty set of extreme points of $%
K$ (cf. the Krein-Milman theorem \cite[Theorem 3.23]{Rudin}). Now, by
Equation (\ref{ZDZDZDZD}), observe that the complement of $\mathcal{D}$ (\ref%
{Zbis}) in $\mathbf{CK}(\mathcal{X}^{\ast })$ equals%
\begin{equation}
\mathbf{CK}\left( \mathcal{X}^{\ast }\right) \backslash \mathcal{D=}%
\bigcup_{D,m\in \mathbb{N}}\mathcal{F}_{D,m}\ .  \label{union}
\end{equation}%
Therefore, $\mathcal{D}$ is a $G_{\delta }$ subset of $\mathbf{CK}(\mathcal{X%
}^{\ast })$ if $\mathcal{F}_{D,m}$ is a weak$^{\ast }$-Hausdorff-closed set
for any $D,m\in \mathbb{N}$.

By Theorem \ref{Solution selfbaby copy(4)+1}, the weak$^{\ast }$-Hausdorff
hypertopology of $\mathbf{CK}_{D}(\mathcal{X}^{\ast })$ is metrizable and $%
\mathbf{CK}_{D}(\mathcal{X}^{\ast })$, being weak$^{\ast }$%
-Hausdorff-compact, is a weak$^{\ast }$-Hausdorff-closed subset of the
Hausdorff hyperspace $\mathbf{CK}(\mathcal{X}^{\ast })$ (see Corollary \ref%
{convexity corrolary} and \cite[Chapter 5, Theorem 7]{topology}). So, fix $%
D,m\in \mathbb{N}$ and take any sequence $(K_{n})_{n\in \mathbb{N}}\subseteq 
\mathcal{F}_{D,m}$ converging with respect to the weak$^{\ast }$-Hausdorff
hypertopology to $K_{\infty }\in $ $\mathbf{CK}_{D}(\mathcal{X}^{\ast })$.
For any $n\in \mathbb{N}$, there is $\omega _{n}\in K_{n}$ such that $%
B\left( \omega _{n},1/m\right) \cap \mathcal{E}(K_{n})=\emptyset $. By
metrizability and weak$^{\ast }$\ compactness of the ball $\mathbf{B}(0,D)$
and Corollary \ref{Solution selfbaby copy(4)}, there is a subsequence $%
(\omega _{n_{k}})_{k\in \mathbb{N}}$ converging to some $\omega _{\infty
}\in K_{\infty }$. Assume that, for some $\varepsilon \in (0,1/m)$, there \
is $\sigma _{\infty }\in \mathcal{E}(K_{\infty })$ such that%
\begin{equation*}
d\left( \omega _{\infty },\sigma _{\infty }\right) \leq \frac{1}{m}%
-\varepsilon \ .
\end{equation*}%
By the Mazur theorem (see, e.g., \cite[Theorem 1.20]{Phelps-Asplund}), the
Straszewicz theorem extended to all weak Asplund spaces \cite[Theorem 6.2]%
{Phelps-Asplund} and the Milman theorem \cite[Theorem 10.13]{BruPedra2}, the
set of exposed points of $K_{\infty }$ is weak$^{\ast }$-dense in $\mathcal{E%
}(K_{\infty })$. As a consequence, we can assume without loss of generality
that $\sigma _{\infty }$ is an exposed point. In particular, there is $A\in 
\mathcal{X}$ such that 
\begin{equation}
\max_{\sigma \in K_{\infty }}\mathrm{Re}\{\hat{A}(\sigma )\}=\hat{A}\left(
\sigma _{\infty }\right) \ ,  \label{itupeva1}
\end{equation}%
with $\sigma _{\infty }$ being the \emph{unique} maximizer in $K_{\infty }$.
Recall that $\hat{A}$ is the map $\sigma \mapsto \sigma (A)$ from $\mathcal{X%
}^{\ast }$ to $\mathbb{C}$ (cf. (\ref{sdfsdfkljsdlfkj})). Consider now the
sets%
\begin{equation*}
\mathcal{M}_{n}\doteq \left\{ \tilde{\sigma}\in K_{n}:\max_{\sigma \in K_{n}}%
\mathrm{Re}\{\hat{A}(\sigma )\}=\hat{A}\left( \tilde{\sigma}\right) \right\}
\ ,\qquad n\in \mathbb{N}\ .
\end{equation*}%
By affinity and weak$^{\ast }$-continuity of the function $\hat{A}$,
together with the weak$^{\ast }$-compactness of $K_{n}$, the set $\mathcal{M}%
_{n}$ is a convex weak$^{\ast }$-compact subset of $K_{n}$ for any $n\in 
\mathbb{N}$. In fact, $\mathcal{M}_{n}$ is a (weak$^{\ast }$-closed)\ face%
\footnote{%
It means that, if $\sigma \in \mathcal{M}_{n}$ is a finite convex
combination of elements $\sigma _{j}\in K_{n}$ then all $\sigma _{j}\in 
\mathcal{M}_{n}$.} of $K_{n}$ and thus, any extreme point of $\mathcal{M}%
_{n} $ belongs to $\mathcal{E}(K_{n})$. So, pick any extreme point $\sigma
_{n}\in \mathcal{E}(K_{n})$ of $\mathcal{M}_{n}$ for each $n\in \mathbb{N}$.
Since 
\begin{eqnarray*}
\max_{\sigma \in K_{n}}\mathrm{Re}\{\hat{A}(\sigma )\}-\max_{\tilde{\sigma}%
\in K_{\infty }}\mathrm{Re}\{\hat{A}(\tilde{\sigma})\} &=&\max_{\sigma \in
K_{n}}\min_{\tilde{\sigma}\in K_{\infty }}\mathrm{Re}\{\hat{A}(\sigma -%
\tilde{\sigma})\}\leq \max_{\sigma \in K_{n}}\min_{\tilde{\sigma}\in
K_{\infty }}\left\vert \left( \sigma -\tilde{\sigma}\right) \left( A\right)
\right\vert \ , \\
\max_{\tilde{\sigma}\in K_{\infty }}\mathrm{Re}\{\hat{A}(\tilde{\sigma}%
)\}-\max_{\sigma \in K_{n}}\mathrm{Re}\{\hat{A}(\sigma )\} &=&\max_{\tilde{%
\sigma}\in K_{\infty }}\min_{\sigma \in K_{n}}\mathrm{Re}\{\hat{A}(\tilde{%
\sigma}-\sigma )\}\leq \max_{\tilde{\sigma}\in K_{\infty }}\min_{\sigma \in
K_{n}}\left\vert \left( \sigma -\tilde{\sigma}\right) \left( A\right)
\right\vert \ ,
\end{eqnarray*}%
we deduce from Definition \ref{hypertopology} and the weak$^{\ast }$%
-Hausdorff convergence of $(K_{n})_{n\in \mathbb{N}}$ to $K_{\infty }$ that%
\begin{equation*}
\lim_{n\rightarrow \infty }\mathrm{Re}\{\hat{A}(\sigma
_{n})\}=\lim_{n\rightarrow \infty }\max_{\sigma \in K_{n}}\mathrm{Re}\{\hat{A%
}(\sigma )\}=\max_{\sigma \in K_{\infty }}\mathrm{Re}\{\hat{A}(\sigma )\}=%
\hat{A}\left( \sigma _{\infty }\right) \ .
\end{equation*}%
Therefore, keeping in mind the convergence of the subsequence $(\omega
_{n_{k}})_{k\in \mathbb{N}}$ towards $\omega _{\infty }\in K_{\infty }$,
there is a subsequence $(\sigma _{n_{k(l)}})_{l\in \mathbb{N}}$ of $(\sigma
_{n_{k}})_{k\in \mathbb{N}}$ (itself being a subsequence of $(\sigma
_{n})_{n\in \mathbb{N}}$) converging to $\sigma _{\infty }$, as it is the 
\emph{unique} maximizer of (\ref{itupeva1}) and $\hat{A}$ is weak$^{\ast }$%
-continuous. Since, for any $l\in \mathbb{N}$, 
\begin{multline*}
d(\sigma _{n_{k(l)}},\omega _{n_{k(l)}})\leq d(\sigma _{\infty },\omega
_{\infty })+d(\omega _{\infty },\omega _{n_{k(l)}})+d(\sigma
_{n_{k(l)}},\sigma _{\infty }) \\
\leq \frac{1}{m}-\varepsilon +d(\omega _{\infty },\omega
_{n_{k(l)}})+d(\sigma _{n_{k(l)}},\sigma _{\infty })
\end{multline*}%
with $\varepsilon \in (0,1/m)$ and $\sigma _{n}\in \mathcal{E}(K_{n})$ for $%
n\in \mathbb{N}$, we thus arrive at a contradiction. Therefore, $K_{\infty
}\in \mathcal{F}_{D,m}$. This means that $\mathcal{F}_{D,m}$ is a weak$%
^{\ast }$-Hausdorff-closed set for any $D,m\in \mathbb{N}$ and hence, the
countable union (\ref{union}) is a $F_{\sigma }$ set with complement being $%
\mathcal{D}$. The assertion follows, as the complement of an $F_{\sigma }$
set is a $G_{\delta }$ set.
\end{proof}

To show that $\mathcal{D}$ is weak$^{\ast }$-Hausdorff dense in the
hyperspace $\mathbf{CK}(\mathcal{X}^{\ast })$, like in the proof of \cite[%
Theorem 4.3]{FonfLindenstrauss} and in contrast with \cite{Klee}, we design
elements of $\mathcal{D}$ that approximate $K\in \mathbf{CK}(\mathcal{X}%
^{\ast })$ by using a procedure that is very similar to the construction of
the Poulsen simplex \cite{Poulsen}. Note however that Poulsen used the
existence of orthonormal bases in infinite-dimensional Hilbert spaces%
\footnote{%
In \cite{Poulsen}, Poulsen uses the Hilbert space $\ell ^{2}(\mathbb{N})$ to
construct his example of a convex compact set (in fact a simplex) with dense
extreme boundary.}. Here, the Hahn-Banach separation theorem \cite[Theorem
3.4 (b)]{Rudin} replaces the orthogonality property coming from the Hilbert
space structure. In all previous results \cite{Klee,FonfLindenstrauss} on
the density of convex compact sets with dense extreme boundary, the norm
topology is used, while the primordial topology is here the weak$^{\ast }$
topology. In this context, the metrizability of weak$^{\ast }$ and weak$%
^{\ast }$-Hausdorff topologies on norm-closed balls is pivotal. See Theorem %
\ref{Solution selfbaby copy(4)+1}. We give now the precise assertion along
with its proof:

\begin{theorem}[Weak$^{\ast }$-Hausdorf density of $\mathcal{D}$]
\label{Solution selfbaby copy(5)+0}\mbox{ }\newline
Let $\mathcal{X}$ be an infinite-dimensional separable Banach space. Then, $%
\mathcal{D}$ is a weak$^{\ast }$-Hausdorff dense subset of $\mathbf{CK}(%
\mathcal{X}^{\ast })$.
\end{theorem}

\begin{proof}
Let $\mathcal{X}$ be an infinite-dimensional separable Banach space and fix
once and for all a convex weak$^{\ast }$-compact subset $K\in \mathbf{CK}(%
\mathcal{X}^{\ast })$. The construction of convex weak$^{\ast }$-compact
sets in $\mathcal{D}$ approximating $K$ is done in several steps:\medskip

\noindent \underline{Step 0:} By Lemma \ref{dddddddddddddddddd}, $K$ belongs
to some norm-closed ball $\mathbf{B}(0,D)$ of radius $D\in \mathbb{R}^{+}$,
in other words, $K\in \mathbf{CK}_{D}(\mathcal{X}^{\ast })$, see (\ref{ZD})-(%
\ref{norm ball}). Therefore, we can use the metric $d$ defined by (\ref%
{metrics0}) and generating the weak$^{\ast }$ topology on $\mathbf{B}(0,D)$.
Then, for any fixed $\varepsilon \in \mathbb{R}^{+}$, there is a finite set $%
\{\omega _{j}\}_{j=1}^{n_{\varepsilon }}\subseteq K$, $n_{\varepsilon }\in 
\mathbb{N}$, such that 
\begin{equation}
K\subseteq \bigcup\limits_{j=1}^{n_{\varepsilon }}B\left( \omega
_{j},\varepsilon \right) \ ,  \label{dense1}
\end{equation}%
where $B\left( \omega ,r\right) \subseteq \mathbf{B}(0,D)$ denotes the weak$%
^{\ast }$-open ball (\ref{ball weak}) of radius $r\in \mathbb{R}^{+}$
centered at $\omega \in \mathcal{X}^{\ast }$. We then define the convex weak$%
^{\ast }$-compact set 
\begin{equation}
K_{0}\doteq \mathrm{co}\left\{ \omega _{1},\ldots ,\omega _{n_{\varepsilon
}}\right\} \subseteq \mathrm{span}\{\omega _{1},\ldots ,\omega
_{n_{\varepsilon }}\}\ .  \label{K0}
\end{equation}%
By (\ref{metric1}) and (\ref{dense1}), note that 
\begin{equation}
d_{H}(K,K_{0})\leq \varepsilon \ .  \label{K0bis}
\end{equation}%
\smallskip

\noindent \underline{Step 1:} Observe that the ball $\mathbf{B}(0,D)$ is weak%
$^{\ast }$-separable, by its weak$^{\ast }$ compactness (the Banach-Alaoglu
theorem \cite[Theorem 3.15]{Rudin}) and metrizability (cf. separability of $%
\mathcal{X}$ and \cite[Theorem 3.16]{Rudin}). Take any weak$^{\ast }$-dense
countable set $\{\varrho _{0,k}\}_{k\in \mathbb{N}}$ of $K_{0}$. By infinite
dimensionality of $\mathcal{X}^{\ast }$, there is $\sigma _{1}\in \mathcal{X}%
^{\ast }\backslash \mathrm{span}\{\omega _{1},\ldots ,\omega
_{n_{\varepsilon }}\}$ with 
\begin{equation}
\left\Vert \sigma _{1}\right\Vert _{\mathcal{X}^{\ast }}=D\ .
\label{D simple}
\end{equation}%
As in the proof of Proposition \ref{convexity lemma copy(1)}, recall that $%
\mathcal{X}^{\ast }$, endowed with the weak$^{\ast }$ topology, is a locally
convex (Hausdorff) space with $\mathcal{X}$ as its dual. Since $\{\sigma
_{1}\}$ is a convex weak$^{\ast }$-compact set and $\mathrm{span}\{\omega
_{1},\ldots ,\omega _{n_{\varepsilon }}\}$ is convex and weak$^{\ast }$%
-closed \cite[Theorem 1.42]{Rudin}, we infer from the Hahn-Banach separation
theorem \cite[Theorem 3.4 (b)]{Rudin} the existence of $A_{1}\in \mathcal{X}$
such that 
\begin{equation*}
\sup \left\{ \mathrm{Re}\left\{ \sigma \left( A_{1}\right) \right\} :\sigma
\in \mathrm{span}\{\omega _{1},\ldots ,\omega _{n_{\varepsilon }}\}\right\} <%
\mathrm{Re}\left\{ \sigma _{1}\left( A_{1}\right) \right\} \ .
\end{equation*}%
Since $\mathrm{span}\{\omega _{1},\ldots ,\omega _{n_{\varepsilon }}\}$ is a
linear space, observe that 
\begin{equation}
\mathrm{Re}\left\{ \sigma \left( A_{1}\right) \right\} =0\ ,\qquad \sigma
\in \mathrm{span}\{\omega _{1},\ldots ,\omega _{n_{\varepsilon }}\}\ .
\label{cool1}
\end{equation}%
Thus, by rescaling $A_{1}\in \mathcal{X}$, we can assume without loss of
generality that%
\begin{equation}
\mathrm{Re}\left\{ \sigma _{1}\left( A_{1}\right) \right\} =1\ .
\label{cool2}
\end{equation}%
Let 
\begin{equation}
\omega _{n_{\varepsilon }+1}\doteq \left( 1-\lambda _{1}\right) \varpi
_{1}+\lambda _{1}\sigma _{1}\ ,\qquad \text{with}\qquad \lambda _{1}\doteq
\min \left\{ 1,2^{-2}D^{-1}\varepsilon \right\} ,\ \varpi _{1}\doteq \varrho
_{0,1}\in K_{0}\ .  \label{omega1}
\end{equation}%
In contrast with the proof of \cite[Theorem 4.3]{FonfLindenstrauss}, we use
a convex combination to automatically ensure that $\left\Vert \omega
_{n_{\varepsilon }+1}\right\Vert _{\mathcal{X}^{\ast }}\leq D$, by convexity
of the (norm-closed) ball $\mathbf{B}(0,D)$. The inequality $\lambda
_{1}\leq 2^{-2}D^{-1}\varepsilon $ yields%
\begin{equation}
d\left( \omega _{n_{\varepsilon }+1},\varpi _{1}\right) \leq \left\Vert
\omega _{n_{\varepsilon }+1}-\varpi _{1}\right\Vert _{\mathcal{X}^{\ast
}}\leq 2^{-1}\varepsilon \ .  \label{dense2dense2}
\end{equation}%
Define the new convex weak$^{\ast }$-compact set%
\begin{equation*}
K_{1}\doteq \mathrm{co}\left\{ \omega _{1},\ldots ,\omega _{n_{\varepsilon
}+1}\right\} \subseteq \mathrm{span}\{\omega _{1},\ldots ,\omega
_{n_{\varepsilon }+1}\}\ .
\end{equation*}%
Observe that $\omega _{n_{\varepsilon }+1}$ is an exposed point of $K_{1}$,
by (\ref{cool1}) and (\ref{cool2}). By (\ref{metric1}), (\ref{K0}) and (\ref%
{dense2dense2}), note that $d_{H}(K_{0},K_{1})\leq 2^{-1}\varepsilon $,
which, by the triangle inequality and (\ref{K0bis}), yields 
\begin{equation}
d_{H}(K,K_{1})\leq \left( 1+2^{-1}\right) \varepsilon  \label{K1bis}
\end{equation}%
for an arbitrary (but previously fixed) $\varepsilon \in \mathbb{R}^{+}$%
.\medskip

\noindent \underline{Step 2:} Take any weak$^{\ast }$ dense countable set $%
\{\varrho _{1,k}\}_{k\in \mathbb{N}}$ of $K_{1}$. By infinite dimensionality
of $\mathcal{X}^{\ast }$, there is $\sigma _{2}\in \mathcal{X}^{\ast
}\backslash \mathrm{span}\{\omega _{1},\ldots ,\omega _{n_{\varepsilon
}+1}\} $ with 
\begin{equation}
\left\Vert \sigma _{2}\right\Vert _{\mathcal{X}^{\ast }}=\min \left\{
D,2^{-1}\left\Vert A_{1}\right\Vert _{\mathcal{X}}^{-1}\lambda _{1}\right\}
\ .  \label{sdfsdf}
\end{equation}%
As before, we deduce from the Hahn-Banach separation theorem \cite[Theorem
3.4 (b)]{Rudin} the existence of $A_{2}\in \mathcal{X}$ such that 
\begin{equation}
\mathrm{Re}\left\{ \sigma _{2}\left( A_{2}\right) \right\} =1\qquad \text{and%
}\qquad \mathrm{Re}\left\{ \sigma \left( A_{2}\right) \right\} =0\ ,\qquad
\sigma \in \mathrm{span}\{\omega _{1},\ldots ,\omega _{n_{\varepsilon
}+1}\}\ .  \label{cool3}
\end{equation}%
Let%
\begin{equation}
\omega _{n_{\varepsilon }+2}\doteq \left( 1-\lambda _{2}\right) \varpi
_{2}+\lambda _{2}\sigma _{2}\ ,\qquad \text{with}\qquad \lambda _{2}\doteq
\min \left\{ 1,2^{-3}D^{-1}\varepsilon \right\} ,\ \varpi _{2}\doteq \varrho
_{1,1}\in K_{1}\ .  \label{omega2}
\end{equation}%
In this case, similar to Inequality (\ref{dense2dense2}), 
\begin{equation}
d\left( \omega _{n_{\varepsilon }+2},\varpi _{2}\right) \leq \left\Vert
\omega _{n_{\varepsilon }+2}-\varpi _{2}\right\Vert _{\mathcal{X}^{\ast
}}\leq 2^{-2}\varepsilon \ .  \label{dense2dense3}
\end{equation}%
Define the new convex weak$^{\ast }$-compact set%
\begin{equation*}
K_{2}\doteq \mathrm{co}\left\{ \omega _{1},\ldots ,\omega _{n_{\varepsilon
}+2}\right\} \subseteq \mathrm{span}\{\omega _{1},\ldots ,\omega
_{n_{\varepsilon }+2}\}\ .
\end{equation*}%
By (\ref{cool3}), $\omega _{n_{\varepsilon }+2}$ is an exposed point of $%
K_{2}$, but it is not obvious that the exposed point $\omega
_{n_{\varepsilon }+1}$ of $K_{1}$ is still an exposed point of $K_{2}$, with
respect to $A_{1}\in \mathcal{X}$. This property is a consequence of%
\begin{equation*}
\mathrm{Re}\left\{ \omega _{n_{\varepsilon }+2}\left( A_{1}\right) \right\}
=\left( 1-\lambda _{2}\right) \mathrm{Re}\left\{ \varpi _{2}\left(
A_{1}\right) \right\} +\lambda _{2}\mathrm{Re}\left\{ \sigma _{2}\left(
A_{1}\right) \right\} <\mathrm{Re}\left\{ \omega _{n_{\varepsilon }+1}\left(
A_{1}\right) \right\} =\lambda _{1}\ ,
\end{equation*}%
(see (\ref{cool1}), (\ref{omega1}) and (\ref{omega2})), which holds true
because 
\begin{equation*}
\mathrm{Re}\left\{ \sigma _{2}\left( A_{1}\right) \right\} \leq
2^{-1}\lambda _{1}<\lambda _{1}\ ,
\end{equation*}%
by Equation (\ref{sdfsdf}). By (\ref{metric1}), (\ref{K1bis}) and (\ref%
{dense2dense3}) together with the triangle inequality,%
\begin{equation*}
d_{H}(K,K_{2})\leq \left( 1+2^{-1}+2^{-2}\right) \varepsilon
\end{equation*}%
for an arbitrary (but previously fixed) $\varepsilon \in \mathbb{R}^{+}$%
.\medskip

\noindent \underline{Step $n\rightarrow \infty $:} We now iterate the above
procedure, ensuring, at each step $n\geq 3$, that the addition of the
element 
\begin{equation}
\omega _{n_{\varepsilon }+n}\doteq \left( 1-\lambda _{n}\right) \varpi
_{n}+\lambda _{n}\sigma _{n}\ ,\qquad \text{with}\qquad \lambda _{n}\doteq
\min \left\{ 1,2^{-(n+1)}D^{-1}\varepsilon \right\} \ ,
\label{definition omegan}
\end{equation}%
in order to define the convex weak$^{\ast }$-compact set 
\begin{equation}
K_{n}\doteq \mathrm{co}\left\{ \omega _{1},\ldots ,\omega _{n_{\varepsilon
}+n}\right\} \subseteq \mathrm{span}\{\omega _{1},\ldots ,\omega
_{n_{\varepsilon }+n}\}\ ,  \label{Kn}
\end{equation}%
does not destroy the property\ of the elements $\omega _{n_{\varepsilon
}+1},\ldots ,\omega _{n_{\varepsilon }+n-1}$ being exposed. To this end, for
any $n\geq 2$, we choose $\sigma _{n}\in \mathcal{X}^{\ast }\backslash 
\mathrm{span}\{\omega _{1},\ldots ,\omega _{n_{\varepsilon }+n-1}\}$ such
that 
\begin{equation}
\left\Vert \sigma _{n}\right\Vert _{\mathcal{X}^{\ast }}=\min \left\{
D,2^{-1}\left\Vert A_{1}\right\Vert _{\mathcal{X}}^{-1}\lambda _{1},\ldots
,2^{-1}\left\Vert A_{n-1}\right\Vert _{\mathcal{X}}^{-1}\lambda
_{n-1}\right\} \ .  \label{toto}
\end{equation}%
Compare with (\ref{D simple}) and (\ref{sdfsdf}). Here, for any $j\in
\{1,\ldots ,n-1\}$, $A_{j}\in \mathcal{X}$ satisfies%
\begin{equation}
\mathrm{Re}\left\{ \sigma _{j}\left( A_{j}\right) \right\} =1\qquad \text{and%
}\qquad \mathrm{Re}\left\{ \sigma \left( A_{j}\right) \right\} =0\ ,\qquad
\sigma \in \mathrm{span}\{\omega _{1},\ldots ,\omega _{n_{\varepsilon
}+j-1}\}\ .  \label{totototo}
\end{equation}%
Compare with (\ref{cool1})-(\ref{cool2}) and (\ref{cool3}). We also have to
conveniently choose $\varpi _{n}\in K_{n-1}$ in order to get the asserted
weak$^{\ast }$ density. Like in the proof of \cite[Theorem 4.3]%
{FonfLindenstrauss} the sequence $(\varpi _{n})_{n\in \mathbb{N}}$ is chosen
such that 
\begin{equation*}
\{\varpi _{n}\}_{n\in \mathbb{N}}=\left\{ \varrho _{n,k}\right\} _{n\in 
\mathbb{N}_{0},k\in \mathbb{N}}
\end{equation*}%
and all the functionals $\varrho _{n,k}$ appear infinitely many times in the
sequence $(\varpi _{n})_{n\in \mathbb{N}}$. In this case, we obtain a weak$%
^{\ast }$-dense set $\{\omega _{n}\}_{n\in \mathbb{N}}$ in the convex weak$%
^{\ast }$-compact set 
\begin{equation}
K_{\infty }\doteq \overline{\mathrm{co}\left\{ \{\omega _{n}\}_{n\in \mathbb{%
N}}\right\} }\in \mathbf{CK}_{D}\left( \mathcal{X}^{\ast }\right) \ ,
\label{equaion}
\end{equation}%
which, by construction, satisfies 
\begin{equation*}
d_{H}(K,K_{\infty })\leq \sum_{n=0}^{\infty }2^{-n}\varepsilon =2\varepsilon
\end{equation*}%
for an arbitrary (but previously fixed) $\varepsilon \in \mathbb{R}^{+}$%
.\medskip

\noindent \underline{Step $n=\infty $:} It remains to verify that $\omega
_{n_{\varepsilon }+j}$, $j\in \mathbb{N}$, are exposed points of $K_{\infty
} $, whence $K_{\infty }\in \mathcal{D}$. By (\ref{definition omegan}) with $%
\varpi _{n}\in K_{n-1}$ (see (\ref{Kn})), for each natural number $n\geq j+1$%
, there are $\alpha _{n,j-1}^{(j)},\ldots ,\alpha _{n,n}^{(j)}\in \lbrack
0,1]$ and $\rho _{n}^{(j)}\in \mathrm{co}\left\{ \omega _{1},\ldots ,\omega
_{n_{\varepsilon }+j-1}\right\} $ such that 
\begin{equation}
\alpha _{n,j-1}^{(j)}+\alpha _{n,j}^{(j)}+\sum_{k=j+1}^{n}\alpha
_{n,k}^{(j)}\lambda _{k}=1\quad \text{and}\quad \omega _{n_{\varepsilon
}+n}=\alpha _{n,j-1}^{(j)}\rho _{n}^{(j)}+\alpha _{n,j}^{(j)}\omega
_{n_{\varepsilon }+j}+\sum_{k=j+1}^{n}\alpha _{n,k}^{(j)}\lambda _{k}\sigma
_{k}\ .  \label{inequality ddfdfinequality ddfdf}
\end{equation}%
Additionally, define $\alpha _{n,k}^{(j)}\doteq 1$ for all natural numbers $%
k\geq n$ while $\alpha _{n,k}^{(j)}\doteq 0$ for $k\in \mathbb{N}_{0}$ such
that $k\leq j-2$. Using (\ref{toto}), (\ref{totototo}) and (\ref{inequality
ddfdfinequality ddfdf}), at fixed $j\in \mathbb{N}$, we thus obtain that%
\begin{eqnarray}
\mathrm{Re}\left\{ \omega _{n_{\varepsilon }+n}\left( A_{j}\right) \right\}
&=&\alpha _{n,j}^{(j)}\mathrm{Re}\left\{ \omega _{n_{\varepsilon }+j}\left(
A_{j}\right) \right\} +\sum_{k=j+1}^{n}\alpha _{n,k}^{(j)}\lambda _{k}%
\mathrm{Re}\left\{ \sigma _{k}\left( A_{j}\right) \right\}  \notag \\
&\leq &\lambda _{j}\left( 1-2^{-1}\sum_{k=j+1}^{n}\alpha _{n,k}^{(j)}\lambda
_{k}\right)  \label{inequality ddfdf}
\end{eqnarray}%
for any $n\geq j+1$, while, for any natural number $n\leq j-1$, 
\begin{equation*}
\mathrm{Re}\left\{ \omega _{n_{\varepsilon }+n}\left( A_{j}\right) \right\}
=0\ ,
\end{equation*}%
using (\ref{totototo}). Fix $j\in \mathbb{N}$ and let $\omega _{\infty }\in
K_{\infty }$ be a solution to the variational problem%
\begin{equation}
\max_{\sigma \in K_{\infty }}\mathrm{Re}\left\{ \sigma \left( A_{j}\right)
\right\} =\mathrm{Re}\left\{ \omega _{\infty }\left( A_{j}\right) \right\}
\geq \mathrm{Re}\left\{ \omega _{n_{\varepsilon }+j}\left( A_{j}\right)
\right\} =\lambda _{j}\ .  \label{maximum}
\end{equation}%
($K_{\infty }$ is weak$^{\ast }$-compact.) By weak$^{\ast }$-density of $%
\{\omega _{n}\}_{n\in \mathbb{N}}$ in $K_{\infty }$, there is a sequence $%
(\omega _{n_{\varepsilon }+n_{l}})_{l\in \mathbb{N}}$ converging to $\omega
_{\infty }$ in the weak$^{\ast }$ topology. Since $K_{j}$ is weak$^{\ast }$%
-compact and $\alpha _{n,k}^{(j)}\in \lbrack 0,1]$ for all $k\in \mathbb{N}%
_{0}$ and $n,j\in \mathbb{N}$, by using a standard argument with a so-called
diagonal subsequence, we can choose the sequence $(n_{l})_{l\in \mathbb{N}}$
such that $(\rho _{n_{l}}^{(j)})$ weak$^{\ast }$-converges to $\rho _{\infty
}^{(j)}\in K_{j-1}$, and $(\alpha _{n_{l},k}^{(j)})_{l\in \mathbb{N}}$ has a
limit for any fixed $k\in \mathbb{N}_{0}$ and $j\in \mathbb{N}$. Using (\ref%
{definition omegan}), (\ref{inequality ddfdf}) and the inequality 
\begin{equation*}
\sum_{k=j+1}^{n_{l}}\alpha _{n_{l},k}^{(j)}\lambda _{k}\leq
D^{-1}\varepsilon \sum_{k=j+1}^{\infty
}2^{-(k+1)}=2^{-(j+1)}D^{-1}\varepsilon
\end{equation*}%
together with Lebesgue's dominated convergence theorem, we thus obtain that 
\begin{equation*}
\mathrm{Re}\left\{ \omega _{\infty }\left( A_{j}\right) \right\}
=\lim_{l\rightarrow \infty }\mathrm{Re}\left\{ \omega _{n_{\varepsilon
}+n_{l}}\left( A_{j}\right) \right\} \leq \lambda _{j}\left(
1-2^{-1}\sum_{k=j+1}^{\infty }\lambda _{k}\lim_{l\rightarrow \infty }\alpha
_{n_{l},k}^{(j)}\right) \ .
\end{equation*}%
Because of (\ref{maximum}), it follows that%
\begin{equation*}
\lim_{l\rightarrow \infty }\alpha _{n_{l},k}^{(j)}=0\ ,\qquad k\in
\{j+1,\ldots ,\infty \}\ ,
\end{equation*}%
leading to $\omega _{\infty }\in K_{j}$, by (\ref{definition omegan}), (\ref%
{inequality ddfdfinequality ddfdf})\ and Lebesgue's dominated convergence
theorem. (Recall that $K_{j}$ is defined by (\ref{Kn}) for $n=j\in \mathbb{N}
$.) Since $\omega _{n_{\varepsilon }+j}$ is by construction the unique
maximizer of 
\begin{equation*}
\max_{\sigma \in K_{j}}\mathrm{Re}\left\{ \sigma \left( A_{j}\right)
\right\} =\mathrm{Re}\left\{ \omega _{n_{\varepsilon }+j}\left( A_{j}\right)
\right\}
\end{equation*}%
and (\ref{maximum}) holds true with $\omega _{\infty }\in K_{j}$, we deduce
that $\omega _{\infty }=\omega _{n_{\varepsilon }+j}$, which is thus an
exposed point of $K_{\infty }$ for any $j\in \mathbb{N}$.
\end{proof}

Our proof differs in several important aspects from the one of \cite[Theorem
4.3]{FonfLindenstrauss}, even if it has the same general structure, inspired
by Poulsen's construction \cite{Poulsen}, as already mentioned. To be more
precise, as compared to the proof of \cite[Theorem 4.3]{FonfLindenstrauss}, 
\emph{Step 0} is new and is a direct consequence of the compactness and
metrizability of $K$, a property not assumed in \cite[Theorem 4.3]%
{FonfLindenstrauss}. \emph{Step 1} to \emph{Step }$n\rightarrow \infty $ are
similar to what is done in \cite{FonfLindenstrauss}, but with the essential
difference that convex combinations are used to produce new (strongly)\
exposed points and the required bounds on $\{\lambda _{n},\sigma
_{n}\}_{n\in \mathbb{N}}$ are thus quite different. Compare Equations (\ref%
{definition omegan}) and (\ref{toto}) with the bounds on $\upsilon
_{1},\upsilon _{2},\upsilon _{3}$ given in \cite[p. 27-29]{FonfLindenstrauss}%
, at parameters $r_{1}(t),r_{2}(t),r_{3}(t)=1$. In particular, \cite[Lemma
4.2]{FonfLindenstrauss}, which is essential to prove that the Poulsen-type
construction leads to a dense set of (strongly) exposed points in \cite[%
Theorem 4.3]{FonfLindenstrauss}, is \emph{never} used here. Instead, we use
other direct estimates on convex combinations to deduce this property. This
corresponds to \emph{Step }$n=\infty $.

Note finally that \cite[Theorem 4.3]{FonfLindenstrauss} shows the density of
convex compact sets with dense set of \emph{strongly} exposed points. A
strongly exposed point $\sigma _{0}$ in some convex set $K\subseteq \mathcal{%
X}^{\ast }$ is an\emph{\ }exposed point for some $A\in \mathcal{X}$ with the
additional property that any minimizing net of the real part of $\hat{A}$
(cf. (\ref{sdfsdfkljsdlfkj})) has to converge to $\sigma _{0}$ in the weak$%
^{\ast }$ topology\footnote{%
One should not mistake the notion of strongly exposed points discussed here
for the notion of weak$^{\ast }$ strongly exposed points of \cite[Definition
5.8]{Phelps-Asplund} where a weak$^{\ast }$ strongly exposed point is a (weak%
$^{\ast }$) exposed point with the additional property that any minimizing
net of the real part of $\hat{A}$ has to converge to $\sigma _{0}$ \emph{in
the norm topology} of $\mathcal{X}^{\ast }$.}. Observe that the only weak$%
^{\ast }$ accumulation point of such a minimizing net is the exposed point $%
\sigma _{0}$, by weak$^{\ast }$ continuity of $\hat{A}$. If $K$ is weak$%
^{\ast }$-compact, this yields that any minimizing net converges to $\sigma
_{0}$ in the weak$^{\ast }$ topology. In other words, any exposed point is 
\emph{automatically} strongly exposed in all convex weak$^{\ast }$-compact
sets $K\in \mathbf{CK}(\mathcal{X}^{\ast })$.

\section{Technical Proofs\label{Well-posedness sect copy(1)}}

The aim of this section is to prove Theorems \ref{theorem sdfkjsdklfjsdklfj
copy(3)} and \ref{proposition dynamic classique II}. In fact, we prove here
stronger results than these theorems. The proof of Theorem \ref{theorem
sdfkjsdklfjsdklfj copy(3)} is done in six lemmata and two corollaries. The
proof of Theorem \ref{proposition dynamic classique II} is a direct
consequence of Corollary \ref{Corollary bije+cocylbaby copy(1)}.

We start with a useful estimate on the norm-continuous two-para%
%TCIMACRO{\TeXButton{\-}{\-}}%
%BeginExpansion
\-%
%EndExpansion
meter family $(T_{t,s}^{\xi })_{s,t\in \mathbb{R}}$ of $\ast $-automorphisms
of $\mathcal{X}$ defined by the non-auto%
%TCIMACRO{\TeXButton{\-}{\-}}%
%BeginExpansion
\-%
%EndExpansion
nomous evolution equations (\ref{flow baby1})-(\ref{flow baby1bis}).

\begin{lemma}[Continuity of quantum dynamics]
\label{Solution selfbaby copy(1)}\mbox{ }\newline
Let $\mathcal{X}$ be a unital $C^{\ast }$-algebra. For any $h\in C_{b}\left( 
\mathbb{R};\mathfrak{Y}\left( \mathbb{R}\right) \right) $, $\xi _{1},\xi
_{2}\in C\left( \mathbb{R};E\right) $ and $s_{1},s_{2},t_{1},t_{2}\in 
\mathbb{R}$,%
\begin{multline*}
\left\Vert T_{t_{2},s_{2}}^{\xi _{2}}-T_{t_{1},s_{1}}^{\xi _{1}}\right\Vert
_{\mathcal{B}\left( \mathcal{X}\right) }\leq 2\left( \left\vert
t_{2}-t_{1}\right\vert +\left\vert s_{2}-s_{1}\right\vert \right) \left\Vert
h\right\Vert _{C_{b}\left( \mathbb{R};\mathfrak{Y}\left( \mathbb{R}\right)
\right) } \\
+2\int_{s_{2}}^{t_{2}}\left\Vert \mathrm{D}h\left( \alpha ;\xi _{1}\left(
\alpha \right) \right) -\mathrm{D}h\left( \alpha ;\xi _{2}\left( \alpha
\right) \right) \right\Vert _{\mathcal{X}}\mathrm{d}\alpha \ .
\end{multline*}
\end{lemma}

\begin{proof}
Fix $h\in C_{b}\left( \mathbb{R};\mathfrak{Y}\left( \mathbb{R}\right)
\right) $, $\xi _{1},\xi _{2}\in C\left( \mathbb{R};E\right) $ and $%
s_{1},s_{2},t_{1},t_{2}\in \mathbb{R}$. Via (\ref{babyreverse}), observe that%
\begin{equation*}
T_{t_{2},s_{2}}^{\xi _{2}}-T_{t_{1},s_{1}}^{\xi _{1}}=T_{t_{1},s_{1}}^{\xi
_{1}}\circ (T_{t_{2},t_{1}}^{\xi _{1}}-\mathbf{1}_{\mathcal{X}%
})+(T_{s_{1},s_{2}}^{\xi _{1}}-\mathbf{1}_{\mathcal{X}})\circ
T_{t_{2},s_{1}}^{\xi _{1}}+T_{t_{2},s_{2}}^{\xi _{2}}-T_{t_{2},s_{2}}^{\xi
_{1}}\ .
\end{equation*}%
Using (\ref{flow baby1})-(\ref{flow baby1bis}) together with (\ref%
{babyreverse}), we thus obtain the equality 
\begin{equation}
T_{t_{2},s_{2}}^{\xi _{2}}-T_{t_{1},s_{1}}^{\xi
_{1}}=\int_{t_{1}}^{t_{2}}T_{\alpha ,s_{1}}^{\xi _{1}}\circ X_{\alpha }^{\xi
_{1}\left( \alpha \right) }\mathrm{d}\alpha +\int_{s_{2}}^{s_{1}}X_{\alpha
}^{\xi _{1}\left( \alpha \right) }\circ T_{t_{2},\alpha }^{\xi _{1}}\mathrm{d%
}\alpha +\int_{s_{2}}^{t_{2}}T_{\alpha ,s_{2}}^{\xi _{2}}\circ \left(
X_{\alpha }^{\xi _{2}\left( \alpha \right) }-X_{\alpha }^{\xi _{1}\left(
\alpha \right) }\right) \circ T_{t_{2},\alpha }^{\xi _{1}}\mathrm{d}\alpha \
.  \label{T0}
\end{equation}%
For any $\xi \in C\left( \mathbb{R};E\right) $, $(T_{t,s}^{\xi })_{s,t\in 
\mathbb{R}}$ is a two-para%
%TCIMACRO{\TeXButton{\-}{\-}}%
%BeginExpansion
\-%
%EndExpansion
meter family of $\ast $-automorphisms of $\mathcal{X}$ and the generator $%
X_{t}^{\xi \left( t\right) }$ defined by (\ref{flow baby0}) has its operator
norm bounded by (\ref{bounded baby}). Therefore, the sum of the first two
terms in the right hand side of (\ref{T0}) is bounded by 
\begin{equation*}
2\left\vert t_{2}-t_{1}\right\vert \left\Vert h\right\Vert _{C_{b}\left( 
\mathbb{R};\mathfrak{Y}\left( \mathbb{R}\right) \right) }+2\left\vert
s_{2}-s_{1}\right\vert \left\Vert h\right\Vert _{C_{b}\left( \mathbb{R};%
\mathfrak{Y}\left( \mathbb{R}\right) \right) }\ ,
\end{equation*}%
while the last term in (\ref{T0}) is bounded by%
\begin{equation*}
2\int_{s_{2}}^{t_{2}}\left\Vert \mathrm{D}h\left( \alpha ;\xi _{1}\left(
\alpha \right) \right) -\mathrm{D}h\left( \alpha ;\xi _{2}\left( \alpha
\right) \right) \right\Vert _{\mathcal{X}}\mathrm{d}\alpha \ .
\end{equation*}
\end{proof}

We start now more specifically with the proof of Theorem \ref{theorem
sdfkjsdklfjsdklfj copy(3)}, by showing the existence and uniqueness of the
solution to the self-consistency equation. To this end, we basically use the
Banach fixed point theorem.

In contrast with Section \ref{Hausdorff Hypertopology}, note that, below,
the dual $\mathcal{X}^{\ast }$ of the unital $C^{\ast }$-algebra $\mathcal{X}
$ is always equipped with the usual norm for linear functionals on a normed
space. In particular, $\mathcal{X}^{\ast }$ is in this case a Banach space.
The set $E$ of states is a weak$^{\ast }$-compact subset of $\mathcal{X}%
^{\ast }$ in the weak$^{\ast }$ topology, but not in the norm topology,
unless $\mathcal{X}$ is finite-dimensional. This issue leads us to introduce
Conditions (a)-(b) of Theorem \ref{theorem sdfkjsdklfjsdklfj copy(3)}, that
is:

\begin{condition}
\label{Conditions (a)-(b)}\mbox{ }\newline
\emph{(a)} Let $\mathcal{X}$ be a unital $C^{\ast }$-algebra and $\mathfrak{B%
}$ a finite-dimensional real subspace of $\mathcal{X}^{\mathbb{R}}$%
.\smallskip \newline
\emph{(b)} Take $h\in C_{b}\left( \mathbb{R};\mathfrak{Y}\left( \mathbb{R}%
\right) \right) $ and a constant $D_{0}\in \mathbb{R}^{+}$ such that, for
all $t\in {\mathbb{R}}$,%
\begin{equation*}
\left\Vert \mathrm{D}h(t;\rho )-\mathrm{D}h(t;\tilde{\rho})\right\Vert _{%
\mathcal{X}}\leq D_{0}\sup_{B\in \mathfrak{B},\left\Vert B\right\Vert
=1}\left\vert \left( \rho -\tilde{\rho}\right) \left( B\right) \right\vert \
,\qquad \rho ,\tilde{\rho}\in E\ .
\end{equation*}
\end{condition}

\noindent We are now in a position to show the existence and uniqueness of
the solution to the self-consistency equation:

\begin{lemma}[Self-consistency equations]
\label{Solution selfbaby}\mbox{ }\newline
Under Condition \ref{Conditions (a)-(b)}, for any $s\in \mathbb{R}$ and $%
\rho \in E$, there is a unique solution $\varpi _{\rho ,s}\in C\left( 
\mathbb{R};E\right) $ to the following equation in $\xi \in C\left( \mathbb{R%
};E\right) $: 
\begin{equation}
\forall t\in {\mathbb{R}}:\qquad \xi \left( t\right) =\rho \circ
T_{t,s}^{\xi }\ .  \label{self consitence equation1baby}
\end{equation}%
Moreover, $\varpi _{\rho ,s}(t)=\varpi _{\varpi _{\rho ,s}(r),r}(t)$ for any 
$r,s,t\in {\mathbb{R}}$.
\end{lemma}

\begin{proof}
We prove the existence and uniqueness of a solution to (\ref{self consitence
equation1baby}) by using the Banach fixed point theorem, similar to the
Picard-Lindel\"{o}f theory for ODEs, keeping in mind that $E$ is endowed
with the weak$^{\ast }$ topology: Pick a function $h\in C_{b}\left( \mathbb{R%
};\mathfrak{Y}\left( \mathbb{R}\right) \right) $, an initial time $s\in 
\mathbb{R}$ and a state $\rho \in E$. For $\epsilon \in \mathbb{R}^{+}$,
define the map $\mathfrak{F}$ from 
\begin{equation*}
\mathcal{C}_{\epsilon ,s}\doteq C\left( [s-\epsilon ,s+\epsilon ];\mathcal{X}%
^{\ast }\right) \cap C\left( [s-\epsilon ,s+\epsilon ];E\right)
\end{equation*}%
to itself by 
\begin{equation}
\mathfrak{F}\left( \xi \right) \left( t\right) \doteq \rho \circ
T_{t,s}^{\xi }\ ,\qquad t\in \left[ s-\epsilon ,s+\epsilon \right] \ .
\label{labelbaby}
\end{equation}%
The continuity of $\mathfrak{F}\left( \xi \right) $ in the Banach space $%
C\left( [s-\epsilon ,s+\epsilon ];\mathcal{X}^{\ast }\right) $ can directly
be read from Lemma \ref{Solution selfbaby copy(1)} and Condition \ref%
{Conditions (a)-(b)} (b). The same also yields the contractivity of $%
\mathfrak{F}$ for sufficiently small $\epsilon \in \mathbb{R}^{+}$,
uniformly with respect to $s\in \mathbb{R}$ and $\rho \in E$. Using the
Banach fixed point theorem, there is a unique solution $\varpi _{\rho ,s}$
to $\mathfrak{F}\left( \xi \right) =\xi $ in $\mathcal{C}_{\epsilon ,s}$. By
exactly the same arguments, observe that, for each $r\in \left[ s-\epsilon
,s+\epsilon \right] $, the following self-consistency equation 
\begin{equation}
\forall t\in \lbrack r-\tilde{\epsilon},r+\tilde{\epsilon}]:\qquad \xi
\left( t\right) =\varpi _{\rho ,s}\left( r\right) \circ T_{t,r}^{\xi }\ ,
\label{solutionplus}
\end{equation}%
has also a unique solution $\varpi _{\varpi _{\rho ,s}(r),r}$ in $\mathcal{C}%
_{\tilde{\epsilon},r}$ for any $\tilde{\epsilon}\in (0,\epsilon ]$. By the
reverse cocycle property (\ref{babyreverse}), at fixed $s\in \mathbb{R}$ and 
$\rho \in E$, $\varpi _{\rho ,s}$ solves (\ref{solutionplus}) for any $r\in
(s-\epsilon ,s+\epsilon )$ and $t\in \lbrack s-\tilde{\epsilon},s+\tilde{%
\epsilon}]$ with $\tilde{\epsilon}=\epsilon -|s-r|\in \mathbb{R}^{+}$,
whence 
\begin{equation}
\varpi _{\rho ,s}(t)=\varpi _{\varpi _{\rho ,s}(r),r}(t)\ ,\qquad r\in
(s-\epsilon ,s+\epsilon ),\ t\in \lbrack s-\tilde{\epsilon},s+\tilde{\epsilon%
}]\ .  \label{causalite}
\end{equation}

Now, assume the existence and uniqueness of a solution $\varpi _{\rho ,s}$
to $\mathfrak{F}\left( \xi \right) =\xi $ in $\mathcal{C}_{\epsilon _{0},s}$
for some parameter $\epsilon _{0}\in \mathbb{R}^{+}$. Take $r\in (s-\epsilon
_{0},s-\epsilon _{0}+\epsilon )\cup (s+\epsilon _{0}-\epsilon ,s+\epsilon
_{0})$. By combining the existence and uniqueness of a solution $\varpi
_{\varpi _{\rho ,s}(r),r}$ to (\ref{solutionplus}) in $\mathcal{C}_{\tilde{%
\epsilon},r}$ together with the reverse cocycle property (\ref{babyreverse}%
), we deduce that 
\begin{equation*}
\varpi _{\rho ,s}(t)=\varpi _{\varpi _{\rho ,s}(r),r}(t)\ ,\qquad t\in
(s-\epsilon _{0},s+\epsilon _{0})\ ,
\end{equation*}%
as well as the existence of a unique solution $\varpi _{\rho ,s}$ to $%
\mathfrak{F}\left( \xi \right) =\xi $ in $\mathcal{C}_{\epsilon
_{0}+\epsilon ,s}$. As a consequence, one can infer from a contradiction
argument the existence and uniqueness of a solution in $C\left( \mathbb{R};%
\mathcal{X}^{\ast }\right) \cap C\left( \mathbb{R};E\right) $ to (\ref{self
consitence equation1baby}). Moreover, this solution must satisfy the
equality $\varpi _{\rho ,s}(t)=\varpi _{\varpi _{\rho ,s}(r),r}(t)$ for any $%
r,s,t\in {\mathbb{R}}$.

Finally, to prove uniqueness in $C\left( \mathbb{R};E\right) $, we observe
from Lemma \ref{Solution selfbaby copy(1)} and Condition \ref{Conditions
(a)-(b)} that any solution in $C\left( \mathbb{R};E\right) $ (i.e.,
continuous with respect to the weak$^{\ast }$ topology in $E$) to (\ref{self
consitence equation1baby}) is automatically in $C\left( \mathbb{R};\mathcal{X%
}^{\ast }\right) $ (i.e., continuous with respect to the norm topology in $%
\mathcal{X}^{\ast }$).
\end{proof}

\begin{corollary}[Bijectivity of the solution to the self-consistency
equation]
\label{bijectivitybaby}\mbox{ }\newline
Under Condition \ref{Conditions (a)-(b)}, for any $s,t\in \mathbb{R}$, $%
\varpi _{s}\left( t\right) \equiv (\varpi _{\rho ,s}\left( t\right) )_{\rho
\in E}$ is a bijective map from $E$ to itself.
\end{corollary}

\begin{proof}
This is a straightforward consequence of Lemma \ref{Solution selfbaby}, in
particular the equality $\varpi _{\rho ,s}(t)=\varpi _{\varpi _{\rho
,s}(r),r}(t)$ for any $r,s,t\in {\mathbb{R}}$.
\end{proof}

\begin{lemma}[Differentiability of the solution -- I]
\label{Differentiability2baby}\mbox{ }\newline
Under Condition \ref{Conditions (a)-(b)}, for $s\in \mathbb{R}$ and $\rho
\in E$, $\varpi _{\rho ,s}\in C^{1}\left( \mathbb{R};\mathcal{X}^{\ast
}\right) $ with derivative given by%
\begin{equation*}
\partial _{t}\varpi _{\rho ,s}\left( t\right) =\rho \circ T_{t,s}^{\varpi
_{\rho ,s}}\circ X_{t}^{\varpi _{\rho ,s}\left( t\right) }\ ,\qquad t\in 
\mathbb{R}\ .
\end{equation*}
\end{lemma}

\begin{proof}
This is a direct consequence of Equation (\ref{flow baby1}) together with
Lemma \ref{Solution selfbaby}.
\end{proof}

\begin{lemma}[Continuity with respect to the initial condition]
\label{lemma contnuitybaby}\mbox{ }\newline
Under Condition \ref{Conditions (a)-(b)}, for any $s,t\in \mathbb{R}$, $%
\varpi _{s}\left( t\right) \equiv (\varpi _{\rho ,s}\left( t\right) )_{\rho
\in E}\in C\left( E;E\right) $.
\end{lemma}

\begin{proof}
Take $s\in \mathbb{R}$ and two states $\rho _{1},\rho _{2}\in E$. Then,
define the quantity%
\begin{equation*}
\mathbf{X}\left( \epsilon \right) \doteq \max_{t\in \left[ s-\epsilon
,s+\epsilon \right] }\max_{B\in \mathfrak{B},\left\Vert B\right\Vert
=1}\left\vert \left( \varpi _{\rho _{1},s}\left( t\right) -\varpi _{\rho
_{2},s}\left( t\right) \right) \left( B\right) \right\vert \ ,\qquad
\epsilon \in \mathbb{R}^{+}\ .
\end{equation*}%
Because $\varpi _{\rho ,s}\left( t\right) =\rho \circ T_{t,s}^{\varpi _{\rho
,s}}$ (Lemma \ref{Solution selfbaby}) with $(T_{t,s}^{\xi })_{s,t\in \mathbb{%
R}}$ being a family of $\ast $-automorphisms of $\mathcal{X}$ for any $\xi
\in C\left( \mathbb{R};E\right) $, this positive number is bounded by 
\begin{equation}
\mathbf{X}\left( \epsilon \right) \leq \max_{B\in \mathfrak{B},\left\Vert
B\right\Vert =1}\left\vert \left( \rho _{1}-\rho _{2}\right) \circ
T_{t,s}^{\varpi _{\rho _{1},s}}\left( B\right) \right\vert +\mathbf{Y}\left(
\epsilon \right) \ ,  \label{X0}
\end{equation}%
where 
\begin{equation}
\mathbf{Y}\left( \epsilon \right) \doteq \max_{t\in \left[ s-\epsilon
,s+\epsilon \right] }\left\Vert T_{t,s}^{\varpi _{\rho
_{1},s}}-T_{t,s}^{\varpi _{\rho _{2},s}}\right\Vert _{\mathcal{B}\left( 
\mathcal{X}\right) }\ .  \label{Y0}
\end{equation}%
By Lemma \ref{Solution selfbaby copy(1)}, the last quantity is bounded by%
\begin{equation}
\mathbf{Y}\left( \epsilon \right) \leq 2\max_{t\in \left[ s-\epsilon
,s+\epsilon \right] }\left\{ \int_{s}^{t}\left\Vert \mathrm{D}h\left( \alpha
;\varpi _{\rho _{1},s}\left( \alpha \right) \right) -\mathrm{D}h\left(
\alpha ;\varpi _{\rho _{2},s}\left( \alpha \right) \right) \right\Vert _{%
\mathcal{X}}\mathrm{d}\alpha \right\} \ ,  \label{X}
\end{equation}%
which, together with Condition \ref{Conditions (a)-(b)} (b), leads to 
\begin{equation}
\mathbf{Y}\left( \epsilon \right) \leq 2D_{0}\epsilon \mathbf{X}\left(
\epsilon \right) \ ,\qquad \epsilon \in \mathbb{R}^{+}\ .
\label{gronwalbaby0}
\end{equation}%
By Inequality (\ref{X0}), it follows that 
\begin{equation}
\left( 1-2D_{0}\epsilon \right) \mathbf{X}\left( \epsilon \right) \leq
\max_{B\in \mathfrak{B},\left\Vert B\right\Vert =1}\left\vert \left( \rho
_{1}-\rho _{2}\right) \circ T_{t,s}^{\varpi _{\rho _{1},s}}\left( B\right)
\right\vert \ ,\qquad \epsilon \in \mathbb{R}^{+}\ .  \label{gronwalbaby}
\end{equation}%
Now, we combine $\varpi _{\rho ,s}\left( t\right) =\rho \circ
T_{t,s}^{\varpi _{\rho ,s}}$ with (\ref{Y0}) and (\ref{gronwalbaby0})-(\ref%
{gronwalbaby}) to get the inequality%
\begin{eqnarray}
\left\vert \varpi _{\rho _{1},s}\left( t\right) \left( A\right) -\varpi
_{\rho _{2},s}\left( t\right) \left( A\right) \right\vert &\leq &\left\vert
\left( \rho _{1}-\rho _{2}\right) \circ T_{t,s}^{\varpi _{\rho
_{1},s}}\left( A\right) \right\vert  \label{gronwalbabyfinal} \\
&&+\frac{2D_{0}\epsilon \left\Vert A\right\Vert _{\mathcal{X}}}{\left(
1-2D_{0}\epsilon \right) }\max_{B\in \mathfrak{B},\left\Vert B\right\Vert
=1}\left\vert \left( \rho _{1}-\rho _{2}\right) \circ T_{t,s}^{\varpi _{\rho
_{1},s}}\left( B\right) \right\vert  \notag
\end{eqnarray}%
for any $s\in \mathbb{R}$, $\rho _{1},\rho _{2}\in E$, $A\in \mathcal{X}$, $%
\epsilon \in \left( 0,D_{0}/2\right) $ and $t\in \left[ s-\epsilon
,s+\epsilon \right] $. By finite dimensionality of $\mathfrak{B}$ (Condition %
\ref{Conditions (a)-(b)} (a)), the norm and weak$^{\ast }$ topologies of $%
\mathfrak{B}^{\ast }$ are the same and the weak$^{\ast }$ continuity
property of $\varpi _{s}\left( t\right) $ follows from (\ref%
{gronwalbabyfinal}) for any times $s\in \mathbb{R}$ and $t\in \left[
s-\epsilon ,s+\epsilon \right] $, provided $\epsilon <D_{0}/2$. Using now
the equality $\varpi _{\rho ,s}(t)=\varpi _{\varpi _{\rho ,s}(r),r}(t)$ for
any $r,s,t\in {\mathbb{R}}$ (Lemma \ref{Solution selfbaby}), we thus deduce
the weak$^{\ast }$ continuity of $\varpi _{s}\left( t\right) $ for all times 
$s,t\in \mathbb{R}$.
\end{proof}

\begin{corollary}[Solution to the self-consistency equation as homeomorphism
family]
\label{Corollary bije+cocylbaby}\mbox{ }\newline
Under Condition \ref{Conditions (a)-(b)}, at any fixed times $s,t\in \mathbb{%
R}$, $\varpi _{s}\left( t\right) \equiv (\varpi _{\rho ,s}\left( t\right)
)_{\rho \in E}\in \mathrm{Aut}\left( E\right) $, i.e., $\varpi _{s}\left(
t\right) $ is a automorphism of the state space $E$. Moreover, it satisfies
a cocycle property:%
\begin{equation}
\forall s,r,t\in \mathbb{R}:\qquad \varpi _{s}\left( t\right) =\varpi
_{r}\left( t\right) \circ \varpi _{s}\left( r\right) \ .
\label{cocycle propertybaby}
\end{equation}
\end{corollary}

\begin{proof}
By Corollary \ref{bijectivitybaby} and Lemma \ref{lemma contnuitybaby}, for
any $s,t\in \mathbb{R}$, $\varpi _{s}\left( t\right) $ is a weak$^{\ast }$%
-continuous bijective map from $E$ to itself. Recall that $E$ is the
(Hausdorff) topological space of all states on $\mathcal{X}$ with the weak$%
^{\ast }$ topology. It is weak$^{\ast }$-compact. Therefore, the inverse of $%
\varpi _{s}\left( t\right) $ is also weak$^{\ast }$-continuous. Equation (%
\ref{cocycle propertybaby}) is only another way to write the equality $%
\varpi _{\rho ,s}(t)=\varpi _{\varpi _{\rho ,s}(r),r}(t)$ of Lemma \ref%
{Solution selfbaby}.
\end{proof}

Recall that the set $\mathrm{Aut}\left( E\right) $ of all automorphisms (or
self-homeomorphisms) of $E$ is endowed with the topology of uniform
convergence of weak$^{\ast }$-continuous functions from $E$ to itself. See (%
\ref{uniform convergence weak*}). Having this in mind, we obtain now the
following lemma:

\begin{lemma}[Well-posedness of the self-consistency equation]
\label{lemma wellbaby}\mbox{ }\newline
Under Condition \ref{Conditions (a)-(b)}, for any $s\in \mathbb{R}$,%
\begin{equation*}
\mathbf{\varpi }_{s}\equiv (\varpi _{s}\left( t\right) )_{t\in \mathbb{R}%
}\equiv ((\varpi _{\rho ,s}\left( t\right) )_{\rho \in E})_{t\in \mathbb{R}%
}\in C\left( \mathbb{R};\mathrm{Aut}\left( E\right) \right) \ .
\end{equation*}
\end{lemma}

\begin{proof}
Take any net $(t_{j})_{j\in J}\subseteq \mathbb{R}$ converging to some
arbitrary time $t\in \mathbb{R}$. Assume that $\varpi _{s}\left(
t_{j}\right) $ does not converge to $\varpi _{s}\left( t\right) $, in the
topology of uniform convergence of weak$^{\ast }$-continuous functions. In
this case, by (\ref{uniform convergence weak*}), there is a net $(\rho
_{j})_{j\in J}\subseteq E$ of states, $A\in \mathcal{X}$ and $\varepsilon
\in \mathbb{R}^{+}$ such that 
\begin{equation}
\liminf_{j\in J}\left\vert \left[ \varpi _{\rho _{j},s}\left( t_{j}\right)
-\varpi _{\rho _{j},s}\left( t\right) \right] \left( A\right) \right\vert
\geq \varepsilon >0\ .  \label{machin2baby}
\end{equation}%
By weak$^{\ast }$ compactness of $E$, there is a subnet $(\rho
_{j_{i}})_{i\in I}$ weak$^{\ast }$-converging to some $\rho \in E$. By
Lemmata \ref{Solution selfbaby}, \ref{lemma contnuitybaby} and Inequality (%
\ref{machin2baby}), it follows that 
\begin{equation*}
\liminf_{i\in I}\left\vert \left[ \rho _{j_{i}}\circ T_{t_{j_{i}},s}^{\varpi
_{\rho _{j_{i}},s}}-\varpi _{\rho ,s}\left( t\right) \right] \left( A\right)
\right\vert \geq \varepsilon >0\ .
\end{equation*}%
Using (\ref{Y0}) and (\ref{gronwalbaby0})-(\ref{gronwalbaby}) together with
the reverse cocycle property (\ref{babyreverse}) and the fact that $%
(T_{t,s}^{\xi })_{s,t\in \mathbb{R}}$ is a family of $\ast $-automorphisms
of $\mathcal{X}$ for any $\xi \in C\left( \mathbb{R};E\right) $, we thus
deduce from the last inequality that 
\begin{equation}
\liminf_{i\in I}\left\vert \left[ \rho _{j_{i}}\circ T_{t_{j_{i}},s}^{\varpi
_{\rho ,s}}-\varpi _{\rho ,s}\left( t\right) \right] \left( A\right)
\right\vert \geq \varepsilon >0\ .  \label{estimate}
\end{equation}%
This is a contradiction because $(T_{t,s}^{\varpi _{\rho ,s}})_{s,t\in {%
\mathbb{R}}}$ is a norm-continuous two-para%
%TCIMACRO{\TeXButton{\-}{\-}}%
%BeginExpansion
\-%
%EndExpansion
meter family. Hence, for any $A\in \mathcal{X}$,%
\begin{equation*}
\lim_{i\in I}\rho _{j_{i}}\circ T_{t_{j_{i}},s}^{\varpi _{\rho ,s}}\left(
A\right) =\rho \circ T_{t,s}^{\varpi _{\rho ,s}}\left( A\right) =\left[
\varpi _{\rho ,s}\left( t\right) \right] \left( A\right) \ .
\end{equation*}
\end{proof}

\begin{lemma}[Joint continuity with respect to initial and final times]
\label{lemma wellbabybaby}\mbox{ }\newline
Under Condition \ref{Conditions (a)-(b)}, the solution to the
self-consistency equation is jointly continuous with respect to initial and
final times: 
\begin{equation*}
\mathbf{\varpi }\equiv (\mathbf{\varpi }_{s})_{s\in \mathbb{R}}\equiv
(\varpi _{s}\left( t\right) )_{s,t\in \mathbb{R}}\equiv ((\varpi _{\rho
,s}\left( t\right) )_{\rho \in E})_{s,t\in \mathbb{R}}\in C\left( \mathbb{R}%
^{2};\mathrm{Aut}\left( E\right) \right) \ .
\end{equation*}
\end{lemma}

\begin{proof}
We use again the Banach fixed point theorem: Fix $s\in \mathbb{R}$ and $%
\epsilon \in \mathbb{R}^{+}$. Similar to (\ref{labelbaby}), we define the
map $\mathfrak{F}$ from 
\begin{equation*}
C\left( [s-\epsilon ,s+\epsilon ]^{2};\mathcal{X}^{\ast }\right) \cap
C\left( [s-\epsilon ,s+\epsilon ]^{2};E\right)
\end{equation*}%
to itself by%
\begin{equation*}
\mathfrak{F}\left( \zeta \right) \left( r,t\right) \doteq \rho \circ
T_{t,r}^{\zeta \left( r,\cdot \right) },\qquad r,t\in \left[ s-\epsilon
,s+\epsilon \right] \ ,
\end{equation*}%
where 
\begin{equation*}
\zeta \left( r,\cdot \right) \in C\left( [s-\epsilon ,s+\epsilon ];\mathcal{X%
}^{\ast }\right) \cap C\left( [s-\epsilon ,s+\epsilon ];E\right)
\end{equation*}%
is the function defined, at fixed $r\in \lbrack s-\epsilon ,s+\epsilon ]$,
by $\zeta \left( r,t\right) $ for any $t\in \lbrack s-\epsilon ,s+\epsilon ]$%
. By Lemma \ref{Solution selfbaby copy(1)} and\ Condition \ref{Conditions
(a)-(b)} (b), $\mathfrak{F}$ is a contraction for sufficiently small times $%
\epsilon \in \mathbb{R}^{+}$ and we use similar arguments as in the proof of
Lemma \ref{Solution selfbaby} to show the existence of a unique solution $%
\gimel $ to the following equation in $\zeta \in C\left( [s-\epsilon
,s+\epsilon ]^{2};E\right) $:%
\begin{equation*}
\forall r,t\in \left[ s-\epsilon ,s+\epsilon \right] :\qquad \zeta \left(
r,t\right) =\rho \circ T_{t,r}^{\zeta \left( r,\cdot \right) }\ .
\end{equation*}%
By uniqueness of the solution to (\ref{self consitence equation1baby}) in $%
C\left( \mathbb{R};E\right) $ at any fixed $s\in \mathbb{R}$, $\varpi _{\rho
,r}\left( t\right) =\gimel \left( r,t\right) $ for any $r,t\in \lbrack
s-\epsilon ,s+\epsilon ]$. By Corollary \ref{Corollary bije+cocylbaby} and
Lemma \ref{lemma wellbaby}, it follows that%
\begin{equation*}
(\varpi _{\rho ,s}\left( t\right) )_{s,t\in \mathbb{R}}\in C\left( \mathbb{R}%
^{2};E\right) \ ,\qquad \rho \in E\ .
\end{equation*}%
Finally, by similar compactness arguments as in the proof of Lemma \ref%
{lemma wellbaby}, we deduce the assertion.
\end{proof}

\begin{lemma}[Differentiability of the solution -- II]
\label{lemma contnuity copy(1)}\mbox{ }\newline
Fix $n\in \mathbb{N}$, $g\in C_{b}\left( \mathbb{R};C_{b}^{3}\left( \mathbb{R%
}^{n},\mathbb{R}\right) \right) $, $\left\{ B_{j}\right\}
_{j=1}^{n}\subseteq \mathcal{X}^{\mathbb{R}}$ and 
\begin{equation}
h\left( t;\rho \right) \doteq g\left( t;\rho \left( B_{1}\right) ,\ldots
,\rho \left( B_{n}\right) \right) \ ,\qquad t\in \mathbb{R},\ \rho \in E\ .
\label{h assumtions}
\end{equation}%
Then, for any $s,t\in \mathbb{R}$ and $A\in \mathcal{X}$, 
\begin{equation}
(\varpi _{\rho ,s}\left( t\right) (A))_{\rho \in E}\equiv (\varpi _{\rho
,s}\left( t,A\right) )_{\rho \in E}\in C^{1}\left( E;\mathbb{C}\right)
\label{eqsdfkljsdklfj}
\end{equation}%
and, for any $\upsilon \in E$,%
\begin{equation}
\left[ \mathrm{d}\varpi _{\rho ,s}\left( t,A\right) \right] \left( \upsilon
\right) =\upsilon \left( \mathrm{D}\varpi _{\rho ,s}\left( t,A\right)
\right) =\left( \upsilon -\rho \right) \circ T_{t,s}^{\varpi _{\rho
,s}}\left( A\right) +\mathfrak{\maltese }_{A}\left[ \mathrm{d}\varpi _{\rho
,s}\left( \cdot ,\cdot \right) \left( \upsilon \right) \right]  \notag
\end{equation}%
where, for any continuous function $\xi :\mathbb{R}\times \mathcal{X}%
\rightarrow \mathbb{C}$, 
\begin{eqnarray}
\mathfrak{\maltese }_{A}\left[ \xi \right] &\doteq
&\sum_{j,k=1}^{n}\int_{s}^{t}\mathrm{d}\alpha \ \xi \left( \alpha
,B_{k}\right) \rho \circ T_{\alpha ,s}^{\varpi _{\rho ,s}}\left( i\left[
B_{j},T_{t,\alpha }^{\varpi _{\rho ,s}}\left( A\right) \right] \right)
\label{def dysonsbaby} \\
&&\times \partial _{x_{k}}\partial _{x_{j}}g\left( \alpha ;\varpi _{\rho
,s}\left( \alpha ,B_{1}\right) ,\ldots ,\varpi _{\rho ,s}\left( \alpha
,B_{n}\right) \right) \ .  \notag
\end{eqnarray}%
Moreover, for all $s\in \mathbb{R}$ and $\rho \in E$, the map $(t,A)\mapsto 
\mathrm{D}\varpi _{\rho ,s}\left( t,A\right) $ from $\mathbb{R}\times 
\mathcal{X}$ to $\mathcal{X}$ is continuous.
\end{lemma}

\begin{proof}
Fix all parameters of the lemma. Observe first that Taylor's theorem applied
to $\partial _{x_{j}}g\left( t\right) $ for each $t\in \mathbb{R}$ and $j\in
\{1,\ldots ,n\}$ yields that, for all $x,y\in \mathbb{R}^{n}$,%
\begin{equation}
\partial _{x_{j}}g\left( t;y\right) =\partial _{x_{j}}g\left( t;x\right)
+\sum_{k=1}^{n}\left( y_{k}-x_{k}\right) \left( \partial _{x_{k}}\partial
_{x_{j}}g\left( t;x_{1},\ldots ,x_{n}\right) +r_{k}\left( t,x,y\right)
\right)  \label{taylor1}
\end{equation}%
where, for any $k\in \{1,\ldots ,n\}$, $r_{k}\left( \cdot ,\cdot ,\cdot
\right) $ is a continuous real-valued function on $\mathbb{R}\times \mathbb{R%
}^{n}\times \mathbb{R}^{n}$ such that 
\begin{equation}
\lim_{y\rightarrow x}r_{k}\left( t,x,y\right) =0\ ,  \label{taylor2}
\end{equation}%
uniformly for $t$ and $x$ in a compact set. Note additionally that the
function $h$, as defined by (\ref{h assumtions}), satisfies Condition \ref%
{Conditions (a)-(b)}.

For any $s,t\in \mathbb{R}$, $\rho ,\upsilon \in E$, $\lambda \in (0,1]$ and 
$A\in \mathcal{X}$, we infer from Lemma \ref{Solution selfbaby} that%
\begin{eqnarray*}
\beth \left( \lambda ,t,A;\upsilon \right) &\doteq &\lambda ^{-1}\left(
\varpi _{\left( 1-\lambda \right) \rho +\lambda \upsilon ,s}\left(
t,A\right) -\varpi _{\rho ,s}\left( t,A\right) \right) \\
&=&\left( \upsilon -\rho \right) \circ T_{t,s}^{\varpi _{\rho ,s}}\left(
A\right) +\lambda ^{-1}\left( \left( 1-\lambda \right) \rho +\lambda
\upsilon \right) \circ \left( T_{t,s}^{\varpi _{\left( 1-\lambda \right)
\rho +\lambda \upsilon ,s}}-T_{t,s}^{\varpi _{\rho ,s}}\right) \left(
A\right) \ .
\end{eqnarray*}%
Through Equations (\ref{Ynbis}), (\ref{flow baby0}), (\ref{T0}), (\ref{h
assumtions}) and (\ref{taylor1}), we deduce that 
\begin{align}
\beth \left( \lambda ,t,A;\upsilon \right) & =\left( \upsilon -\rho \right)
\circ T_{t,s}^{\varpi _{\rho ,s}}\left( A\right)
+\sum_{j,k=1}^{n}\int_{s}^{t}\mathrm{d}\alpha \ \beth \left( \lambda ,\alpha
,B_{k};\upsilon \right)  \label{truc a la con chat} \\
& \qquad \qquad \qquad \times \left( \left( 1-\lambda \right) \rho +\lambda
\upsilon \right) \circ T_{\alpha ,s}^{\varpi _{\left( 1-\lambda \right) \rho
+\lambda \upsilon ,s}}\left( i\left[ B_{j},T_{t,\alpha }^{\varpi _{\rho
,s}}\left( A\right) \right] \right)  \notag \\
& \qquad \qquad \qquad \qquad \times \left( \partial _{x_{k}}\partial
_{x_{j}}g\left( \alpha ;\mathbf{x}\left( 0,\alpha \right) \right)
+r_{k}\left( \alpha ;\mathbf{x}\left( 0,\alpha \right) ,\mathbf{x}\left(
\lambda ,\alpha \right) \right) \right) \ ,  \notag
\end{align}%
where%
\begin{equation*}
\mathbf{x}\left( \lambda ,\alpha \right) \doteq \left( \varpi _{\left(
1-\lambda \right) \rho +\lambda \upsilon ,s}\left( \alpha ,B_{1}\right)
,\ldots ,\varpi _{\left( 1-\lambda \right) \rho +\lambda \upsilon ,s}\left(
\alpha ,B_{n}\right) \right) \in \mathbb{R}^{n}\ .
\end{equation*}%
From Equation (\ref{truc a la con chat}), one sees that $\beth \left(
\lambda ,t,A;\upsilon \right) $ is given by a Dyson-type series which is
absolutely summable, uniformly with respect to $\lambda \in (0,1]$, because $%
(T_{t,s}^{\xi })_{s,t\in \mathbb{R}}$ is a family of $\ast $-automorphisms
of $\mathcal{X}$ for any $\xi \in C\left( \mathbb{R};E\right) $. By Lemmata %
\ref{Solution selfbaby copy(1)} and \ref{lemma wellbaby} together with
Condition \ref{Conditions (a)-(b)} (b), 
\begin{equation*}
\lim_{\lambda \rightarrow 0^{+}}\left( \left( 1-\lambda \right) \rho
+\lambda \upsilon \right) \circ T_{\alpha ,s}^{\varpi _{\left( 1-\lambda
\right) \rho +\lambda \upsilon ,s}}\left( i\left[ B_{j},T_{t,\alpha
}^{\varpi _{\rho ,s}}\left( A\right) \right] \right) =\rho \circ T_{\alpha
,s}^{\varpi _{\rho ,s}}\left( i\left[ B_{j},T_{t,\alpha }^{\varpi _{\rho
,s}}\left( A\right) \right] \right)
\end{equation*}%
while 
\begin{equation*}
\lim_{\lambda \rightarrow 0^{+}}r_{k}\left( \alpha ,\mathbf{x}\left(
0,\alpha \right) ,\mathbf{x}\left( \lambda ,\alpha \right) \right) =0\ ,
\end{equation*}%
using (\ref{taylor2}). (Both limits are uniform for $\alpha $ in a compact
set.) Hence, we deduce from Lebesgue's dominated convergence theorem that 
\begin{equation}
\beth \left( 0,t,A;\upsilon \right) \doteq \lim_{\lambda \rightarrow
0^{+}}\beth \left( \lambda ,t,A;\upsilon \right) =\lim_{\lambda \rightarrow
0^{+}}\lambda ^{-1}\left( \varpi _{\left( 1-\lambda \right) \rho +\lambda
\upsilon ,s}\left( t,A\right) -\varpi _{\rho ,s}\left( t,A\right) \right)
\label{lebesgue chat}
\end{equation}%
exists for all $s,t\in \mathbb{R}$, $\rho ,\upsilon \in E$ and $A\in 
\mathcal{X}$, as given by a Dyson-type series. In particular, for any $%
\upsilon \in E$, the complex-valued function $(t,A)\mapsto \beth \left(
0,t,A,\upsilon \right) $ on $\mathbb{R}\times \mathcal{X}$ is the unique
solution in $\xi \in C\left( \mathbb{R}\times \mathcal{X};\mathbb{C}\right) $
to the equation 
\begin{equation}
\xi \left( t,A\right) =\left( \upsilon -\rho \right) \circ T_{t,s}^{\varpi
_{\rho ,s}}\left( A\right) +\mathfrak{\maltese }_{A}\left[ \xi \right]
\label{integral equation 0}
\end{equation}%
with $\mathfrak{\maltese }_{A}$ defined by (\ref{def dysonsbaby}). Compare
with (\ref{truc a la con chat}) taken at $\lambda =0$. Note that the
integral equation 
\begin{align}
\mathfrak{D}\left( t,A\right) & =T_{t,s}^{\varpi _{\rho ,s}}\left( A\right)
-\rho \circ T_{t,s}^{\varpi _{\rho ,s}}\left( A\right) \mathfrak{1}%
+\sum_{j,k=1}^{n}\int_{s}^{t}\mathrm{d}\alpha \ \mathfrak{D}\left( \alpha
,B_{k}\right)  \label{chat0} \\
& \qquad \qquad \qquad \times \rho \circ T_{\alpha ,s}^{\varpi _{\rho
,s}}\left( i\left[ B_{j},T_{t,\alpha }^{\varpi _{\rho ,s}}\left( A\right) %
\right] \right) \partial _{x_{k}}\partial _{x_{j}}g\left( \alpha ;\mathbf{x}%
\left( 0,\alpha \right) \right)  \notag
\end{align}%
uniquely determines, by absolutely summable (in $\mathcal{X}$) Dyson-type
series, a continuous map $(t,A)\mapsto \mathfrak{D}\left( t,A\right) $ from $%
\mathbb{R}\times \mathcal{X}$ to $\mathcal{X}$, which, by (\ref{integral
equation 0}), satisfies 
\begin{equation}
\upsilon \left( \mathfrak{D}\left( t,A\right) \right) =\beth \left(
0,t,A;\upsilon \right) \doteq \lim_{\lambda \rightarrow 0^{+}}\lambda
^{-1}\left( \varpi _{\left( 1-\lambda \right) \rho +\lambda \upsilon
,s}\left( t,A\right) -\varpi _{\rho ,s}\left( t,A\right) \right)
\label{chat1}
\end{equation}%
for all $s,t\in \mathbb{R}$, $\rho ,\upsilon \in E$ and $A\in \mathcal{X}$.
By Definition \ref{convex Frechet derivative}, the assertion follows.
\end{proof}

\begin{lemma}[Differentiability of the solution -- III]
\label{Differentiability copy(1)}\mbox{ }\newline
Under the assumptions of Lemma \ref{lemma contnuity copy(1)}, for any $t\in 
\mathbb{R}$, $\rho \in E$ and $A\in \mathcal{X}$, 
\begin{equation*}
(\varpi _{\rho ,s}\left( t\right) (A))_{s\in \mathbb{R}}\equiv (\varpi
_{\rho ,s}\left( t,A\right) )_{s\in \mathbb{R}}\in C^{1}\left( \mathbb{R};%
\mathbb{C}\right)
\end{equation*}%
with derivative given, for any $A\in \mathcal{X}$, by%
\begin{equation}
\partial _{s}\varpi _{\rho ,s}\left( t,A\right) =-\rho \circ X_{s}^{\rho
}\circ T_{t,s}^{\varpi _{\rho ,s}}\left( A\right) +\mathfrak{\maltese }_{A}%
\left[ \partial _{s}\varpi _{\rho ,s}\right] \ .  \label{equation integral}
\end{equation}%
Here, $\mathfrak{\maltese }_{A}$ is defined by (\ref{def dysonsbaby}) and $%
(t,A)\mapsto \partial _{s}\varpi _{\rho ,s}(t,A)$ is a continuous function
on $\mathbb{R}\times \mathcal{X}$.
\end{lemma}

\begin{proof}
By Lemma \ref{Solution selfbaby}, for any $\rho \in E$, $s,t\in \mathbb{R}$, 
$A\in \mathcal{X}$ and $\varepsilon \in \mathbb{R}\backslash \{0\}$, 
\begin{eqnarray*}
\tilde{\beth}\left( \varepsilon ,t,A\right) &\doteq &\varepsilon ^{-1}\left(
\varpi _{\rho ,s+\varepsilon }\left( t,A\right) -\varpi _{\rho ,s}\left(
t,A\right) \right) \\
&=&\varepsilon ^{-1}\rho \circ (T_{t,s+\varepsilon }^{\varpi _{\rho
,s}}-T_{t,s}^{\varpi _{\rho ,s}})\left( A\right) +\varepsilon ^{-1}\rho
\circ (T_{t,s+\varepsilon }^{\varpi _{\rho ,s+\varepsilon
}}-T_{t,s+\varepsilon }^{\varpi _{\rho ,s}})\left( A\right) \ .
\end{eqnarray*}%
Similar to (\ref{truc a la con chat}), via Equations (\ref{Ynbis}), (\ref%
{flow baby0}), (\ref{T0}), (\ref{h assumtions}) and (\ref{taylor1}) we
deduce that%
\begin{align}
\tilde{\beth}\left( \varepsilon ,t,A\right) & =\varepsilon ^{-1}\rho \circ
(T_{t,s+\varepsilon }^{\varpi _{\rho ,s}}-T_{t,s}^{\varpi _{\rho ,s}})\left(
A\right) +\sum_{j,k=1}^{n}\int_{s+\varepsilon }^{t}\mathrm{d}\alpha \ \tilde{%
\beth}\left( \varepsilon ,\alpha ,B_{k}\right)  \label{sdkljsdlkjf} \\
& \times \rho \circ T_{\alpha ,s}^{\varpi _{\rho ,s+\varepsilon }}\left( i 
\left[ B_{j},T_{t,\alpha }^{\varpi _{\rho ,s}}\left( A\right) \right]
\right) \left( \partial _{x_{k}}\partial _{x_{j}}g\left( \alpha ;\mathbf{y}%
\left( 0,\alpha \right) \right) +r_{k}\left( \alpha ;\mathbf{y}\left(
0,\alpha \right) ,\mathbf{y}\left( \varepsilon ,\alpha \right) \right)
\right)  \notag
\end{align}%
with 
\begin{equation*}
\mathbf{y}\left( \varepsilon ,\alpha \right) \doteq \left( \varpi _{\rho
,s+\varepsilon }\left( \alpha ,B_{1}\right) ,\ldots ,\varpi _{\rho
,s+\varepsilon }\left( \alpha ,B_{n}\right) \right) \in \mathbb{R}^{n}\ .
\end{equation*}%
Again, one sees from Equation (\ref{sdkljsdlkjf}) that $\tilde{\beth}\left(
\varepsilon ,t,A\right) $ is given by a Dyson-type series which is
absolutely summable, uniformly with respect to $\varepsilon $ in a bounded
set. Recall that $(T_{t,s}^{\xi })_{s,t\in \mathbb{R}}$ is a norm-continuous
two-parameter family of $\ast $-automorphisms of $\mathcal{X}$ satisfying in 
$\mathcal{B}(\mathcal{X})$ the non-autonomous evolution equation (\ref{flow
baby1bis}) for any fixed $\xi \in C\left( \mathbb{R};E\right) $. Therefore,
similar to (\ref{lebesgue chat}), by Equation (\ref{taylor2}), Lemmata \ref%
{Solution selfbaby copy(1)} and \ref{lemma wellbabybaby} together with
Condition \ref{Conditions (a)-(b)} (b) and Lebesgue's dominated convergence
theorem, we deduce that%
\begin{equation*}
\partial _{s}\varpi _{\rho ,s}\left( t,A\right) \doteq \lim_{\varepsilon
\rightarrow 0}\tilde{\beth}\left( \varepsilon ,t,A\right) =\lim_{\varepsilon
\rightarrow 0}\varepsilon ^{-1}\left( \varpi _{\rho ,s+\varepsilon }\left(
t,A\right) -\varpi _{\rho ,s}\left( t,A\right) \right)
\end{equation*}%
exists for all $s,t\in \mathbb{R}$, $\rho \in E$ and $A\in \mathcal{X}$, as
given by a Dyson-type series. In particular, the complex-valued function $%
(t,A)\mapsto \partial _{s}\varpi _{\rho ,s}\left( t,A\right) $ on $\mathbb{R}%
\times \mathcal{X}$ is the unique solution in $\xi \in C\left( \mathbb{R}%
\times \mathcal{X};\mathbb{C}\right) $ to the equation 
\begin{equation}
\xi \left( t,A\right) =-\rho \circ X_{s}^{\rho }\circ T_{t,s}^{\varpi _{\rho
,s}}\left( A\right) +\mathfrak{\maltese }_{A}\left[ \xi \right]
\label{toto final}
\end{equation}%
with $\mathfrak{\maltese }_{A}$ defined by (\ref{def dysonsbaby}). Compare
with (\ref{sdkljsdlkjf}) taken at $\varepsilon =0$.
\end{proof}

We conclude this section with the derivation of Liouville's equation for the
time-evolution of (elementary) continuous and affine functions defined by (%
\ref{fA}), from which Theorem \ref{proposition dynamic classique II} is
deduced.

\begin{corollary}[Liouville's equation for affine functions]
\label{Corollary bije+cocylbaby copy(1)}\mbox{ }\newline
Under the assumptions of Lemma \ref{lemma contnuity copy(1)}, 
\begin{equation*}
\partial _{s}V_{t,s}^{h}(\hat{A})=-\{h\left( s\right) ,V_{t,s}^{h}(\hat{A}%
)\}\ ,\qquad s,t\in \mathbb{R},\ A\in \mathcal{X},
\end{equation*}%
with $\hat{A}\in \mathfrak{C}$ being the elementary continuous and affine
function defined by (\ref{fA}). In particular, both side of the equation are
well-defined functions in $\mathfrak{C}$.
\end{corollary}

\begin{proof}
Fix $s\in \mathbb{R}$ and $\rho \in E$. By (\ref{fA}) and (\ref{classical
evolution familybaby}), note that 
\begin{equation*}
V_{t,s}^{h}(\hat{A})=\varpi _{\rho ,s}\left( t\right) \left( A\right) \equiv
\varpi _{\rho ,s}\left( t,A\right) \ ,\qquad t\in \mathbb{R},\ A\in \mathcal{%
X}\ .
\end{equation*}%
By Lemma \ref{lemma contnuity copy(1)}, the map 
\begin{equation*}
(t,A)\mapsto \mathrm{D}V_{t,s}^{h}(\hat{A})\left( \rho \right) =\mathrm{D}%
\varpi _{\rho ,s}\left( t;A\right) =\mathfrak{D}\left( t,A\right)
\end{equation*}%
from $\mathbb{R}\times \mathcal{X}$ to $\mathcal{X}$ is continuous. See also
Definition \ref{convex Frechet derivative} and Equation (\ref{clear2}).
Therefore, the map 
\begin{equation*}
\left( t,A\right) \mapsto -\{h\left( s\right) ,V_{t,s}^{h}(\hat{A})\}\left(
\rho \right) \doteq -\rho \left( i\left[ \mathrm{D}h\left( s;\rho \right) ,%
\mathrm{D}V_{t,s}^{h}(\hat{A})\left( \rho \right) \right] \right)
\end{equation*}%
from $\mathbb{R}\times \mathcal{X}$ to $\mathbb{C}$ is a well-defined
continuous function. See Definition \ref{convex Frechet derivative copy(1)}
and (\ref{extension}). By (\ref{chat0}), it solves Equation (\ref{toto final}%
), like the well-defined continuous map 
\begin{equation*}
(t,A)\mapsto \partial _{s}V_{t,s}^{h}(\hat{A})\left( \rho \right) =\partial
_{s}\varpi _{\rho ,s}\left( t\right) \left( A\right) \equiv \partial
_{s}\varpi _{\rho ,s}\left( t,A\right)
\end{equation*}%
from $\mathbb{R}\times \mathcal{X}$ to $\mathbb{C}$ (Lemma \ref%
{Differentiability copy(1)}). By uniqueness of the solution to (\ref{toto
final}), the assertion follows.
\end{proof}

\section{Appendix: Liminal, Postliminal and Antiliminal $C^{\ast }$-Algebras 
\label{Liminal appendix}}

\noindent \textit{La m\^{e}me structure qui, si vous montez, comporte une
distance, si vous descendez, n'en comporte pas.\footnote{%
Engl.: \textit{The same structure which, if you go up, contains a distance,
if you go down, does not contain it.} See \cite[p. 38]{Libera}. This
citation refers to the highly political and theological issues of
hierarchies in Christianity, as discussed in the Late Middle Ages. Indeed,
for some theologians of the XIIIe-XIVe centuries like Giles of Rome, the
increasing hierarchy refers to the existence of an order, implying in
particular a distance (cf. \textit{\textquotedblleft potentia dei
ordinaria\textquotedblright }). From the top down, the relation can be
direct, immediate, without distance (cf. \textquotedblleft \textit{potentia
dei absoluta}\textquotedblright ). In the mathematical context, from the
bottom up, we have in mind the ordering of measures having same barycenter $%
\rho $ in a compact convex space $K$ to arrive, by \textquotedblleft
removing\textquotedblright\ the mass farther away from $\rho $, at a
(maximal) measure only supported by extreme points, as proved and stated in
the Choquet(-Bishop-de Leeuw) Theorem. See, e.g., (\ref{affine decomposition}%
). From the top down, we have in mind the disconcerting property that, for
some $K$, extreme points are meanwhile dense, i.e., any $x\in K$ is
arbitrarily close to an extreme point. }}\smallskip

\hfill A. de Libera, 2015\bigskip

As explained in \cite[p. 99]{Dixmier}, the notion of \emph{liminal} $C^{\ast
}$-algebras was first introduced in 1951 by Kaplansky under the name of 
\emph{CCR}-algebras. Remark that, in this context, \emph{CCR} does not mean
\textquotedblleft Canonical Commutation Relations\textquotedblright\ but
\textquotedblleft Completely Continuous Representations\textquotedblright ,
\textquotedblleft completely continuous\textquotedblright\ standing for
\textquotedblleft compact\textquotedblright . \emph{CCR} usually means
nowadays \textquotedblleft Canonical Commutation
Relations\textquotedblright\ and thus, like Dixmier in his textbook \cite%
{Dixmier} on $C^{\ast }$-algebras, we rather prefer the terminology
\textquotedblleft liminal\textquotedblright . This concept is strongly
related to the $C^{\ast }$-algebra $\mathcal{K}(\mathcal{H})$ of compact
operators acting on a Hilbert space $\mathcal{H}$ via the concept of $%
C^{\ast }$-algebra representations. See also \cite{Blackadar} for a recent
compendium on operator algebras.

Recall that a \emph{representation} on the Hilbert space $\mathcal{H}$ of a $%
C^{\ast }$-algebra $\mathcal{X}$ is, by definition \cite[Definition 2.3.2]%
{BrattelliRobinsonI}, a $\ast $-homomorphism $\mathbf{\pi }$ from $\mathcal{X%
}$ to the unital $C^{\ast }$-algebra $\mathcal{B}(\mathcal{H})$ of all
bounded linear operators acting on $\mathcal{H}$. Injective representations
are called \emph{faithful}. The representation of a $C^{\ast }$-algebra $%
\mathcal{X}$ is not unique:\ For any representation $\mathbf{\pi }:\mathcal{X%
}\rightarrow \mathcal{B}(\mathcal{H})$, we can construct another one by
doubling the Hilbert space $\mathcal{H}$ and the map $\mathbf{\pi }$, via a
direct sum $\mathcal{H}_{1}\oplus \mathcal{H}_{2}$ of two copies $\mathcal{H}%
_{1},\mathcal{H}_{2}$ of $\mathcal{H}$. Thus, recall also the notion of
\textquotedblleft minimal\textquotedblright\ representations of $C^{\ast }$%
-algebras: If $\mathbf{\pi }:\mathcal{X}\rightarrow \mathcal{B}(\mathcal{H})$
is a representation of a $C^{\ast }$-algebra $\mathcal{X}$ on the Hilbert
space $\mathcal{H}$, we say that it is \emph{irreducible}, whenever $\{0\}$
and $\mathcal{H}$ are the only closed subspaces of $\mathcal{H}$ which are
invariant with respect to any operator of $\mathbf{\pi }(\mathcal{X}%
)\subseteq \mathcal{B}(\mathcal{H})$.

Every $C^{\ast }$-algebra which is isomorphic to the $C^{\ast }$-algebra $%
\mathcal{K}(\mathcal{H})$ of all compact operators acting on some Hilbert
space $\mathcal{H}$ is said to be \emph{elementary}. The concept of \emph{%
liminal} $C^{\ast }$-algebras generalizes this notion (see \cite[Definition
4.2.1]{Dixmier} or \cite[Section IV.1.3.1]{Blackadar}):

\begin{definition}[Liminal $C^{\ast }$-algebras]
\label{liminal}\mbox{ }\newline
A $C^{\ast }$-algebra $\mathcal{X}$ is called liminal if, for every
irreducible representation $\pi $ of $\mathcal{X}$\ and each $A\in \mathcal{X%
}$, $\pi (A)$ is compact.
\end{definition}

\noindent All finite-dimensional $C^{\ast }$-algebras are of course liminal.
All commutative $C^{\ast }$-algebras are also liminal. See \cite[4.2.1-4.2.2]%
{Dixmier} or \cite[Examples IV.1.3.3]{Blackadar}. Note that the set of
elements of a $C^{\ast }$-algebra $\mathcal{X}$ whose images under any
irreducible representation are compact operators is the largest liminal
closed two-sided ideal of $\mathcal{X}$, by \cite[Proposition 4.2.6]{Dixmier}%
.

Later, Kaplansky and Glimm also introduced the term \emph{GCR}\footnote{%
The definition of \emph{GCR} given in \cite[Section IV.1.3.1]{Blackadar} is
different from the original one.} for a generalization of Definition \ref%
{liminal}, much later replaced by \emph{postliminal} (see \cite[Section 4.3.1%
]{Dixmier} or \cite[Section IV.1.3.1]{Blackadar}). On the one hand, observe
that the $C^{\ast }$-algebra $\mathcal{K}(\mathcal{H})$ of compact operators
acting on a Hilbert space $\mathcal{H}$ is a closed two-sided ideal of the $%
C^{\ast }$-algebra $\mathcal{B}(\mathcal{H})$ of bounded operators acting on 
$\mathcal{H}$. On the other hand, from a closed, self-adjoint two-sided
ideal $\mathcal{I}$ of a $C^{\ast }$-algebra $\mathcal{X}$ and the quotient $%
\mathcal{X}/\mathcal{I}$, we can construct a $C^{\ast }$-algebra. Keeping
this information in mind, the notion of postliminal $C^{\ast }$-algebras are
defined as follows \cite[Section 4.3.1]{Dixmier}:

\begin{definition}[Postliminal $C^{\ast }$-algebras]
\label{liminal copy(1)}\mbox{ }\newline
A $C^{\ast }$-algebra $\mathcal{X}$ is postliminal if every non-zero
quotient $C^{\ast }$-algebra of $\mathcal{X}$ possesses a non-zero liminal
closed two-sided ideal.
\end{definition}

\noindent All liminal $C^{\ast }$-algebras are postliminal, by \cite[%
Proposition 4.2.4]{Dixmier}, but the converse is false.

Kaplansky and Glimm named important $C^{\ast }$-algebras that are not \emph{%
GCR}\ (postliminal), \emph{NGCR} $C^{\ast }$-algebras. Such algebras were
later called \emph{antiliminal} \cite[Section 4.3.1]{Dixmier} (see also \cite%
[Section IV.1.3.1]{Blackadar}):

\begin{definition}[Antiliminal $C^{\ast }$-algebras]
\label{liminal copy(2)}\mbox{ }\newline
A $C^{\ast }$-algebra $\mathcal{X}$ is antiliminal if the zero ideal is its
only liminal closed two-sided ideal.
\end{definition}

\noindent Remark that a quotient $C^{\ast }$-algebra of an antiliminal $%
C^{\ast }$-algebra is not antiliminal, in general.

If the image of $\mathcal{X}$ by an \emph{irreducible} representation $\pi $
would intersect the set of compact operators, then the set of compact
operators would automatically be included in $\pi \left( \mathcal{X}\right) $%
, by \cite[Corollary 4.1.10]{Dixmier}. In other words, the image of a $%
C^{\ast }$-algebra by an irreducible representation either contains the set
of compact operators or does not intersect it. Antiliminal and separable $%
C^{\ast }$-algebras are related to the second situation \cite[Theorem 1 (b)]%
{Glimm61}:

\begin{theorem}[Glimm]
\label{Glimm thm}\mbox{ }\newline
Let $\mathcal{X}$ be a separable\footnote{%
In the non-separable situation, any of the assertions (ii)-(iv) yields (i),
by \cite[Theorem 1 (c)]{Glimm61}.} $C^{\ast }$-algebra. Then, the following
conditions are equivalent: \newline
\emph{(i)} $\mathcal{X}$ is antiliminal.\newline
\emph{(ii)} $\mathcal{X}$ has a faithful type II representation. \newline
\emph{(iii)} $\mathcal{X}$ has a faithful type III representation. \newline
\emph{(iv)} $\mathcal{X}$ has a faithful representation which is a direct
sum of a family of representations of $\mathcal{X}$ whose range does not
contain the compact operators.
\end{theorem}

\begin{proof}
By \cite[Proposition 1.8.5]{Dixmier}, an antiliminal $C^{\ast }$-algebra
does not possess any postliminal closed two-sided ideal, apart from the zero
ideal. Therefore, the theorem is a direct consequence of \cite[Theorem 1 (b)]%
{Glimm61}, keeping in mind that \textquotedblleft completely
continuous\textquotedblright\ in \cite{Glimm61}\ is a synonym of
\textquotedblleft compact\textquotedblright .
\end{proof}

\noindent In other words, no irreducible representation of an antiliminal
separable $C^{\ast }$-algebra $\mathcal{X}$ contains the compact operators
on the representation space. Additionally, antiliminal $C^{\ast }$-algebras
are directly related with von Neumann algebras of type II and III, while
postliminal $C^{\ast }$-algebras are directly associated with von Neumann
algebras of type I, by \cite[Theorem 1 (a), (c)]{Glimm61}.

Antiliminal (unital) $C^{\ast }$-algebras $\mathcal{X}$ have a set $E$ of
states with fairly complicated geometrical structure, similar to the Poulsen
simplex \cite{Poulsen}. Recall that $E$ is a weak$^{\ast }$-compact convex
subset of $\mathcal{X}^{\ast }$ with (nonempty) set of extreme points
denoted by $\mathcal{E}(E)$. See (\ref{closure of the convex hull}). Then,
one has the following result \cite[Lemma 11.2.4]{Dixmier}:

\begin{lemma}[Weak$^{\ast }$ density of the set of extremes states]
\label{liminal copy(3)}\mbox{ }\newline
Let $\mathcal{X}$ be a antiliminal unital $C^{\ast }$-algebra. Assume that
any two (different) non-zero closed two-sided ideals of $\mathcal{X}$ always
have a non-zero intersection. Then $E=\overline{\mathcal{E}(E)}$, in the weak%
$^{\ast }$ topology.
\end{lemma}

$C^{\ast }$-algebras $\mathcal{X}$ relevant for mathematical physics often
have a faithful type III representation. See, e.g., \cite[Section 4]{Powers}
for UHF (uniformly hyperfinite) algebras. Note that every $\ast $%
-representation of a UHF algebra, like for instance a CAR algebra, is
faithful. For key statements on representations of CAR $C^{\ast }$-algebras,
see, e.g., \cite[Theorem 2.4]{Murakami} or \cite[Theorem 12.3.8]{Kadison}
and references therein. Hence, by Theorem \ref{Glimm thm}, many $C^{\ast }$%
-algebras $\mathcal{X}$ with physical applications are antiliminal. Note
also that they are generally separable and \emph{simple}:

\begin{definition}[Simple $C^{\ast }$-algebras]
\label{simple def}\mbox{ }\newline
A $C^{\ast }$-algebra $\mathcal{X}$ is simple if the only closed two-sided
ideals of $\mathcal{X}$ are the trivial sets $\{0\}$ and $\mathcal{X}$.
\end{definition}

\noindent $C^{\ast }$-algebras $\mathcal{B}(\mathcal{H})$ of all (bounded)\
linear operators acting on some finite-dimensional Hilbert space $\mathcal{H}
$ are of course simple. See, e.g., \cite[Corollary 4.1.7]{Dixmier}. However,
finite-dimensional $C^{\ast }$-algebras are not generally simple, but
semisimple only, as direct sums of simple algebras.

In mathematical physics, unital $C^{\ast }$-algebras of infinitely extended
(quantum) systems are usually built from a family of local
finite-dimensional $C^{\ast }$-subalgebras. It refers to approximately
finite-dimensional (AF) $C^{\ast }$-algebras, originally introduced in 1972
by Bratteli \cite{Bratteliquasi}. See also the quasi-local algebras \cite[%
Definition 2.6.3]{BrattelliRobinsonI}. AF $C^{\ast }$-algebras used in
physics are usually simple, by \cite[Corollary 2.6.19]{BrattelliRobinsonI},
because they are generally constructed from simple local algebras (typically
as some inductive limit, with respect to boxes $\Lambda $, of a family of $%
C^{\ast }$-algebras $\mathcal{B}(\mathcal{H}_{\Lambda })$, with $\mathrm{dim}%
\mathcal{H}_{\Lambda }\mathcal{<\infty }$, like, for instance, Cuntz,
lattice CAR or quantum-spin $C^{\ast }$-algebras). Therefore, by Lemma \ref%
{liminal copy(3)}, for infinitely extended quantum systems, like fermions on
the lattice or quantum-spin systems, the corresponding set $E$ of states has
a \emph{dense} subset of extreme points. This fact is well-known and already
discussed in \cite[p. 226]{Bratteliquasi}. See also \cite[Example 4.1.31]%
{BrattelliRobinsonI} for a direct proof in the context of the so-called UHF
(uniformly hyperfinite) $C^{\ast }$-algebras \cite[Examples 2.6.12]%
{BrattelliRobinsonI}.\bigskip

\noindent \textit{Acknowledgments:} This work is supported by CNPq
(308337/2017-4), FAPESP (2017/22340-9), as well as by the Basque Government
through the grant IT641-13 and the BERC 2018-2021 program, and by the
Spanish Ministry of Science, Innovation and Universities: BCAM Severo Ochoa
accreditation SEV-2017-0718, MTM2017-82160-C2-2-P. We thank S. Rodrigues for
pointing out the reference \cite{extra-ref}. Finally, special thanks to S.
Breteaux for having pointed out many typos and who meanwhile suggested
various improvements during his very detailed reading of the first version
of this paper.

\noindent \textbf{Jean-Bernard Bru} \newline
Departamento de Matem\'{a}ticas\newline
Facultad de Ciencia y Tecnolog\'{\i}a\newline
Universidad del Pa\'{\i}s Vasco\newline
Apartado 644, 48080 Bilbao \medskip \newline
BCAM - Basque Center for Applied Mathematics\newline
Mazarredo, 14. \newline
48009 Bilbao\medskip \newline
IKERBASQUE, Basque Foundation for Science\newline
48011, Bilbao\bigskip

\noindent \textbf{Walter de Siqueira Pedra} \newline
Departamento de F\'{\i}sica Matem\'{a}tica\newline
Instituto de F\'{\i}sica,\newline
Universidade de S\~{a}o Paulo\newline
Rua do Mat\~{a}o 1371\newline
CEP 05508-090 S\~{a}o Paulo, SP Brasil

\end{document}